\documentclass{fac}

\usepackage{graphicx}
\usepackage{amssymb}
\usepackage{amsfonts}
\usepackage{amsmath}
\usepackage{dsfont}
\usepackage{MnSymbol}
\usepackage{mathtools}
\usepackage{epsfig}
\usepackage{epstopdf}
\usepackage{tikz}
\usepackage{amsthm}
\ifprodtf \else \usepackage{latexsym}\fi

\newcommand\black{\ensuremath{\blacktriangleright}}
\newcommand\white{\ensuremath{\vartriangleright}}

\newif\ifamsfontsloaded
\ifprodtf
  \newcommand\whbl{\white\kern-.1em--\kern-.1em\black}
  \newcommand\blwh{\black\kern-.1em--\kern-.1em\white}
  \newcommand\blbl{\black\kern-.1em--\kern-.1em\black}
  \newcommand\whwh{\white\kern-.1em--\kern-.1em\white}
  \amsfontsloadedtrue
\else
  \checkfont{msam10}
  \iffontfound
    \IfFileExists{amssymb.sty}
      {\usepackage{amssymb}\amsfontsloadedtrue
       \newcommand\whbl{\white\kern-.125em--\kern-.125em\black}%
       \newcommand\blwh{\black\kern-.125em--\kern-.125em\white}%
       \newcommand\blbl{\black\kern-.125em--\kern-.125em\black}%
       \newcommand\whwh{\white\kern-.125em--\kern-.125em\white}}
      {}
  \fi
\fi

\newtheorem{theorem}{Theorem}[section]
\newtheorem{definition}[theorem]{Definition}

\title[Draft of Truly Concurrent Process Algebra with Timing]
      {Truly Concurrent Process Algebra with Timing}

\author[Yong Wang]
    {Yong Wang\\
     College of Computer Science and Technology,\\
     Faculty of Information Technology,\\
     Beijing University of Technology, Beijing, China\\
     }

\correspond{Yong Wang, Pingleyuan 100, Chaoyang District, Beijing, China.
            e-mail: wangy@bjut.edu.cn}

\pubyear{2017}
\pagerange{\pageref{firstpage}--\pageref{lastpage}}

\begin{document}
\label{firstpage}

\makecorrespond

\maketitle

\begin{abstract}
We extend truly concurrent process algebra APTC with timing related properties. Just like ACP with timing, APTC with timing also has four parts: discrete relative timing, discrete absolute timing, continuous relative timing and continuous absolute timing.
\end{abstract}


\section{Introduction}{\label{int}}

In true concurrency, there are various structures, such as Petri net, event structure and domain et al \cite{ES1} \cite{ES2} \cite{CM}, to model true concurrency. There are also some kinds of bisimulations to capture the behavioral equivalence between these structures, including pomset bisimulation, step bisimulation, history-preserving (hp-) bisimulation, and the finest hereditary history-preserving (hhp-) bisimulation \cite{HHP1} \cite{HHP2}. Based on these truly concurrent semantics models, several logics to relate the logic syntaxes and the semantics models, such as the reversible logic \cite{RL1} \cite{RL2}, the truly concurrent logic SFL \cite{SFL} based on the interleaving mu-calculi \cite{MUC}, and a uniform logic \cite{LTC1} \cite{LTC2} to cover the above truly concurrent bisimulations. we also discussed the weakly truly concurrent bisimulations and their logics \cite{WTC}, that is, which are related to true concurrency with silent step $\tau$.

Process algebras CCS \cite{CCS} \cite{CC} and ACP \cite{ALNC} \cite{ACP} are based on the interleaving bisimulation. For the lack of process algebras based on truly concurrent bisimulations, we developed a calculus for true concurrency CTC \cite{CTC}, an axiomatization for true concurrency APTC \cite{ATC} and a calculus of truly concurrent mobile processes $\pi_{tc}$ \cite{PITC}, which are corresponding to CCS, ACP and $\pi$ based on interleaving bisimulation. There are correspondence between APTC and process algebra ACP \cite{ACP}, in this paper, we extend APTC with timing related properties. Just like ACP with timing \cite{T1} \cite{T2} \cite{T3}, APTC with timing also has four parts: discrete relative timing, discrete absolute timing, continuous relative timing and continuous absolute timing.

This paper is organized as follows. In section \ref{bg}, we introduce some preliminaries on APTC and timing. In section \ref{drt}, section \ref{dat}, section \ref{srt} and section \ref{sat}, we introduce APTC with discrete relative timing, APTC with discrete absolute timing, APTC with continuous relative timing and APTC with continuous absolute timing, respectively. We introduce recursion and abstraction in section \ref{rec} and section \ref{abs}. We take an example to illustrate the usage of APTC with timing in section \ref{app}. The extension mechanism is discussed in section \ref{ext}. Finally, in section \ref{con}, we conclude our work.

\section{Backgrounds}\label{bg}

For the convenience of the readers, we introduce some backgrounds about our previous work on truly concurrent process algebra \cite{ATC} \cite{CTC} \cite{PITC}, and also timing \cite{T1} \cite{T2} \cite{T3} in traditional process algebra ACP \cite{ACP}.

\subsection{Truly Concurrent Process Algebra}\label{tcpa}

In this subsection, we introduce the preliminaries on truly concurrent process algebra \cite{ATC} \cite{CTC} \cite{PITC}, which is based on truly concurrent operational semantics.

For this paper is an extension to APTC with timing, in the following, we introduce APTC briefly, for details, please refer to APTC \cite{ATC}.

APTC eliminates the differences of structures of transition system, event structure, etc, and discusses their behavioral equivalences. It considers that there are two kinds of causality relations: the chronological order modeled by the sequential composition and the causal order between different parallel branches modeled by the communication merge. It also considers that there exist two kinds of confliction relations: the structural confliction modeled by the alternative composition and the conflictions in different parallel branches which should be eliminated. Based on conservative extension, there are four modules in APTC: BATC (Basic Algebra for True Concurrency), APTC (Algebra for Parallelism in True Concurrency), recursion and abstraction.

\subsubsection{Basic Algebra for True Concurrency}

BATC has sequential composition $\cdot$ and alternative composition $+$ to capture the chronological ordered causality and the structural confliction. The constants are ranged over $A$, the set of atomic actions. The algebraic laws on $\cdot$ and $+$ are sound and complete modulo truly concurrent bisimulation equivalences (including pomset bisimulation, step bisimulation, hp-bisimulation and hhp-bisimulation).

\begin{definition}[Prime event structure with silent event]\label{PES}
Let $\Lambda$ be a fixed set of labels, ranged over $a,b,c,\cdots$ and $\tau$. A ($\Lambda$-labelled) prime event structure with silent event $\tau$ is a tuple $\mathcal{E}=\langle \mathbb{E}, \leq, \sharp, \lambda\rangle$, where $\mathbb{E}$ is a denumerable set of events, including the silent event $\tau$. Let $\hat{\mathbb{E}}=\mathbb{E}\backslash\{\tau\}$, exactly excluding $\tau$, it is obvious that $\hat{\tau^*}=\epsilon$, where $\epsilon$ is the empty event. Let $\lambda:\mathbb{E}\rightarrow\Lambda$ be a labelling function and let $\lambda(\tau)=\tau$. And $\leq$, $\sharp$ are binary relations on $\mathbb{E}$, called causality and conflict respectively, such that:

\begin{enumerate}
  \item $\leq$ is a partial order and $\lceil e \rceil = \{e'\in \mathbb{E}|e'\leq e\}$ is finite for all $e\in \mathbb{E}$. It is easy to see that $e\leq\tau^*\leq e'=e\leq\tau\leq\cdots\leq\tau\leq e'$, then $e\leq e'$.
  \item $\sharp$ is irreflexive, symmetric and hereditary with respect to $\leq$, that is, for all $e,e',e''\in \mathbb{E}$, if $e\sharp e'\leq e''$, then $e\sharp e''$.
\end{enumerate}

Then, the concepts of consistency and concurrency can be drawn from the above definition:

\begin{enumerate}
  \item $e,e'\in \mathbb{E}$ are consistent, denoted as $e\frown e'$, if $\neg(e\sharp e')$. A subset $X\subseteq \mathbb{E}$ is called consistent, if $e\frown e'$ for all $e,e'\in X$.
  \item $e,e'\in \mathbb{E}$ are concurrent, denoted as $e\parallel e'$, if $\neg(e\leq e')$, $\neg(e'\leq e)$, and $\neg(e\sharp e')$.
\end{enumerate}
\end{definition}

The prime event structure without considering silent event $\tau$ is the original one in \cite{ES1} \cite{ES2} \cite{CM}.

\begin{definition}[Configuration]
Let $\mathcal{E}$ be a PES. A (finite) configuration in $\mathcal{E}$ is a (finite) consistent subset of events $C\subseteq \mathcal{E}$, closed with respect to causality (i.e. $\lceil C\rceil=C$). The set of finite configurations of $\mathcal{E}$ is denoted by $\mathcal{C}(\mathcal{E})$. We let $\hat{C}=C\backslash\{\tau\}$.
\end{definition}

A consistent subset of $X\subseteq \mathbb{E}$ of events can be seen as a pomset. Given $X, Y\subseteq \mathbb{E}$, $\hat{X}\sim \hat{Y}$ if $\hat{X}$ and $\hat{Y}$ are isomorphic as pomsets. In the following of the paper, we say $C_1\sim C_2$, we mean $\hat{C_1}\sim\hat{C_2}$.

\begin{definition}[Pomset transitions and step]
Let $\mathcal{E}$ be a PES and let $C\in\mathcal{C}(\mathcal{E})$, and $\emptyset\neq X\subseteq \mathbb{E}$, if $C\cap X=\emptyset$ and $C'=C\cup X\in\mathcal{C}(\mathcal{E})$, then $C\xrightarrow{X} C'$ is called a pomset transition from $C$ to $C'$. When the events in $X$ are pairwise concurrent, we say that $C\xrightarrow{X}C'$ is a step.
\end{definition}

\begin{definition}[Pomset, step bisimulation]\label{PSB}
Let $\mathcal{E}_1$, $\mathcal{E}_2$ be PESs. A pomset bisimulation is a relation $R\subseteq\mathcal{C}(\mathcal{E}_1)\times\mathcal{C}(\mathcal{E}_2)$, such that if $(C_1,C_2)\in R$, and $C_1\xrightarrow{X_1}C_1'$ then $C_2\xrightarrow{X_2}C_2'$, with $X_1\subseteq \mathbb{E}_1$, $X_2\subseteq \mathbb{E}_2$, $X_1\sim X_2$ and $(C_1',C_2')\in R$, and vice-versa. We say that $\mathcal{E}_1$, $\mathcal{E}_2$ are pomset bisimilar, written $\mathcal{E}_1\sim_p\mathcal{E}_2$, if there exists a pomset bisimulation $R$, such that $(\emptyset,\emptyset)\in R$. By replacing pomset transitions with steps, we can get the definition of step bisimulation. When PESs $\mathcal{E}_1$ and $\mathcal{E}_2$ are step bisimilar, we write $\mathcal{E}_1\sim_s\mathcal{E}_2$.
\end{definition}

\begin{definition}[Posetal product]
Given two PESs $\mathcal{E}_1$, $\mathcal{E}_2$, the posetal product of their configurations, denoted $\mathcal{C}(\mathcal{E}_1)\overline{\times}\mathcal{C}(\mathcal{E}_2)$, is defined as

$$\{(C_1,f,C_2)|C_1\in\mathcal{C}(\mathcal{E}_1),C_2\in\mathcal{C}(\mathcal{E}_2),f:C_1\rightarrow C_2 \textrm{ isomorphism}\}.$$

A subset $R\subseteq\mathcal{C}(\mathcal{E}_1)\overline{\times}\mathcal{C}(\mathcal{E}_2)$ is called a posetal relation. We say that $R$ is downward closed when for any $(C_1,f,C_2),(C_1',f',C_2')\in \mathcal{C}(\mathcal{E}_1)\overline{\times}\mathcal{C}(\mathcal{E}_2)$, if $(C_1,f,C_2)\subseteq (C_1',f',C_2')$ pointwise and $(C_1',f',C_2')\in R$, then $(C_1,f,C_2)\in R$.

For $f:X_1\rightarrow X_2$, we define $f[x_1\mapsto x_2]:X_1\cup\{x_1\}\rightarrow X_2\cup\{x_2\}$, $z\in X_1\cup\{x_1\}$,(1)$f[x_1\mapsto x_2](z)=
x_2$,if $z=x_1$;(2)$f[x_1\mapsto x_2](z)=f(z)$, otherwise. Where $X_1\subseteq \mathbb{E}_1$, $X_2\subseteq \mathbb{E}_2$, $x_1\in \mathbb{E}_1$, $x_2\in \mathbb{E}_2$.
\end{definition}

\begin{definition}[(Hereditary) history-preserving bisimulation]\label{HHPB}
A history-preserving (hp-) bisimulation is a posetal relation $R\subseteq\mathcal{C}(\mathcal{E}_1)\overline{\times}\mathcal{C}(\mathcal{E}_2)$ such that if $(C_1,f,C_2)\in R$, and $C_1\xrightarrow{e_1} C_1'$, then $C_2\xrightarrow{e_2} C_2'$, with $(C_1',f[e_1\mapsto e_2],C_2')\in R$, and vice-versa. $\mathcal{E}_1,\mathcal{E}_2$ are history-preserving (hp-)bisimilar and are written $\mathcal{E}_1\sim_{hp}\mathcal{E}_2$ if there exists a hp-bisimulation $R$ such that $(\emptyset,\emptyset,\emptyset)\in R$.

A hereditary history-preserving (hhp-)bisimulation is a downward closed hp-bisimulation. $\mathcal{E}_1,\mathcal{E}_2$ are hereditary history-preserving (hhp-)bisimilar and are written $\mathcal{E}_1\sim_{hhp}\mathcal{E}_2$.
\end{definition}

In the following, let $e_1, e_2, e_1', e_2'\in \mathbb{E}$, and let variables $x,y,z$ range over the set of terms for true concurrency, $p,q,s$ range over the set of closed terms. The set of axioms of BATC consists of the laws given in Table \ref{AxiomsForBATC}.

\begin{center}
    \begin{table}
        \begin{tabular}{@{}ll@{}}
            \hline No. &Axiom\\
            $A1$ & $x+ y = y+ x$\\
            $A2$ & $(x+ y)+ z = x+ (y+ z)$\\
            $A3$ & $x+ x = x$\\
            $A4$ & $(x+ y)\cdot z = x\cdot z + y\cdot z$\\
            $A5$ & $(x\cdot y)\cdot z = x\cdot(y\cdot z)$\\
        \end{tabular}
        \caption{Axioms of BATC}
        \label{AxiomsForBATC}
    \end{table}
\end{center}

We give the operational transition rules of operators $\cdot$ and $+$ as Table \ref{TRForBATC} shows. And the predicate $\xrightarrow{e}\surd$ represents successful termination after execution of the event $e$.

\begin{center}
    \begin{table}
        $$\frac{}{e\xrightarrow{e}\surd}$$
        $$\frac{x\xrightarrow{e}\surd}{x+ y\xrightarrow{e}\surd} \quad\frac{x\xrightarrow{e}x'}{x+ y\xrightarrow{e}x'} \quad\frac{y\xrightarrow{e}\surd}{x+ y\xrightarrow{e}\surd} \quad\frac{y\xrightarrow{e}y'}{x+ y\xrightarrow{e}y'}$$
        $$\frac{x\xrightarrow{e}\surd}{x\cdot y\xrightarrow{e} y} \quad\frac{x\xrightarrow{e}x'}{x\cdot y\xrightarrow{e}x'\cdot y}$$
        \caption{Transition rules of BATC}
        \label{TRForBATC}
    \end{table}
\end{center}

\begin{theorem}[Soundness of BATC modulo truly concurrent bisimulation equivalences]\label{SBATC}
The axiomatization of BATC is sound modulo truly concurrent bisimulation equivalences $\sim_{p}$, $\sim_{s}$, $\sim_{hp}$ and $\sim_{hhp}$. That is,

\begin{enumerate}
  \item let $x$ and $y$ be BATC terms. If BATC $\vdash x=y$, then $x\sim_{p} y$;
  \item let $x$ and $y$ be BATC terms. If BATC $\vdash x=y$, then $x\sim_{s} y$;
  \item let $x$ and $y$ be BATC terms. If BATC $\vdash x=y$, then $x\sim_{hp} y$;
  \item let $x$ and $y$ be BATC terms. If BATC $\vdash x=y$, then $x\sim_{hhp} y$.
\end{enumerate}

\end{theorem}

\begin{theorem}[Completeness of BATC modulo truly concurrent bisimulation equivalences]\label{CBATC}
The axiomatization of BATC is complete modulo truly concurrent bisimulation equivalences $\sim_{p}$, $\sim_{s}$, $\sim_{hp}$ and $\sim_{hhp}$. That is,

\begin{enumerate}
  \item let $p$ and $q$ be closed BATC terms, if $p\sim_{p} q$ then $p=q$;
  \item let $p$ and $q$ be closed BATC terms, if $p\sim_{s} q$ then $p=q$;
  \item let $p$ and $q$ be closed BATC terms, if $p\sim_{hp} q$ then $p=q$;
  \item let $p$ and $q$ be closed BATC terms, if $p\sim_{hhp} q$ then $p=q$.
\end{enumerate}

\end{theorem}

\subsubsection{Algebra for Parallelism in True Concurrency}

APTC uses the whole parallel operator $\between$, the auxiliary binary parallel $\parallel$ to model parallelism, and the communication merge $\mid$ to model communications among different parallel branches, and also the unary conflict elimination operator $\Theta$ and the binary unless operator $\triangleleft$ to eliminate conflictions among different parallel branches. Since a communication may be blocked, a new constant called deadlock $\delta$ is extended to $A$, and also a new unary encapsulation operator $\partial_H$ is introduced to eliminate $\delta$, which may exist in the processes. The algebraic laws on these operators are also sound and complete modulo truly concurrent bisimulation equivalences (including pomset bisimulation, step bisimulation, hp-bisimulation, but not hhp-bisimulation). Note that, the parallel operator $\parallel$ in a process cannot be eliminated by deductions on the process using axioms of APTC, but other operators can eventually be steadied by $\cdot$, $+$ and $\parallel$, this is also why truly concurrent bisimulations are called an \emph{truly concurrent} semantics.

We design the axioms of APTC in Table \ref{AxiomsForAPTC}, including algebraic laws of parallel operator $\parallel$, communication operator $\mid$, conflict elimination operator $\Theta$ and unless operator $\triangleleft$, encapsulation operator $\partial_H$, the deadlock constant $\delta$, and also the whole parallel operator $\between$.

\begin{center}
    \begin{table}
        \begin{tabular}{@{}ll@{}}
            \hline No. &Axiom\\
            $A6$ & $x+ \delta = x$\\
            $A7$ & $\delta\cdot x =\delta$\\
            $P1$ & $x\between y = x\parallel y + x\mid y$\\
            $P2$ & $x\parallel y = y \parallel x$\\
            $P3$ & $(x\parallel y)\parallel z = x\parallel (y\parallel z)$\\
            $P4$ & $e_1\parallel (e_2\cdot y) = (e_1\parallel e_2)\cdot y$\\
            $P5$ & $(e_1\cdot x)\parallel e_2 = (e_1\parallel e_2)\cdot x$\\
            $P6$ & $(e_1\cdot x)\parallel (e_2\cdot y) = (e_1\parallel e_2)\cdot (x\between y)$\\
            $P7$ & $(x+ y)\parallel z = (x\parallel z)+ (y\parallel z)$\\
            $P8$ & $x\parallel (y+ z) = (x\parallel y)+ (x\parallel z)$\\
            $P9$ & $\delta\parallel x = \delta$\\
            $P10$ & $x\parallel \delta = \delta$\\
            $C11$ & $e_1\mid e_2 = \gamma(e_1,e_2)$\\
            $C12$ & $e_1\mid (e_2\cdot y) = \gamma(e_1,e_2)\cdot y$\\
            $C13$ & $(e_1\cdot x)\mid e_2 = \gamma(e_1,e_2)\cdot x$\\
            $C14$ & $(e_1\cdot x)\mid (e_2\cdot y) = \gamma(e_1,e_2)\cdot (x\between y)$\\
            $C15$ & $(x+ y)\mid z = (x\mid z) + (y\mid z)$\\
            $C16$ & $x\mid (y+ z) = (x\mid y)+ (x\mid z)$\\
            $C17$ & $\delta\mid x = \delta$\\
            $C18$ & $x\mid\delta = \delta$\\
            $CE19$ & $\Theta(e) = e$\\
            $CE20$ & $\Theta(\delta) = \delta$\\
            $CE21$ & $\Theta(x+ y) = \Theta(x)\triangleleft y + \Theta(y)\triangleleft x$\\
            $CE22$ & $\Theta(x\cdot y)=\Theta(x)\cdot\Theta(y)$\\
            $CE23$ & $\Theta(x\parallel y) = ((\Theta(x)\triangleleft y)\parallel y)+ ((\Theta(y)\triangleleft x)\parallel x)$\\
            $CE24$ & $\Theta(x\mid y) = ((\Theta(x)\triangleleft y)\mid y)+ ((\Theta(y)\triangleleft x)\mid x)$\\
            $U25$ & $(\sharp(e_1,e_2))\quad e_1\triangleleft e_2 = \tau$\\
            $U26$ & $(\sharp(e_1,e_2),e_2\leq e_3)\quad e_1\triangleleft e_3 = e_1$\\
            $U27$ & $(\sharp(e_1,e_2),e_2\leq e_3)\quad e3\triangleleft e_1 = \tau$\\
            $U28$ & $e\triangleleft \delta = e$\\
            $U29$ & $\delta \triangleleft e = \delta$\\
            $U30$ & $(x+ y)\triangleleft z = (x\triangleleft z)+ (y\triangleleft z)$\\
            $U31$ & $(x\cdot y)\triangleleft z = (x\triangleleft z)\cdot (y\triangleleft z)$\\
            $U32$ & $(x\parallel y)\triangleleft z = (x\triangleleft z)\parallel (y\triangleleft z)$\\
            $U33$ & $(x\mid y)\triangleleft z = (x\triangleleft z)\mid (y\triangleleft z)$\\
            $U34$ & $x\triangleleft (y+ z) = (x\triangleleft y)\triangleleft z$\\
            $U35$ & $x\triangleleft (y\cdot z)=(x\triangleleft y)\triangleleft z$\\
            $U36$ & $x\triangleleft (y\parallel z) = (x\triangleleft y)\triangleleft z$\\
            $U37$ & $x\triangleleft (y\mid z) = (x\triangleleft y)\triangleleft z$\\
            $D1$ & $e\notin H\quad\partial_H(e) = e$\\
            $D2$ & $e\in H\quad \partial_H(e) = \delta$\\
            $D3$ & $\partial_H(\delta) = \delta$\\
            $D4$ & $\partial_H(x+ y) = \partial_H(x)+\partial_H(y)$\\
            $D5$ & $\partial_H(x\cdot y) = \partial_H(x)\cdot\partial_H(y)$\\
            $D6$ & $\partial_H(x\parallel y) = \partial_H(x)\parallel\partial_H(y)$\\
        \end{tabular}
        \caption{Axioms of APTC}
        \label{AxiomsForAPTC}
    \end{table}
\end{center}

we give the transition rules of APTC in Table \ref{TRForAPTC}, it is suitable for all truly concurrent behavioral equivalence, including pomset bisimulation, step bisimulation, hp-bisimulation and hhp-bisimulation.

\begin{center}
    \begin{table}
        $$\frac{x\xrightarrow{e_1}\surd\quad y\xrightarrow{e_2}\surd}{x\parallel y\xrightarrow{\{e_1,e_2\}}\surd} \quad\frac{x\xrightarrow{e_1}x'\quad y\xrightarrow{e_2}\surd}{x\parallel y\xrightarrow{\{e_1,e_2\}}x'}$$
        $$\frac{x\xrightarrow{e_1}\surd\quad y\xrightarrow{e_2}y'}{x\parallel y\xrightarrow{\{e_1,e_2\}}y'} \quad\frac{x\xrightarrow{e_1}x'\quad y\xrightarrow{e_2}y'}{x\parallel y\xrightarrow{\{e_1,e_2\}}x'\between y'}$$
        $$\frac{x\xrightarrow{e_1}\surd\quad y\xrightarrow{e_2}\surd}{x\mid y\xrightarrow{\gamma(e_1,e_2)}\surd} \quad\frac{x\xrightarrow{e_1}x'\quad y\xrightarrow{e_2}\surd}{x\mid y\xrightarrow{\gamma(e_1,e_2)}x'}$$
        $$\frac{x\xrightarrow{e_1}\surd\quad y\xrightarrow{e_2}y'}{x\mid y\xrightarrow{\gamma(e_1,e_2)}y'} \quad\frac{x\xrightarrow{e_1}x'\quad y\xrightarrow{e_2}y'}{x\mid y\xrightarrow{\gamma(e_1,e_2)}x'\between y'}$$
        $$\frac{x\xrightarrow{e_1}\surd\quad (\sharp(e_1,e_2))}{\Theta(x)\xrightarrow{e_1}\surd} \quad\frac{x\xrightarrow{e_2}\surd\quad (\sharp(e_1,e_2))}{\Theta(x)\xrightarrow{e_2}\surd}$$
        $$\frac{x\xrightarrow{e_1}x'\quad (\sharp(e_1,e_2))}{\Theta(x)\xrightarrow{e_1}\Theta(x')} \quad\frac{x\xrightarrow{e_2}x'\quad (\sharp(e_1,e_2))}{\Theta(x)\xrightarrow{e_2}\Theta(x')}$$
        $$\frac{x\xrightarrow{e_1}\surd \quad y\nrightarrow^{e_2}\quad (\sharp(e_1,e_2))}{x\triangleleft y\xrightarrow{\tau}\surd}
        \quad\frac{x\xrightarrow{e_1}x' \quad y\nrightarrow^{e_2}\quad (\sharp(e_1,e_2))}{x\triangleleft y\xrightarrow{\tau}x'}$$
        $$\frac{x\xrightarrow{e_1}\surd \quad y\nrightarrow^{e_3}\quad (\sharp(e_1,e_2),e_2\leq e_3)}{x\triangleleft y\xrightarrow{e_1}\surd}
        \quad\frac{x\xrightarrow{e_1}x' \quad y\nrightarrow^{e_3}\quad (\sharp(e_1,e_2),e_2\leq e_3)}{x\triangleleft y\xrightarrow{e_1}x'}$$
        $$\frac{x\xrightarrow{e_3}\surd \quad y\nrightarrow^{e_2}\quad (\sharp(e_1,e_2),e_1\leq e_3)}{x\triangleleft y\xrightarrow{\tau}\surd}
        \quad\frac{x\xrightarrow{e_3}x' \quad y\nrightarrow^{e_2}\quad (\sharp(e_1,e_2),e_1\leq e_3)}{x\triangleleft y\xrightarrow{\tau}x'}$$
        $$\frac{x\xrightarrow{e}\surd}{\partial_H(x)\xrightarrow{e}\surd}\quad (e\notin H)\quad\quad\frac{x\xrightarrow{e}x'}{\partial_H(x)\xrightarrow{e}\partial_H(x')}\quad(e\notin H)$$
        \caption{Transition rules of APTC}
        \label{TRForAPTC}
    \end{table}
\end{center}

\begin{theorem}[Soundness of APTC modulo truly concurrent bisimulation equivalences]\label{SAPTC}
The axiomatization of APTC is sound modulo truly concurrent bisimulation equivalences $\sim_{p}$, $\sim_{s}$, and $\sim_{hp}$. That is,

\begin{enumerate}
  \item let $x$ and $y$ be APTC terms. If APTC $\vdash x=y$, then $x\sim_{p} y$;
  \item let $x$ and $y$ be APTC terms. If APTC $\vdash x=y$, then $x\sim_{s} y$;
  \item let $x$ and $y$ be APTC terms. If APTC $\vdash x=y$, then $x\sim_{hp} y$.
\end{enumerate}

\end{theorem}

\begin{theorem}[Completeness of APTC modulo truly concurrent bisimulation equivalences]\label{CAPTC}
The axiomatization of APTC is complete modulo truly concurrent bisimulation equivalences $\sim_{p}$, $\sim_{s}$, and $\sim_{hp}$. That is,

\begin{enumerate}
  \item let $p$ and $q$ be closed APTC terms, if $p\sim_{p} q$ then $p=q$;
  \item let $p$ and $q$ be closed APTC terms, if $p\sim_{s} q$ then $p=q$;
  \item let $p$ and $q$ be closed APTC terms, if $p\sim_{hp} q$ then $p=q$.
\end{enumerate}

\end{theorem}

\subsubsection{Recursion}

To model infinite computation, recursion is introduced into APTC. In order to obtain a sound and complete theory, guarded recursion and linear recursion are needed. The corresponding axioms are RSP (Recursive Specification Principle) and RDP (Recursive Definition Principle), RDP says the solutions of a recursive specification can represent the behaviors of the specification, while RSP says that a guarded recursive specification has only one solution, they are sound with respect to APTC with guarded recursion modulo several truly concurrent bisimulation equivalences (including pomset bisimulation, step bisimulation and hp-bisimulation), and they are complete with respect to APTC with linear recursion modulo several truly concurrent bisimulation equivalences (including pomset bisimulation, step bisimulation and hp-bisimulation). In the following, $E,F,G$ are recursion specifications, $X,Y,Z$ are recursive variables.

For a guarded recursive specifications $E$ with the form

$$X_1=t_1(X_1,\cdots,X_n)$$
$$\cdots$$
$$X_n=t_n(X_1,\cdots,X_n)$$

the behavior of the solution $\langle X_i|E\rangle$ for the recursion variable $X_i$ in $E$, where $i\in\{1,\cdots,n\}$, is exactly the behavior of their right-hand sides $t_i(X_1,\cdots,X_n)$, which is captured by the two transition rules in Table \ref{TRForGR}.

\begin{center}
    \begin{table}
        $$\frac{t_i(\langle X_1|E\rangle,\cdots,\langle X_n|E\rangle)\xrightarrow{\{e_1,\cdots,e_k\}}\surd}{\langle X_i|E\rangle\xrightarrow{\{e_1,\cdots,e_k\}}\surd}$$
        $$\frac{t_i(\langle X_1|E\rangle,\cdots,\langle X_n|E\rangle)\xrightarrow{\{e_1,\cdots,e_k\}} y}{\langle X_i|E\rangle\xrightarrow{\{e_1,\cdots,e_k\}} y}$$
        \caption{Transition rules of guarded recursion}
        \label{TRForGR}
    \end{table}
\end{center}

The $RDP$ (Recursive Definition Principle) and the $RSP$ (Recursive Specification Principle) are shown in Table \ref{RDPRSP}.

\begin{center}
\begin{table}
  \begin{tabular}{@{}ll@{}}
\hline No. &Axiom\\
  $RDP$ & $\langle X_i|E\rangle = t_i(\langle X_1|E,\cdots,X_n|E\rangle)\quad (i\in\{1,\cdots,n\})$\\
  $RSP$ & if $y_i=t_i(y_1,\cdots,y_n)$ for $i\in\{1,\cdots,n\}$, then $y_i=\langle X_i|E\rangle \quad(i\in\{1,\cdots,n\})$\\
\end{tabular}
\caption{Recursive definition and specification principle}
\label{RDPRSP}
\end{table}
\end{center}

\begin{theorem}[Soundness of $APTC$ with guarded recursion]\label{SAPTCR}
Let $x$ and $y$ be $APTC$ with guarded recursion terms. If $APTC\textrm{ with guarded recursion}\vdash x=y$, then
\begin{enumerate}
  \item $x\sim_{s} y$;
  \item $x\sim_{p} y$;
  \item $x\sim_{hp} y$.
\end{enumerate}
\end{theorem}

\begin{theorem}[Completeness of $APTC$ with linear recursion]\label{CAPTCR}
Let $p$ and $q$ be closed $APTC$ with linear recursion terms, then,
\begin{enumerate}
  \item if $p\sim_{s} q$ then $p=q$;
  \item if $p\sim_{p} q$ then $p=q$;
  \item if $p\sim_{hp} q$ then $p=q$.
\end{enumerate}
\end{theorem}

\subsubsection{Abstraction}

To abstract away internal implementations from the external behaviors, a new constant $\tau$ called silent step is added to $A$, and also a new unary abstraction operator $\tau_I$ is used to rename actions in $I$ into $\tau$ (the resulted APTC with silent step and abstraction operator is called $\textrm{APTC}_{\tau}$). The recursive specification is adapted to guarded linear recursion to prevent infinite $\tau$-loops specifically. The axioms of $\tau$ and $\tau_I$ are sound modulo rooted branching truly concurrent bisimulation equivalences (several kinds of weakly truly concurrent bisimulation equivalences, including rooted branching pomset bisimulation, rooted branching step bisimulation and rooted branching hp-bisimulation). To eliminate infinite $\tau$-loops caused by $\tau_I$ and obtain the completeness, CFAR (Cluster Fair Abstraction Rule) is used to prevent infinite $\tau$-loops in a constructible way.

\begin{definition}[Weak pomset transitions and weak step]
Let $\mathcal{E}$ be a PES and let $C\in\mathcal{C}(\mathcal{E})$, and $\emptyset\neq X\subseteq \hat{\mathbb{E}}$, if $C\cap X=\emptyset$ and $\hat{C'}=\hat{C}\cup X\in\mathcal{C}(\mathcal{E})$, then $C\xRightarrow{X} C'$ is called a weak pomset transition from $C$ to $C'$, where we define $\xRightarrow{e}\triangleq\xrightarrow{\tau^*}\xrightarrow{e}\xrightarrow{\tau^*}$. And $\xRightarrow{X}\triangleq\xrightarrow{\tau^*}\xrightarrow{e}\xrightarrow{\tau^*}$, for every $e\in X$. When the events in $X$ are pairwise concurrent, we say that $C\xRightarrow{X}C'$ is a weak step.
\end{definition}

\begin{definition}[Branching pomset, step bisimulation]\label{BPSB}
Assume a special termination predicate $\downarrow$, and let $\surd$ represent a state with $\surd\downarrow$. Let $\mathcal{E}_1$, $\mathcal{E}_2$ be PESs. A branching pomset bisimulation is a relation $R\subseteq\mathcal{C}(\mathcal{E}_1)\times\mathcal{C}(\mathcal{E}_2)$, such that:
 \begin{enumerate}
   \item if $(C_1,C_2)\in R$, and $C_1\xrightarrow{X}C_1'$ then
   \begin{itemize}
     \item either $X\equiv \tau^*$, and $(C_1',C_2)\in R$;
     \item or there is a sequence of (zero or more) $\tau$-transitions $C_2\xrightarrow{\tau^*} C_2^0$, such that $(C_1,C_2^0)\in R$ and $C_2^0\xRightarrow{X}C_2'$ with $(C_1',C_2')\in R$;
   \end{itemize}
   \item if $(C_1,C_2)\in R$, and $C_2\xrightarrow{X}C_2'$ then
   \begin{itemize}
     \item either $X\equiv \tau^*$, and $(C_1,C_2')\in R$;
     \item or there is a sequence of (zero or more) $\tau$-transitions $C_1\xrightarrow{\tau^*} C_1^0$, such that $(C_1^0,C_2)\in R$ and $C_1^0\xRightarrow{X}C_1'$ with $(C_1',C_2')\in R$;
   \end{itemize}
   \item if $(C_1,C_2)\in R$ and $C_1\downarrow$, then there is a sequence of (zero or more) $\tau$-transitions $C_2\xrightarrow{\tau^*}C_2^0$ such that $(C_1,C_2^0)\in R$ and $C_2^0\downarrow$;
   \item if $(C_1,C_2)\in R$ and $C_2\downarrow$, then there is a sequence of (zero or more) $\tau$-transitions $C_1\xrightarrow{\tau^*}C_1^0$ such that $(C_1^0,C_2)\in R$ and $C_1^0\downarrow$.
 \end{enumerate}

We say that $\mathcal{E}_1$, $\mathcal{E}_2$ are branching pomset bisimilar, written $\mathcal{E}_1\approx_{bp}\mathcal{E}_2$, if there exists a branching pomset bisimulation $R$, such that $(\emptyset,\emptyset)\in R$.

By replacing pomset transitions with steps, we can get the definition of branching step bisimulation. When PESs $\mathcal{E}_1$ and $\mathcal{E}_2$ are branching step bisimilar, we write $\mathcal{E}_1\approx_{bs}\mathcal{E}_2$.
\end{definition}

\begin{definition}[Rooted branching pomset, step bisimulation]\label{RBPSB}
Assume a special termination predicate $\downarrow$, and let $\surd$ represent a state with $\surd\downarrow$. Let $\mathcal{E}_1$, $\mathcal{E}_2$ be PESs. A branching pomset bisimulation is a relation $R\subseteq\mathcal{C}(\mathcal{E}_1)\times\mathcal{C}(\mathcal{E}_2)$, such that:
 \begin{enumerate}
   \item if $(C_1,C_2)\in R$, and $C_1\xrightarrow{X}C_1'$ then $C_2\xrightarrow{X}C_2'$ with $C_1'\approx_{bp}C_2'$;
   \item if $(C_1,C_2)\in R$, and $C_2\xrightarrow{X}C_2'$ then $C_1\xrightarrow{X}C_1'$ with $C_1'\approx_{bp}C_2'$;
   \item if $(C_1,C_2)\in R$ and $C_1\downarrow$, then $C_2\downarrow$;
   \item if $(C_1,C_2)\in R$ and $C_2\downarrow$, then $C_1\downarrow$.
 \end{enumerate}

We say that $\mathcal{E}_1$, $\mathcal{E}_2$ are rooted branching pomset bisimilar, written $\mathcal{E}_1\approx_{rbp}\mathcal{E}_2$, if there exists a rooted branching pomset bisimulation $R$, such that $(\emptyset,\emptyset)\in R$.

By replacing pomset transitions with steps, we can get the definition of rooted branching step bisimulation. When PESs $\mathcal{E}_1$ and $\mathcal{E}_2$ are rooted branching step bisimilar, we write $\mathcal{E}_1\approx_{rbs}\mathcal{E}_2$.
\end{definition}

\begin{definition}[Branching (hereditary) history-preserving bisimulation]\label{BHHPB}
Assume a special termination predicate $\downarrow$, and let $\surd$ represent a state with $\surd\downarrow$. A branching history-preserving (hp-) bisimulation is a weakly posetal relation $R\subseteq\mathcal{C}(\mathcal{E}_1)\overline{\times}\mathcal{C}(\mathcal{E}_2)$ such that:

 \begin{enumerate}
   \item if $(C_1,f,C_2)\in R$, and $C_1\xrightarrow{e_1}C_1'$ then
   \begin{itemize}
     \item either $e_1\equiv \tau$, and $(C_1',f[e_1\mapsto \tau],C_2)\in R$;
     \item or there is a sequence of (zero or more) $\tau$-transitions $C_2\xrightarrow{\tau^*} C_2^0$, such that $(C_1,f,C_2^0)\in R$ and $C_2^0\xrightarrow{e_2}C_2'$ with $(C_1',f[e_1\mapsto e_2],C_2')\in R$;
   \end{itemize}
   \item if $(C_1,f,C_2)\in R$, and $C_2\xrightarrow{e_2}C_2'$ then
   \begin{itemize}
     \item either $X\equiv \tau$, and $(C_1,f[e_2\mapsto \tau],C_2')\in R$;
     \item or there is a sequence of (zero or more) $\tau$-transitions $C_1\xrightarrow{\tau^*} C_1^0$, such that $(C_1^0,f,C_2)\in R$ and $C_1^0\xrightarrow{e_1}C_1'$ with $(C_1',f[e_2\mapsto e_1],C_2')\in R$;
   \end{itemize}
   \item if $(C_1,f,C_2)\in R$ and $C_1\downarrow$, then there is a sequence of (zero or more) $\tau$-transitions $C_2\xrightarrow{\tau^*}C_2^0$ such that $(C_1,f,C_2^0)\in R$ and $C_2^0\downarrow$;
   \item if $(C_1,f,C_2)\in R$ and $C_2\downarrow$, then there is a sequence of (zero or more) $\tau$-transitions $C_1\xrightarrow{\tau^*}C_1^0$ such that $(C_1^0,f,C_2)\in R$ and $C_1^0\downarrow$.
 \end{enumerate}

$\mathcal{E}_1,\mathcal{E}_2$ are branching history-preserving (hp-)bisimilar and are written $\mathcal{E}_1\approx_{bhp}\mathcal{E}_2$ if there exists a branching hp-bisimulation $R$ such that $(\emptyset,\emptyset,\emptyset)\in R$.

A branching hereditary history-preserving (hhp-)bisimulation is a downward closed branching hhp-bisimulation. $\mathcal{E}_1,\mathcal{E}_2$ are branching hereditary history-preserving (hhp-)bisimilar and are written $\mathcal{E}_1\approx_{bhhp}\mathcal{E}_2$.
\end{definition}

\begin{definition}[Rooted branching (hereditary) history-preserving bisimulation]\label{RBHHPB}
Assume a special termination predicate $\downarrow$, and let $\surd$ represent a state with $\surd\downarrow$. A rooted branching history-preserving (hp-) bisimulation is a weakly posetal relation $R\subseteq\mathcal{C}(\mathcal{E}_1)\overline{\times}\mathcal{C}(\mathcal{E}_2)$ such that:

 \begin{enumerate}
   \item if $(C_1,f,C_2)\in R$, and $C_1\xrightarrow{e_1}C_1'$, then $C_2\xrightarrow{e_2}C_2'$ with $C_1'\approx_{bhp}C_2'$;
   \item if $(C_1,f,C_2)\in R$, and $C_2\xrightarrow{e_2}C_1'$, then $C_1\xrightarrow{e_1}C_2'$ with $C_1'\approx_{bhp}C_2'$;
   \item if $(C_1,f,C_2)\in R$ and $C_1\downarrow$, then $C_2\downarrow$;
   \item if $(C_1,f,C_2)\in R$ and $C_2\downarrow$, then $C_1\downarrow$.
 \end{enumerate}

$\mathcal{E}_1,\mathcal{E}_2$ are rooted branching history-preserving (hp-)bisimilar and are written $\mathcal{E}_1\approx_{rbhp}\mathcal{E}_2$ if there exists rooted a branching hp-bisimulation $R$ such that $(\emptyset,\emptyset,\emptyset)\in R$.

A rooted branching hereditary history-preserving (hhp-)bisimulation is a downward closed rooted branching hhp-bisimulation. $\mathcal{E}_1,\mathcal{E}_2$ are rooted branching hereditary history-preserving (hhp-)bisimilar and are written $\mathcal{E}_1\approx_{rbhhp}\mathcal{E}_2$.
\end{definition}

The axioms and transition rules of $\textrm{APTC}_{\tau}$ are shown in Table \ref{AxiomsForTau} and Table \ref{TRForTau}.

\begin{center}
\begin{table}
  \begin{tabular}{@{}ll@{}}
\hline No. &Axiom\\
  $B1$ & $e\cdot\tau=e$\\
  $B2$ & $e\cdot(\tau\cdot(x+y)+x)=e\cdot(x+y)$\\
  $B3$ & $x\parallel\tau=x$\\
  $TI1$ & $e\notin I\quad \tau_I(e)=e$\\
  $TI2$ & $e\in I\quad \tau_I(e)=\tau$\\
  $TI3$ & $\tau_I(\delta)=\delta$\\
  $TI4$ & $\tau_I(x+y)=\tau_I(x)+\tau_I(y)$\\
  $TI5$ & $\tau_I(x\cdot y)=\tau_I(x)\cdot\tau_I(y)$\\
  $TI6$ & $\tau_I(x\parallel y)=\tau_I(x)\parallel\tau_I(y)$\\
  $CFAR$ & If $X$ is in a cluster for $I$ with exits \\
           & $\{(a_{11}\parallel\cdots\parallel a_{1i})Y_1,\cdots,(a_{m1}\parallel\cdots\parallel a_{mi})Y_m, b_{11}\parallel\cdots\parallel b_{1j},\cdots,b_{n1}\parallel\cdots\parallel b_{nj}\}$, \\
           & then $\tau\cdot\tau_I(\langle X|E\rangle)=$\\
           & $\tau\cdot\tau_I((a_{11}\parallel\cdots\parallel a_{1i})\langle Y_1|E\rangle+\cdots+(a_{m1}\parallel\cdots\parallel a_{mi})\langle Y_m|E\rangle+b_{11}\parallel\cdots\parallel b_{1j}+\cdots+b_{n1}\parallel\cdots\parallel b_{nj})$\\
\end{tabular}
\caption{Axioms of $\textrm{APTC}_{\tau}$}
\label{AxiomsForTau}
\end{table}
\end{center}

\begin{center}
    \begin{table}
        $$\frac{}{\tau\xrightarrow{\tau}\surd}$$
        $$\frac{x\xrightarrow{e}\surd}{\tau_I(x)\xrightarrow{e}\surd}\quad e\notin I
        \quad\quad\frac{x\xrightarrow{e}x'}{\tau_I(x)\xrightarrow{e}\tau_I(x')}\quad e\notin I$$

        $$\frac{x\xrightarrow{e}\surd}{\tau_I(x)\xrightarrow{\tau}\surd}\quad e\in I
        \quad\quad\frac{x\xrightarrow{e}x'}{\tau_I(x)\xrightarrow{\tau}\tau_I(x')}\quad e\in I$$
        \caption{Transition rule of $\textrm{APTC}_{\tau}$}
        \label{TRForTau}
    \end{table}
\end{center}

\begin{theorem}[Soundness of $APTC_{\tau}$ with guarded linear recursion]\label{SAPTCABS}
Let $x$ and $y$ be $APTC_{\tau}$ with guarded linear recursion terms. If $APTC_{\tau}$ with guarded linear recursion $\vdash x=y$, then
\begin{enumerate}
  \item $x\approx_{rbs} y$;
  \item $x\approx_{rbp} y$;
  \item $x\approx_{rbhp} y$.
\end{enumerate}
\end{theorem}

\begin{theorem}[Soundness of $CFAR$]\label{SCFAR}
$CFAR$ is sound modulo rooted branching truly concurrent bisimulation equivalences $\approx_{rbs}$, $\approx_{rbp}$ and $\approx_{rbhp}$.
\end{theorem}

\begin{theorem}[Completeness of $APTC_{\tau}$ with guarded linear recursion and $CFAR$]\label{CCFAR}
Let $p$ and $q$ be closed $APTC_{\tau}$ with guarded linear recursion and $CFAR$ terms, then,
\begin{enumerate}
  \item if $p\approx_{rbs} q$ then $p=q$;
  \item if $p\approx_{rbp} q$ then $p=q$;
  \item if $p\approx_{rbhp} q$ then $p=q$.
\end{enumerate}
\end{theorem}

\subsection{Timing}{\label{timing}}

Process algebra with timing \cite{T1} \cite{T2} \cite{T3} can be used to describe or analyze systems with time-dependent behaviors, which is an extension of process algebra ACP \cite{ACP}. The timing of actions is either relative or absolute, and the time scale on which the time is measured is either discrete or continuous. The four resulted theories are generalizations of ACP without timing.

This work (truly concurrent process algebra with timing) is a generalization of APTC without timing (see section \ref{tcpa}). Similarly to process algebra with timing \cite{T1} \cite{T2} \cite{T3}, there are also four resulted theories. The four theories with timing (four truly concurrent process algebras with timing and four process algebras with timing) will be explained in details in the following sections, and we do not repeat again in this subsection.

\section{Discrete Relative Timing}{\label{drt}}

In this section, we will introduce a version of APTC with relative timing and time measured on a discrete time scale. Measuring time on  a discrete time scale means that time is divided into time slices and timing of actions is done with respect to the time slices in which they are performed. With respect to relative timing, timing is relative to the execution time of the previous action, and if the previous action does not exist, the start-up time of the whole process.

Like APTC without timing, let us start with a basic algebra for true concurrency called $\textrm{BATC}^{\textrm{drt}}$ (BATC with discrete relative timing). Then we continue with $\textrm{APTC}^{\textrm{drt}}$ (APTC with discrete relative timing).

\subsection{Basic Definitions}

In this subsection, we will introduce some basic definitions about timing. These basic concepts come from \cite{T3}, we introduce them into the backgrounds of true concurrency.

\begin{definition}[Undelayable actions]
Undelayable actions are defined as atomic processes that perform an action in the current time slice and then terminate successfully. We use a constant $\underline{\underline{a}}$ to represent the undelayable action, that is, the atomic process that performs the action $a$ in the current time slice and then terminates successfully.
\end{definition}

\begin{definition}[Undelayable deadlock]
Undelayable deadlock $\underline{\underline{\delta}}$ is an additional process that is neither capable of performing any action nor capable of idling till the next time slice.
\end{definition}

\begin{definition}[Relative delay]
The relative delay of the process $p$ for $n$ ($n\in\mathbb{N}$) time slices is the process that idles till the $n$th-next time slice and then behaves like $p$. The operator $\sigma_{\textrm{rel}}$ is used to represent the relative delay, and let $\sigma^n_{\textrm{rel}}(t) = n \sigma_{\textrm{rel}} t$.
\end{definition}

\begin{definition}[Deadlocked process]
Deadlocked process $\dot{\delta}$ is an additional process that has deadlocked before the current time slice. After a delay of one time slice, the undelayable deadlock $\underline{\underline{\delta}}$ and the deadlocked process $\dot{\delta}$ are indistinguishable from each other.
\end{definition}

\begin{definition}[Truly concurrent bisimulation equivalences with time-related capabilities]\label{TBTTC1}
The following requirement with time-related capabilities is added to truly concurrent bisimulation equivalences $\sim_{p}$, $\sim_{s}$, $\sim_{hp}$ and $\sim_{hhp}$:
\begin{itemize}
  \item if a process is capable of first idling till a certain time slice and next going on as another process, then any equivalent process must be capable of first idling till the same time slice and next going on as a process equivalent to the other process;
  \item if a process has deadlocked before the current time slice, then any equivalent process must have deadlocked before the current time slice.
\end{itemize}
\end{definition}

\begin{definition}[Relative time-out]
The relative time-out $\upsilon_{\textrm{rel}}$ of a process $p$ after $n$ ($n\in\mathbb{N}$) time slices behaves either like the part of $p$ that does not idle till the $n$th-next time slice, or like the deadlocked process after a delay of $n$ time slices if $p$ is capable of idling till the $n$th-next time slice; otherwise, like $p$. And let $\upsilon^n_{\textrm{rel}}(t) = n \upsilon_{\textrm{rel}} t$.
\end{definition}

\begin{definition}[Relative initialization]
The relative initialization $\overline{\upsilon}_{\textrm{rel}}$ of a process $p$ after $n$ ($n\in\mathbb{N}$) time slices behaves like the part of $p$ that idles till the $n$th-next time slice if $p$ is capable of idling till that time slice; otherwise, like the deadlocked process after a delay of $n$ time slices. And we let $\overline{\upsilon}^n_{\textrm{rel}}(t) = n \overline{\upsilon}_{\textrm{rel}} t$.
\end{definition}

\subsection{Basic Algebra for True Concurrency with Discrete Relative Timing}

In this subsection, we will introduce the theory $\textrm{BATC}^{\textrm{drt}}$.

\subsubsection{The Theory $\textrm{BATC}^{\textrm{drt}}$}

\begin{definition}[Signature of $\textrm{BATC}^{\textrm{drt}}$]
The signature of $\textrm{BATC}^{\textrm{drt}}$ consists of the sort $\mathcal{P}_{\textrm{rel}}$ of processes with discrete relative timing, the undelayable action constants $\underline{\underline{a}}: \rightarrow\mathcal{P}_{\textrm{rel}}$ for each $a\in A$, the undelayable deadlock constant $\underline{\underline{\delta}}: \rightarrow \mathcal{P}_{\textrm{rel}}$, the alternative composition operator $+: \mathcal{P}_{\textrm{rel}}\times\mathcal{P}_{\textrm{rel}} \rightarrow \mathcal{P}_{\textrm{rel}}$, the sequential composition operator $\cdot: \mathcal{P}_{\textrm{rel}} \times \mathcal{P}_{\textrm{rel}} \rightarrow \mathcal{P}_{\textrm{rel}}$, the relative delay operator $\sigma_{\textrm{rel}}: \mathbb{N}\times \mathcal{P}_{\textrm{rel}} \rightarrow \mathcal{P}_{\textrm{rel}}$, the deadlocked process constant $\dot{\delta}: \rightarrow \mathcal{P}_{\textrm{rel}}$, the relative time-out operator $\upsilon_{\textrm{rel}}: \mathbb{N}\times\mathcal{P}_{\textrm{rel}} \rightarrow\mathcal{P}_{\textrm{rel}}$ and the relative initialization operator $\overline{\upsilon}_{\textrm{rel}}: \mathbb{N}\times\mathcal{P}_{\textrm{rel}} \rightarrow\mathcal{P}_{\textrm{rel}}$.
\end{definition}

The set of axioms of $\textrm{BATC}^{\textrm{drt}}$ consists of the laws given in Table \ref{AxiomsForBATCDRT}.

\begin{center}
    \begin{table}
        \begin{tabular}{@{}ll@{}}
            \hline No. &Axiom\\
            $A1$ & $x+ y = y+ x$\\
            $A2$ & $(x+ y)+ z = x+ (y+ z)$\\
            $A3$ & $x+ x = x$\\
            $A4$ & $(x+ y)\cdot z = x\cdot z + y\cdot z$\\
            $A5$ & $(x\cdot y)\cdot z = x\cdot(y\cdot z)$\\
            $A6ID$ & $x + \dot{\delta} = x$\\
            $A7ID$ & $\dot{\delta}\cdot x = \dot{\delta}$\\
            $DRT1$ & $\sigma^0_{\textrm{rel}}(x) = x$\\
            $DRT2$ & $\sigma^m_{\textrm{rel}}( \sigma^n_{\textrm{rel}}(x)) = \sigma^{m+n}_{\textrm{rel}}(x)$\\
            $DRT3$ & $\sigma^n_{\textrm{rel}}(x) + \sigma^n_{\textrm{rel}}(y) = \sigma^n_{\textrm{rel}}(x+y)$\\
            $DRT4$ & $\sigma^n_{\textrm{rel}}(x)\cdot y = \sigma^n_{\textrm{rel}}(x\cdot y)$\\
            $DRT7$ & $\sigma^1_{\textrm{rel}}(\dot{\delta}) = \underline{\underline{\delta}}$\\
            $A6DRa$ & $\underline{\underline{a}} + \underline{\underline{\delta}} = \underline{\underline{a}}$\\
            $DRTO0$ & $\upsilon^n_{\textrm{rel}}(\dot{\delta}) = \dot{\delta}$\\
            $DRTO1$ & $\upsilon^0_{\textrm{rel}}(x) = \dot(\delta)$\\
            $DRTO2$ & $\upsilon^{n+1}_{\textrm{rel}}(\underline{\underline{a}}) = \underline{\underline{a}}$\\
            $DRTO3$ & $\upsilon^{m+n}_{\textrm{rel}} (\sigma^n_{\textrm{rel}}(x)) = \sigma^n_{\textrm{rel}}(\upsilon^m_{\textrm{rel}}(x))$\\
            $DRTO4$ & $\upsilon^n_{\textrm{rel}}(x+y) = \upsilon^n_{\textrm{rel}}(x) + \upsilon^n_{\textrm{rel}}(y)$\\
            $DRTO5$ & $\upsilon^n_{\textrm{rel}}(x\cdot y) = \upsilon^n_{\textrm{rel}}(x)\cdot y$\\
            $DRI0$ & $\overline{\upsilon}^n_{\textrm{rel}}(\dot{\delta}) = \sigma^n_{\textrm{rel}}(\dot{\delta})$\\
            $DRI1$ & $\overline{\upsilon}^0_{\textrm{rel}}(x) = x$\\
            $DRI2$ & $\overline{\upsilon}^{n+1}_{\textrm{rel}}(\underline{\underline{a}}) = \sigma^n_{\textrm{rel}}(\underline{\underline{\delta}})$\\
            $DRI3$ & $\overline{\upsilon}^{m+n}_{\textrm{rel}} (\sigma^n_{\textrm{rel}}(x)) = \sigma^n_{\textrm{rel}}(\overline{\upsilon}^m_{\textrm{rel}}(x))$\\
            $DRI4$ & $\overline{\upsilon}^n_{\textrm{rel}}(x+y) = \overline{\upsilon}^n_{\textrm{rel}}(x) + \overline{\upsilon}^n_{\textrm{rel}}(y)$\\
            $DRI5$ & $\overline{\upsilon}^n_{\textrm{rel}}(x\cdot y) = \overline{\upsilon}^n_{\textrm{rel}}(x)\cdot y$\\
        \end{tabular}
        \caption{Axioms of $\textrm{BATC}^{\textrm{drt}}(a\in A_{\delta}, m,n\geq 0)$}
        \label{AxiomsForBATCDRT}
    \end{table}
\end{center}

The operational semantics of $\textrm{BATC}^{\textrm{drt}}$ are defined by the transition rules in Table \ref{TRForBATCDRT}. Where $\uparrow$ is a unary deadlocked predicate, and $t\nuparrow \triangleq \neg(t\uparrow)$; $t\mapsto^m t'$ means that process $t$ is capable of first idling till the $m$th-next time slice, and then proceeding as process $t'$.

\begin{center}
    \begin{table}
        $$\frac{}{\dot{\delta}\uparrow}
        \quad\frac{}{\underline{\underline{a}}\xrightarrow{a}\surd}
        \quad\frac{x\xrightarrow{a}x'}{\sigma^0_{\textrm{rel}}(x) \xrightarrow{a}x'}
        \quad\frac{x\xrightarrow{a}\surd}{\sigma^0_{\textrm{rel}}(x) \xrightarrow{a}\surd}
        \quad\frac{x\uparrow}{\sigma^0_{\textrm{rel}}(x)\uparrow}$$
        $$\frac{}{\sigma^{m+n+1}_{\textrm{rel}}(x)\mapsto^{m} \sigma^{n+1}_{\textrm{rel}}(x)}
        \quad \frac{x\nuparrow}{\sigma^m_{\textrm{rel}}(x)\mapsto^m x}
        \quad\frac{x\mapsto^m x'}{\sigma^n_{\textrm{rel}}(x) \mapsto^{m+n} x'}$$
        $$\frac{x\xrightarrow{a}x'}{x+ y\xrightarrow{a}x'} \quad\frac{y\xrightarrow{a}y'}{x+ y\xrightarrow{a}y'}
        \quad\frac{x\xrightarrow{a}\surd}{x+ y\xrightarrow{a}\surd}
        \quad\frac{y\xrightarrow{a}\surd}{x+ y\xrightarrow{a}\surd}$$
        $$\frac{x\mapsto^{m}x'\quad y\nmapsto^m}{x+ y\mapsto^{m}x'} \quad\frac{x\nmapsto^{m}\quad y\mapsto^m y'}{x+ y\mapsto^{m}y'}
        \quad\frac{x\mapsto^{m}x'\quad y\mapsto^m y'}{x+ y\mapsto^{m}x'+y'}
        \quad\frac{x\uparrow\quad y\uparrow}{x+ y\uparrow}$$
        $$\frac{x\xrightarrow{a}\surd}{x\cdot y\xrightarrow{a} y} \quad\frac{x\xrightarrow{a}x'}{x\cdot y\xrightarrow{a}x'\cdot y}
        \quad \frac{x\mapsto^{m}x'}{x\cdot y\mapsto^{m}x'\cdot y}
        \quad \frac{x\uparrow}{x\cdot y\uparrow}$$
        $$\frac{x\xrightarrow{a}x'}{\upsilon^{n+1}_{\textrm{rel}}(x) \xrightarrow{a}x'}
        \quad\frac{x\xrightarrow{a}\surd}{\upsilon^{n+1}_{\textrm{rel}}(x) \xrightarrow{a}\surd}$$
        $$\frac{x\mapsto^m x'}{\upsilon^{m+n+1}_{\textrm{rel}}(x) \mapsto^m \upsilon^{n+1}_{\textrm{rel}}(x')}
        \quad\frac{}{\upsilon^0_{\textrm{rel}}(x)\uparrow}
        \quad\frac{x\uparrow}{\upsilon^{n+1}_{\textrm{rel}}(x)\uparrow}$$
        $$\frac{x\xrightarrow{a}x'}{\overline{\upsilon}^{0}_{\textrm{rel}}(x) \xrightarrow{a}x'}
        \quad\frac{x\xrightarrow{a}\surd}{\overline{\upsilon}^{0}_{\textrm{rel}}(x) \xrightarrow{a}\surd}
        \quad\frac{x\mapsto^m x'}{\overline{\upsilon}^n_{\textrm{rel}}(x)\mapsto^m x'}n\leq m$$
        $$\frac{x\mapsto^m x'}{\overline{\upsilon}^{m+n+1}_{\textrm{rel}}(x) \mapsto^m \upsilon^{n+1}_{\textrm{rel}}(x')}
        \quad\frac{x\nmapsto^m}{\overline{\upsilon}^{m+n+1}_{\textrm{rel}}(x) \mapsto^m \overline{\upsilon}^{n+1}_{\textrm{rel}}(\dot{\delta})}
        \quad\frac{x\uparrow}{\overline{\upsilon}^{0}_{\textrm{rel}}(x)\uparrow}$$
        \caption{Transition rules of $\textrm{BATC}^{\textrm{drt}}a\in A, m>0,n\geq 0$}
        \label{TRForBATCDRT}
    \end{table}
\end{center}

\subsubsection{Elimination}

\begin{definition}[Basic terms of $\textrm{BATC}^{\textrm{drt}}$]
The set of basic terms of $\textrm{BATC}^{\textrm{drt}}$, $\mathcal{B}(\textrm{BATC}^{\textrm{drt}})$, is inductively defined as follows by two auxiliary sets $\mathcal{B}_0(\textrm{BATC}^{\textrm{drt}})$ and $\mathcal{B}_1(\textrm{BATC}^{\textrm{drt}})$:
\begin{enumerate}
  \item if $a\in A_{\delta}$, then $\underline{\underline{a}} \in \mathcal{B}_1(\textrm{BATC}^{\textrm{drt}})$;
  \item if $a\in A$ and $t\in \mathcal{B}(\textrm{BATC}^{\textrm{drt}})$, then $\underline{\underline{a}}\cdot t \in \mathcal{B}_1(\textrm{BATC}^{\textrm{drt}})$;
  \item if $t,t'\in \mathcal{B}_1(\textrm{BATC}^{\textrm{drt}})$, then $t+t'\in \mathcal{B}_1(\textrm{BATC}^{\textrm{drt}})$;
  \item if $t\in \mathcal{B}_1(\textrm{BATC}^{\textrm{drt}})$, then $t\in \mathcal{B}_0(\textrm{BATC}^{\textrm{drt}})$;
  \item if $n>0$ and $t\in \mathcal{B}_0(\textrm{BATC}^{\textrm{drt}})$, then $\sigma^n_{\textrm{rel}}(t) \in \mathcal{B}_0(\textrm{BATC}^{\textrm{drt}})$;
  \item if $n>0$, $t\in \mathcal{B}_1(\textrm{BATC}^{\textrm{drt}})$ and $t'\in \mathcal{B}_0(\textrm{BATC}^{\textrm{drt}})$, then $t+\sigma^n_{\textrm{rel}}(t') \in \mathcal{B}_0(\textrm{BATC}^{\textrm{drt}})$;
  \item $\dot{\delta}\in \mathcal{B}(\textrm{BATC}^{\textrm{drt}})$;
  \item if $t\in \mathcal{B}_0(\textrm{BATC}^{\textrm{drt}})$, then $t\in \mathcal{B}(\textrm{BATC}^{\textrm{drt}})$.
\end{enumerate}
\end{definition}

\begin{theorem}[Elimination theorem]
Let $p$ be a closed $\textrm{BATC}^{\textrm{drt}}$ term. Then there is a basic $\textrm{BATC}^{\textrm{drt}}$ term $q$ such that $\textrm{BATC}^{\textrm{drt}}\vdash p=q$.
\end{theorem}

\begin{proof}
It is sufficient to induct on the structure of the closed $\textrm{BATC}^{\textrm{drt}}$ term $p$. It can be proven that $p$ combined by the constants and operators of $\textrm{BATC}^{\textrm{drt}}$ exists an equal basic term $q$, and the other operators not included in the basic terms, such as $\upsilon_{\textrm{rel}}$ and $\overline{\upsilon}_{\textrm{rel}}$ can be eliminated.
\end{proof}

\subsubsection{Connections}

\begin{theorem}[Generalization of $\textrm{BATC}^{\textrm{drt}}$]

By the definitions of $a=\underline{\underline{a}}$ for each $a\in A$ and $\delta=\underline{\underline{\delta}}$, $\textrm{BATC}^{\textrm{drt}}$ is a generalization of $BATC$.
\end{theorem}

\begin{proof}
It follows from the following two facts.

\begin{enumerate}
  \item The transition rules of $BATC$ in section \ref{tcpa} are all source-dependent;
  \item The sources of the transition rules of $\textrm{BATC}^{\textrm{drt}}$ contain an occurrence of $\dot{\delta}$, $\underline{\underline{a}}$, $\sigma^n_{\textrm{rel}}$, $\upsilon^n_{\textrm{rel}}$ and $\overline{\upsilon}^n_{\textrm{rel}}$.
\end{enumerate}

So, $BATC$ is an embedding of $\textrm{BATC}^{\textrm{drt}}$, as desired.
\end{proof}

\subsubsection{Congruence}

\begin{theorem}[Congruence of $\textrm{BATC}^{\textrm{drt}}$]
Truly concurrent bisimulation equivalences are all congruences with respect to $\textrm{BATC}^{\textrm{drt}}$. That is,
\begin{itemize}
  \item pomset bisimulation equivalence $\sim_{p}$ is a congruence with respect to $\textrm{BATC}^{\textrm{drt}}$;
  \item step bisimulation equivalence $\sim_{s}$ is a congruence with respect to $\textrm{BATC}^{\textrm{drt}}$;
  \item hp-bisimulation equivalence $\sim_{hp}$ is a congruence with respect to $\textrm{BATC}^{\textrm{drt}}$;
  \item hhp-bisimulation equivalence $\sim_{hhp}$ is a congruence with respect to $\textrm{BATC}^{\textrm{drt}}$.
\end{itemize}
\end{theorem}

\begin{proof}
It is easy to see that $\sim_p$, $\sim_s$, $\sim_{hp}$ and $\sim_{hhp}$ are all equivalent relations on $\textrm{BATC}^{\textrm{drt}}$ terms, it is only sufficient to prove that $\sim_p$, $\sim_s$, $\sim_{hp}$ and $\sim_{hhp}$ are all preserved by the operators $\sigma^n_{\textrm{rel}}$, $\upsilon^n_{\textrm{rel}}$ and $\overline{\upsilon}^n_{\textrm{rel}}$. It is trivial and we omit it.
\end{proof}

\subsubsection{Soundness}

\begin{theorem}[Soundness of $\textrm{BATC}^{\textrm{drt}}$]
The axiomatization of $\textrm{BATC}^{\textrm{drt}}$ is sound modulo truly concurrent bisimulation equivalences $\sim_{p}$, $\sim_{s}$, $\sim_{hp}$ and $\sim_{hhp}$. That is,
\begin{enumerate}
  \item let $x$ and $y$ be $\textrm{BATC}^{\textrm{drt}}$ terms. If $\textrm{BATC}^{\textrm{drt}}\vdash x=y$, then $x\sim_{s} y$;
  \item let $x$ and $y$ be $\textrm{BATC}^{\textrm{drt}}$ terms. If $\textrm{BATC}^{\textrm{drt}}\vdash x=y$, then $x\sim_{p} y$;
  \item let $x$ and $y$ be $\textrm{BATC}^{\textrm{drt}}$ terms. If $\textrm{BATC}^{\textrm{drt}}\vdash x=y$, then $x\sim_{hp} y$;
  \item let $x$ and $y$ be $\textrm{BATC}^{\textrm{drt}}$ terms. If $\textrm{BATC}^{\textrm{drt}}\vdash x=y$, then $x\sim_{hhp} y$.
\end{enumerate}
\end{theorem}

\begin{proof}
Since $\sim_p$, $\sim_s$, $\sim_{hp}$ and $\sim_{hhp}$ are both equivalent and congruent relations, we only need to check if each axiom in Table \ref{AxiomsForBATCDRT} is sound modulo $\sim_p$, $\sim_s$, $\sim_{hp}$ and $\sim_{hhp}$ respectively.

\begin{enumerate}
  \item We only check the soundness of the non-trivial axiom $DRTO3$ modulo $\sim_s$.
        Let $p$ be $\textrm{BATC}^{\textrm{drt}}$ processes, and $\upsilon^{m+n}_{\textrm{rel}} (\sigma^n_{\textrm{rel}}(p)) = \sigma^n_{\textrm{rel}}(\upsilon^m_{\textrm{rel}}(p))$, it is sufficient to prove that $\upsilon^{m+n}_{\textrm{rel}} (\sigma^n_{\textrm{rel}}(p)) \sim_{s} \sigma^n_{\textrm{rel}}(\upsilon^m_{\textrm{rel}}(p))$. By the transition rules of operator $\sigma^n_{\textrm{rel}}$ and $\upsilon^m_{\textrm{rel}}$ in Table \ref{TRForBATCDRT}, we get

        $$\frac{}{\upsilon^{m+n}_{\textrm{rel}} (\sigma^n_{\textrm{rel}}(p))\mapsto^n \upsilon^{m}_{\textrm{rel}} (\sigma^0_{\textrm{rel}}(p))}$$

        $$\frac{}{\sigma^n_{\textrm{rel}}(\upsilon^m_{\textrm{rel}}(p))\mapsto^n \sigma^0_{\textrm{rel}}(\upsilon^m_{\textrm{rel}}(p))}$$

        There are several cases:

        $$\frac{p\xrightarrow{a} \surd}{\upsilon^{m}_{\textrm{rel}} (\sigma^0_{\textrm{rel}}(p))\xrightarrow{a}\surd}$$

        $$\frac{p\xrightarrow{a} \surd}{\sigma^0_{\textrm{rel}}(\upsilon^m_{\textrm{rel}}(p))\xrightarrow{a}\surd}$$

        $$\frac{p\xrightarrow{a} p'}{\upsilon^{m}_{\textrm{rel}} (\sigma^0_{\textrm{rel}}(p))\xrightarrow{a}p'}$$

        $$\frac{p\xrightarrow{a} p'}{\sigma^0_{\textrm{rel}}(\upsilon^m_{\textrm{rel}}(p))\xrightarrow{a}p'}$$

        $$\frac{p \uparrow}{\upsilon^{m}_{\textrm{rel}} (\sigma^0_{\textrm{rel}}(p))\uparrow}$$

        $$\frac{p \uparrow}{\sigma^0_{\textrm{rel}}(\upsilon^m_{\textrm{rel}}(p))\uparrow}$$

        So, we see that each case leads to $\upsilon^{m+n}_{\textrm{rel}} (\sigma^n_{\textrm{rel}}(p)) \sim_{s} \sigma^n_{\textrm{rel}}(\upsilon^m_{\textrm{rel}}(p))$, as desired.
  \item From the definition of pomset bisimulation, we know that pomset bisimulation is defined by pomset transitions, which are labeled by pomsets. In a pomset transition, the events (actions) in the pomset are either within causality relations (defined by $\cdot$) or in concurrency (implicitly defined by $\cdot$ and $+$, and explicitly defined by $\between$), of course, they are pairwise consistent (without conflicts). We have already proven the case that all events are pairwise concurrent (soundness modulo step bisimulation), so, we only need to prove the case of events in causality. Without loss of generality, we take a pomset of $P=\{\underline{\underline{a}},\underline{\underline{b}}:\underline{\underline{a}}\cdot \underline{\underline{b}}\}$. Then the pomset transition labeled by the above $P$ is just composed of one single event transition labeled by $\underline{\underline{a}}$ succeeded by another single event transition labeled by $\underline{\underline{b}}$, that is, $\xrightarrow{P}=\xrightarrow{a}\xrightarrow{b}$.

        Similarly to the proof of soundness modulo step bisimulation equivalence, we can prove that each axiom in Table \ref{AxiomsForBATCDRT} is sound modulo pomset bisimulation equivalence, we omit them.
  \item From the definition of hp-bisimulation, we know that hp-bisimulation is defined on the posetal product $(C_1,f,C_2),f:C_1\rightarrow C_2\textrm{ isomorphism}$. Two process terms $s$ related to $C_1$ and $t$ related to $C_2$, and $f:C_1\rightarrow C_2\textrm{ isomorphism}$. Initially, $(C_1,f,C_2)=(\emptyset,\emptyset,\emptyset)$, and $(\emptyset,\emptyset,\emptyset)\in\sim_{hp}$. When $s\xrightarrow{a}s'$ ($C_1\xrightarrow{a}C_1'$), there will be $t\xrightarrow{a}t'$ ($C_2\xrightarrow{a}C_2'$), and we define $f'=f[a\mapsto a]$. Then, if $(C_1,f,C_2)\in\sim_{hp}$, then $(C_1',f',C_2')\in\sim_{hp}$.

        Similarly to the proof of soundness modulo pomset bisimulation equivalence, we can prove that each axiom in Table \ref{AxiomsForBATCDRT} is sound modulo hp-bisimulation equivalence, we just need additionally to check the above conditions on hp-bisimulation, we omit them.
  \item We just need to add downward-closed condition to the soundness modulo hp-bisimulation equivalence, we omit them.
\end{enumerate}

\end{proof}

\subsubsection{Completeness}

\begin{theorem}[Completeness of $\textrm{BATC}^{\textrm{drt}}$]
The axiomatization of $\textrm{BATC}^{\textrm{drt}}$ is complete modulo truly concurrent bisimulation equivalences $\sim_{p}$, $\sim_{s}$, $\sim_{hp}$ and $\sim_{hhp}$. That is,
\begin{enumerate}
  \item let $p$ and $q$ be closed $\textrm{BATC}^{\textrm{drt}}$ terms, if $p\sim_{s} q$ then $p=q$;
  \item let $p$ and $q$ be closed $\textrm{BATC}^{\textrm{drt}}$ terms, if $p\sim_{p} q$ then $p=q$;
  \item let $p$ and $q$ be closed $\textrm{BATC}^{\textrm{drt}}$ terms, if $p\sim_{hp} q$ then $p=q$;
  \item let $p$ and $q$ be closed $\textrm{BATC}^{\textrm{drt}}$ terms, if $p\sim_{hhp} q$ then $p=q$.
\end{enumerate}

\end{theorem}

\begin{proof}
\begin{enumerate}
  \item Firstly, by the elimination theorem of $\textrm{BATC}^{\textrm{drt}}$, we know that for each closed $\textrm{BATC}^{\textrm{drt}}$ term $p$, there exists a closed basic $\textrm{BATC}^{\textrm{drt}}$ term $p'$, such that $\textrm{BATC}^{\textrm{drt}}\vdash p=p'$, so, we only need to consider closed basic $\textrm{BATC}^{\textrm{drt}}$ terms.

        The basic terms modulo associativity and commutativity (AC) of conflict $+$ (defined by axioms $A1$ and $A2$ in Table \ref{AxiomsForBATCDRT}), and this equivalence is denoted by $=_{AC}$. Then, each equivalence class $s$ modulo AC of $+$ has the following normal form

        $$s_1+\cdots+ s_k$$

        with each $s_i$ either an atomic event or of the form $t_1\cdot t_2$, and each $s_i$ is called the summand of $s$.

        Now, we prove that for normal forms $n$ and $n'$, if $n\sim_{s} n'$ then $n=_{AC}n'$. It is sufficient to induct on the sizes of $n$ and $n'$. We can get $n=_{AC} n'$.

        Finally, let $s$ and $t$ be basic terms, and $s\sim_s t$, there are normal forms $n$ and $n'$, such that $s=n$ and $t=n'$. The soundness theorem of $\textrm{BATC}^{\textrm{drt}}$ modulo step bisimulation equivalence yields $s\sim_s n$ and $t\sim_s n'$, so $n\sim_s s\sim_s t\sim_s n'$. Since if $n\sim_s n'$ then $n=_{AC}n'$, $s=n=_{AC}n'=t$, as desired.
  \item This case can be proven similarly, just by replacement of $\sim_{s}$ by $\sim_{p}$.
  \item This case can be proven similarly, just by replacement of $\sim_{s}$ by $\sim_{hp}$.
  \item This case can be proven similarly, just by replacement of $\sim_{s}$ by $\sim_{hhp}$.
\end{enumerate}
\end{proof}

\subsection{Algebra for Parallelism in True Concurrency with Discrete Relative Timing}

In this subsection, we will introduce $\textrm{APTC}^{\textrm{drt}}$.

\subsubsection{The Theory $\textrm{APTC}^{\textrm{drt}}$}

\begin{definition}[Signature of $\textrm{APTC}^{\textrm{drt}}$]
The signature of $\textrm{APTC}^{\textrm{drt}}$ consists of the signature of $\textrm{BATC}^{\textrm{drt}}$, and the whole parallel composition operator $\between: \mathcal{P}_{\textrm{rel}}\times\mathcal{P}_{\textrm{rel}} \rightarrow \mathcal{P}_{\textrm{rel}}$, the parallel operator $\parallel: \mathcal{P}_{\textrm{rel}}\times\mathcal{P}_{\textrm{rel}} \rightarrow \mathcal{P}_{\textrm{rel}}$, the communication merger operator $\mid: \mathcal{P}_{\textrm{rel}}\times\mathcal{P}_{\textrm{rel}} \rightarrow \mathcal{P}_{\textrm{rel}}$, and the encapsulation operator $\partial_H: \mathcal{P}_{\textrm{rel}} \rightarrow \mathcal{P}_{\textrm{rel}}$ for all $H\subseteq A$.
\end{definition}

The set of axioms of $\textrm{APTC}^{\textrm{drt}}$ consists of the laws given in Table \ref{AxiomsForAPTCDRT}.

\begin{center}
    \begin{table}
        \begin{tabular}{@{}ll@{}}
            \hline No. &Axiom\\
            $P1$ & $x\between y = x\parallel y + x\mid y$\\
            $P2$ & $x\parallel y = y \parallel x$\\
            $P3$ & $(x\parallel y)\parallel z = x\parallel (y\parallel z)$\\
            $P4DR$ & $\underline{\underline{a}}\parallel (\underline {\underline{b}}\cdot y) = (\underline{\underline{a}}\parallel \underline{\underline{b}})\cdot y$\\
            $P5DR$ & $(\underline{\underline{a}}\cdot x)\parallel \underline{\underline{b}} = (\underline{\underline{a}}\parallel \underline{\underline{b}})\cdot x$\\
            $P6DR$ & $(\underline{\underline{a}}\cdot x)\parallel (\underline{\underline{b}}\cdot y) = (\underline{\underline{a}}\parallel \underline{\underline{b}})\cdot (x\between y)$\\
            $P7$ & $(x+ y)\parallel z = (x\parallel z)+ (y\parallel z)$\\
            $P8$ & $x\parallel (y+ z) = (x\parallel y)+ (x\parallel z)$\\
            $DRP9ID$ & $(\upsilon^1_{\textrm{rel}}(x)+ \underline{\underline{\delta}})\parallel \sigma^{n+1}_{\textrm{rel}}(y) = \underline{\underline{\delta}}$\\
            $DRP10ID$ & $\sigma^{n+1}_{\textrm{rel}}(x)\parallel (\upsilon^1_{\textrm{rel}}(y)+ \underline{\underline{\delta}}) = \underline{\underline{\delta}}$\\
            $DRP11$ & $\sigma^n_{\textrm{rel}}(x) \parallel \sigma^n_{\textrm{rel}}(y) = \sigma^n_{\textrm{rel}}(x\parallel y)$\\
            $PID12$ & $\dot{\delta}\parallel x = \dot{\delta}$\\
            $PID13$ & $x\parallel \dot{\delta} = \dot{\delta}$\\
            $C14DR$ & $\underline{\underline{a}}\mid \underline{\underline{b}} = \gamma(\underline{\underline{a}}, \underline{\underline{b}})$\\
            $C15DR$ & $\underline{\underline{a}}\mid (\underline{\underline{b}}\cdot y) = \gamma(\underline{\underline{a}}, \underline{\underline{b}})\cdot y$\\
            $C16DR$ & $(\underline{\underline{a}}\cdot x)\mid \underline{\underline{b}} = \gamma(\underline{\underline{a}}, \underline{\underline{b}})\cdot x$\\
            $C17DR$ & $(\underline{\underline{a}}\cdot x)\mid (\underline{\underline{b}}\cdot y) = \gamma(\underline{\underline{a}}, \underline{\underline{b}})\cdot (x\between y)$\\
            $C18$ & $(x+ y)\mid z = (x\mid z) + (y\mid z)$\\
            $C19$ & $x\mid (y+ z) = (x\mid y)+ (x\mid z)$\\
            $DRC20ID$ & $(\upsilon^1_{\textrm{rel}}(x)+ \underline{\underline{\delta}})\mid \sigma^{n+1}_{\textrm{rel}}(y) = \underline{\underline{\delta}}$\\
            $DRC21ID$ & $\sigma^{n+1}_{\textrm{rel}}(x)\mid (\upsilon^1_{\textrm{rel}}(y)+ \underline{\underline{\delta}}) = \underline{\underline{\delta}}$\\
            $DRC22$ & $\sigma^n_{\textrm{rel}}(x) \mid \sigma^n_{\textrm{rel}}(y) = \sigma^n_{\textrm{rel}}(x\mid y)$\\
            $CID23$ & $\dot{\delta}\mid x = \dot{\delta}$\\
            $CID24$ & $x\mid\dot{\delta} = \dot{\delta}$\\
            $CE25DR$ & $\Theta(\underline{\underline{a}}) = \underline{\underline{a}}$\\
            $CE26DRID$ & $\Theta(\dot{\delta}) = \dot{\delta}$\\
            $CE27$ & $\Theta(x+ y) = \Theta(x)\triangleleft y + \Theta(y)\triangleleft x$\\
            $CE28$ & $\Theta(x\cdot y)=\Theta(x)\cdot\Theta(y)$\\
            $CE29$ & $\Theta(x\parallel y) = ((\Theta(x)\triangleleft y)\parallel y)+ ((\Theta(y)\triangleleft x)\parallel x)$\\
            $CE30$ & $\Theta(x\mid y) = ((\Theta(x)\triangleleft y)\mid y)+ ((\Theta(y)\triangleleft x)\mid x)$\\
            $U31DRID$ & $(\sharp(\underline{\underline{a}},\underline{\underline{b}}))\quad \underline{\underline{a}}\triangleleft \underline{\underline{b}} = \underline{\underline{\tau}}$\\
            $U32DRID$ & $(\sharp(\underline{\underline{a}},\underline{\underline{b}}),\underline{\underline{b}}\leq \underline{\underline{c}})\quad \underline{\underline{a}}\triangleleft \underline{\underline{c}} = \underline{\underline{a}}$\\
            $U33DRID$ & $(\sharp(\underline{\underline{a}},\underline{\underline{b}}),\underline{\underline{b}}\leq \underline{\underline{c}})\quad \underline{\underline{c}}\triangleleft \underline{\underline{a}} = \underline{\underline{\tau}}$\\
            $U34DRID$ & $\underline{\underline{a}}\triangleleft \underline{\underline{\delta}} = \underline{\underline{a}}$\\
            $U35DRID$ & $\underline{\underline{\delta}} \triangleleft \underline{\underline{a}} = \underline{\underline{\delta}}$\\
            $U36$ & $(x+ y)\triangleleft z = (x\triangleleft z)+ (y\triangleleft z)$\\
            $U37$ & $(x\cdot y)\triangleleft z = (x\triangleleft z)\cdot (y\triangleleft z)$\\
            $U38$ & $(x\parallel y)\triangleleft z = (x\triangleleft z)\parallel (y\triangleleft z)$\\
            $U39$ & $(x\mid y)\triangleleft z = (x\triangleleft z)\mid (y\triangleleft z)$\\
            $U40$ & $x\triangleleft (y+ z) = (x\triangleleft y)\triangleleft z$\\
            $U41$ & $x\triangleleft (y\cdot z)=(x\triangleleft y)\triangleleft z$\\
            $U42$ & $x\triangleleft (y\parallel z) = (x\triangleleft y)\triangleleft z$\\
            $U43$ & $x\triangleleft (y\mid z) = (x\triangleleft y)\triangleleft z$\\
            $D1DRID$ & $\underline{\underline{a}}\notin H\quad\partial_H(\underline{\underline{a}}) = \underline{\underline{a}}$\\
            $D2DRID$ & $\underline{\underline{a}}\in H\quad \partial_H(\underline{\underline{a}}) = \underline{\underline{\delta}}$\\
            $D3DRID$ & $\partial_H(\dot{\delta}) = \dot{\delta}$\\
            $D4$ & $\partial_H(x+ y) = \partial_H(x)+\partial_H(y)$\\
            $D5$ & $\partial_H(x\cdot y) = \partial_H(x)\cdot\partial_H(y)$\\
            $D6$ & $\partial_H(x\parallel y) = \partial_H(x)\parallel\partial_H(y)$\\
            $DRD7$ & $\partial_H(\sigma^n_{\textrm{rel}}(x)) = \sigma^n_{\textrm{rel}}(\partial_H(x))$\\
        \end{tabular}
        \caption{Axioms of $\textrm{APTC}^{\textrm{drt}}(a,b,c\in A_{\delta}, n\geq 0)$}
        \label{AxiomsForAPTCDRT}
    \end{table}
\end{center}

The operational semantics of $\textrm{APTC}^{\textrm{drt}}$ are defined by the transition rules in Table \ref{TRForAPTCDRT}.

\begin{center}
    \begin{table}
        $$\frac{x\xrightarrow{a}\surd\quad y\xrightarrow{b}\surd}{x\parallel y\xrightarrow{\{a,b\}}\surd} \quad\frac{x\xrightarrow{a}x'\quad y\xrightarrow{b}\surd}{x\parallel y\xrightarrow{\{a,b\}}x'}$$
        $$\frac{x\xrightarrow{a}\surd\quad y\xrightarrow{b}y'}{x\parallel y\xrightarrow{\{a,b\}}y'} \quad\frac{x\xrightarrow{a}x'\quad y\xrightarrow{b}y'}{x\parallel y\xrightarrow{\{a,b\}}x'\between y'}$$
        $$\frac{x\mapsto^{m}x'\quad y\mapsto^{m}y'}{x\parallel y\mapsto^{m}x'\parallel y'} \quad\frac{x\uparrow}{x\parallel y\uparrow} \quad\frac{y\uparrow}{x\parallel y\uparrow}$$
        $$\frac{x\xrightarrow{a}\surd\quad y\xrightarrow{b}\surd}{x\mid y\xrightarrow{\gamma(a,b)}\surd} \quad\frac{x\xrightarrow{a}x'\quad y\xrightarrow{b}\surd}{x\mid y\xrightarrow{\gamma(a,b)}x'}$$
        $$\frac{x\xrightarrow{a}\surd\quad y\xrightarrow{b}y'}{x\mid y\xrightarrow{\gamma(a,b)}y'} \quad\frac{x\xrightarrow{a}x'\quad y\xrightarrow{b}y'}{x\mid y\xrightarrow{\gamma(a,b)}x'\between y'}$$
        $$\frac{x\mapsto^{m}x'\quad y\mapsto^{m}y'}{x\mid y\mapsto^{m}x'\mid y'} \quad\frac{x\uparrow}{x\mid y\uparrow} \quad\frac{y\uparrow}{x\mid y\uparrow}$$
        $$\frac{x\xrightarrow{a}\surd\quad (\sharp(a,b))}{\Theta(x)\xrightarrow{a}\surd} \quad\frac{x\xrightarrow{b}\surd\quad (\sharp(a,b))}{\Theta(x)\xrightarrow{b}\surd}$$
        $$\frac{x\xrightarrow{a}x'\quad (\sharp(a,b))}{\Theta(x)\xrightarrow{a}\Theta(x')} \quad\frac{x\xrightarrow{b}x'\quad (\sharp(a,b))}{\Theta(x)\xrightarrow{b}\Theta(x')}$$
        $$\frac{x\mapsto^{m}x'}{\Theta(x)\mapsto^{m}\Theta(x')} \quad\frac{x\uparrow}{\Theta(x)\uparrow}$$
        $$\frac{x\xrightarrow{a}\surd \quad y\nrightarrow^{b}\quad (\sharp(a,b))}{x\triangleleft y\xrightarrow{\tau}\surd}
        \quad\frac{x\xrightarrow{a}x' \quad y\nrightarrow^{b}\quad (\sharp(a,b))}{x\triangleleft y\xrightarrow{\tau}x'}$$
        $$\frac{x\xrightarrow{a}\surd \quad y\nrightarrow^{c}\quad (\sharp(a,b),b\leq c)}{x\triangleleft y\xrightarrow{a}\surd}
        \quad\frac{x\xrightarrow{a}x' \quad y\nrightarrow^{c}\quad (\sharp(a,b),b\leq c)}{x\triangleleft y\xrightarrow{a}x'}$$
        $$\frac{x\xrightarrow{c}\surd \quad y\nrightarrow^{b}\quad (\sharp(a,b),a\leq c)}{x\triangleleft y\xrightarrow{\tau}\surd}
        \quad\frac{x\xrightarrow{c}x' \quad y\nrightarrow^{b}\quad (\sharp(a,b),a\leq c)}{x\triangleleft y\xrightarrow{\tau}x'}$$
        $$\frac{x\mapsto^{m}x'\quad y\mapsto^{m}y'}{x\triangleleft y\mapsto^{m}x'\triangleleft y'} \quad\frac{x\uparrow}{x\triangleleft y\uparrow}$$
        $$\frac{x\xrightarrow{a}\surd}{\partial_H(x)\xrightarrow{a}\surd}\quad (e\notin H)\quad\frac{x\xrightarrow{a}x'}{\partial_H(x)\xrightarrow{a}\partial_H(x')}\quad(e\notin H)$$
        $$\frac{x\mapsto^{m}x'}{\partial_H(x)\mapsto^{m}\partial_H(x')}\quad(e\notin H)\quad\frac{x\uparrow}{\partial_H(x)\uparrow}$$
    \caption{Transition rules of $\textrm{APTC}^{\textrm{drt}}(a,b,c\in A, m>0)$}
    \label{TRForAPTCDRT}
    \end{table}
\end{center}

\subsubsection{Elimination}

\begin{definition}[Basic terms of $\textrm{APTC}^{\textrm{drt}}$]
The set of basic terms of $\textrm{APTC}^{\textrm{drt}}$, $\mathcal{B}(\textrm{APTC}^{\textrm{drt}})$, is inductively defined as follows by two auxiliary sets $\mathcal{B}_0(\textrm{APTC}^{\textrm{drt}})$ and $\mathcal{B}_1(\textrm{APTC}^{\textrm{drt}})$:
\begin{enumerate}
  \item if $a\in A_{\delta}$, then $\underline{\underline{a}} \in \mathcal{B}_1(\textrm{APTC}^{\textrm{drt}})$;
  \item if $a\in A$ and $t\in \mathcal{B}(\textrm{APTC}^{\textrm{drt}})$, then $\underline{\underline{a}}\cdot t \in \mathcal{B}_1(\textrm{APTC}^{\textrm{drt}})$;
  \item if $t,t'\in \mathcal{B}_1(\textrm{APTC}^{\textrm{drt}})$, then $t+t'\in \mathcal{B}_1(\textrm{APTC}^{\textrm{drt}})$;
  \item if $t,t'\in \mathcal{B}_1(\textrm{APTC}^{\textrm{drt}})$, then $t\parallel t'\in \mathcal{B}_1(\textrm{APTC}^{\textrm{drt}})$;
  \item if $t\in \mathcal{B}_1(\textrm{APTC}^{\textrm{drt}})$, then $t\in \mathcal{B}_0(\textrm{APTC}^{\textrm{drt}})$;
  \item if $n>0$ and $t\in \mathcal{B}_0(\textrm{APTC}^{\textrm{drt}})$, then $\sigma^n_{\textrm{rel}}(t) \in \mathcal{B}_0(\textrm{APTC}^{\textrm{drt}})$;
  \item if $n>0$, $t\in \mathcal{B}_1(\textrm{APTC}^{\textrm{drt}})$ and $t'\in \mathcal{B}_0(\textrm{APTC}^{\textrm{drt}})$, then $t+\sigma^n_{\textrm{rel}}(t') \in \mathcal{B}_0(\textrm{APTC}^{\textrm{drt}})$;
  \item $\dot{\delta}\in \mathcal{B}(\textrm{APTC}^{\textrm{drt}})$;
  \item if $t\in \mathcal{B}_0(\textrm{APTC}^{\textrm{drt}})$, then $t\in \mathcal{B}(\textrm{APTC}^{\textrm{drt}})$.
\end{enumerate}
\end{definition}

\begin{theorem}[Elimination theorem]
Let $p$ be a closed $\textrm{APTC}^{\textrm{drt}}$ term. Then there is a basic $\textrm{APTC}^{\textrm{drt}}$ term $q$ such that $\textrm{APTC}^{\textrm{drt}}\vdash p=q$.
\end{theorem}

\begin{proof}
It is sufficient to induct on the structure of the closed $\textrm{APTC}^{\textrm{drt}}$ term $p$. It can be proven that $p$ combined by the constants and operators of $\textrm{APTC}^{\textrm{drt}}$ exists an equal basic term $q$, and the other operators not included in the basic terms, such as $\upsilon_{\textrm{rel}}$, $\overline{\upsilon}_{\textrm{rel}}$, $\between$, $\mid$, $\partial_H$, $\Theta$ and $\triangleleft$ can be eliminated.
\end{proof}

\subsubsection{Connections}

\begin{theorem}[Generalization of $\textrm{APTC}^{\textrm{drt}}$]
\begin{enumerate}
  \item By the definitions of $a=\underline{\underline{a}}$ for each $a\in A$ and $\delta=\underline{\underline{\delta}}$, $\textrm{APTC}^{\textrm{drt}}$ is a generalization of $APTC$.
  \item $\textrm{APTC}^{\textrm{drt}}$ is a generalization of $\textrm{BATC}^{\textrm{drt}}$¡£
\end{enumerate}

\end{theorem}

\begin{proof}
\begin{enumerate}
  \item It follows from the following two facts.

    \begin{enumerate}
      \item The transition rules of $APTC$ in section \ref{tcpa} are all source-dependent;
      \item The sources of the transition rules of $\textrm{APTC}^{\textrm{drt}}$ contain an occurrence of $\dot{\delta}$, $\underline{\underline{a}}$, $\sigma^n_{\textrm{rel}}$, $\upsilon^n_{\textrm{rel}}$ and $\overline{\upsilon}^n_{\textrm{rel}}$.
    \end{enumerate}

    So, $APTC$ is an embedding of $\textrm{APTC}^{\textrm{drt}}$, as desired.
    \item It follows from the following two facts.

    \begin{enumerate}
      \item The transition rules of $\textrm{BATC}^{\textrm{drt}}$ are all source-dependent;
      \item The sources of the transition rules of $\textrm{APTC}^{\textrm{drt}}$ contain an occurrence of $\between$, $\parallel$, $\mid$, $\Theta$, $\triangleleft$, $\partial_H$.
    \end{enumerate}

    So, $\textrm{BATC}^{\textrm{drt}}$ is an embedding of $\textrm{APTC}^{\textrm{drt}}$, as desired.
\end{enumerate}
\end{proof}

\subsubsection{Congruence}

\begin{theorem}[Congruence of $\textrm{APTC}^{\textrm{drt}}$]
Truly concurrent bisimulation equivalences $\sim_p$, $\sim_s$ and $\sim_{hp}$ are all congruences with respect to $\textrm{APTC}^{\textrm{drt}}$. That is,
\begin{itemize}
  \item pomset bisimulation equivalence $\sim_{p}$ is a congruence with respect to $\textrm{APTC}^{\textrm{drt}}$;
  \item step bisimulation equivalence $\sim_{s}$ is a congruence with respect to $\textrm{APTC}^{\textrm{drt}}$;
  \item hp-bisimulation equivalence $\sim_{hp}$ is a congruence with respect to $\textrm{APTC}^{\textrm{drt}}$.
\end{itemize}
\end{theorem}

\begin{proof}
It is easy to see that $\sim_p$, $\sim_s$, and $\sim_{hp}$ are all equivalent relations on $\textrm{APTC}^{\textrm{drt}}$ terms, it is only sufficient to prove that $\sim_p$, $\sim_s$, and $\sim_{hp}$ are all preserved by the operators $\sigma^n_{\textrm{rel}}$, $\upsilon^n_{\textrm{rel}}$ and $\overline{\upsilon}^n_{\textrm{rel}}$. It is trivial and we omit it.
\end{proof}

\subsubsection{Soundness}

\begin{theorem}[Soundness of $\textrm{APTC}^{\textrm{drt}}$]
The axiomatization of $\textrm{APTC}^{\textrm{drt}}$ is sound modulo truly concurrent bisimulation equivalences $\sim_{p}$, $\sim_{s}$, and $\sim_{hp}$. That is,
\begin{enumerate}
  \item let $x$ and $y$ be $\textrm{APTC}^{\textrm{drt}}$ terms. If $\textrm{APTC}^{\textrm{drt}}\vdash x=y$, then $x\sim_{s} y$;
  \item let $x$ and $y$ be $\textrm{APTC}^{\textrm{drt}}$ terms. If $\textrm{APTC}^{\textrm{drt}}\vdash x=y$, then $x\sim_{p} y$;
  \item let $x$ and $y$ be $\textrm{APTC}^{\textrm{drt}}$ terms. If $\textrm{APTC}^{\textrm{drt}}\vdash x=y$, then $x\sim_{hp} y$.
\end{enumerate}
\end{theorem}

\begin{proof}
Since $\sim_p$, $\sim_s$, and $\sim_{hp}$ are both equivalent and congruent relations, we only need to check if each axiom in Table \ref{AxiomsForAPTCDRT} is sound modulo $\sim_p$, $\sim_s$, and $\sim_{hp}$ respectively.

\begin{enumerate}
  \item We only check the soundness of the non-trivial axiom $DRP11$ modulo $\sim_s$.
        Let $p,q$ be $\textrm{APTC}^{\textrm{drt}}$ processes, and $\sigma^n_{\textrm{rel}}(p) \parallel \sigma^n_{\textrm{rel}}(q) = \sigma^n_{\textrm{rel}}(p\parallel q)$, it is sufficient to prove that $\sigma^n_{\textrm{rel}}(p) \parallel \sigma^n_{\textrm{rel}}(q) \sim_{s} \sigma^n_{\textrm{rel}}(p\parallel q)$. By the transition rules of operator $\sigma^n_{\textrm{rel}}$ and $\parallel$ in Table \ref{TRForBATCDRT}, we get

        $$\frac{}{\sigma^n_{\textrm{rel}}(p) \parallel \sigma^n_{\textrm{rel}}(q)\mapsto^n \sigma^0_{\textrm{rel}}(p) \parallel \sigma^0_{\textrm{rel}}(q)}$$

        $$\frac{}{\sigma^n_{\textrm{rel}}(p\parallel q)\mapsto^n \sigma^0_{\textrm{rel}}(p\parallel q)}$$

        There are several cases:

        $$\frac{p\xrightarrow{a} \surd\quad q\xrightarrow{b}\surd}{\sigma^0_{\textrm{rel}}(p) \parallel \sigma^0_{\textrm{rel}}(q)\xrightarrow{\{a,b\}}\surd}$$

        $$\frac{p\xrightarrow{a} \surd\quad q\xrightarrow{b}\surd}{\sigma^0_{\textrm{rel}}(p\parallel q)\xrightarrow{\{a,b\}}\surd}$$

        $$\frac{p\xrightarrow{a} p'\quad q\xrightarrow{b}\surd}{\sigma^0_{\textrm{rel}}(p) \parallel \sigma^0_{\textrm{rel}}(q)\xrightarrow{\{a,b\}}p'}$$

        $$\frac{p\xrightarrow{a} p'\quad q\xrightarrow{b}\surd}{\sigma^0_{\textrm{rel}}(p\parallel q)\xrightarrow{\{a,b\}}p'}$$

        $$\frac{p\xrightarrow{a} \surd\quad q\xrightarrow{b}q'}{\sigma^0_{\textrm{rel}}(p) \parallel \sigma^0_{\textrm{rel}}(q)\xrightarrow{\{a,b\}}q'}$$

        $$\frac{p\xrightarrow{a} \surd\quad q\xrightarrow{b}q'}{\sigma^0_{\textrm{rel}}(p\parallel q)\xrightarrow{\{a,b\}}q'}$$

        $$\frac{p\xrightarrow{a} p'\quad q\xrightarrow{b}q'}{\sigma^0_{\textrm{rel}}(p) \parallel \sigma^0_{\textrm{rel}}(q)\xrightarrow{\{a,b\}}p'\between q'}$$

        $$\frac{p\xrightarrow{a} p'\quad q\xrightarrow{b}q'}{\sigma^0_{\textrm{rel}}(p\parallel q)\xrightarrow{\{a,b\}}p'\between q'}$$

        $$\frac{p \uparrow}{\sigma^0_{\textrm{rel}}(p) \parallel \sigma^0_{\textrm{rel}}(q)\uparrow}$$

        $$\frac{p\uparrow}{\sigma^0_{\textrm{rel}}(p\parallel q)\uparrow}$$

        $$\frac{q \uparrow}{\sigma^0_{\textrm{rel}}(p) \parallel \sigma^0_{\textrm{rel}}(q)\uparrow}$$

        $$\frac{q\uparrow}{\sigma^0_{\textrm{rel}}(p\parallel q)\uparrow}$$

        So, we see that each case leads to $\sigma^n_{\textrm{rel}}(p) \parallel \sigma^n_{\textrm{rel}}(q) \sim_{s} \sigma^n_{\textrm{rel}}(p\parallel q)$, as desired.
  \item From the definition of pomset bisimulation, we know that pomset bisimulation is defined by pomset transitions, which are labeled by pomsets. In a pomset transition, the events (actions) in the pomset are either within causality relations (defined by $\cdot$) or in concurrency (implicitly defined by $\cdot$ and $+$, and explicitly defined by $\between$), of course, they are pairwise consistent (without conflicts). We have already proven the case that all events are pairwise concurrent (soundness modulo step bisimulation), so, we only need to prove the case of events in causality. Without loss of generality, we take a pomset of $P=\{\underline{\underline{a}},\underline{\underline{b}}:\underline{\underline{a}}\cdot \underline{\underline{b}}\}$. Then the pomset transition labeled by the above $P$ is just composed of one single event transition labeled by $\underline{\underline{a}}$ succeeded by another single event transition labeled by $\underline{\underline{b}}$, that is, $\xrightarrow{P}=\xrightarrow{a}\xrightarrow{b}$.

        Similarly to the proof of soundness modulo step bisimulation equivalence, we can prove that each axiom in Table \ref{AxiomsForAPTCDRT} is sound modulo pomset bisimulation equivalence, we omit them.
  \item From the definition of hp-bisimulation, we know that hp-bisimulation is defined on the posetal product $(C_1,f,C_2),f:C_1\rightarrow C_2\textrm{ isomorphism}$. Two process terms $s$ related to $C_1$ and $t$ related to $C_2$, and $f:C_1\rightarrow C_2\textrm{ isomorphism}$. Initially, $(C_1,f,C_2)=(\emptyset,\emptyset,\emptyset)$, and $(\emptyset,\emptyset,\emptyset)\in\sim_{hp}$. When $s\xrightarrow{a}s'$ ($C_1\xrightarrow{a}C_1'$), there will be $t\xrightarrow{a}t'$ ($C_2\xrightarrow{a}C_2'$), and we define $f'=f[a\mapsto a]$. Then, if $(C_1,f,C_2)\in\sim_{hp}$, then $(C_1',f',C_2')\in\sim_{hp}$.

        Similarly to the proof of soundness modulo pomset bisimulation equivalence, we can prove that each axiom in Table \ref{AxiomsForAPTCDRT} is sound modulo hp-bisimulation equivalence, we just need additionally to check the above conditions on hp-bisimulation, we omit them.
\end{enumerate}

\end{proof}

\subsubsection{Completeness}

\begin{theorem}[Completeness of $\textrm{APTC}^{\textrm{drt}}$]
The axiomatization of $\textrm{APTC}^{\textrm{drt}}$ is complete modulo truly concurrent bisimulation equivalences $\sim_{p}$, $\sim_{s}$, and $\sim_{hp}$. That is,
\begin{enumerate}
  \item let $p$ and $q$ be closed $\textrm{APTC}^{\textrm{drt}}$ terms, if $p\sim_{s} q$ then $p=q$;
  \item let $p$ and $q$ be closed $\textrm{APTC}^{\textrm{drt}}$ terms, if $p\sim_{p} q$ then $p=q$;
  \item let $p$ and $q$ be closed $\textrm{APTC}^{\textrm{drt}}$ terms, if $p\sim_{hp} q$ then $p=q$.
\end{enumerate}

\end{theorem}

\begin{proof}
\begin{enumerate}
  \item Firstly, by the elimination theorem of $\textrm{APTC}^{\textrm{drt}}$, we know that for each closed $\textrm{APTC}^{\textrm{drt}}$ term $p$, there exists a closed basic $\textrm{APTC}^{\textrm{drt}}$ term $p'$, such that $\textrm{APTC}^{\textrm{drt}}\vdash p=p'$, so, we only need to consider closed basic $\textrm{APTC}^{\textrm{drt}}$ terms.

        The basic terms modulo associativity and commutativity (AC) of conflict $+$ (defined by axioms $A1$ and $A2$ in Table \ref{AxiomsForBATCDRT}) and associativity and commutativity (AC) of parallel $\parallel$ (defined by axioms $P2$ and $P3$ in Table \ref{AxiomsForAPTCDRT}), and these equivalences is denoted by $=_{AC}$. Then, each equivalence class $s$ modulo AC of $+$ and $\parallel$ has the following normal form

        $$s_1+\cdots+ s_k$$

        with each $s_i$ either an atomic event or of the form

        $$t_1\cdot\cdots\cdot t_m$$

        with each $t_j$ either an atomic event or of the form

        $$u_1\parallel\cdots\parallel u_n$$

        with each $u_l$ an atomic event, and each $s_i$ is called the summand of $s$.

        Now, we prove that for normal forms $n$ and $n'$, if $n\sim_{s} n'$ then $n=_{AC}n'$. It is sufficient to induct on the sizes of $n$ and $n'$. We can get $n=_{AC} n'$.

        Finally, let $s$ and $t$ be basic $\textrm{APTC}^{\textrm{drt}}$ terms, and $s\sim_s t$, there are normal forms $n$ and $n'$, such that $s=n$ and $t=n'$. The soundness theorem modulo step bisimulation equivalence yields $s\sim_s n$ and $t\sim_s n'$, so $n\sim_s s\sim_s t\sim_s n'$. Since if $n\sim_s n'$ then $n=_{AC}n'$, $s=n=_{AC}n'=t$, as desired.
  \item This case can be proven similarly, just by replacement of $\sim_{s}$ by $\sim_{p}$.
  \item This case can be proven similarly, just by replacement of $\sim_{s}$ by $\sim_{hp}$.
\end{enumerate}
\end{proof}

\section{Discrete Absolute Timing}{\label{dat}}

In this section, we will introduce a version of APTC with absolute timing and time measured on a discrete time scale. Measuring time on  a discrete time scale means that time is divided into time slices and timing of actions is done with respect to the time slices in which they are performed. While in absolute timing, all timing is counted from the start of the whole process.

Like APTC without timing, let us start with a basic algebra for true concurrency called $\textrm{BATC}^{\textrm{dat}}$ (BATC with discrete absolute timing). Then we continue with $\textrm{APTC}^{\textrm{dat}}$ (APTC with discrete absolute timing).

\subsection{Basic Definitions}

In this subsection, we will introduce some basic definitions about timing. These basic concepts come from \cite{T3}, we introduce them into the backgrounds of true concurrency.

\begin{definition}[Undelayable actions]
Undelayable actions are defined as atomic processes that perform an action in the current time slice and then terminate successfully. We use a constant $\underline{a}$ to represent the undelayable action, that is, the atomic process that performs the action $a$ in the current time slice and then terminates successfully.
\end{definition}

\begin{definition}[Undelayable deadlock]
Undelayable deadlock $\underline{\delta}$ is an additional process that is neither capable of performing any action nor capable of idling till after time slice 1.
\end{definition}

\begin{definition}[Absolute delay]
The absolute delay of the process $p$ for $n$ ($n\in\mathbb{N}$) time slices is the process that idles $n$ time slices longer than $p$ and otherwise behaves like $p$. The operator $\sigma_{\textrm{abs}}$ is used to represent the absolute delay, and let $\sigma^n_{\textrm{abs}}(t) = n \sigma_{\textrm{abs}} t$.
\end{definition}

\begin{definition}[Deadlocked process]
Deadlocked process $\dot{\delta}$ is an additional process that has deadlocked before time slice 1. After a delay of one time slice, the undelayable deadlock $\underline{\delta}$ and the deadlocked process $\dot{\delta}$ are indistinguishable from each other.
\end{definition}

\begin{definition}[Truly concurrent bisimulation equivalences with time-related capabilities]\label{TBTTC2}
The following requirement with time-related capabilities is added to truly concurrent bisimulation equivalences $\sim_{p}$, $\sim_{s}$, $\sim_{hp}$ and $\sim_{hhp}$ and Definition \ref{TBTTC1}:
\begin{itemize}
  \item in case of absolute timing, the requirements in Definition \ref{TBTTC1} apply to the capabilities in a certain time slice.
\end{itemize}
\end{definition}

\begin{definition}[Absolute time-out]
The absolute time-out $\upsilon_{\textrm{abs}}$ of a process $p$ at time $n$ ($n\in\mathbb{N}$) behaves either like the part of $p$ that does not idle till time slice $n+1$, or like the deadlocked process after a delay of $n$ time slices if $p$ is capable of idling till time slice $n+1$; otherwise, like $p$. And let $\upsilon^n_{\textrm{abs}}(t) = n \upsilon_{\textrm{abs}} t$.
\end{definition}

\begin{definition}[Absolute initialization]
The absolute initialization $\overline{\upsilon}_{\textrm{abs}}$ of a process $p$ at time $n$ ($n\in\mathbb{N}$) behaves like the part of $p$ that idles till time slice $n+1$ if $p$ is capable of idling till that time slice; otherwise, like the deadlocked process after a delay of $n$ time slices. And we let $\overline{\upsilon}^n_{\textrm{abs}}(t) = n \overline{\upsilon}_{\textrm{abs}} t$.
\end{definition}

\subsection{Basic Algebra for True Concurrency with Discrete Absolute Timing}

In this subsection, we will introduce the theory $\textrm{BATC}^{\textrm{dat}}$.

\subsubsection{The Theory $\textrm{BATC}^{\textrm{dat}}$}

\begin{definition}[Signature of $\textrm{BATC}^{\textrm{dat}}$]
The signature of $\textrm{BATC}^{\textrm{dat}}$ consists of the sort $\mathcal{P}_{\textrm{abs}}$ of processes with discrete absolute timing, the undelayable action constants $\underline{a}: \rightarrow\mathcal{P}_{\textrm{abs}}$ for each $a\in A$, the undelayable deadlock constant $\underline{\delta}: \rightarrow \mathcal{P}_{\textrm{abs}}$, the alternative composition operator $+: \mathcal{P}_{\textrm{abs}}\times\mathcal{P}_{\textrm{abs}} \rightarrow \mathcal{P}_{\textrm{abs}}$, the sequential composition operator $\cdot: \mathcal{P}_{\textrm{abs}} \times \mathcal{P}_{\textrm{abs}} \rightarrow \mathcal{P}_{\textrm{abs}}$, the absolute delay operator $\sigma_{\textrm{abs}}: \mathbb{N}\times \mathcal{P}_{\textrm{abs}} \rightarrow \mathcal{P}_{\textrm{abs}}$, the deadlocked process constant $\dot{\delta}: \rightarrow \mathcal{P}_{\textrm{abs}}$, the absolute time-out operator $\upsilon_{\textrm{abs}}: \mathbb{N}\times\mathcal{P}_{\textrm{abs}} \rightarrow\mathcal{P}_{\textrm{abs}}$ and the absolute initialization operator $\overline{\upsilon}_{\textrm{abs}}: \mathbb{N}\times\mathcal{P}_{\textrm{abs}} \rightarrow\mathcal{P}_{\textrm{abs}}$.
\end{definition}

The set of axioms of $\textrm{BATC}^{\textrm{dat}}$ consists of the laws given in Table \ref{AxiomsForBATCDAT}.

\begin{center}
    \begin{table}
        \begin{tabular}{@{}ll@{}}
            \hline No. &Axiom\\
            $A1$ & $x+ y = y+ x$\\
            $A2$ & $(x+ y)+ z = x+ (y+ z)$\\
            $A3$ & $x+ x = x$\\
            $A4$ & $(x+ y)\cdot z = x\cdot z + y\cdot z$\\
            $A5$ & $(x\cdot y)\cdot z = x\cdot(y\cdot z)$\\
            $A6ID$ & $x + \dot{\delta} = x$\\
            $A7ID$ & $\dot{\delta}\cdot x = \dot{\delta}$\\
            $DAT1$ & $\sigma^0_{\textrm{abs}}(x) = \overline{\upsilon}^0_{\textrm{abs}}(x)$\\
            $DAT2$ & $\sigma^m_{\textrm{abs}}( \sigma^n_{\textrm{abs}}(x)) = \sigma^{m+n}_{\textrm{abs}}(x)$\\
            $DAT3$ & $\sigma^n_{\textrm{abs}}(x) + \sigma^n_{\textrm{abs}}(y) = \sigma^n_{\textrm{abs}}(x+y)$\\
            $DAT4$ & $\sigma^n_{\textrm{abs}}(x)\cdot \upsilon^n_{\textrm{abs}}(y) = \sigma^n_{\textrm{abs}}(x\cdot \dot{\delta})$\\
            $DAT5$ & $\sigma^n_{\textrm{abs}}(x)\cdot (\upsilon^n_{\textrm{abs}}(y) +\sigma^n_{\textrm{abs}}(z)) = \sigma^n_{\textrm{abs}}(x\cdot\overline{\upsilon}^0_{\textrm{abs}}(z))$\\
            $DAT6$ & $\sigma^n_{\textrm{abs}}(\dot{\delta})\cdot x = \sigma^n_{\textrm{abs}}(\dot{\delta})$\\
            $DAT7$ & $\sigma^1_{\textrm{abs}}(\dot{\delta}) = \underline{\delta}$\\
            $A6DAa$ & $\underline{a} + \underline{\delta} = \underline{a}$\\
            $DATO0$ & $\upsilon^n_{\textrm{abs}}(\dot{\delta}) = \dot{\delta}$\\
            $DATO1$ & $\upsilon^0_{\textrm{abs}}(x) = \dot(\delta)$\\
            $DATO2$ & $\upsilon^{n+1}_{\textrm{abs}}(\underline{a}) = \underline{a}$\\
            $DATO3$ & $\upsilon^{m+n}_{\textrm{abs}} (\sigma^n_{\textrm{abs}}(x)) = \sigma^n_{\textrm{abs}}(\upsilon^m_{\textrm{abs}}(x))$\\
            $DATO4$ & $\upsilon^n_{\textrm{abs}}(x+y) = \upsilon^n_{\textrm{abs}}(x) + \upsilon^n_{\textrm{abs}}(y)$\\
            $DATO5$ & $\upsilon^n_{\textrm{abs}}(x\cdot y) = \upsilon^n_{\textrm{abs}}(x)\cdot y$\\
            $DAI0$ & $\overline{\upsilon}^0_{\textrm{abs}}(\dot{\delta}) = \dot{\delta}$\\
            $DAI1$ & $\overline{\upsilon}^{n+1}_{\textrm{abs}}(\dot{\delta}) = \sigma^{n+1}_{\textrm{abs}}(\dot{\delta})$\\
            $DAI2$ & $\overline{\upsilon}^{n+1}_{\textrm{abs}}(\underline{a}) = \sigma^{n+1}_{\textrm{abs}}(\dot{\delta})$\\
            $DAI3$ & $\overline{\upsilon}^{m+n}_{\textrm{abs}} (\sigma^n_{\textrm{abs}}(x)) = \sigma^n_{\textrm{abs}}(\overline{\upsilon}^m_{\textrm{abs}}(\overline{\upsilon}^0_{\textrm{abs}}(x)))$\\
            $DAI4$ & $\overline{\upsilon}^n_{\textrm{abs}}(x+y) = \overline{\upsilon}^n_{\textrm{abs}}(x) + \overline{\upsilon}^n_{\textrm{abs}}(y)$\\
            $DAI5$ & $\overline{\upsilon}^n_{\textrm{abs}}(x\cdot y) = \overline{\upsilon}^n_{\textrm{abs}}(x)\cdot y$\\
        \end{tabular}
        \caption{Axioms of $\textrm{BATC}^{\textrm{dat}}(a\in A_{\delta}, m,n\geq 0)$}
        \label{AxiomsForBATCDAT}
    \end{table}
\end{center}

The operational semantics of $\textrm{BATC}^{\textrm{dat}}$ are defined by the transition rules in Table \ref{TRForBATCDAT}. The transition rules are defined on $\langle t, n\rangle$, where $t$ is a term and $n\in\mathbb{N}$. Where $\uparrow$ is a unary deadlocked predicate, and $\langle t, n\rangle\nuparrow \triangleq \neg(\langle t, n\rangle\uparrow)$; $\langle t,n\rangle\mapsto^m \langle t',n'\rangle$ means that process $t$ is capable of first idling till the $m$th-next time slice, and then proceeding as process $t'$ and $m+n=n'$.

\begin{center}
    \begin{table}
        $$\frac{}{\langle\dot{\delta},n\rangle\uparrow}\quad\frac{}{\langle \underline{\delta},n+1\rangle\uparrow}
        \quad\frac{}{\langle\underline{a},0\rangle\xrightarrow{a}\langle\surd,0\rangle} \quad \frac{}{\langle \underline{a},n+1\rangle\uparrow}$$

        $$\frac{\langle x,n\rangle\xrightarrow{a}\langle x',n\rangle}{\langle\sigma^0_{\textrm{abs}}(x),n\rangle \xrightarrow{a}\langle x',n\rangle}
        \quad\frac{\langle x,n\rangle\xrightarrow{a}\langle x',n\rangle}{\langle\sigma^m_{\textrm{abs}}(x),n+m\rangle \xrightarrow{a}\langle \sigma^m_{\textrm{abs}}(x'),n+m\rangle}$$

        $$\frac{\langle x,n\rangle\xrightarrow{a}\langle\surd,n\rangle}{\langle\sigma^{n'}_{\textrm{abs}}(x),n+n'\rangle \xrightarrow{a}\langle\surd,n+n'\rangle}
        \quad\frac{\langle x,n\rangle\uparrow}{\langle\sigma^{n'}_{\textrm{abs}}(x),n+n'\rangle\uparrow}$$

        $$\frac{}{\langle\sigma^{n'+m}_{\textrm{abs}}(x),n\rangle\mapsto^{m} \langle\sigma^{n'+m}_{\textrm{abs}}(x),n+m\rangle}(n'>n)
        \quad \frac{\langle x,0\rangle\nuparrow}{\langle\sigma^{n'+m}_{\textrm{abs}}(x),n'\rangle\mapsto^m \sigma^{n'+m}_{\textrm{abs}}(x),n'+m\rangle}$$

        $$\frac{\langle x,n\rangle\mapsto^m \langle x, n+m\rangle}{\langle\sigma^{n'}_{\textrm{abs}}(x),n+n'\rangle \mapsto^{m} \langle\sigma^{n'}_{\textrm{abs}}(x),n+n'+m\rangle}$$

        $$\frac{\langle x,n\rangle\mapsto^m \langle x, n+m\rangle}{\langle\sigma^{n'}_{\textrm{abs}}(x),n\rangle \mapsto^{n'+m} \langle\sigma^{n'}_{\textrm{abs}}(x),n+n'+m\rangle}$$

        $$\frac{\langle x,n\rangle\xrightarrow{a}\langle x',n\rangle}{\langle x+ y,n\rangle\xrightarrow{a}\langle x',n\rangle}
        \quad\frac{\langle y,n\rangle\xrightarrow{a}\langle y',n\rangle}{\langle x+ y,n\rangle\xrightarrow{a}\langle y',n\rangle}$$

        $$\frac{\langle x,n\rangle\xrightarrow{a}\langle \surd,n\rangle}{\langle x+ y,n\rangle\xrightarrow{a}\langle \surd,n\rangle}
        \quad\frac{\langle y,n\rangle\xrightarrow{a}\langle \surd,n\rangle}{\langle x+ y,n\rangle\xrightarrow{a}\langle \surd,n\rangle}$$

        $$\frac{\langle x,n\rangle\mapsto^{m}\langle x,n+m\rangle}{\langle x+ y,n\rangle\mapsto^{m}\langle x+y,n+m\rangle} \quad\frac{\langle y,n\rangle\mapsto^m \langle y,n+m\rangle}{\langle x+ y,n\rangle\mapsto^{m}\langle x+y,n+m\rangle}
        \quad\frac{\langle x,n\rangle\uparrow\quad \langle y,n\rangle\uparrow}{\langle x+ y,n\rangle\uparrow}$$

        $$\frac{\langle x,n\rangle\xrightarrow{a}\langle\surd,n\rangle}{\langle x\cdot y,n\rangle\xrightarrow{a} \langle y,n\rangle}
        \quad\frac{\langle x,n\rangle\xrightarrow{a}\langle x',n\rangle}{\langle x\cdot y,n\rangle\xrightarrow{a}\langle x'\cdot y,n\rangle}$$

        $$\frac{\langle x,n\rangle\mapsto^{m}\langle x,n+m\rangle}{\langle x\cdot y,n\rangle\mapsto^{m}\langle x\cdot y,n+m\rangle}
        \quad \frac{\langle x,n\rangle\uparrow}{\langle x\cdot y,n\rangle\uparrow}$$

        $$\frac{\langle x,n\rangle\xrightarrow{a}\langle x',n\rangle}{\langle\upsilon^{n'}_{\textrm{abs}}(x),n\rangle \xrightarrow{a}\langle x',n\rangle}(n'>n)
        \quad\frac{\langle x,n\rangle\xrightarrow{a}\langle\surd,n\rangle}{\langle\upsilon^{n'}_{\textrm{abs}}(x),n\rangle \xrightarrow{a}\langle\surd,n\rangle}(n'>n)$$

        $$\frac{\langle x,n\rangle\mapsto^m \langle x,n+m\rangle}{\langle\upsilon^{n'}_{\textrm{abs}}(x),n\rangle \mapsto^m \langle\upsilon^{n'}_{\textrm{abs}}(x),n+m\rangle}(n'>n+m)$$

        $$\frac{}{\langle\upsilon^{n'}_{\textrm{abs}}(x),n\rangle\uparrow}(n'\leq n)
        \quad\frac{\langle x,n\rangle\uparrow}{\langle\upsilon^{n'}_{\textrm{abs}}(x),n\rangle\uparrow}(n'>n)$$

        $$\frac{\langle x,n\rangle\xrightarrow{a}\langle x',n\rangle}{\langle\overline{\upsilon}^{n'}_{\textrm{abs}}(x),n\rangle \xrightarrow{a}\langle x',n\rangle}(n'\leq n)
        \quad\frac{\langle x,n\rangle\xrightarrow{a}\langle\surd,n\rangle}{\langle\overline{\upsilon}^{n'}_{\textrm{abs}}(x),n\rangle \xrightarrow{a}\langle\surd,n\rangle}(n'\leq n)$$

        $$\frac{}{\langle\overline{\upsilon}^{n'+m}_{\textrm{abs}}(x),n\rangle\mapsto^m \langle\overline{\upsilon}^{n'+m}_{\textrm{abs}}(x),n+m}(n'>n)$$

        $$\frac{\langle x,n'+m\rangle\nuparrow}{\langle\overline{\upsilon}^{n'+m}_{\textrm{abs}}(x),n'\rangle \mapsto^m \langle\overline{\upsilon}^{n'+m}_{\textrm{abs}}(x),n'+m\rangle}$$

        $$\frac{\langle x,n\rangle\mapsto^m\langle x,n+m\rangle}{\langle\overline{\upsilon}^{n'}_{\textrm{abs}}(x),n\rangle \mapsto^m \langle\overline{\upsilon}^{n'}_{\textrm{abs}}(x),n+m\rangle}(n'\leq n+m)
        \quad\frac{\langle x,n\rangle\uparrow}{\langle\overline{\upsilon}^{n'}_{\textrm{abs}}(x),n\rangle\uparrow}(n'\leq n)$$
        \caption{Transition rules of $\textrm{BATC}^{\textrm{dat}}(a\in A, m>0, n,n'\geq 0)$}
        \label{TRForBATCDAT}
    \end{table}
\end{center}

\subsubsection{Elimination}

\begin{definition}[Basic terms of $\textrm{BATC}^{\textrm{dat}}$]
The set of basic terms of $\textrm{BATC}^{\textrm{dat}}$, $\mathcal{B}(\textrm{BATC}^{\textrm{dat}})$, is inductively defined as follows by two auxiliary sets $\mathcal{B}_0(\textrm{BATC}^{\textrm{dat}})$ and $\mathcal{B}_1(\textrm{BATC}^{\textrm{dat}})$:
\begin{enumerate}
  \item if $a\in A_{\delta}$, then $\underline{a} \in \mathcal{B}_1(\textrm{BATC}^{\textrm{dat}})$;
  \item if $a\in A$ and $t\in \mathcal{B}(\textrm{BATC}^{\textrm{dat}})$, then $\underline{a}\cdot t \in \mathcal{B}_1(\textrm{BATC}^{\textrm{dat}})$;
  \item if $t,t'\in \mathcal{B}_1(\textrm{BATC}^{\textrm{dat}})$, then $t+t'\in \mathcal{B}_1(\textrm{BATC}^{\textrm{dat}})$;
  \item if $t\in \mathcal{B}_1(\textrm{BATC}^{\textrm{dat}})$, then $t\in \mathcal{B}_0(\textrm{BATC}^{\textrm{dat}})$;
  \item if $n>0$ and $t\in \mathcal{B}_0(\textrm{BATC}^{\textrm{dat}})$, then $\sigma^n_{\textrm{abs}}(t) \in \mathcal{B}_0(\textrm{BATC}^{\textrm{dat}})$;
  \item if $n>0$, $t\in \mathcal{B}_1(\textrm{BATC}^{\textrm{dat}})$ and $t'\in \mathcal{B}_0(\textrm{BATC}^{\textrm{dat}})$, then $t+\sigma^n_{\textrm{abs}}(t') \in \mathcal{B}_0(\textrm{BATC}^{\textrm{dat}})$;
  \item $\dot{\delta}\in \mathcal{B}(\textrm{BATC}^{\textrm{dat}})$;
  \item if $t\in \mathcal{B}_0(\textrm{BATC}^{\textrm{dat}})$, then $t\in \mathcal{B}(\textrm{BATC}^{\textrm{dat}})$.
\end{enumerate}
\end{definition}

\begin{theorem}[Elimination theorem]
Let $p$ be a closed $\textrm{BATC}^{\textrm{dat}}$ term. Then there is a basic $\textrm{BATC}^{\textrm{dat}}$ term $q$ such that $\textrm{BATC}^{\textrm{dat}}\vdash p=q$.
\end{theorem}

\begin{proof}
It is sufficient to induct on the structure of the closed $\textrm{BATC}^{\textrm{dat}}$ term $p$. It can be proven that $p$ combined by the constants and operators of $\textrm{BATC}^{\textrm{dat}}$ exists an equal basic term $q$, and the other operators not included in the basic terms, such as $\upsilon_{\textrm{abs}}$ and $\overline{\upsilon}_{\textrm{abs}}$ can be eliminated.
\end{proof}

\subsubsection{Connections}

\begin{theorem}[Generalization of $\textrm{BATC}^{\textrm{dat}}$]

By the definitions of $a=\underline{a}$ for each $a\in A$ and $\delta=\underline{\delta}$, $\textrm{BATC}^{\textrm{dat}}$ is a generalization of $BATC$.
\end{theorem}

\begin{proof}
It follows from the following two facts.

\begin{enumerate}
  \item The transition rules of $BATC$ in section \ref{tcpa} are all source-dependent;
  \item The sources of the transition rules of $\textrm{BATC}^{\textrm{dat}}$ contain an occurrence of $\dot{\delta}$, $\underline{a}$, $\sigma^n_{\textrm{abs}}$, $\upsilon^n_{\textrm{abs}}$ and $\overline{\upsilon}^n_{\textrm{abs}}$.
\end{enumerate}

So, $BATC$ is an embedding of $\textrm{BATC}^{\textrm{dat}}$, as desired.
\end{proof}

\subsubsection{Congruence}

\begin{theorem}[Congruence of $\textrm{BATC}^{\textrm{dat}}$]
Truly concurrent bisimulation equivalences are all congruences with respect to $\textrm{BATC}^{\textrm{dat}}$. That is,
\begin{itemize}
  \item pomset bisimulation equivalence $\sim_{p}$ is a congruence with respect to $\textrm{BATC}^{\textrm{dat}}$;
  \item step bisimulation equivalence $\sim_{s}$ is a congruence with respect to $\textrm{BATC}^{\textrm{dat}}$;
  \item hp-bisimulation equivalence $\sim_{hp}$ is a congruence with respect to $\textrm{BATC}^{\textrm{dat}}$;
  \item hhp-bisimulation equivalence $\sim_{hhp}$ is a congruence with respect to $\textrm{BATC}^{\textrm{dat}}$.
\end{itemize}
\end{theorem}

\begin{proof}
It is easy to see that $\sim_p$, $\sim_s$, $\sim_{hp}$ and $\sim_{hhp}$ are all equivalent relations on $\textrm{BATC}^{\textrm{dat}}$ terms, it is only sufficient to prove that $\sim_p$, $\sim_s$, $\sim_{hp}$ and $\sim_{hhp}$ are all preserved by the operators $\sigma^n_{\textrm{abs}}$, $\upsilon^n_{\textrm{abs}}$ and $\overline{\upsilon}^n_{\textrm{abs}}$. It is trivial and we omit it.
\end{proof}

\subsubsection{Soundness}

\begin{theorem}[Soundness of $\textrm{BATC}^{\textrm{dat}}$]
The axiomatization of $\textrm{BATC}^{\textrm{dat}}$ is sound modulo truly concurrent bisimulation equivalences $\sim_{p}$, $\sim_{s}$, $\sim_{hp}$ and $\sim_{hhp}$. That is,
\begin{enumerate}
  \item let $x$ and $y$ be $\textrm{BATC}^{\textrm{dat}}$ terms. If $\textrm{BATC}^{\textrm{dat}}\vdash x=y$, then $x\sim_{s} y$;
  \item let $x$ and $y$ be $\textrm{BATC}^{\textrm{dat}}$ terms. If $\textrm{BATC}^{\textrm{dat}}\vdash x=y$, then $x\sim_{p} y$;
  \item let $x$ and $y$ be $\textrm{BATC}^{\textrm{dat}}$ terms. If $\textrm{BATC}^{\textrm{dat}}\vdash x=y$, then $x\sim_{hp} y$;
  \item let $x$ and $y$ be $\textrm{BATC}^{\textrm{dat}}$ terms. If $\textrm{BATC}^{\textrm{dat}}\vdash x=y$, then $x\sim_{hhp} y$.
\end{enumerate}
\end{theorem}

\begin{proof}
Since $\sim_p$, $\sim_s$, $\sim_{hp}$ and $\sim_{hhp}$ are both equivalent and congruent relations, we only need to check if each axiom in Table \ref{AxiomsForBATCDAT} is sound modulo $\sim_p$, $\sim_s$, $\sim_{hp}$ and $\sim_{hhp}$ respectively.

\begin{enumerate}
  \item We only check the soundness of the non-trivial axiom $DATO3$ modulo $\sim_s$.
        Let $p$ be $\textrm{BATC}^{\textrm{dat}}$ processes, and $\upsilon^{m+n}_{\textrm{abs}} (\sigma^n_{\textrm{abs}}(p)) = \sigma^n_{\textrm{abs}}(\upsilon^m_{\textrm{abs}}(p))$, it is sufficient to prove that $\upsilon^{m+n}_{\textrm{abs}} (\sigma^n_{\textrm{abs}}(p)) \sim_{s} \sigma^n_{\textrm{abs}}(\upsilon^m_{\textrm{abs}}(p))$. By the transition rules of operator $\sigma^n_{\textrm{abs}}$ and $\upsilon^m_{\textrm{abs}}$ in Table \ref{TRForBATCDAT}, we get

        $$\frac{\langle p,0\rangle\nuparrow}{\langle\upsilon^{m+n}_{\textrm{abs}}(\sigma^n_{\textrm{abs}}(p)),n'\rangle\mapsto^n \langle\upsilon^{m}_{\textrm{abs}} (\sigma^n_{\textrm{abs}}(p)),n'+n\rangle}$$

        $$\frac{\langle p,0\rangle\nuparrow}{\langle\sigma^n_{\textrm{abs}}(\upsilon^m_{\textrm{abs}}(p)),n'\rangle\mapsto^n \langle\sigma^n_{\textrm{abs}}(\upsilon^m_{\textrm{abs}}(p)),n'+n\rangle}$$

        There are several cases:

        $$\frac{\langle p,n'\rangle\xrightarrow{a} \langle\surd,n'\rangle}{\langle\upsilon^{m}_{\textrm{abs}} (\sigma^n_{\textrm{abs}}(p)),n'+n\rangle\xrightarrow{a}\langle\surd,n'+n\rangle}$$

        $$\frac{\langle p,n'\rangle\xrightarrow{a} \langle\surd,n'\rangle}{\langle\sigma^n_{\textrm{abs}}(\upsilon^m_{\textrm{abs}}(p)),n'+n\rangle\xrightarrow{a}\langle\surd,n'+n\rangle}$$

        $$\frac{\langle p,n'\rangle\xrightarrow{a} \langle p',n'\rangle}{\langle\upsilon^{m}_{\textrm{abs}} (\sigma^n_{\textrm{abs}}(p)),n'+n\rangle\xrightarrow{a}\langle\sigma^n_{\textrm{abs}}(p'),n'+n\rangle}$$

        $$\frac{\langle p,n'\rangle\xrightarrow{a} \langle p',n'\rangle}{\langle\sigma^n_{\textrm{abs}}(\upsilon^m_{\textrm{abs}}(p)),n'+n\rangle\xrightarrow{a}\langle\sigma^n_{\textrm{abs}}(p'),n'+n\rangle}$$

        $$\frac{\langle p,n'\rangle\uparrow}{\langle\upsilon^{m}_{\textrm{abs}} (\sigma^n_{\textrm{abs}}(p)),n'+n\rangle\uparrow}$$

        $$\frac{\langle p,n'\rangle\uparrow}{\langle\sigma^n_{\textrm{abs}}(\upsilon^m_{\textrm{abs}}(p)),n'+n\rangle\uparrow}$$

        So, we see that each case leads to $\upsilon^{m+n}_{\textrm{abs}} (\sigma^n_{\textrm{abs}}(p)) \sim_{s} \sigma^n_{\textrm{abs}}(\upsilon^m_{\textrm{abs}}(p))$, as desired.
  \item From the definition of pomset bisimulation, we know that pomset bisimulation is defined by pomset transitions, which are labeled by pomsets. In a pomset transition, the events (actions) in the pomset are either within causality relations (defined by $\cdot$) or in concurrency (implicitly defined by $\cdot$ and $+$, and explicitly defined by $\between$), of course, they are pairwise consistent (without conflicts). We have already proven the case that all events are pairwise concurrent (soundness modulo step bisimulation), so, we only need to prove the case of events in causality. Without loss of generality, we take a pomset of $P=\{\underline{a},\underline{b}:\underline{a}\cdot \underline{b}\}$. Then the pomset transition labeled by the above $P$ is just composed of one single event transition labeled by $\underline{a}$ succeeded by another single event transition labeled by $\underline{b}$, that is, $\xrightarrow{P}=\xrightarrow{a}\xrightarrow{b}$.

        Similarly to the proof of soundness modulo step bisimulation equivalence, we can prove that each axiom in Table \ref{AxiomsForBATCDAT} is sound modulo pomset bisimulation equivalence, we omit them.
  \item From the definition of hp-bisimulation, we know that hp-bisimulation is defined on the posetal product $(C_1,f,C_2),f:C_1\rightarrow C_2\textrm{ isomorphism}$. Two process terms $s$ related to $C_1$ and $t$ related to $C_2$, and $f:C_1\rightarrow C_2\textrm{ isomorphism}$. Initially, $(C_1,f,C_2)=(\emptyset,\emptyset,\emptyset)$, and $(\emptyset,\emptyset,\emptyset)\in\sim_{hp}$. When $s\xrightarrow{a}s'$ ($C_1\xrightarrow{a}C_1'$), there will be $t\xrightarrow{a}t'$ ($C_2\xrightarrow{a}C_2'$), and we define $f'=f[a\mapsto a]$. Then, if $(C_1,f,C_2)\in\sim_{hp}$, then $(C_1',f',C_2')\in\sim_{hp}$.

        Similarly to the proof of soundness modulo pomset bisimulation equivalence, we can prove that each axiom in Table \ref{AxiomsForBATCDAT} is sound modulo hp-bisimulation equivalence, we just need additionally to check the above conditions on hp-bisimulation, we omit them.
  \item We just need to add downward-closed condition to the soundness modulo hp-bisimulation equivalence, we omit them.
\end{enumerate}

\end{proof}

\subsubsection{Completeness}

\begin{theorem}[Completeness of $\textrm{BATC}^{\textrm{dat}}$]
The axiomatization of $\textrm{BATC}^{\textrm{dat}}$ is complete modulo truly concurrent bisimulation equivalences $\sim_{p}$, $\sim_{s}$, $\sim_{hp}$ and $\sim_{hhp}$. That is,
\begin{enumerate}
  \item let $p$ and $q$ be closed $\textrm{BATC}^{\textrm{dat}}$ terms, if $p\sim_{s} q$ then $p=q$;
  \item let $p$ and $q$ be closed $\textrm{BATC}^{\textrm{dat}}$ terms, if $p\sim_{p} q$ then $p=q$;
  \item let $p$ and $q$ be closed $\textrm{BATC}^{\textrm{dat}}$ terms, if $p\sim_{hp} q$ then $p=q$;
  \item let $p$ and $q$ be closed $\textrm{BATC}^{\textrm{dat}}$ terms, if $p\sim_{hhp} q$ then $p=q$.
\end{enumerate}

\end{theorem}

\begin{proof}
\begin{enumerate}
  \item Firstly, by the elimination theorem of $\textrm{BATC}^{\textrm{dat}}$, we know that for each closed $\textrm{BATC}^{\textrm{dat}}$ term $p$, there exists a closed basic $\textrm{BATC}^{\textrm{dat}}$ term $p'$, such that $\textrm{BATC}^{\textrm{dat}}\vdash p=p'$, so, we only need to consider closed basic $\textrm{BATC}^{\textrm{dat}}$ terms.

        The basic terms modulo associativity and commutativity (AC) of conflict $+$ (defined by axioms $A1$ and $A2$ in Table \ref{AxiomsForBATCDAT}), and this equivalence is denoted by $=_{AC}$. Then, each equivalence class $s$ modulo AC of $+$ has the following normal form

        $$s_1+\cdots+ s_k$$

        with each $s_i$ either an atomic event or of the form $t_1\cdot t_2$, and each $s_i$ is called the summand of $s$.

        Now, we prove that for normal forms $n$ and $n'$, if $n\sim_{s} n'$ then $n=_{AC}n'$. It is sufficient to induct on the sizes of $n$ and $n'$. We can get $n=_{AC} n'$.

        Finally, let $s$ and $t$ be basic terms, and $s\sim_s t$, there are normal forms $n$ and $n'$, such that $s=n$ and $t=n'$. The soundness theorem of $\textrm{BATC}^{\textrm{dat}}$ modulo step bisimulation equivalence yields $s\sim_s n$ and $t\sim_s n'$, so $n\sim_s s\sim_s t\sim_s n'$. Since if $n\sim_s n'$ then $n=_{AC}n'$, $s=n=_{AC}n'=t$, as desired.
  \item This case can be proven similarly, just by replacement of $\sim_{s}$ by $\sim_{p}$.
  \item This case can be proven similarly, just by replacement of $\sim_{s}$ by $\sim_{hp}$.
  \item This case can be proven similarly, just by replacement of $\sim_{s}$ by $\sim_{hhp}$.
\end{enumerate}
\end{proof}

\subsection{Algebra for Parallelism in True Concurrency with Discrete Absolute Timing}

In this subsection, we will introduce $\textrm{APTC}^{\textrm{dat}}$.

\subsubsection{The Theory $\textrm{APTC}^{\textrm{dat}}$}

\begin{definition}[Signature of $\textrm{APTC}^{\textrm{dat}}$]
The signature of $\textrm{APTC}^{\textrm{dat}}$ consists of the signature of $\textrm{BATC}^{\textrm{dat}}$, and the whole parallel composition operator $\between: \mathcal{P}_{\textrm{abs}}\times\mathcal{P}_{\textrm{abs}} \rightarrow \mathcal{P}_{\textrm{abs}}$, the parallel operator $\parallel: \mathcal{P}_{\textrm{abs}}\times\mathcal{P}_{\textrm{abs}} \rightarrow \mathcal{P}_{\textrm{abs}}$, the communication merger operator $\mid: \mathcal{P}_{\textrm{abs}}\times\mathcal{P}_{\textrm{abs}} \rightarrow \mathcal{P}_{\textrm{abs}}$, and the encapsulation operator $\partial_H: \mathcal{P}_{\textrm{abs}} \rightarrow \mathcal{P}_{\textrm{abs}}$ for all $H\subseteq A$.
\end{definition}

The set of axioms of $\textrm{APTC}^{\textrm{dat}}$ consists of the laws given in Table \ref{AxiomsForAPTCDAT}.

\begin{center}
    \begin{table}
        \begin{tabular}{@{}ll@{}}
            \hline No. &Axiom\\
            $P1$ & $x\between y = x\parallel y + x\mid y$\\
            $P2$ & $x\parallel y = y \parallel x$\\
            $P3$ & $(x\parallel y)\parallel z = x\parallel (y\parallel z)$\\
            $P4DA$ & $\underline{a}\parallel (\underline{b}\cdot y) = (\underline{a}\parallel \underline{b})\cdot y$\\
            $P5DA$ & $(\underline{a}\cdot x)\parallel \underline{b} = (\underline{a}\parallel \underline{b})\cdot x$\\
            $P6DA$ & $(\underline{a}\cdot x)\parallel (\underline{b}\cdot y) = (\underline{a}\parallel \underline{b})\cdot (x\between y)$\\
            $P7$ & $(x+ y)\parallel z = (x\parallel z)+ (y\parallel z)$\\
            $P8$ & $x\parallel (y+ z) = (x\parallel y)+ (x\parallel z)$\\
            $DAP9ID$ & $(\upsilon^1_{\textrm{abs}}(x)+ \underline{\delta})\parallel \sigma^{n+1}_{\textrm{abs}}(y) = \underline{\delta}$\\
            $DAP10ID$ & $\sigma^{n+1}_{\textrm{abs}}(x)\parallel (\upsilon^1_{\textrm{abs}}(y)+ \underline{\delta}) = \underline{\delta}$\\
            $DAP11$ & $\sigma^n_{\textrm{abs}}(x) \parallel \sigma^n_{\textrm{abs}}(y) = \sigma^n_{\textrm{abs}}(x\parallel y)$\\
            $PID12$ & $\dot{\delta}\parallel x = \dot{\delta}$\\
            $PID13$ & $x\parallel \dot{\delta} = \dot{\delta}$\\
            $C14DA$ & $\underline{a}\mid \underline{b} = \gamma(\underline{a}, \underline{b})$\\
            $C15DA$ & $\underline{a}\mid (\underline{b}\cdot y) = \gamma(\underline{a}, \underline{b})\cdot y$\\
            $C16DA$ & $(\underline{a}\cdot x)\mid \underline{b} = \gamma(\underline{a}, \underline{b})\cdot x$\\
            $C17DA$ & $(\underline{a}\cdot x)\mid (\underline{b}\cdot y) = \gamma(\underline{a}, \underline{b})\cdot (x\between y)$\\
            $C18$ & $(x+ y)\mid z = (x\mid z) + (y\mid z)$\\
            $C19$ & $x\mid (y+ z) = (x\mid y)+ (x\mid z)$\\
            $DAC20ID$ & $(\upsilon^1_{\textrm{abs}}(x)+ \underline{\delta})\mid \sigma^{n+1}_{\textrm{abs}}(y) = \underline{\delta}$\\
            $DAC21ID$ & $\sigma^{n+1}_{\textrm{abs}}(x)\mid (\upsilon^1_{\textrm{abs}}(y)+ \underline{\delta}) = \underline{\delta}$\\
            $DAC22$ & $\sigma^n_{\textrm{abs}}(x) \mid \sigma^n_{\textrm{abs}}(y) = \sigma^n_{\textrm{abs}}(x\mid y)$\\
            $CID23$ & $\dot{\delta}\mid x = \dot{\delta}$\\
            $CID24$ & $x\mid\dot{\delta} = \dot{\delta}$\\
            $CE25DA$ & $\Theta(\underline{a}) = \underline{a}$\\
            $CE26DAID$ & $\Theta(\dot{\delta}) = \dot{\delta}$\\
            $CE27$ & $\Theta(x+ y) = \Theta(x)\triangleleft y + \Theta(y)\triangleleft x$\\
            $CE28$ & $\Theta(x\cdot y)=\Theta(x)\cdot\Theta(y)$\\
            $CE29$ & $\Theta(x\parallel y) = ((\Theta(x)\triangleleft y)\parallel y)+ ((\Theta(y)\triangleleft x)\parallel x)$\\
            $CE30$ & $\Theta(x\mid y) = ((\Theta(x)\triangleleft y)\mid y)+ ((\Theta(y)\triangleleft x)\mid x)$\\
            $U31DAID$ & $(\sharp(\underline{a},\underline{b}))\quad \underline{a}\triangleleft \underline{b} = \underline{\tau}$\\
            $U32DAID$ & $(\sharp(\underline{a},\underline{b}),\underline{b}\leq \underline{c})\quad \underline{a}\triangleleft \underline{c} = \underline{a}$\\
            $U33DAID$ & $(\sharp(\underline{a},\underline{b}),\underline{b}\leq \underline{c})\quad \underline{c}\triangleleft \underline{a} = \underline{\tau}$\\
            $U34DAID$ & $\underline{a}\triangleleft \underline{\delta} = \underline{a}$\\
            $U35DAID$ & $\underline{\delta} \triangleleft \underline{a} = \underline{\delta}$\\
            $U36$ & $(x+ y)\triangleleft z = (x\triangleleft z)+ (y\triangleleft z)$\\
            $U37$ & $(x\cdot y)\triangleleft z = (x\triangleleft z)\cdot (y\triangleleft z)$\\
            $U38$ & $(x\parallel y)\triangleleft z = (x\triangleleft z)\parallel (y\triangleleft z)$\\
            $U39$ & $(x\mid y)\triangleleft z = (x\triangleleft z)\mid (y\triangleleft z)$\\
            $U40$ & $x\triangleleft (y+ z) = (x\triangleleft y)\triangleleft z$\\
            $U41$ & $x\triangleleft (y\cdot z)=(x\triangleleft y)\triangleleft z$\\
            $U42$ & $x\triangleleft (y\parallel z) = (x\triangleleft y)\triangleleft z$\\
            $U43$ & $x\triangleleft (y\mid z) = (x\triangleleft y)\triangleleft z$\\
            $D1DAID$ & $\underline{a}\notin H\quad\partial_H(\underline{a}) = \underline{a}$\\
            $D2DAID$ & $\underline{a}\in H\quad \partial_H(\underline{a}) = \underline{\delta}$\\
            $D3DAID$ & $\partial_H(\dot{\delta}) = \dot{\delta}$\\
            $D4$ & $\partial_H(x+ y) = \partial_H(x)+\partial_H(y)$\\
            $D5$ & $\partial_H(x\cdot y) = \partial_H(x)\cdot\partial_H(y)$\\
            $D6$ & $\partial_H(x\parallel y) = \partial_H(x)\parallel\partial_H(y)$\\
            $DAD7$ & $\partial_H(\sigma^n_{\textrm{abs}}(x)) = \sigma^n_{\textrm{abs}}(\partial_H(x))$\\
        \end{tabular}
        \caption{Axioms of $\textrm{APTC}^{\textrm{dat}}(a,b,c\in A_{\delta}, n\geq 0)$}
        \label{AxiomsForAPTCDAT}
    \end{table}
\end{center}

The operational semantics of $\textrm{APTC}^{\textrm{dat}}$ are defined by the transition rules in Table \ref{TRForAPTCDAT}.

\begin{center}
    \begin{table}
        $$\frac{\langle x,n\rangle\xrightarrow{a}\langle\surd,n\rangle\quad \langle y,n\rangle\xrightarrow{b}\langle\surd,n\rangle}{\langle x\parallel y,n\rangle\xrightarrow{\{a,b\}}\langle\surd,n\rangle} \quad\frac{\langle x,n\rangle\xrightarrow{a}\langle x',n\rangle\quad \langle y,n\rangle\xrightarrow{b}\langle\surd,n\rangle}{\langle x\parallel y,n\rangle\xrightarrow{\{a,b\}}\langle x',n\rangle}$$

        $$\frac{\langle x,n\rangle\xrightarrow{a}\langle\surd,n\rangle\quad \langle y,n\rangle\xrightarrow{b}\langle y',n\rangle}{\langle x\parallel y,n\rangle\xrightarrow{\{a,b\}}\langle y',n\rangle} \quad\frac{\langle x,n\rangle\xrightarrow{a}\langle x',n\rangle\quad \langle y,n\rangle\xrightarrow{b}\langle y',n\rangle}{\langle x\parallel y,n\rangle\xrightarrow{\{a,b\}}\langle x'\between y',n\rangle}$$

        $$\frac{\langle x,n\rangle\mapsto^{m}\langle x,n+m\rangle\quad \langle y,n\rangle\mapsto^{m}\langle y,n+m\rangle}{\langle x\parallel y,n\rangle\mapsto^{m}\langle x\parallel y,n+m\rangle} \quad\frac{\langle x,n\rangle\uparrow}{\langle x\parallel y,n\rangle\uparrow} \quad\frac{\langle y,n\rangle\uparrow}{\langle x\parallel y,n\rangle\uparrow}$$

        $$\frac{\langle x,n\rangle\xrightarrow{a}\langle \surd,n\rangle\quad \langle y,n\rangle\xrightarrow{b}\langle\surd,n\rangle}{\langle x\mid y,n\rangle\xrightarrow{\gamma(a,b)}\langle\surd,n\rangle} \quad\frac{\langle x,n\rangle\xrightarrow{a}\langle x',n\rangle\quad \langle y,n\rangle\xrightarrow{b}\langle \surd,n\rangle}{\langle x\mid y,n\rangle\xrightarrow{\gamma(a,b)}\langle x',n\rangle}$$

        $$\frac{\langle x,n\rangle\xrightarrow{a}\langle\surd,n\rangle\quad \langle y,n\rangle\xrightarrow{b}\langle y',n\rangle}{\langle x\mid y,n\rangle\xrightarrow{\gamma(a,b)}\langle y',n\rangle} \quad\frac{\langle x,n\rangle\xrightarrow{a}\langle x',n\rangle\quad \langle y,n\rangle\xrightarrow{b}\langle y',n\rangle}{\langle x\mid y,n\rangle\xrightarrow{\gamma(a,b)}\langle x'\between y',n\rangle}$$

        $$\frac{\langle x,n\rangle\mapsto^{m}\langle x,n+m\rangle\quad \langle y,n\rangle\mapsto^{m}\langle y,n+m\rangle}{\langle x\mid y,n\rangle\mapsto^{m}\langle x\mid y,n+m\rangle} \quad\frac{\langle x,n\rangle\uparrow}{\langle x\mid y,n\rangle\uparrow} \quad\frac{\langle y,n\rangle\uparrow}{\langle x\mid y,n\rangle\uparrow}$$

        $$\frac{\langle x,n\rangle\xrightarrow{a}\langle\surd,n\rangle\quad (\sharp(a,b))}{\langle\Theta(x),n\rangle\xrightarrow{a}\langle\surd,n\rangle} \quad\frac{\langle x,n\rangle\xrightarrow{b}\langle\surd,n\rangle\quad (\sharp(a,b))}{\langle\Theta(x),n\rangle\xrightarrow{b}\langle\surd,n\rangle}$$

        $$\frac{\langle x,n\rangle\xrightarrow{a}\langle x',n\rangle\quad (\sharp(a,b))}{\langle\Theta(x),n\rangle\xrightarrow{a}\langle\Theta(x'),n\rangle} \quad\frac{\langle x,n\rangle\xrightarrow{b}\langle x',n\rangle\quad (\sharp(a,b))}{\langle\Theta(x),n\rangle\xrightarrow{b}\langle\Theta(x'),n\rangle}$$

        $$\frac{\langle x,n\rangle\mapsto^{m}\langle x,n+m\rangle}{\langle\Theta(x),n\rangle\mapsto^{m}\langle\Theta(x),n+m\rangle} \quad\frac{\langle x,n\rangle\uparrow}{\langle\Theta(x),n\rangle\uparrow}$$

        $$\frac{\langle x,n\rangle\xrightarrow{a}\langle\surd,n\rangle \quad \langle y,n\rangle\nrightarrow^{b}\quad (\sharp(a,b))}{\langle x\triangleleft y,n\rangle\xrightarrow{\tau}\langle\surd,n\rangle}
        \quad\frac{\langle x,n\rangle\xrightarrow{a}\langle x',n\rangle \quad \langle y,n\rangle\nrightarrow^{b}\quad (\sharp(a,b))}{\langle x\triangleleft y,n\rangle\xrightarrow{\tau}\langle x',n\rangle}$$

        $$\frac{\langle x,n\rangle\xrightarrow{a}\langle\surd,n\rangle \quad \langle y,n\rangle\nrightarrow^{c}\quad (\sharp(a,b),b\leq c)}{\langle x\triangleleft y,n\rangle\xrightarrow{a}\langle \surd,n\rangle}
        \quad\frac{x\xrightarrow{a}x' \quad y\nrightarrow^{c}\quad (\sharp(a,b),b\leq c)}{x\triangleleft y\xrightarrow{a}x'}$$

        $$\frac{\langle x,n\rangle\xrightarrow{c}\langle\surd,n\rangle \quad \langle y,n\rangle\nrightarrow^{b}\quad (\sharp(a,b),a\leq c)}{\langle x\triangleleft y,n\rangle\xrightarrow{\tau}\langle\surd,n\rangle}
        \quad\frac{\langle x,n\rangle\xrightarrow{c}\langle x',n\rangle \quad \langle y,n\rangle\nrightarrow^{b}\quad (\sharp(a,b),a\leq c)}{\langle x\triangleleft y,n\rangle\xrightarrow{\tau}\langle x',n\rangle}$$

        $$\frac{\langle x,n\rangle\mapsto^{m}\langle x,n+m\rangle\quad \langle y,n\rangle\mapsto^{m}\langle y,n+m\rangle}{\langle x\triangleleft y,n\rangle\mapsto^{m}\langle x\triangleleft y,n+m\rangle} \quad\frac{\langle x,n\rangle\uparrow}{\langle x\triangleleft y,n\rangle\uparrow}$$

        $$\frac{\langle x,n\rangle\xrightarrow{a}\langle\surd,n\rangle}{\langle\partial_H(x),n\rangle\xrightarrow{a}\langle\surd,n\rangle}\quad (e\notin H)\quad\frac{\langle x,n\rangle\xrightarrow{a}\langle x',n\rangle}{\langle\partial_H(x),n\rangle\xrightarrow{a}\langle\partial_H(x'),n\rangle}\quad(e\notin H)$$

        $$\frac{\langle x,n\rangle\mapsto^{m}\langle x,n+m\rangle}{\langle\partial_H(x),n\rangle\mapsto^{m}\langle\partial_H(x'),n+m\rangle}\quad(e\notin H)\quad\frac{\langle x,n\rangle\uparrow}{\langle\partial_H(x),n\rangle\uparrow}$$
    \caption{Transition rules of $\textrm{APTC}^{\textrm{dat}}(a,b,c\in A, m>0, n\geq 0)$}
    \label{TRForAPTCDAT}
    \end{table}
\end{center}

\subsubsection{Elimination}

\begin{definition}[Basic terms of $\textrm{APTC}^{\textrm{dat}}$]
The set of basic terms of $\textrm{APTC}^{\textrm{dat}}$, $\mathcal{B}(\textrm{APTC}^{\textrm{dat}})$, is inductively defined as follows by two auxiliary sets $\mathcal{B}_0(\textrm{APTC}^{\textrm{dat}})$ and $\mathcal{B}_1(\textrm{APTC}^{\textrm{dat}})$:
\begin{enumerate}
  \item if $a\in A_{\delta}$, then $\underline{a} \in \mathcal{B}_1(\textrm{APTC}^{\textrm{dat}})$;
  \item if $a\in A$ and $t\in \mathcal{B}(\textrm{APTC}^{\textrm{dat}})$, then $\underline{a}\cdot t \in \mathcal{B}_1(\textrm{APTC}^{\textrm{dat}})$;
  \item if $t,t'\in \mathcal{B}_1(\textrm{APTC}^{\textrm{dat}})$, then $t+t'\in \mathcal{B}_1(\textrm{APTC}^{\textrm{dat}})$;
  \item if $t,t'\in \mathcal{B}_1(\textrm{APTC}^{\textrm{dat}})$, then $t\parallel t'\in \mathcal{B}_1(\textrm{APTC}^{\textrm{dat}})$;
  \item if $t\in \mathcal{B}_1(\textrm{APTC}^{\textrm{dat}})$, then $t\in \mathcal{B}_0(\textrm{APTC}^{\textrm{dat}})$;
  \item if $n>0$ and $t\in \mathcal{B}_0(\textrm{APTC}^{\textrm{dat}})$, then $\sigma^n_{\textrm{abs}}(t) \in \mathcal{B}_0(\textrm{APTC}^{\textrm{dat}})$;
  \item if $n>0$, $t\in \mathcal{B}_1(\textrm{APTC}^{\textrm{dat}})$ and $t'\in \mathcal{B}_0(\textrm{APTC}^{\textrm{dat}})$, then $t+\sigma^n_{\textrm{abs}}(t') \in \mathcal{B}_0(\textrm{APTC}^{\textrm{dat}})$;
  \item $\dot{\delta}\in \mathcal{B}(\textrm{APTC}^{\textrm{dat}})$;
  \item if $t\in \mathcal{B}_0(\textrm{APTC}^{\textrm{dat}})$, then $t\in \mathcal{B}(\textrm{APTC}^{\textrm{dat}})$.
\end{enumerate}
\end{definition}

\begin{theorem}[Elimination theorem]
Let $p$ be a closed $\textrm{APTC}^{\textrm{dat}}$ term. Then there is a basic $\textrm{APTC}^{\textrm{dat}}$ term $q$ such that $\textrm{APTC}^{\textrm{dat}}\vdash p=q$.
\end{theorem}

\begin{proof}
It is sufficient to induct on the structure of the closed $\textrm{APTC}^{\textrm{dat}}$ term $p$. It can be proven that $p$ combined by the constants and operators of $\textrm{APTC}^{\textrm{dat}}$ exists an equal basic term $q$, and the other operators not included in the basic terms, such as $\upsilon_{\textrm{abs}}$, $\overline{\upsilon}_{\textrm{abs}}$, $\between$, $\mid$, $\partial_H$, $\Theta$ and $\triangleleft$ can be eliminated.
\end{proof}

\subsubsection{Connections}

\begin{theorem}[Generalization of $\textrm{APTC}^{\textrm{dat}}$]
\begin{enumerate}
  \item By the definitions of $a=\underline{a}$ for each $a\in A$ and $\delta=\underline{\delta}$, $\textrm{APTC}^{\textrm{dat}}$ is a generalization of $APTC$.
  \item $\textrm{APTC}^{\textrm{dat}}$ is a generalization of $\textrm{BATC}^{\textrm{dat}}$¡£
\end{enumerate}

\end{theorem}

\begin{proof}
\begin{enumerate}
  \item It follows from the following two facts.

    \begin{enumerate}
      \item The transition rules of $APTC$ in section \ref{tcpa} are all source-dependent;
      \item The sources of the transition rules of $\textrm{APTC}^{\textrm{dat}}$ contain an occurrence of $\dot{\delta}$, $\underline{a}$, $\sigma^n_{\textrm{abs}}$, $\upsilon^n_{\textrm{abs}}$ and $\overline{\upsilon}^n_{\textrm{abs}}$.
    \end{enumerate}

    So, $APTC$ is an embedding of $\textrm{APTC}^{\textrm{dat}}$, as desired.
    \item It follows from the following two facts.

    \begin{enumerate}
      \item The transition rules of $\textrm{BATC}^{\textrm{dat}}$ are all source-dependent;
      \item The sources of the transition rules of $\textrm{APTC}^{\textrm{dat}}$ contain an occurrence of $\between$, $\parallel$, $\mid$, $\Theta$, $\triangleleft$, $\partial_H$.
    \end{enumerate}

    So, $\textrm{BATC}^{\textrm{dat}}$ is an embedding of $\textrm{APTC}^{\textrm{dat}}$, as desired.
\end{enumerate}
\end{proof}

\subsubsection{Congruence}

\begin{theorem}[Congruence of $\textrm{APTC}^{\textrm{dat}}$]
Truly concurrent bisimulation equivalences $\sim_p$, $\sim_s$ and $\sim_{hp}$ are all congruences with respect to $\textrm{APTC}^{\textrm{dat}}$. That is,
\begin{itemize}
  \item pomset bisimulation equivalence $\sim_{p}$ is a congruence with respect to $\textrm{APTC}^{\textrm{dat}}$;
  \item step bisimulation equivalence $\sim_{s}$ is a congruence with respect to $\textrm{APTC}^{\textrm{dat}}$;
  \item hp-bisimulation equivalence $\sim_{hp}$ is a congruence with respect to $\textrm{APTC}^{\textrm{dat}}$.
\end{itemize}
\end{theorem}

\begin{proof}
It is easy to see that $\sim_p$, $\sim_s$, and $\sim_{hp}$ are all equivalent relations on $\textrm{APTC}^{\textrm{dat}}$ terms, it is only sufficient to prove that $\sim_p$, $\sim_s$, and $\sim_{hp}$ are all preserved by the operators $\sigma^n_{\textrm{abs}}$, $\upsilon^n_{\textrm{abs}}$ and $\overline{\upsilon}^n_{\textrm{abs}}$. It is trivial and we omit it.
\end{proof}

\subsubsection{Soundness}

\begin{theorem}[Soundness of $\textrm{APTC}^{\textrm{dat}}$]
The axiomatization of $\textrm{APTC}^{\textrm{dat}}$ is sound modulo truly concurrent bisimulation equivalences $\sim_{p}$, $\sim_{s}$, and $\sim_{hp}$. That is,
\begin{enumerate}
  \item let $x$ and $y$ be $\textrm{APTC}^{\textrm{dat}}$ terms. If $\textrm{APTC}^{\textrm{dat}}\vdash x=y$, then $x\sim_{s} y$;
  \item let $x$ and $y$ be $\textrm{APTC}^{\textrm{dat}}$ terms. If $\textrm{APTC}^{\textrm{dat}}\vdash x=y$, then $x\sim_{p} y$;
  \item let $x$ and $y$ be $\textrm{APTC}^{\textrm{dat}}$ terms. If $\textrm{APTC}^{\textrm{dat}}\vdash x=y$, then $x\sim_{hp} y$.
\end{enumerate}
\end{theorem}

\begin{proof}
Since $\sim_p$, $\sim_s$, and $\sim_{hp}$ are both equivalent and congruent relations, we only need to check if each axiom in Table \ref{AxiomsForAPTCDAT} is sound modulo $\sim_p$, $\sim_s$, and $\sim_{hp}$ respectively.

\begin{enumerate}
  \item We only check the soundness of the non-trivial axiom $DAP11$ modulo $\sim_s$.
        Let $p,q$ be $\textrm{APTC}^{\textrm{dat}}$ processes, and $\sigma^n_{\textrm{abs}}(p) \parallel \sigma^n_{\textrm{abs}}(q) = \sigma^n_{\textrm{abs}}(p\parallel q)$, it is sufficient to prove that $\sigma^n_{\textrm{abs}}(p) \parallel \sigma^n_{\textrm{abs}}(q) \sim_{s} \sigma^n_{\textrm{abs}}(p\parallel q)$. By the transition rules of operator $\sigma^n_{\textrm{abs}}$ and $\parallel$ in Table \ref{TRForAPTCDAT}, we get

        $$\frac{\langle p,0\rangle\nuparrow}{\langle\sigma^n_{\textrm{abs}}(p) \parallel \sigma^n_{\textrm{abs}}(q),n'\rangle\mapsto^n \langle\sigma^n_{\textrm{abs}}(p) \parallel \sigma^n_{\textrm{abs}}(q),n'+n\rangle}$$

        $$\frac{\langle p,0\rangle\nuparrow}{\langle\sigma^n_{\textrm{abs}}(p\parallel q),n'\rangle\mapsto^n \langle\sigma^n_{\textrm{abs}}(p\parallel q),n'+n\rangle}$$

        There are several cases:

        $$\frac{\langle p,n'\rangle\xrightarrow{a} \langle\surd,n'\rangle\quad \langle q,n'\rangle\xrightarrow{b}\langle\surd,n'\rangle}{\langle\sigma^n_{\textrm{abs}}(p) \parallel \sigma^n_{\textrm{abs}}(q),n'+n\rangle\xrightarrow{\{a,b\}}\langle\surd,n'+n\rangle}$$

        $$\frac{\langle p,n'\rangle\xrightarrow{a} \langle\surd,n'\rangle\quad \langle q,n'\rangle\xrightarrow{b}\langle\surd,n'\rangle}{\langle\sigma^n_{\textrm{abs}}(p\parallel q),n'+n\rangle\xrightarrow{\{a,b\}}\langle\surd,n'+n\rangle}$$

        $$\frac{\langle p,n'\rangle\xrightarrow{a} \langle p',n'\rangle\quad \langle q,n'\rangle\xrightarrow{b}\langle\surd,n'\rangle}{\langle\sigma^n_{\textrm{abs}}(p) \parallel \sigma^n_{\textrm{abs}}(q),n'+n\rangle\xrightarrow{\{a,b\}}\langle\sigma^n_{\textrm{abs}}(p'),n'+n\rangle}$$

        $$\frac{\langle p,n'\rangle\xrightarrow{a} \langle p',n'\rangle\quad \langle q,n'\rangle\xrightarrow{b}\langle\surd,n'\rangle}{\langle\sigma^n_{\textrm{abs}}(p\parallel q),n'+n\rangle\xrightarrow{\{a,b\}}\langle\sigma^n_{\textrm{abs}}(p'),n'+n\rangle}$$

        $$\frac{\langle p,n'\rangle\xrightarrow{a} \langle\surd,n'\rangle\quad \langle q,n'\rangle\xrightarrow{b}\langle q',n'\rangle}{\langle\sigma^n_{\textrm{abs}}(p) \parallel \sigma^n_{\textrm{abs}}(q),n'+n\rangle\xrightarrow{\{a,b\}}\langle\sigma^n_{\textrm{abs}}(q'),n'+n\rangle}$$

        $$\frac{\langle p,n'\rangle\xrightarrow{a} \langle\surd,n'\rangle\quad \langle q,n'\rangle\xrightarrow{b}\langle q',n'\rangle}{\langle\sigma^n_{\textrm{abs}}(p\parallel q),n'+n\rangle\xrightarrow{\{a,b\}}\langle\sigma^n_{\textrm{abs}}(q'),n'+n\rangle}$$

        $$\frac{\langle p,n'\rangle\xrightarrow{a} \langle p',n'\rangle\quad \langle q,n'\rangle\xrightarrow{b}\langle q',n'\rangle}{\langle\sigma^n_{\textrm{abs}}(p) \parallel \sigma^n_{\textrm{abs}}(q),n'+n\rangle\xrightarrow{\{a,b\}}\langle\sigma^n_{\textrm{abs}}(p')\between \sigma^n_{\textrm{abs}}(q'),n'+n\rangle}$$

        $$\frac{\langle p,n'\rangle\xrightarrow{a} \langle p',n'\rangle\quad \langle q,n'\rangle\xrightarrow{b}\langle q',n'\rangle}{\langle\sigma^n_{\textrm{abs}}(p\parallel q),n'+n\rangle\xrightarrow{\{a,b\}}\langle\sigma^n_{\textrm{abs}}(p'\between q'),n'+n\rangle}$$

        $$\frac{\langle p,n'\rangle \uparrow}{\langle\sigma^n_{\textrm{abs}}(p) \parallel \sigma^n_{\textrm{abs}}(q),n'+n\rangle\uparrow}$$

        $$\frac{\langle p,n'\rangle\uparrow}{\langle\sigma^n_{\textrm{abs}}(p\parallel q),n'+n\rangle\uparrow}$$

        $$\frac{\langle q,n'\rangle \uparrow}{\langle\sigma^n_{\textrm{abs}}(p) \parallel \sigma^n_{\textrm{abs}}(q),n'+n\rangle\uparrow}$$

        $$\frac{\langle q,n'\rangle\uparrow}{\langle\sigma^n_{\textrm{abs}}(p\parallel q),n'+n\rangle\uparrow}$$

        So, we see that each case leads to $\sigma^n_{\textrm{abs}}(p) \parallel \sigma^n_{\textrm{abs}}(q) \sim_{s} \sigma^n_{\textrm{abs}}(p\parallel q)$, as desired.
  \item From the definition of pomset bisimulation, we know that pomset bisimulation is defined by pomset transitions, which are labeled by pomsets. In a pomset transition, the events (actions) in the pomset are either within causality relations (defined by $\cdot$) or in concurrency (implicitly defined by $\cdot$ and $+$, and explicitly defined by $\between$), of course, they are pairwise consistent (without conflicts). We have already proven the case that all events are pairwise concurrent (soundness modulo step bisimulation), so, we only need to prove the case of events in causality. Without loss of generality, we take a pomset of $P=\{\underline{a},\underline{b}:\underline{a}\cdot \underline{b}\}$. Then the pomset transition labeled by the above $P$ is just composed of one single event transition labeled by $\underline{a}$ succeeded by another single event transition labeled by $\underline{b}$, that is, $\xrightarrow{P}=\xrightarrow{a}\xrightarrow{b}$.

        Similarly to the proof of soundness modulo step bisimulation equivalence, we can prove that each axiom in Table \ref{AxiomsForAPTCDAT} is sound modulo pomset bisimulation equivalence, we omit them.
  \item From the definition of hp-bisimulation, we know that hp-bisimulation is defined on the posetal product $(C_1,f,C_2),f:C_1\rightarrow C_2\textrm{ isomorphism}$. Two process terms $s$ related to $C_1$ and $t$ related to $C_2$, and $f:C_1\rightarrow C_2\textrm{ isomorphism}$. Initially, $(C_1,f,C_2)=(\emptyset,\emptyset,\emptyset)$, and $(\emptyset,\emptyset,\emptyset)\in\sim_{hp}$. When $s\xrightarrow{a}s'$ ($C_1\xrightarrow{a}C_1'$), there will be $t\xrightarrow{a}t'$ ($C_2\xrightarrow{a}C_2'$), and we define $f'=f[a\mapsto a]$. Then, if $(C_1,f,C_2)\in\sim_{hp}$, then $(C_1',f',C_2')\in\sim_{hp}$.

        Similarly to the proof of soundness modulo pomset bisimulation equivalence, we can prove that each axiom in Table \ref{AxiomsForAPTCDAT} is sound modulo hp-bisimulation equivalence, we just need additionally to check the above conditions on hp-bisimulation, we omit them.
\end{enumerate}

\end{proof}

\subsubsection{Completeness}

\begin{theorem}[Completeness of $\textrm{APTC}^{\textrm{dat}}$]
The axiomatization of $\textrm{APTC}^{\textrm{dat}}$ is complete modulo truly concurrent bisimulation equivalences $\sim_{p}$, $\sim_{s}$, and $\sim_{hp}$. That is,
\begin{enumerate}
  \item let $p$ and $q$ be closed $\textrm{APTC}^{\textrm{dat}}$ terms, if $p\sim_{s} q$ then $p=q$;
  \item let $p$ and $q$ be closed $\textrm{APTC}^{\textrm{dat}}$ terms, if $p\sim_{p} q$ then $p=q$;
  \item let $p$ and $q$ be closed $\textrm{APTC}^{\textrm{dat}}$ terms, if $p\sim_{hp} q$ then $p=q$.
\end{enumerate}

\end{theorem}

\begin{proof}
\begin{enumerate}
  \item Firstly, by the elimination theorem of $\textrm{APTC}^{\textrm{dat}}$, we know that for each closed $\textrm{APTC}^{\textrm{dat}}$ term $p$, there exists a closed basic $\textrm{APTC}^{\textrm{dat}}$ term $p'$, such that $\textrm{APTC}^{\textrm{dat}}\vdash p=p'$, so, we only need to consider closed basic $\textrm{APTC}^{\textrm{dat}}$ terms.

        The basic terms modulo associativity and commutativity (AC) of conflict $+$ (defined by axioms $A1$ and $A2$ in Table \ref{AxiomsForBATCDAT}) and associativity and commutativity (AC) of parallel $\parallel$ (defined by axioms $P2$ and $P3$ in Table \ref{AxiomsForAPTCDAT}), and these equivalences is denoted by $=_{AC}$. Then, each equivalence class $s$ modulo AC of $+$ and $\parallel$ has the following normal form

        $$s_1+\cdots+ s_k$$

        with each $s_i$ either an atomic event or of the form

        $$t_1\cdot\cdots\cdot t_m$$

        with each $t_j$ either an atomic event or of the form

        $$u_1\parallel\cdots\parallel u_n$$

        with each $u_l$ an atomic event, and each $s_i$ is called the summand of $s$.

        Now, we prove that for normal forms $n$ and $n'$, if $n\sim_{s} n'$ then $n=_{AC}n'$. It is sufficient to induct on the sizes of $n$ and $n'$. We can get $n=_{AC} n'$.

        Finally, let $s$ and $t$ be basic $\textrm{APTC}^{\textrm{dat}}$ terms, and $s\sim_s t$, there are normal forms $n$ and $n'$, such that $s=n$ and $t=n'$. The soundness theorem modulo step bisimulation equivalence yields $s\sim_s n$ and $t\sim_s n'$, so $n\sim_s s\sim_s t\sim_s n'$. Since if $n\sim_s n'$ then $n=_{AC}n'$, $s=n=_{AC}n'=t$, as desired.
  \item This case can be proven similarly, just by replacement of $\sim_{s}$ by $\sim_{p}$.
  \item This case can be proven similarly, just by replacement of $\sim_{s}$ by $\sim_{hp}$.
\end{enumerate}
\end{proof}

\subsection{Discrete Initial Abstraction}

In this subsection, we will introduce $\textrm{APTC}^{\textrm{dat}}$ with discrete initial abstraction called $\textrm{APTC}^{\textrm{dat}}\surd$.

\subsubsection{Basic Definition}

\begin{definition}[Discrete initial abstraction]
Discrete initial abstraction $\surd_d$ is an abstraction mechanism to form functions from natural numbers to processes with absolute timing, that map each natural number $n$ to a process initialized at time $n$.
\end{definition}

\subsubsection{The Theory $\textrm{APTC}^{\textrm{dat}}\surd$}

\begin{definition}[Signature of $\textrm{APTC}^{\textrm{dat}}\surd$]
The signature of $\textrm{APTC}^{\textrm{dat}}\surd$ consists of the signature of $\textrm{APTC}^{\textrm{dat}}$, and the discrete initial abstraction operator $\surd_d: \mathbb{N}.\mathcal{P}^*_{\textrm{abs}}\rightarrow\mathcal{P}^*_{\textrm{abs}}$. Where $\mathcal{P}^*_{\textrm{abs}}$ is the sorts with discrete initial abstraction.
\end{definition}

The set of axioms of $\textrm{APTC}^{\textrm{dat}}\surd$ consists of the laws given in Table \ref{AxiomsForAPTCDATDIA}. Where $i,j,\cdots$ are variables of sort $\mathbb{N}$, $F,G,\cdots$ are variables of sort $\mathbb{N}.\mathcal{P}^*_{\textrm{abs}}$, $K,L,\cdots$ are variables of sort $\mathbb{N},\mathbb{N}.\mathcal{P}^*_{\textrm{abs}}$, and we write $\surd_d i.t$ for $\surd_d(i.t)$.

\begin{center}
    \begin{table}
        \begin{tabular}{@{}ll@{}}
            \hline No. &Axiom\\
            $DIA1$ & $\surd_d i.F(i) = \surd_d j.F(j)$\\
            $DIA2$ & $\overline{\upsilon}^n_{\textrm{abs}}(\surd_d i.F(i)) = \overline{\upsilon}^n_{\textrm{abs}}(F(n))$\\
            $DIA3$ & $\surd_d i.(\surd_d j.K(i,j)) = \surd_d i.K(i,i)$\\
            $DIA4$ & $x = \surd_d i.x$\\
            $DIA5$ & $(\forall i\in\mathbb{N}.\overline{\upsilon}^i_{\textrm{abs}}(x) = \overline{\upsilon}^i_{\textrm{abs}}(y))\Rightarrow x=y$\\
            $DIA6$ & $\sigma^n_{\textrm{abs}}(\underline{a})\cdot x = \sigma^n_{\textrm{abs}}(\underline{a})\cdot\overline{\upsilon}^n_{\textrm{abs}}(x)$\\
            $DIA7$ & $\sigma^n_{\textrm{abs}}(\surd_d i.F(i)) = \sigma^n_{\textrm{abs}}(F(0))$\\
            $DIA8$ & $(\surd_d i.F(i)) + x = \surd_d i.(F(i) + \overline{\upsilon}^i_{\textrm{abs}}(x))$\\
            $DIA9$ & $(\surd_d i.F(i)) \cdot x = \surd_d i.(F(i) \cdot x)$\\
            $DIA10$ & $\upsilon^n_{\textrm{abs}}(\surd_d i.F(i)) = \surd_d i.\upsilon^n_{\textrm{abs}}(F(i))$\\
            $DIA11$ & $(\surd_d i.F(i)) \parallel x = \surd_d i.(F(i) \parallel\overline{\upsilon}^i_{\textrm{abs}}(x))$\\
            $DIA12$ & $x\parallel(\surd_d i.F(i)) = \surd_d i.(\overline{\upsilon}^i_{\textrm{abs}}(x)\parallel F(i))$\\
            $DIA13$ & $(\surd_d i.F(i)) \mid x = \surd_d i.(F(i) \mid\overline{\upsilon}^i_{\textrm{abs}}(x))$\\
            $DIA14$ & $x\mid(\surd_d i.F(i)) = \surd_d i.(\overline{\upsilon}^i_{\textrm{abs}}(x)\mid F(i))$\\
            $DIA15$ & $\Theta(\surd_d i.F(i)) = \surd_d i.\Theta(F(i))$\\
            $DIA16$ & $(\surd_d i.F(i)) \triangleleft x = \surd_d i.(F(i) \triangleleft x)$\\
            $DIA17$ & $\partial_H(\surd_d i.F(i)) = \surd_d i.\partial_H(F(i))$\\
        \end{tabular}
        \caption{Axioms of $\textrm{APTC}^{\textrm{dat}}\surd(n\geq 0)$}
        \label{AxiomsForAPTCDATDIA}
    \end{table}
\end{center}

It sufficient to extend bisimulations $\mathcal{CI}/\sim$ of $APTC^{\textrm{dat}}$ to

$$(\mathcal{CI}/\sim)^* = \{f:\mathbb{N}\rightarrow \mathcal{CI}/\sim | \forall i\in\mathbb{N}.f(i) = \overline{\upsilon}^i_{\textrm{abs}}(f(i))\}$$

and define the constants and operators of $APTC^{\textrm{dat}}\surd$ on $(\mathcal{CI}/\sim)^*$ as in Table \ref{DefsForAPTCDATDIA}, and the $*: \mathcal{CI}/\sim\times (\mathcal{CI}/\sim)^*\rightarrow\mathcal{CI}/\sim$ is defined in Table \ref{StarForAPTCDATDIA}.

\begin{center}
    \begin{table}
        \begin{tabular}{@{}ll@{}}
            \hline
            $\dot{\delta} = \lambda j.\dot{\delta}$ & $\overline{\upsilon}^i_{\textrm{abs}}(f) = f(i)$\\
            $\underline{a} = \lambda j.\overline{\upsilon}^j_{\textrm{abs}}(\underline{a}) (a\in A_{\delta})$ & $f\between g = \lambda j.(f(j)\between g(j))$\\
            $\sigma^i_{\textrm{abs}}(f) = \lambda_j.\overline{\upsilon}^j_{\textrm{abs}}(\sigma^i_{\textrm{abs}}(f(0)))$ & $f\parallel g = \lambda j.(f(j)\parallel g(j))$\\
            $f+g = \lambda j.(f(j)+g(j))$ & $f\mid g = \lambda j.(f(j)\mid g(j))$\\
            $f\cdot g = \lambda j.(f(j)*g)$ & $\partial_H(f) = \lambda j.\partial_H(f(j))$\\
            $\Theta(f) = \lambda j.\Theta(f(j))$ & $f\triangleleft g = \lambda j. (f(j)\triangleleft g(j))$\\
            $\upsilon^i_{\textrm{abs}}(f) = \lambda j.\overline{\upsilon}^j_{\textrm{abs}}(\upsilon^i_{\textrm{abs}}f(j)))$ & $\surd_d(\varphi) = \lambda j.\overline{\upsilon}^j_{\textrm{abs}}(\varphi(j))$\\
        \end{tabular}
        \caption{Definitions of $\textrm{APTC}^{\textrm{dat}}$ on $(\mathcal{CI}/\sim)^*$}
        \label{DefsForAPTCDATDIA}
    \end{table}
\end{center}

\begin{center}
    \begin{table}
        \begin{tabular}{@{}ll@{}}
            \hline
            $\dot{\delta}* f = \dot{\delta}$\\
            $\underline{a}* f = \underline{a}\cdot f(0)(a\in A_{\delta})$\\
            $\sigma^i_{\textrm{abs}}(p)*f = \sigma^i_{\textrm{abs}}(p*\lambda j.f(i+j))$\\
            $(p+q)*f = (p*f) + (q*f)$\\
            $(p\cdot q)*f = p\cdot(q*f)$\\
            $(p\parallel q)*f = (p*f)\parallel (q* f)$\\
        \end{tabular}
        \caption{Definitions of $*$}
        \label{StarForAPTCDATDIA}
    \end{table}
\end{center}

\subsubsection{Connections}

\begin{center}
    \begin{table}
        \begin{tabular}{@{}ll@{}}
            \hline
            $\underline{\underline{a}} = \surd_d j.\sigma^j_{\textrm{abs}}(\underline{a})(a\in A)$\\
            $\underline{\underline{\delta}} = \surd_d j.\sigma^j_{\textrm{abs}}(\underline{\delta})$\\
            $\sigma^i_{\textrm{abs}}(x) = \surd_d j.\overline{\upsilon}^{i+j}_{\textrm{abs}}(x)$\\
            $\upsilon^i_{\textrm{abs}}(x) = \surd_d j.\upsilon^{i+j}_{\textrm{abs}}(\overline{\upsilon}^j_{\textrm{abs}}(x))$\\
            $\overline{\upsilon}^i_{\textrm{abs}}(x) = \surd_d j.\overline{\upsilon}^{i+j}_{\textrm{abs}}(\overline{\upsilon}^j_{\textrm{abs}}(x))$\\
        \end{tabular}
        \caption{Definitions of constants and operators of $ACTC^{\textrm{drt}}$ in $\textrm{APTC}^{\textrm{dat}}\surd$}
        \label{DefsForAPTCDATDIADRT}
    \end{table}
\end{center}

\begin{center}
    \begin{table}
        \begin{tabular}{@{}ll@{}}
            \hline
            $DPTST0$ & $\mu(\dot{\delta}) = \dot{\delta}$\\
            $DPTST1$ & $\mu(\underline{a}) = \dot{\delta}$\\
            $DPTST2$ & $\mu(\sigma^{n+1}_{\textrm{abs}}(x)) = \sigma^{n}_{\textrm{abs}}(x)$\\
            $DPTST3$ & $\mu(x+y) = \mu(x) + \mu(y)$\\
            $DPTST4$ & $\mu(x\cdot y) = \mu(x)\cdot\mu(y)$\\
            $DPTST5$ & $\mu(x\parallel y) = \mu(x)\parallel\mu(y)$\\
            $DPTST6$ & $\mu(x) = \surd_d i.\mu(\overline{\upsilon}^{i+1}_{\textrm{abs}}(x))$\\
        \end{tabular}
        \caption{Axioms of time spectrum tail $(a\in A_{\delta}m n\geq 0)$}
        \label{MUForAPTCDATDIADRT}
    \end{table}
\end{center}

\begin{center}
    \begin{table}
        $$\frac{\langle x, n+1\rangle \xrightarrow{a}\langle x',n+1\rangle}{\langle \mu(x),n\rangle\xrightarrow{a}\langle\mu(x'),n\rangle}
        \quad\frac{\langle x, n+1\rangle \xrightarrow{a}\langle \surd,n+1\rangle}{\langle \mu(x),n\rangle\xrightarrow{a}\langle\surd,n\rangle}$$

        $$\frac{\langle x, n+1\rangle \mapsto^{m}\langle x,n+m+1\rangle}{\langle \mu(x),n\rangle\mapsto^{m}\langle\mu(x),n+m\rangle}
        \quad\frac{\langle x, n+1\rangle\uparrow}{\langle\mu(x),n\rangle\uparrow}\quad\frac{\langle x,0\rangle\nmapsto^1}{\langle\mu(x),n\rangle\uparrow}$$
        \caption{Transition rules of time spectrum tail $(a\in A, m>0, n\geq 0)$}
        \label{TRMUForAPTCDATDIADRT}
    \end{table}
\end{center}

\begin{center}
    \begin{table}
        \begin{tabular}{@{}ll@{}}
            \hline
            $\mu(f) = \lambda k.\mu(f(k+1))$\\
        \end{tabular}
        \caption{Definition of time spectrum tail on $(\mathcal{CI}/\sim)^*$}
        \label{DEFMUForAPTCDATDIADRT}
    \end{table}
\end{center}

\begin{theorem}[Generalization of $\textrm{APTC}^{\textrm{dat}}\surd$]
\begin{enumerate}
  \item By the definitions of constants and operators of $ACTC^{\textrm{drt}}$ in $\textrm{APTC}^{\textrm{dat}}\surd$ in Table \ref{DefsForAPTCDATDIADRT}, a relatively timed process with discrete initialization abstraction of the time spectrum tail operator $\mu: \mathcal{P}^*_{\textrm{abs}}\rightarrow \mathcal{P}^*_{\textrm{abs}}$ in Table \ref{MUForAPTCDATDIADRT}, Table \ref{TRMUForAPTCDATDIADRT} and Table \ref{DEFMUForAPTCDATDIADRT}, $\textrm{APTC}^{\textrm{dat}}\surd$ is a generalization of $APTC^{\textrm{drt}}$.
  \item $\textrm{APTC}^{\textrm{dat}}\surd$ is a generalization of $\textrm{BATC}^{\textrm{dat}}$¡£
\end{enumerate}

\end{theorem}

\begin{proof}
\begin{enumerate}
  \item It follows from the following two facts. By the definitions of constants and operators of $ACTC^{\textrm{drt}}$ in $\textrm{APTC}^{\textrm{dat}}\surd$ in Table \ref{DefsForAPTCDATDIADRT}, a relatively timed process with discrete initialization abstraction of the time spectrum tail operator $\mu: \mathcal{P}^*_{\textrm{abs}}\rightarrow \mathcal{P}^*_{\textrm{abs}}$ in Table \ref{MUForAPTCDATDIADRT}, Table \ref{TRMUForAPTCDATDIADRT} and Table \ref{DEFMUForAPTCDATDIADRT},

    \begin{enumerate}
      \item the transition rules of $ACTC^{\textrm{drt}}$ are all source-dependent;
      \item the sources of the transition rules of $\textrm{APTC}^{\textrm{dat}}\surd$ contain an occurrence of $\surd_d$.
    \end{enumerate}

    So, $ACTC^{\textrm{drt}}$ is an embedding of $\textrm{APTC}^{\textrm{dat}}\surd$, as desired.
    \item It follows from the following two facts.

    \begin{enumerate}
      \item The transition rules of $\textrm{APTC}^{\textrm{dat}}$ are all source-dependent;
      \item The sources of the transition rules of $\textrm{APTC}^{\textrm{dat}}\surd$ contain an occurrence of $\surd_d$.
    \end{enumerate}

    So, $\textrm{APTC}^{\textrm{dat}}$ is an embedding of $\textrm{APTC}^{\textrm{dat}}\surd$, as desired.
\end{enumerate}
\end{proof}

\subsubsection{Congruence}

\begin{theorem}[Congruence of $\textrm{APTC}^{\textrm{dat}}\surd$]
Truly concurrent bisimulation equivalences $\sim_p$, $\sim_s$ and $\sim_{hp}$ are all congruences with respect to $\textrm{APTC}^{\textrm{dat}}\surd$. That is,
\begin{itemize}
  \item pomset bisimulation equivalence $\sim_{p}$ is a congruence with respect to $\textrm{APTC}^{\textrm{dat}}\surd$;
  \item step bisimulation equivalence $\sim_{s}$ is a congruence with respect to $\textrm{APTC}^{\textrm{dat}}\surd$;
  \item hp-bisimulation equivalence $\sim_{hp}$ is a congruence with respect to $\textrm{APTC}^{\textrm{dat}}\surd$.
\end{itemize}
\end{theorem}

\begin{proof}
It is easy to see that $\sim_p$, $\sim_s$, and $\sim_{hp}$ are all equivalent relations on $\textrm{APTC}^{\textrm{dat}}\surd$ terms, it is only sufficient to prove that $\sim_p$, $\sim_s$, and $\sim_{hp}$ are all preserved by the operators $\sigma^n_{\textrm{abs}}$, $\upsilon^n_{\textrm{abs}}$ and $\overline{\upsilon}^n_{\textrm{abs}}$. It is trivial and we omit it.
\end{proof}

\subsubsection{Soundness}

\begin{theorem}[Soundness of $\textrm{APTC}^{\textrm{dat}}\surd$]
The axiomatization of $\textrm{APTC}^{\textrm{dat}}\surd$ is sound modulo truly concurrent bisimulation equivalences $\sim_{p}$, $\sim_{s}$, and $\sim_{hp}$. That is,
\begin{enumerate}
  \item let $x$ and $y$ be $\textrm{APTC}^{\textrm{dat}}\surd$ terms. If $\textrm{APTC}^{\textrm{dat}}\surd\vdash x=y$, then $x\sim_{s} y$;
  \item let $x$ and $y$ be $\textrm{APTC}^{\textrm{dat}}\surd$ terms. If $\textrm{APTC}^{\textrm{dat}}\surd\vdash x=y$, then $x\sim_{p} y$;
  \item let $x$ and $y$ be $\textrm{APTC}^{\textrm{dat}}\surd$ terms. If $\textrm{APTC}^{\textrm{dat}}\surd\vdash x=y$, then $x\sim_{hp} y$.
\end{enumerate}
\end{theorem}

\begin{proof}
Since $\sim_p$, $\sim_s$, and $\sim_{hp}$ are both equivalent and congruent relations, we only need to check if each axiom in Table \ref{AxiomsForAPTCDATDIA} is sound modulo $\sim_p$, $\sim_s$, and $\sim_{hp}$ respectively.

\begin{enumerate}
  \item Each axiom in Table \ref{AxiomsForAPTCDATDIA} can be checked that it is sound modulo step bisimulation equivalence, by $\lambda$-definitions in Table \ref{DefsForAPTCDATDIA}, Table \ref{StarForAPTCDATDIA}. We omit them.
  \item From the definition of pomset bisimulation, we know that pomset bisimulation is defined by pomset transitions, which are labeled by pomsets. In a pomset transition, the events (actions) in the pomset are either within causality relations (defined by $\cdot$) or in concurrency (implicitly defined by $\cdot$ and $+$, and explicitly defined by $\between$), of course, they are pairwise consistent (without conflicts). We have already proven the case that all events are pairwise concurrent (soundness modulo step bisimulation), so, we only need to prove the case of events in causality. Without loss of generality, we take a pomset of $P=\{\underline{a},\underline{b}:\underline{a}\cdot \underline{b}\}$. Then the pomset transition labeled by the above $P$ is just composed of one single event transition labeled by $\underline{a}$ succeeded by another single event transition labeled by $\underline{b}$, that is, $\xrightarrow{P}=\xrightarrow{a}\xrightarrow{b}$.

        Similarly to the proof of soundness modulo step bisimulation equivalence, we can prove that each axiom in Table \ref{AxiomsForAPTCDATDIA} is sound modulo pomset bisimulation equivalence, we omit them.
  \item From the definition of hp-bisimulation, we know that hp-bisimulation is defined on the posetal product $(C_1,f,C_2),f:C_1\rightarrow C_2\textrm{ isomorphism}$. Two process terms $s$ related to $C_1$ and $t$ related to $C_2$, and $f:C_1\rightarrow C_2\textrm{ isomorphism}$. Initially, $(C_1,f,C_2)=(\emptyset,\emptyset,\emptyset)$, and $(\emptyset,\emptyset,\emptyset)\in\sim_{hp}$. When $s\xrightarrow{a}s'$ ($C_1\xrightarrow{a}C_1'$), there will be $t\xrightarrow{a}t'$ ($C_2\xrightarrow{a}C_2'$), and we define $f'=f[a\mapsto a]$. Then, if $(C_1,f,C_2)\in\sim_{hp}$, then $(C_1',f',C_2')\in\sim_{hp}$.

        Similarly to the proof of soundness modulo pomset bisimulation equivalence, we can prove that each axiom in Table \ref{AxiomsForAPTCDATDIA} is sound modulo hp-bisimulation equivalence, we just need additionally to check the above conditions on hp-bisimulation, we omit them.
\end{enumerate}

\end{proof}

\subsection{Time-Dependent Conditions}

In this subsection, we will introduce $\textrm{APTC}^{\textrm{dat}}\surd$ with time-dependent conditions called $\textrm{APTC}^{\textrm{dat}}\surd\textrm{C}$.

\subsubsection{Basic Definition}

\begin{definition}[Time-dependent conditions]
The basic kinds of time-dependent conditions are in-time-slice and in-time-slice-greater-than. In-time-slice $n$ $(n\in\mathbb{N})$ is the condition that holds only in time slice $n$ and in-time-slice-greater-than $n$ $(n\in\mathbb{N})$ is the condition that holds in all time slices greater than $n$. $\mathbf{t}$ is as the truth and $\mathbf{f}$ is as falsity.
\end{definition}

\subsubsection{The Theory $\textrm{APTC}^{\textrm{dat}}\surd\textrm{C}$}

\begin{definition}[Signature of $\textrm{APTC}^{\textrm{dat}}\surd\textrm{C}$]
The signature of $\textrm{APTC}^{\textrm{dat}}\surd\textrm{C}$ consists of the signature of $\textrm{APTC}^{\textrm{dat}}\surd$, and the in-time-slice operator $\mathbf{sl}: \mathbb{N}\rightarrow\mathbb{B}^*$, the in-time-slice-greater-than operator $\mathbf{sl}_>: \mathbb{N}\rightarrow\mathbb{B}^*$, the logical constants and operators $\mathbf{t}:\rightarrow\mathbb{B}^*$, $\mathbf{f}:\rightarrow\mathbb{B}^*$, $\neg: \mathbb{B}^*\rightarrow \mathbb{B}^*$, $\vee: \mathbb{B}^*\times\mathbb{B}^*\rightarrow \mathbb{B}^*$, $\wedge: \mathbb{B}^*\times\mathbb{B}^*\rightarrow \mathbb{B}^*$, the absolute initialization operator $\overline{\upsilon}_{\textrm{abs}}: \mathbb{N}\times\mathbb{B}^*\rightarrow\mathbb{B}^*$, the discrete initial abstraction operator $\surd_d: \mathbb{N}.\mathbb{B}^*\rightarrow \mathbb{B}^*$, and the conditional operator $::\rightarrow: \mathbb{B}^*\times\mathcal{P}^*_{\textrm{abs}}\rightarrow\mathcal{P}^*_{\textrm{abs}}$. Where $\mathbb{B}^*$ is the sort of time-dependent conditions.
\end{definition}

The set of axioms of $\textrm{APTC}^{\textrm{dat}}\surd\textrm{C}$ consists of the laws given in Table \ref{LOForAPTCDATDIAC}, Table \ref{LOForAPTCDATDIAC1} and Table \ref{LOForAPTCDATDIAC2}. Where $b$ is a condition.

\begin{center}
    \begin{table}
        \begin{tabular}{@{}ll@{}}
            \hline No. &Axiom\\
            $BOOL1$ & $\neg\mathbf{t} = \mathbf{f}$\\
            $BOOL2$ & $\neg\mathbf{f} = \mathbf{t}$\\
            $BOOL3$ & $\neg\neg b = b$\\
            $BOOL4$ & $\mathbf{t}\vee b = \mathbf{t}$\\
            $BOOL5$ & $\mathbf{f}\vee b = b$\\
            $BOOL6$ & $b\vee b'=b'\vee b$\\
            $BOOL7$ & $b\wedge b' = \neg(\neg b\vee\neg b')$\\
        \end{tabular}
        \caption{Axioms of logical operators}
        \label{LOForAPTCDATDIAC}
    \end{table}
\end{center}

\begin{center}
    \begin{table}
        \begin{tabular}{@{}ll@{}}
            \hline No. &Axiom\\
            $CDAI1$ & $\overline{\upsilon}^n_{\textrm{abs}}(\mathbf{t}) = \mathbf{t}$\\
            $CDAI2$ & $\overline{\upsilon}^n_{\textrm{abs}}(\mathbf{f}) = \mathbf{f}$\\
            $CDAI3$ & $\overline{\upsilon}^n_{\textrm{abs}}(\mathbf{sl}(n+1)) = \mathbf{t}$\\
            $CDAI4$ & $\overline{\upsilon}^{n+m}_{\textrm{abs}}(\mathbf{sl}(n)) = \mathbf{f}$\\
            $CDAI5$ & $\overline{\upsilon}^n_{\textrm{abs}}(\mathbf{sl}(n+m+2)) = \mathbf{f}$\\
            $CDAI6$ & $\overline{\upsilon}^{n+m}_{\textrm{abs}}(\mathbf{sl}_>(n)) = \mathbf{t}$\\
            $CDAI7$ & $\overline{\upsilon}^{n}_{\textrm{abs}}(\mathbf{sl}_>(n+m+1)) = \mathbf{f}$\\
            $CDAI8$ & $\overline{\upsilon}^n_{\textrm{abs}}(\neg b) = \neg\overline{\upsilon}^n_{\textrm{abs}}(b)$\\
            $CDAI9$ & $\overline{\upsilon}^n_{\textrm{abs}}(b \wedge b') = \overline{\upsilon}^n_{\textrm{abs}}(b)\wedge \overline{\upsilon}^n_{\textrm{abs}}(b')$\\
            $CDAI10$ & $\overline{\upsilon}^n_{\textrm{abs}}(b \vee b') = \overline{\upsilon}^n_{\textrm{abs}}(b)\vee \overline{\upsilon}^n_{\textrm{abs}}(b')$\\

            $CDIA1$ & $\surd_d i.C(i) = \surd_d j.D(j)$\\
            $CDIA2$ & $\overline{\upsilon}^n_{\textrm{abs}}(\surd_d i.C(i)) = \overline{\upsilon}^n_{\textrm{abs}}(C(n))$\\
            $CDIA3$ & $\surd_d i.(\surd_d j.E(i,j)) = \surd_d i.E(i,i)$\\
            $CDIA4$ & $b = \surd_d i.b$\\
            $CDIA5$ & $(\forall i\in\mathbb{N}.\overline{\upsilon}^i_{\textrm{asb}}(b)=\overline{\upsilon}^i_{\textrm{abs}}(b'))\Rightarrow b = b'$\\
            $CDIA6$ & $\neg(\surd_d i.C(i)) = \surd_d i.\neg C(i)$\\
            $CDIA7$ & $(\surd_d i.C(i))\wedge b = \surd_d i.(C(i)\wedge \overline{\upsilon}^i_{\textrm{abs}}(b))$\\
            $CDIA8$ & $(\surd_d i.C(i))\vee b = \surd_d i.(C(i)\vee \overline{\upsilon}^i_{\textrm{abs}}(b))$\\
        \end{tabular}
        \caption{Axioms of conditions $(m,n\geq 0)$}
        \label{LOForAPTCDATDIAC1}
    \end{table}
\end{center}

\begin{center}
    \begin{table}
        \begin{tabular}{@{}ll@{}}
            \hline No. &Axiom\\
            $SGC1$ & $\mathbf{t}::\rightarrow x = x$\\
            $SGC2ID$ & $\mathbf{f}::\rightarrow x = \dot{\delta}$\\
            $DASGC1$ & $\overline{\upsilon}^n_{\textrm{abs}}(b::\rightarrow x) = \overline{\upsilon}^n_{\textrm{abs}}(b)::\rightarrow\overline{\upsilon}^n_{\textrm{abs}} (x) + \sigma^n_{\textrm{abs}}(\dot{\delta})$\\
            $DASGC2$ & $x = \sum_{k\in[0,n]}(\mathbf{sl}(k+1)::\rightarrow\overline{\upsilon}^k_{\textrm{abs}}(x)) + \mathbf{sl}_>(n+1)::\rightarrow x$\\
            $SGC3ID$ & $b::\rightarrow\dot{\delta} = \dot{\delta}$\\
            $DASGC3$ & $b::\rightarrow \sigma^n_{\textrm{abs}}(x) + \sigma^n_{\textrm{abs}}(\dot{\delta}) = \surd_d i.\sigma^n_{\textrm{abs}}(\overline{\upsilon}^i_{\textrm{abs}}(b)::\rightarrow x)$\\
            $SGC4$ & $b::\rightarrow(x+y)=b::\rightarrow x + b::\rightarrow y$\\
            $SGC5$ & $b::\rightarrow x\cdot y = (b::\rightarrow x)\cdot y$\\
            $SGC6$ & $(b\vee b')::\rightarrow x = b::\rightarrow x + b'::\rightarrow x$\\
            $SGC7$ & $b::\rightarrow(b'::\rightarrow x) = (b\wedge b')::\rightarrow x$\\
            $DASGC4$ & $b::\rightarrow \upsilon^n_{\textrm{abs}}(x) = \upsilon^n_{\textrm{abs}}(b::\rightarrow x)$\\
            $DASGC5$ & $b::\rightarrow(x\parallel y) = (b::\rightarrow x)\parallel (b::\rightarrow y)$\\
            $DASGC6$ & $b::\rightarrow(x\mid y) = (b::\rightarrow x)\mid (b::\rightarrow y)$\\
            $DASGC7$ & $b::\rightarrow\Theta(x) = \Theta(b::\rightarrow x)$\\
            $DASGC8$ & $b::\rightarrow(x\triangleleft y) = (b::\rightarrow x)\triangleleft (b::\rightarrow y)$\\
            $DASGC9$ & $b::\rightarrow\partial_H(x) = \partial_H(b::\rightarrow x)$\\
            $DASGC10$ & $b::\rightarrow(\surd_d i.F(i)) = \surd_d i.(\overline{\upsilon}^i_{\textrm{abs}}(b)::\rightarrow F(i))$\\
            $DASGC11$ & $(\surd_d i.C(i))::\rightarrow x = \surd_d i.(C(i)::\rightarrow \overline{\upsilon}^i_{\textrm{abs}}(x))$\\
        \end{tabular}
        \caption{Axioms of conditionals $(n\geq 0)$}
        \label{LOForAPTCDATDIAC2}
    \end{table}
\end{center}

The operational semantics of $\textrm{APTC}^{\textrm{dat}}\surd\textrm{C}$ are defined by the transition rules in Table \ref{TRForAPTCDATDIAC} and Table \ref{DefsForAPTCDATDIAC}.

\begin{center}
    \begin{table}
        $$\frac{\langle x,n\rangle\xrightarrow{a}\langle x',n\rangle}{\langle \mathbf{t}::\rightarrow x, n\rangle\xrightarrow{a}\langle x',n\rangle}
        \quad\frac{\langle x,n\rangle\xrightarrow{a}\langle \surd,n\rangle}{\langle \mathbf{t}::\rightarrow x, n\rangle\xrightarrow{a}\langle \surd,n\rangle}$$

        $$\frac{\langle x,n\rangle\mapsto^{m}\langle x,n+m\rangle}{\langle \mathbf{t}::\rightarrow x, n\rangle\mapsto^{m}\langle \mathbf{t}::\rightarrow x,n+m\rangle}
        \quad\frac{\langle x, n\rangle\uparrow}{\langle \mathbf{t}::\rightarrow x,n\rangle\uparrow}
        \quad\frac{}{\langle \mathbf{f}::\rightarrow x,n\rangle\uparrow}$$
    \caption{Transition rules of $\textrm{APTC}^{\textrm{dat}}\surd\textrm{C}(a\in A, m>0, n\geq 0)$}
    \label{TRForAPTCDATDIAC}
    \end{table}
\end{center}

\begin{center}
    \begin{table}
        \begin{tabular}{@{}ll@{}}
            \hline
            $c::\rightarrow f = \lambda j.(c(j)::\rightarrow f(j))$ & $\neg c = \lambda j.\neg(c(j))$\\
            $\mathbf{t} = \lambda j.\mathbf{t}$ & $c\wedge d = \lambda j.(c(j)\wedge d(j))$\\
            $\mathbf{f} = \lambda j.\mathbf{f}$ & $c\vee d = \lambda j.(c(j)\vee d(j))$\\
            $\mathbf{sl}(i) = \lambda j.(\textrm{if }j+1=i \textrm{ then }t \textrm{ else } f)$ & $\overline{\upsilon}^i_{\textrm{abs}}(c) = c(i)$\\
            $\mathbf{sl}_>(i) = \lambda j.(\textrm{if }j+1>i \textrm{ then }t \textrm{ else } f)$ & $\surd_d^*(\gamma) = \lambda j.\overline{\upsilon}^j_{\textrm{abs}}(\gamma(j))$\\
        \end{tabular}
        \caption{Definitions of conditional operator on $(\mathcal{CI}/\sim)^*$}
        \label{DefsForAPTCDATDIAC}
    \end{table}
\end{center}

\subsubsection{Elimination}

\begin{definition}[Basic terms of $\textrm{APTC}^{\textrm{dat}}\surd\textrm{C}$]
The set of basic terms of $\textrm{APTC}^{\textrm{dat}}\surd\textrm{C}$, $\mathcal{B}(\textrm{APTC}^{\textrm{dat}}\surd\textrm{C})$, is inductively defined as follows by two auxiliary sets $\mathcal{B}_0(\textrm{APTC}^{\textrm{dat}}\surd\textrm{C})$ and $\mathcal{B}_1(\textrm{APTC}^{\textrm{dat}}\surd\textrm{C})$:
\begin{enumerate}
  \item if $a\in A_{\delta}$, then $\underline{a} \in \mathcal{B}_1(\textrm{APTC}^{\textrm{dat}}\surd\textrm{C})$;
  \item if $a\in A$ and $t\in \mathcal{B}(\textrm{APTC}^{\textrm{dat}}\surd\textrm{C})$, then $\underline{a}\cdot t \in \mathcal{B}_1(\textrm{APTC}^{\textrm{dat}}\surd\textrm{C})$;
  \item if $t,t'\in \mathcal{B}_1(\textrm{APTC}^{\textrm{dat}}\surd\textrm{C})$, then $t+t'\in \mathcal{B}_1(\textrm{APTC}^{\textrm{dat}}\surd\textrm{C})$;
  \item if $t,t'\in \mathcal{B}_1(\textrm{APTC}^{\textrm{dat}}\surd\textrm{C})$, then $t\parallel t'\in \mathcal{B}_1(\textrm{APTC}^{\textrm{dat}}\surd\textrm{C})$;
  \item if $t\in \mathcal{B}_1(\textrm{APTC}^{\textrm{dat}}\surd\textrm{C})$, then $t\in \mathcal{B}_0(\textrm{APTC}^{\textrm{dat}}\surd\textrm{C})$;
  \item if $n>0$ and $t\in \mathcal{B}_0(\textrm{APTC}^{\textrm{dat}}\surd\textrm{C})$, then $\sigma^n_{\textrm{abs}}(t) \in \mathcal{B}_0(\textrm{APTC}^{\textrm{dat}}\surd\textrm{C})$;
  \item if $n>0$, $t\in \mathcal{B}_1(\textrm{APTC}^{\textrm{dat}}\surd\textrm{C})$ and $t'\in \mathcal{B}_0(\textrm{APTC}^{\textrm{dat}}\surd\textrm{C})$, then $t+\sigma^n_{\textrm{abs}}(t') \in \mathcal{B}_0(\textrm{APTC}^{\textrm{dat}}\surd\textrm{C})$;
  \item if $n>0$ and $t\in \mathcal{B}_0(\textrm{APTC}^{\textrm{dat}}\surd\textrm{C})$, then $\surd_d n.t(n) \in \mathcal{B}_0(\textrm{APTC}^{\textrm{dat}}\surd\textrm{C})$;
  \item $\dot{\delta}\in \mathcal{B}(\textrm{APTC}^{\textrm{dat}}\surd\textrm{C})$;
  \item if $t\in \mathcal{B}_0(\textrm{APTC}^{\textrm{dat}}\surd\textrm{C})$, then $t\in \mathcal{B}(\textrm{APTC}^{\textrm{dat}}\surd\textrm{C})$.
\end{enumerate}
\end{definition}

\begin{theorem}[Elimination theorem]
Let $p$ be a closed $\textrm{APTC}^{\textrm{dat}}\surd\textrm{C}$ term. Then there is a basic $\textrm{APTC}^{\textrm{dat}}\surd\textrm{C}$ term $q$ such that $\textrm{APTC}^{\textrm{dat}}\surd\textrm{C}\vdash p=q$.
\end{theorem}

\begin{proof}
It is sufficient to induct on the structure of the closed $\textrm{APTC}^{\textrm{dat}}\surd\textrm{C}$ term $p$. It can be proven that $p$ combined by the constants and operators of $\textrm{APTC}^{\textrm{dat}}\surd\textrm{C}$ exists an equal basic term $q$, and the other operators not included in the basic terms, such as $\upsilon_{\textrm{abs}}$, $\overline{\upsilon}_{\textrm{abs}}$, $\between$, $\mid$, $\partial_H$, $\Theta$, $\triangleleft$, and the constants and operators related to conditions can be eliminated.
\end{proof}

\subsubsection{Congruence}

\begin{theorem}[Congruence of $\textrm{APTC}^{\textrm{dat}}\surd\textrm{C}$]
Truly concurrent bisimulation equivalences $\sim_p$, $\sim_s$ and $\sim_{hp}$ are all congruences with respect to $\textrm{APTC}^{\textrm{dat}}\surd\textrm{C}$. That is,
\begin{itemize}
  \item pomset bisimulation equivalence $\sim_{p}$ is a congruence with respect to $\textrm{APTC}^{\textrm{dat}}\surd\textrm{C}$;
  \item step bisimulation equivalence $\sim_{s}$ is a congruence with respect to $\textrm{APTC}^{\textrm{dat}}\surd\textrm{C}$;
  \item hp-bisimulation equivalence $\sim_{hp}$ is a congruence with respect to $\textrm{APTC}^{\textrm{dat}}\surd\textrm{C}$.
\end{itemize}
\end{theorem}

\begin{proof}
It is easy to see that $\sim_p$, $\sim_s$, and $\sim_{hp}$ are all equivalent relations on $\textrm{APTC}^{\textrm{dat}}\surd\textrm{C}$ terms, it is only sufficient to prove that $\sim_p$, $\sim_s$, and $\sim_{hp}$ are all preserved by the operators $\sigma^n_{\textrm{abs}}$, $\upsilon^n_{\textrm{abs}}$ and $\overline{\upsilon}^n_{\textrm{abs}}$. It is trivial and we omit it.
\end{proof}

\subsubsection{Soundness}

\begin{theorem}[Soundness of $\textrm{APTC}^{\textrm{dat}}\surd\textrm{C}$]
The axiomatization of $\textrm{APTC}^{\textrm{dat}}\surd\textrm{C}$ is sound modulo truly concurrent bisimulation equivalences $\sim_{p}$, $\sim_{s}$, and $\sim_{hp}$. That is,
\begin{enumerate}
  \item let $x$ and $y$ be $\textrm{APTC}^{\textrm{dat}}\surd\textrm{C}$ terms. If $\textrm{APTC}^{\textrm{dat}}\surd\textrm{C}\vdash x=y$, then $x\sim_{s} y$;
  \item let $x$ and $y$ be $\textrm{APTC}^{\textrm{dat}}\surd\textrm{C}$ terms. If $\textrm{APTC}^{\textrm{dat}}\surd\textrm{C}\vdash x=y$, then $x\sim_{p} y$;
  \item let $x$ and $y$ be $\textrm{APTC}^{\textrm{dat}}\surd\textrm{C}$ terms. If $\textrm{APTC}^{\textrm{dat}}\surd\textrm{C}\vdash x=y$, then $x\sim_{hp} y$.
\end{enumerate}
\end{theorem}

\begin{proof}
Since $\sim_p$, $\sim_s$, and $\sim_{hp}$ are both equivalent and congruent relations, we only need to check if each axiom in Table \ref{LOForAPTCDATDIAC}, Table \ref{LOForAPTCDATDIAC1} and Table \ref{LOForAPTCDATDIAC2} is sound modulo $\sim_p$, $\sim_s$, and $\sim_{hp}$ respectively.

\begin{enumerate}
  \item Each axiom in Table \ref{LOForAPTCDATDIAC}, Table \ref{LOForAPTCDATDIAC1} and Table \ref{LOForAPTCDATDIAC2} can be checked that it is sound modulo step bisimulation equivalence, by transition rules of conditionals in Table \ref{TRForAPTCDATDIAC}. We omit them.
  \item From the definition of pomset bisimulation, we know that pomset bisimulation is defined by pomset transitions, which are labeled by pomsets. In a pomset transition, the events (actions) in the pomset are either within causality relations (defined by $\cdot$) or in concurrency (implicitly defined by $\cdot$ and $+$, and explicitly defined by $\between$), of course, they are pairwise consistent (without conflicts). We have already proven the case that all events are pairwise concurrent (soundness modulo step bisimulation), so, we only need to prove the case of events in causality. Without loss of generality, we take a pomset of $P=\{\underline{a},\underline{b}:\underline{a}\cdot \underline{b}\}$. Then the pomset transition labeled by the above $P$ is just composed of one single event transition labeled by $\underline{a}$ succeeded by another single event transition labeled by $\underline{b}$, that is, $\xrightarrow{P}=\xrightarrow{a}\xrightarrow{b}$.

        Similarly to the proof of soundness modulo step bisimulation equivalence, we can prove that each axiom in Table \ref{AxiomsForAPTCDATDIA} is sound modulo pomset bisimulation equivalence, we omit them.
  \item From the definition of hp-bisimulation, we know that hp-bisimulation is defined on the posetal product $(C_1,f,C_2),f:C_1\rightarrow C_2\textrm{ isomorphism}$. Two process terms $s$ related to $C_1$ and $t$ related to $C_2$, and $f:C_1\rightarrow C_2\textrm{ isomorphism}$. Initially, $(C_1,f,C_2)=(\emptyset,\emptyset,\emptyset)$, and $(\emptyset,\emptyset,\emptyset)\in\sim_{hp}$. When $s\xrightarrow{a}s'$ ($C_1\xrightarrow{a}C_1'$), there will be $t\xrightarrow{a}t'$ ($C_2\xrightarrow{a}C_2'$), and we define $f'=f[a\mapsto a]$. Then, if $(C_1,f,C_2)\in\sim_{hp}$, then $(C_1',f',C_2')\in\sim_{hp}$.

        Similarly to the proof of soundness modulo pomset bisimulation equivalence, we can prove that each axiom in Table \ref{AxiomsForAPTCDATDIA} is sound modulo hp-bisimulation equivalence, we just need additionally to check the above conditions on hp-bisimulation, we omit them.
\end{enumerate}

\end{proof}

\subsubsection{Completeness}

\begin{theorem}[Completeness of $\textrm{APTC}^{\textrm{dat}}\surd\textrm{C}$]
The axiomatization of $\textrm{APTC}^{\textrm{dat}}\surd\textrm{C}$ is complete modulo truly concurrent bisimulation equivalences $\sim_{p}$, $\sim_{s}$, and $\sim_{hp}$. That is,
\begin{enumerate}
  \item let $p$ and $q$ be closed $\textrm{APTC}^{\textrm{dat}}\surd\textrm{C}$ terms, if $p\sim_{s} q$ then $p=q$;
  \item let $p$ and $q$ be closed $\textrm{APTC}^{\textrm{dat}}\surd\textrm{C}$ terms, if $p\sim_{p} q$ then $p=q$;
  \item let $p$ and $q$ be closed $\textrm{APTC}^{\textrm{dat}}\surd\textrm{C}$ terms, if $p\sim_{hp} q$ then $p=q$.
\end{enumerate}

\end{theorem}

\begin{proof}
\begin{enumerate}
  \item Firstly, by the elimination theorem of $\textrm{APTC}^{\textrm{dat}}\surd\textrm{C}$, we know that for each closed $\textrm{APTC}^{\textrm{dat}}\surd\textrm{C}$ term $p$, there exists a closed basic $\textrm{APTC}^{\textrm{dat}}\surd\textrm{C}$ term $p'$, such that $\textrm{APTC}^{\textrm{dat}}\surd\textrm{C}\vdash p=p'$, so, we only need to consider closed basic $\textrm{APTC}^{\textrm{dat}}\surd\textrm{C}$ terms.

        The basic terms modulo associativity and commutativity (AC) of conflict $+$ (defined by axioms $A1$ and $A2$ in Table \ref{AxiomsForBATCDAT}) and associativity and commutativity (AC) of parallel $\parallel$ (defined by axioms $P2$ and $P3$ in Table \ref{AxiomsForAPTCDAT}), and these equivalences is denoted by $=_{AC}$. Then, each equivalence class $s$ modulo AC of $+$ and $\parallel$ has the following normal form

        $$s_1+\cdots+ s_k$$

        with each $s_i$ either an atomic event or of the form

        $$t_1\cdot\cdots\cdot t_m$$

        with each $t_j$ either an atomic event or of the form

        $$u_1\parallel\cdots\parallel u_n$$

        with each $u_l$ an atomic event, and each $s_i$ is called the summand of $s$.

        Now, we prove that for normal forms $n$ and $n'$, if $n\sim_{s} n'$ then $n=_{AC}n'$. It is sufficient to induct on the sizes of $n$ and $n'$. We can get $n=_{AC} n'$.

        Finally, let $s$ and $t$ be basic $\textrm{APTC}^{\textrm{dat}}\surd\textrm{C}$ terms, and $s\sim_s t$, there are normal forms $n$ and $n'$, such that $s=n$ and $t=n'$. The soundness theorem modulo step bisimulation equivalence yields $s\sim_s n$ and $t\sim_s n'$, so $n\sim_s s\sim_s t\sim_s n'$. Since if $n\sim_s n'$ then $n=_{AC}n'$, $s=n=_{AC}n'=t$, as desired.
  \item This case can be proven similarly, just by replacement of $\sim_{s}$ by $\sim_{p}$.
  \item This case can be proven similarly, just by replacement of $\sim_{s}$ by $\sim_{hp}$.
\end{enumerate}
\end{proof}

\section{Continuous Relative Timing}{\label{srt}}

In this section, we will introduce a version of APTC with relative timing and time measured on a continuous time scale. Measuring time on  a continuous time scale means that timing is now done with respect ro time points on a continuous time scale. With respect to relative timing, timing is relative to the execution time of the previous action, and if the previous action does not exist, the start-up time of the whole process.

Like APTC without timing, let us start with a basic algebra for true concurrency called $\textrm{BATC}^{\textrm{srt}}$ (BATC with continuous relative timing). Then we continue with $\textrm{APTC}^{\textrm{srt}}$ (APTC with continuous relative timing).

\subsection{Basic Definitions}

In this subsection, we will introduce some basic definitions about timing. These basic concepts come from \cite{T3}, we introduce them into the backgrounds of true concurrency.

\begin{definition}[Undelayable actions]
Undelayable actions are defined as atomic processes that perform an action and then terminate successfully. We use a constant $\underline{\underline{a}}$ to represent the undelayable action, that is, the atomic process that performs the action $a$ and then terminates successfully.
\end{definition}

\begin{definition}[Undelayable deadlock]
Undelayable deadlock $\tilde{\tilde{\delta}}$ is an additional process that is neither capable of performing any action nor capable of idling beyond the current point of time.
\end{definition}

\begin{definition}[Relative delay]
The relative delay of the process $p$ for a period of time $r$ ($r\in\mathbb{R}^{\geq}$) is the process that idles for a period of time $r$ and then behaves like $p$. The operator $\sigma_{\textrm{rel}}$ is used to represent the relative delay, and let $\sigma^r_{\textrm{rel}}(t) = r \sigma_{\textrm{rel}} t$.
\end{definition}

\begin{definition}[Deadlocked process]
Deadlocked process $\dot{\delta}$ is an additional process that has deadlocked before the current point of time. After a delay of a period of time, the undelayable deadlock $\tilde{\tilde{\delta}}$ and the deadlocked process $\dot{\delta}$ are indistinguishable from each other.
\end{definition}

\begin{definition}[Truly concurrent bisimulation equivalences with time-related capabilities]\label{TBTTC3}
The following requirement with time-related capabilities is added to truly concurrent bisimulation equivalences $\sim_{p}$, $\sim_{s}$, $\sim_{hp}$ and $\sim_{hhp}$:
\begin{itemize}
  \item if a process is capable of first idling till a period of time and next going on as another process, then any equivalent process must be capable of first idling till the same period of time and next going on as a process equivalent to the other process;
  \item if a process has deadlocked before the current point of time, then any equivalent process must have deadlocked before the current point of time.
\end{itemize}
\end{definition}

\begin{definition}[Integration]
Let $f$ be a function from $\mathbb{R}^{\geq}$ to processes with continuous relative timing and $V\subseteq \mathbb{R}^{\geq}$. The integration $\int$ of $f$ over $V$ is the process that behaves like one of the process in $\{f(r)|r\in V\}$.
\end{definition}

\begin{definition}[Relative time-out]
The relative time-out $\upsilon_{\textrm{rel}}$ of a process $p$ after a period of $r$ ($p\in\mathbb{R}^{\geq}$) behaves either like the part of $p$ that does not idle till the $p$th-next time slice, or like the deadlocked process after a delay of $r$ time units if $p$ is capable of idling for the period of time $r$; otherwise, like $p$. And let $\upsilon^r_{\textrm{rel}}(t) = r \upsilon_{\textrm{rel}} t$.
\end{definition}

\begin{definition}[Relative initialization]
The relative initialization $\overline{\upsilon}_{\textrm{rel}}$ of a process $p$ after a period of time $r$ ($r\in\mathbb{R}^{\geq}$) behaves like the part of $p$ that idles for a period of time $r$ if $p$ is capable of idling for $r$; otherwise, like the deadlocked process after a delay of $r$. And we let $\overline{\upsilon}^r_{\textrm{rel}}(t) = r \overline{\upsilon}_{\textrm{rel}} t$.
\end{definition}

\subsection{Basic Algebra for True Concurrency with Continuous Relative Timing}

In this subsection, we will introduce the theory $\textrm{BATC}^{\textrm{srt}}$.

\subsubsection{The Theory $\textrm{BATC}^{\textrm{srt}}$}

\begin{definition}[Signature of $\textrm{BATC}^{\textrm{srt}}$]
The signature of $\textrm{BATC}^{\textrm{srt}}$ consists of the sort $\mathcal{P}_{\textrm{rel}}$ of processes with continuous relative timing, the undelayable action constants $\tilde{\tilde{a}}: \rightarrow\mathcal{P}_{\textrm{rel}}$ for each $a\in A$, the undelayable deadlock constant $\tilde{\tilde{\delta}}: \rightarrow \mathcal{P}_{\textrm{rel}}$, the alternative composition operator $+: \mathcal{P}_{\textrm{rel}}\times\mathcal{P}_{\textrm{rel}} \rightarrow \mathcal{P}_{\textrm{rel}}$, the sequential composition operator $\cdot: \mathcal{P}_{\textrm{rel}} \times \mathcal{P}_{\textrm{rel}} \rightarrow \mathcal{P}_{\textrm{rel}}$, the relative delay operator $\sigma_{\textrm{rel}}: \mathbb{R}^{\geq}\times \mathcal{P}_{\textrm{rel}} \rightarrow \mathcal{P}_{\textrm{rel}}$, the deadlocked process constant $\dot{\delta}: \rightarrow \mathcal{P}_{\textrm{rel}}$, the relative time-out operator $\upsilon_{\textrm{rel}}: \mathbb{R}^{\geq}\times\mathcal{P}_{\textrm{rel}} \rightarrow\mathcal{P}_{\textrm{rel}}$ and the relative initialization operator $\overline{\upsilon}_{\textrm{rel}}: \mathbb{R}^{\geq}\times\mathcal{P}_{\textrm{rel}} \rightarrow\mathcal{P}_{\textrm{rel}}$.
\end{definition}

The set of axioms of $\textrm{BATC}^{\textrm{srt}}$ consists of the laws given in Table \ref{AxiomsForBATCSRT}.

\begin{center}
    \begin{table}
        \begin{tabular}{@{}ll@{}}
            \hline No. &Axiom\\
            $A1$ & $x+ y = y+ x$\\
            $A2$ & $(x+ y)+ z = x+ (y+ z)$\\
            $A3$ & $x+ x = x$\\
            $A4$ & $(x+ y)\cdot z = x\cdot z + y\cdot z$\\
            $A5$ & $(x\cdot y)\cdot z = x\cdot(y\cdot z)$\\
            $A6ID$ & $x + \dot{\delta} = x$\\
            $A7ID$ & $\dot{\delta}\cdot x = \dot{\delta}$\\
            $SRT1$ & $\sigma^0_{\textrm{rel}}(x) = x$\\
            $SRT2$ & $\sigma^q_{\textrm{rel}}( \sigma^p_{\textrm{rel}}(x)) = \sigma^{q+p}_{\textrm{rel}}(x)$\\
            $SRT3$ & $\sigma^p_{\textrm{rel}}(x) + \sigma^p_{\textrm{rel}}(y) = \sigma^p_{\textrm{rel}}(x+y)$\\
            $SRT4$ & $\sigma^p_{\textrm{rel}}(x)\cdot y = \sigma^p_{\textrm{rel}}(x\cdot y)$\\
            $A6SRa$ & $\tilde{\tilde{a}} + \tilde{\tilde{\delta}} = \tilde{\tilde{a}}$\\
            $A6SRb$ & $\sigma^r_{\textrm{rel}}(x) + \tilde{\tilde{\delta}} = \sigma^r_{\textrm{rel}}(x)$\\
            $SRTO0$ & $\upsilon^p_{\textrm{rel}}(\dot{\delta}) = \dot{\delta}$\\
            $SRTO1$ & $\upsilon^0_{\textrm{rel}}(x) = \dot(\delta)$\\
            $SRTO2$ & $\upsilon^{r}_{\textrm{rel}}(\tilde{\tilde{a}}) = \tilde{\tilde{a}}$\\
            $SRTO3$ & $\upsilon^{q+p}_{\textrm{rel}} (\sigma^p_{\textrm{rel}}(x)) = \sigma^p_{\textrm{rel}}(\upsilon^q_{\textrm{rel}}(x))$\\
            $SRTO4$ & $\upsilon^p_{\textrm{rel}}(x+y) = \upsilon^p_{\textrm{rel}}(x) + \upsilon^p_{\textrm{rel}}(y)$\\
            $SRTO5$ & $\upsilon^p_{\textrm{rel}}(x\cdot y) = \upsilon^p_{\textrm{rel}}(x)\cdot y$\\
            $SRI0$ & $\overline{\upsilon}^p_{\textrm{rel}}(\dot{\delta}) = \sigma^p_{\textrm{rel}}(\dot{\delta})$\\
            $SRI1$ & $\overline{\upsilon}^0_{\textrm{rel}}(x) = x$\\
            $SRI2$ & $\overline{\upsilon}^{r}_{\textrm{rel}}(\tilde{\tilde{a}}) = \sigma^p_{\textrm{rel}}(\tilde{\tilde{\delta}})$\\
            $SRI3$ & $\overline{\upsilon}^{q+p}_{\textrm{rel}} (\sigma^p_{\textrm{rel}}(x)) = \sigma^p_{\textrm{rel}}(\overline{\upsilon}^q_{\textrm{rel}}(x))$\\
            $SRI4$ & $\overline{\upsilon}^p_{\textrm{rel}}(x+y) = \overline{\upsilon}^p_{\textrm{rel}}(x) + \overline{\upsilon}^p_{\textrm{rel}}(y)$\\
            $SRI5$ & $\overline{\upsilon}^p_{\textrm{rel}}(x\cdot y) = \overline{\upsilon}^p_{\textrm{rel}}(x)\cdot y$\\
        \end{tabular}
        \caption{Axioms of $\textrm{BATC}^{\textrm{srt}}(a\in A_{\delta}, p,q\geq 0, r>0)$}
        \label{AxiomsForBATCSRT}
    \end{table}
\end{center}

The operational semantics of $\textrm{BATC}^{\textrm{srt}}$ are defined by the transition rules in Table \ref{TRForBATCSRT}. Where $\uparrow$ is a unary deadlocked predicate, and $t\nuparrow \triangleq \neg(t\uparrow)$; $t\mapsto^q t'$ means that process $t$ is capable of first idling for $q$, and then proceeding as process $t'$.

\begin{center}
    \begin{table}
        $$\frac{}{\dot{\delta}\uparrow}
        \quad\frac{}{\tilde{\tilde{a}}\xrightarrow{a}\surd}
        \quad\frac{x\xrightarrow{a}x'}{\sigma^0_{\textrm{rel}}(x) \xrightarrow{a}x'}
        \quad\frac{x\xrightarrow{a}\surd}{\sigma^0_{\textrm{rel}}(x) \xrightarrow{a}\surd}
        \quad\frac{x\uparrow}{\sigma^0_{\textrm{rel}}(x)\uparrow}$$

        $$\frac{}{\sigma^{r+s}_{\textrm{rel}}(x)\mapsto^{r} \sigma^{s}_{\textrm{rel}}(x)}
        \quad \frac{x\nuparrow}{\sigma^r_{\textrm{rel}}(x)\mapsto^r x}
        \quad\frac{x\mapsto^r x'}{\sigma^p_{\textrm{rel}}(x) \mapsto^{r+p} x'}$$

        $$\frac{x\xrightarrow{a}x'}{x+ y\xrightarrow{a}x'} \quad\frac{y\xrightarrow{a}y'}{x+ y\xrightarrow{a}y'}
        \quad\frac{x\xrightarrow{a}\surd}{x+ y\xrightarrow{a}\surd}
        \quad\frac{y\xrightarrow{a}\surd}{x+ y\xrightarrow{a}\surd}$$

        $$\frac{x\mapsto^{r}x'\quad y\nmapsto^r}{x+ y\mapsto^{r}x'} \quad\frac{x\nmapsto^{r}\quad y\mapsto^r y'}{x+ y\mapsto^{r}y'}
        \quad\frac{x\mapsto^{r}x'\quad y\mapsto^r y'}{x+ y\mapsto^{r}x'+y'}
        \quad\frac{x\uparrow\quad y\uparrow}{x+ y\uparrow}$$

        $$\frac{x\xrightarrow{a}\surd}{x\cdot y\xrightarrow{a} y} \quad\frac{x\xrightarrow{a}x'}{x\cdot y\xrightarrow{a}x'\cdot y}
        \quad \frac{x\mapsto^{r}x'}{x\cdot y\mapsto^{r}x'\cdot y}
        \quad \frac{x\uparrow}{x\cdot y\uparrow}$$

        $$\frac{x\xrightarrow{a}x'}{\upsilon^{r}_{\textrm{rel}}(x) \xrightarrow{a}x'}
        \quad\frac{x\xrightarrow{a}\surd}{\upsilon^{r}_{\textrm{rel}}(x) \xrightarrow{a}\surd}$$

        $$\frac{x\mapsto^r x'}{\upsilon^{r+s}_{\textrm{rel}}(x) \mapsto^r \upsilon^{s}_{\textrm{rel}}(x')}
        \quad\frac{}{\upsilon^0_{\textrm{rel}}(x)\uparrow}
        \quad\frac{x\uparrow}{\upsilon^{r}_{\textrm{rel}}(x)\uparrow}$$

        $$\frac{x\xrightarrow{a}x'}{\overline{\upsilon}^{0}_{\textrm{rel}}(x) \xrightarrow{a}x'}
        \quad\frac{x\xrightarrow{a}\surd}{\overline{\upsilon}^{0}_{\textrm{rel}}(x) \xrightarrow{a}\surd}
        \quad\frac{x\mapsto^r x'}{\overline{\upsilon}^p_{\textrm{rel}}(x)\mapsto^r x'}p\leq r$$

        $$\frac{x\mapsto^r x'}{\overline{\upsilon}^{r+s}_{\textrm{rel}}(x) \mapsto^r \upsilon^{s}_{\textrm{rel}}(x')}
        \quad\frac{x\nmapsto^r}{\overline{\upsilon}^{r+s}_{\textrm{rel}}(x) \mapsto^r \overline{\upsilon}^{s}_{\textrm{rel}}(\dot{\delta})}
        \quad\frac{x\uparrow}{\overline{\upsilon}^{0}_{\textrm{rel}}(x)\uparrow}$$
        \caption{Transition rules of $\textrm{BATC}^{\textrm{srt}}(a\in A, p\geq 0, r,s>0)$}
        \label{TRForBATCSRT}
    \end{table}
\end{center}

\subsubsection{Elimination}

\begin{definition}[Basic terms of $\textrm{BATC}^{\textrm{srt}}$]
The set of basic terms of $\textrm{BATC}^{\textrm{srt}}$, $\mathcal{B}(\textrm{BATC}^{\textrm{srt}})$, is inductively defined as follows by two auxiliary sets $\mathcal{B}_0(\textrm{BATC}^{\textrm{srt}})$ and $\mathcal{B}_1(\textrm{BATC}^{\textrm{srt}})$:
\begin{enumerate}
  \item if $a\in A_{\delta}$, then $\tilde{\tilde{a}} \in \mathcal{B}_1(\textrm{BATC}^{\textrm{srt}})$;
  \item if $a\in A$ and $t\in \mathcal{B}(\textrm{BATC}^{\textrm{srt}})$, then $\tilde{\tilde{a}}\cdot t \in \mathcal{B}_1(\textrm{BATC}^{\textrm{srt}})$;
  \item if $t,t'\in \mathcal{B}_1(\textrm{BATC}^{\textrm{srt}})$, then $t+t'\in \mathcal{B}_1(\textrm{BATC}^{\textrm{srt}})$;
  \item if $t\in \mathcal{B}_1(\textrm{BATC}^{\textrm{srt}})$, then $t\in \mathcal{B}_0(\textrm{BATC}^{\textrm{srt}})$;
  \item if $p>0$ and $t\in \mathcal{B}_0(\textrm{BATC}^{\textrm{srt}})$, then $\sigma^p_{\textrm{rel}}(t) \in \mathcal{B}_0(\textrm{BATC}^{\textrm{srt}})$;
  \item if $p>0$, $t\in \mathcal{B}_1(\textrm{BATC}^{\textrm{srt}})$ and $t'\in \mathcal{B}_0(\textrm{BATC}^{\textrm{srt}})$, then $t+\sigma^p_{\textrm{rel}}(t') \in \mathcal{B}_0(\textrm{BATC}^{\textrm{srt}})$;
  \item $\dot{\delta}\in \mathcal{B}(\textrm{BATC}^{\textrm{srt}})$;
  \item if $t\in \mathcal{B}_0(\textrm{BATC}^{\textrm{srt}})$, then $t\in \mathcal{B}(\textrm{BATC}^{\textrm{srt}})$.
\end{enumerate}
\end{definition}

\begin{theorem}[Elimination theorem]
Let $p$ be a closed $\textrm{BATC}^{\textrm{srt}}$ term. Then there is a basic $\textrm{BATC}^{\textrm{srt}}$ term $q$ such that $\textrm{BATC}^{\textrm{srt}}\vdash p=q$.
\end{theorem}

\begin{proof}
It is sufficient to induct on the structure of the closed $\textrm{BATC}^{\textrm{srt}}$ term $p$. It can be proven that $p$ combined by the constants and operators of $\textrm{BATC}^{\textrm{srt}}$ exists an equal basic term $q$, and the other operators not included in the basic terms, such as $\upsilon_{\textrm{rel}}$ and $\overline{\upsilon}_{\textrm{rel}}$ can be eliminated.
\end{proof}

\subsubsection{Connections}

\begin{theorem}[Generalization of $\textrm{BATC}^{\textrm{srt}}$]

By the definitions of $a=\tilde{\tilde{a}}$ for each $a\in A$ and $\delta=\tilde{\tilde{\delta}}$, $\textrm{BATC}^{\textrm{srt}}$ is a generalization of $BATC$.
\end{theorem}

\begin{proof}
It follows from the following two facts.

\begin{enumerate}
  \item The transition rules of $BATC$ in section \ref{tcpa} are all source-dependent;
  \item The sources of the transition rules of $\textrm{BATC}^{\textrm{srt}}$ contain an occurrence of $\dot{\delta}$, $\tilde{\tilde{a}}$, $\sigma^p_{\textrm{rel}}$, $\upsilon^p_{\textrm{rel}}$ and $\overline{\upsilon}^p_{\textrm{rel}}$.
\end{enumerate}

So, $BATC$ is an embedding of $\textrm{BATC}^{\textrm{srt}}$, as desired.
\end{proof}

\subsubsection{Congruence}

\begin{theorem}[Congruence of $\textrm{BATC}^{\textrm{srt}}$]
Truly concurrent bisimulation equivalences are all congruences with respect to $\textrm{BATC}^{\textrm{srt}}$. That is,
\begin{itemize}
  \item pomset bisimulation equivalence $\sim_{p}$ is a congruence with respect to $\textrm{BATC}^{\textrm{srt}}$;
  \item step bisimulation equivalence $\sim_{s}$ is a congruence with respect to $\textrm{BATC}^{\textrm{srt}}$;
  \item hp-bisimulation equivalence $\sim_{hp}$ is a congruence with respect to $\textrm{BATC}^{\textrm{srt}}$;
  \item hhp-bisimulation equivalence $\sim_{hhp}$ is a congruence with respect to $\textrm{BATC}^{\textrm{srt}}$.
\end{itemize}
\end{theorem}

\begin{proof}
It is easy to see that $\sim_p$, $\sim_s$, $\sim_{hp}$ and $\sim_{hhp}$ are all equivalent relations on $\textrm{BATC}^{\textrm{srt}}$ terms, it is only sufficient to prove that $\sim_p$, $\sim_s$, $\sim_{hp}$ and $\sim_{hhp}$ are all preserved by the operators $\sigma^p_{\textrm{rel}}$, $\upsilon^p_{\textrm{rel}}$ and $\overline{\upsilon}^p_{\textrm{rel}}$. It is trivial and we omit it.
\end{proof}

\subsubsection{Soundness}

\begin{theorem}[Soundness of $\textrm{BATC}^{\textrm{srt}}$]
The axiomatization of $\textrm{BATC}^{\textrm{srt}}$ is sound modulo truly concurrent bisimulation equivalences $\sim_{p}$, $\sim_{s}$, $\sim_{hp}$ and $\sim_{hhp}$. That is,
\begin{enumerate}
  \item let $x$ and $y$ be $\textrm{BATC}^{\textrm{srt}}$ terms. If $\textrm{BATC}^{\textrm{srt}}\vdash x=y$, then $x\sim_{s} y$;
  \item let $x$ and $y$ be $\textrm{BATC}^{\textrm{srt}}$ terms. If $\textrm{BATC}^{\textrm{srt}}\vdash x=y$, then $x\sim_{p} y$;
  \item let $x$ and $y$ be $\textrm{BATC}^{\textrm{srt}}$ terms. If $\textrm{BATC}^{\textrm{srt}}\vdash x=y$, then $x\sim_{hp} y$;
  \item let $x$ and $y$ be $\textrm{BATC}^{\textrm{srt}}$ terms. If $\textrm{BATC}^{\textrm{srt}}\vdash x=y$, then $x\sim_{hhp} y$.
\end{enumerate}
\end{theorem}

\begin{proof}
Since $\sim_p$, $\sim_s$, $\sim_{hp}$ and $\sim_{hhp}$ are both equivalent and congruent relations, we only need to check if each axiom in Table \ref{AxiomsForBATCSRT} is sound modulo $\sim_p$, $\sim_s$, $\sim_{hp}$ and $\sim_{hhp}$ respectively.

\begin{enumerate}
  \item We only check the soundness of the non-trivial axiom $SRTO3$ modulo $\sim_s$.
        Let $p$ be $\textrm{BATC}^{\textrm{drt}}$ processes, and $\upsilon^{s+r}_{\textrm{rel}} (\sigma^r_{\textrm{rel}}(p)) = \sigma^r_{\textrm{rel}}(\upsilon^s_{\textrm{rel}}(p))$, it is sufficient to prove that $\upsilon^{s+r}_{\textrm{rel}} (\sigma^r_{\textrm{rel}}(p)) \sim_{s} \sigma^r_{\textrm{rel}}(\upsilon^s_{\textrm{rel}}(p))$. By the transition rules of operator $\sigma^r_{\textrm{rel}}$ and $\upsilon^s_{\textrm{rel}}$ in Table \ref{TRForBATCSRT}, we get

        $$\frac{}{\upsilon^{s+r}_{\textrm{rel}} (\sigma^r_{\textrm{rel}}(p))\mapsto^r \upsilon^{s}_{\textrm{rel}} (\sigma^0_{\textrm{rel}}(p))}$$

        $$\frac{}{\sigma^r_{\textrm{rel}}(\upsilon^s_{\textrm{rel}}(p))\mapsto^r \sigma^0_{\textrm{rel}}(\upsilon^s_{\textrm{rel}}(p))}$$

        There are several cases:

        $$\frac{p\xrightarrow{a} \surd}{\upsilon^{s}_{\textrm{rel}} (\sigma^0_{\textrm{rel}}(p))\xrightarrow{a}\surd}$$

        $$\frac{p\xrightarrow{a} \surd}{\sigma^0_{\textrm{rel}}(\upsilon^s_{\textrm{rel}}(p))\xrightarrow{a}\surd}$$

        $$\frac{p\xrightarrow{a} p'}{\upsilon^{s}_{\textrm{rel}} (\sigma^0_{\textrm{rel}}(p))\xrightarrow{a}p'}$$

        $$\frac{p\xrightarrow{a} p'}{\sigma^0_{\textrm{rel}}(\upsilon^s_{\textrm{rel}}(p))\xrightarrow{a}p'}$$

        $$\frac{p \uparrow}{\upsilon^{s}_{\textrm{rel}} (\sigma^0_{\textrm{rel}}(p))\uparrow}$$

        $$\frac{p \uparrow}{\sigma^0_{\textrm{rel}}(\upsilon^s_{\textrm{rel}}(p))\uparrow}$$

        So, we see that each case leads to $\upsilon^{s+r}_{\textrm{rel}} (\sigma^r_{\textrm{rel}}(p)) \sim_{s} \sigma^r_{\textrm{rel}}(\upsilon^s_{\textrm{rel}}(p))$, as desired.
  \item From the definition of pomset bisimulation, we know that pomset bisimulation is defined by pomset transitions, which are labeled by pomsets. In a pomset transition, the events (actions) in the pomset are either within causality relations (defined by $\cdot$) or in concurrency (implicitly defined by $\cdot$ and $+$, and explicitly defined by $\between$), of course, they are pairwise consistent (without conflicts). We have already proven the case that all events are pairwise concurrent (soundness modulo step bisimulation), so, we only need to prove the case of events in causality. Without loss of generality, we take a pomset of $P=\{\tilde{\tilde{a}},\tilde{\tilde{b}}:\tilde{\tilde{a}}\cdot \tilde{\tilde{b}}\}$. Then the pomset transition labeled by the above $P$ is just composed of one single event transition labeled by $\tilde{\tilde{a}}$ succeeded by another single event transition labeled by $\tilde{\tilde{b}}$, that is, $\xrightarrow{P}=\xrightarrow{a}\xrightarrow{b}$.

        Similarly to the proof of soundness modulo step bisimulation equivalence, we can prove that each axiom in Table \ref{AxiomsForBATCSRT} is sound modulo pomset bisimulation equivalence, we omit them.
  \item From the definition of hp-bisimulation, we know that hp-bisimulation is defined on the posetal product $(C_1,f,C_2),f:C_1\rightarrow C_2\textrm{ isomorphism}$. Two process terms $s$ related to $C_1$ and $t$ related to $C_2$, and $f:C_1\rightarrow C_2\textrm{ isomorphism}$. Initially, $(C_1,f,C_2)=(\emptyset,\emptyset,\emptyset)$, and $(\emptyset,\emptyset,\emptyset)\in\sim_{hp}$. When $s\xrightarrow{a}s'$ ($C_1\xrightarrow{a}C_1'$), there will be $t\xrightarrow{a}t'$ ($C_2\xrightarrow{a}C_2'$), and we define $f'=f[a\mapsto a]$. Then, if $(C_1,f,C_2)\in\sim_{hp}$, then $(C_1',f',C_2')\in\sim_{hp}$.

        Similarly to the proof of soundness modulo pomset bisimulation equivalence, we can prove that each axiom in Table \ref{AxiomsForBATCSRT} is sound modulo hp-bisimulation equivalence, we just need additionally to check the above conditions on hp-bisimulation, we omit them.
  \item We just need to add downward-closed condition to the soundness modulo hp-bisimulation equivalence, we omit them.
\end{enumerate}

\end{proof}

\subsubsection{Completeness}

\begin{theorem}[Completeness of $\textrm{BATC}^{\textrm{srt}}$]
The axiomatization of $\textrm{BATC}^{\textrm{srt}}$ is complete modulo truly concurrent bisimulation equivalences $\sim_{p}$, $\sim_{s}$, $\sim_{hp}$ and $\sim_{hhp}$. That is,
\begin{enumerate}
  \item let $p$ and $q$ be closed $\textrm{BATC}^{\textrm{srt}}$ terms, if $p\sim_{s} q$ then $p=q$;
  \item let $p$ and $q$ be closed $\textrm{BATC}^{\textrm{srt}}$ terms, if $p\sim_{p} q$ then $p=q$;
  \item let $p$ and $q$ be closed $\textrm{BATC}^{\textrm{srt}}$ terms, if $p\sim_{hp} q$ then $p=q$;
  \item let $p$ and $q$ be closed $\textrm{BATC}^{\textrm{srt}}$ terms, if $p\sim_{hhp} q$ then $p=q$.
\end{enumerate}

\end{theorem}

\begin{proof}
\begin{enumerate}
  \item Firstly, by the elimination theorem of $\textrm{BATC}^{\textrm{srt}}$, we know that for each closed $\textrm{BATC}^{\textrm{srt}}$ term $p$, there exists a closed basic $\textrm{BATC}^{\textrm{srt}}$ term $p'$, such that $\textrm{BATC}^{\textrm{dat}}\vdash p=p'$, so, we only need to consider closed basic $\textrm{BATC}^{\textrm{srt}}$ terms.

        The basic terms modulo associativity and commutativity (AC) of conflict $+$ (defined by axioms $A1$ and $A2$ in Table \ref{AxiomsForBATCSRT}), and this equivalence is denoted by $=_{AC}$. Then, each equivalence class $s$ modulo AC of $+$ has the following normal form

        $$s_1+\cdots+ s_k$$

        with each $s_i$ either an atomic event or of the form $t_1\cdot t_2$, and each $s_i$ is called the summand of $s$.

        Now, we prove that for normal forms $n$ and $n'$, if $n\sim_{s} n'$ then $n=_{AC}n'$. It is sufficient to induct on the sizes of $n$ and $n'$. We can get $n=_{AC} n'$.

        Finally, let $s$ and $t$ be basic terms, and $s\sim_s t$, there are normal forms $n$ and $n'$, such that $s=n$ and $t=n'$. The soundness theorem of $\textrm{BATC}^{\textrm{srt}}$ modulo step bisimulation equivalence yields $s\sim_s n$ and $t\sim_s n'$, so $n\sim_s s\sim_s t\sim_s n'$. Since if $n\sim_s n'$ then $n=_{AC}n'$, $s=n=_{AC}n'=t$, as desired.
  \item This case can be proven similarly, just by replacement of $\sim_{s}$ by $\sim_{p}$.
  \item This case can be proven similarly, just by replacement of $\sim_{s}$ by $\sim_{hp}$.
  \item This case can be proven similarly, just by replacement of $\sim_{s}$ by $\sim_{hhp}$.
\end{enumerate}
\end{proof}

\subsection{$\textrm{BATC}^{\textrm{srt}}$ with Integration}

In this subsection, we will introduce the theory $\textrm{BATC}^{\textrm{srt}}$ with integration called $\textrm{BATC}^{\textrm{srt}}\textrm{I}$.

\subsubsection{The Theory $\textrm{BATC}^{\textrm{srt}}\textrm{I}$}

\begin{definition}[Signature of $\textrm{BATC}^{\textrm{srt}}\textrm{I}$]
The signature of $\textrm{BATC}^{\textrm{srt}}\textrm{I}$ consists of the signature of $\textrm{BATC}^{\textrm{srt}}$ and the integration operator $\int: \mathcal{P}(\mathbb{R}^{\geq})\times\mathbb{R}^{\geq}.\mathcal{P}_{\textrm{rel}} \rightarrow\mathcal{P}_{\textrm{rel}}$.
\end{definition}

The set of axioms of $\textrm{BATC}^{\textrm{srt}}\textrm{I}$ consists of the laws given in Table \ref{AxiomsForBATCSRTI}.

\begin{center}
    \begin{table}
        \begin{tabular}{@{}ll@{}}
            \hline No. &Axiom\\
            $INT1$ & $\int_{v\in V}F(v) = \int_{w\in V}F(w)$\\
            $INT2$ & $\int_{v\in\emptyset}F(v) = \dot{\delta}$\\
            $INT3$ & $\int_{v\in\{p\}}F(v) = F(p)$\\
            $INT4$ & $\int_{v\in V\cup W}F(v) = \int_{v\in V}F(v) + \int_{v\in W}F(v)$\\
            $INT5$ & $V\neq\emptyset\Rightarrow\int_{v\in V}x = x$\\
            $INT6$ & $(\forall v\in V.F(v)=G(v))\Rightarrow \int_{v\in V}F(v) = \int_{v\in V}G(v)$\\
            $INT7SRa$ & $\textrm{sup }V = p\Rightarrow\int_{v\in V}\sigma^v_{\textrm{rel}}(\dot{\delta}) = \sigma^p_{\textrm{rel}}(\dot{\delta})$\\
            $INT7SRb$ & $V,W\textrm{ unbounded }\Rightarrow \int_{v\in V}\sigma^v_{\textrm{rel}}(\dot{\delta}) = \int_{v\in W}\sigma^v_{\textrm{rel}}(\dot{\delta})$\\
            $INT8SRa$ & $\textrm{sup }V = p,p\notin V\Rightarrow\int_{v\in V}\sigma^v_{\textrm{rel}}(\tilde{\tilde{\delta}}) = \sigma^p_{\textrm{rel}}(\dot{\delta})$\\
            $INT8SRb$ & $V,W\textrm{ unbounded }\Rightarrow \int_{v\in V}\sigma^v_{\textrm{rel}}(\tilde{\tilde{\delta}}) = \int_{v\in W}\sigma^v_{\textrm{rel}}(\dot{\delta})$\\
            $INT9SR$ & $\textrm{sup }V = p,p\in V\Rightarrow\int_{v\in V}\sigma^v_{\textrm{rel}}(\tilde{\tilde{\delta}}) = \sigma^p_{\textrm{rel}}(\tilde{\tilde{\delta}})$\\
            $INT10SR$ & $\int_{v\in V}\sigma^p_{\textrm{rel}}(F(v)) = \sigma^p_{\textrm{rel}}(\int_{v\in V}F(v))$\\
            $INT11$ & $\int_{v\in V}(F(v)+G(v)) = \int_{v\in V}F(v) + \int_{v\in V}G(v)$\\
            $INT12$ & $\int_{v\in V}(F(v)\cdot x) = (\int_{v\in V}F(v))\cdot x$\\
            $SRTO6$ & $\upsilon^p_{\textrm{rel}}(\int_{v\in V}F(v)) = \int_{v\in V}\upsilon^p_{\textrm{rel}}(F(v))$\\
            $SRI6$ & $\overline{\upsilon}^p_{\textrm{rel}}(\int_{v\in V}F(v)) = \int_{v\in V}\overline{\upsilon}^p_{\textrm{rel}}(F(v))$\\
        \end{tabular}
        \caption{Axioms of $\textrm{BATC}^{\textrm{srt}}\textrm{I}(p\geq 0)$}
        \label{AxiomsForBATCSRTI}
    \end{table}
\end{center}

The operational semantics of $\textrm{BATC}^{\textrm{srt}}\textrm{I}$ are defined by the transition rules in Table \ref{TRForBATCSRTI}.

\begin{center}
    \begin{table}
        $$\frac{F(q)\xrightarrow{a}x'}{\int_{v\in V}F(v)\xrightarrow{a}x'}(q\in V)
        \quad\frac{F(q)\xrightarrow{a}\surd}{\int_{v\in V}F(v)\xrightarrow{a}\surd}(q\in V)$$

        $$\frac{\{F(q)\mapsto^r F_1(q)|q\in V_1\},\cdots,\{F(q)\mapsto^r F_n(q)|q\in V_n\},\{F(q)\nmapsto^r|q\in V_{n+1}\}}{\int_{v\in V}\mapsto^r\int_{v\in V_1}F_1(v)+\cdots+\int_{v\in V_n}F_n(v)}(\{V_1,\cdots,V_n\}\textrm{ partition of }V-V_{n+1}, V_{n+1}\subset V)$$

        $$\frac{\{F(q)\uparrow|q\in V\}}{\int_{v\in V}F(v)\uparrow}$$
        \caption{Transition rules of $\textrm{BATC}^{\textrm{srt}}\textrm{I}(a\in A, p,q\geq 0, r>0)$}
        \label{TRForBATCSRTI}
    \end{table}
\end{center}

\subsubsection{Elimination}

\begin{definition}[Basic terms of $\textrm{BATC}^{\textrm{srt}}\textrm{I}$]
The set of basic terms of $\textrm{BATC}^{\textrm{srt}}\textrm{I}$, $\mathcal{B}(\textrm{BATC}^{\textrm{srt}})$, is inductively defined as follows by two auxiliary sets $\mathcal{B}_0(\textrm{BATC}^{\textrm{srt}}\textrm{I})$ and $\mathcal{B}_1(\textrm{BATC}^{\textrm{srt}}\textrm{I})$:
\begin{enumerate}
  \item if $a\in A_{\delta}$, then $\tilde{\tilde{a}} \in \mathcal{B}_1(\textrm{BATC}^{\textrm{srt}}\textrm{I})$;
  \item if $a\in A$ and $t\in \mathcal{B}(\textrm{BATC}^{\textrm{srt}}\textrm{I})$, then $\tilde{\tilde{a}}\cdot t \in \mathcal{B}_1(\textrm{BATC}^{\textrm{srt}}\textrm{I})$;
  \item if $t,t'\in \mathcal{B}_1(\textrm{BATC}^{\textrm{srt}}\textrm{I})$, then $t+t'\in \mathcal{B}_1(\textrm{BATC}^{\textrm{srt}}\textrm{I})$;
  \item if $t\in \mathcal{B}_1(\textrm{BATC}^{\textrm{srt}}\textrm{I})$, then $t\in \mathcal{B}_0(\textrm{BATC}^{\textrm{srt}}\textrm{I})$;
  \item if $p>0$ and $t\in \mathcal{B}_0(\textrm{BATC}^{\textrm{srt}}\textrm{I})$, then $\sigma^p_{\textrm{rel}}(t) \in \mathcal{B}_0(\textrm{BATC}^{\textrm{srt}}\textrm{I})$;
  \item if $p>0$, $t\in \mathcal{B}_1(\textrm{BATC}^{\textrm{srt}}\textrm{I})$ and $t'\in \mathcal{B}_0(\textrm{BATC}^{\textrm{srt}}\textrm{I})$, then $t+\sigma^p_{\textrm{rel}}(t') \in \mathcal{B}_0(\textrm{BATC}^{\textrm{srt}}\textrm{I})$;
  \item if $t\in \mathcal{B}_0(\textrm{BATC}^{\textrm{srt}}\textrm{I})$, then $\int_{v\in V}(t) \in \mathcal{B}_0(\textrm{BATC}^{\textrm{srt}}\textrm{I})$;
  \item $\dot{\delta}\in \mathcal{B}(\textrm{BATC}^{\textrm{srt}}\textrm{I})$;
  \item if $t\in \mathcal{B}_0(\textrm{BATC}^{\textrm{srt}}\textrm{I})$, then $t\in \mathcal{B}(\textrm{BATC}^{\textrm{srt}}\textrm{I})$.
\end{enumerate}
\end{definition}

\begin{theorem}[Elimination theorem]
Let $p$ be a closed $\textrm{BATC}^{\textrm{srt}}\textrm{I}$ term. Then there is a basic $\textrm{BATC}^{\textrm{srt}}\textrm{I}$ term $q$ such that $\textrm{BATC}^{\textrm{srt}}\vdash p=q$.
\end{theorem}

\begin{proof}
It is sufficient to induct on the structure of the closed $\textrm{BATC}^{\textrm{srt}}\textrm{I}$ term $p$. It can be proven that $p$ combined by the constants and operators of $\textrm{BATC}^{\textrm{srt}}\textrm{I}$ exists an equal basic term $q$, and the other operators not included in the basic terms, such as $\upsilon_{\textrm{rel}}$ and $\overline{\upsilon}_{\textrm{rel}}$ can be eliminated.
\end{proof}

\subsubsection{Connections}

\begin{theorem}[Generalization of $\textrm{BATC}^{\textrm{srt}}\textrm{I}$]
\begin{enumerate}
  \item By the definitions of $a=\int_{v\in[0,\infty)}\sigma^v_{\textrm{rel}}(\tilde{\tilde{a}})$ for each $a\in A$ and $\delta=\int_{v\in[0,\infty)}\sigma^v_{\textrm{rel}}(\tilde{\tilde{\delta}})$, $\textrm{BATC}^{\textrm{srt}}\textrm{I}$ is a generalization of $BATC$.
  \item $\textrm{BATC}^{\textrm{srt}}\textrm{I}$ is a generalization of $BATC^{\textrm{srt}}$.
\end{enumerate}

\end{theorem}

\begin{proof}
\begin{enumerate}
  \item It follows from the following two facts.

    \begin{enumerate}
      \item The transition rules of $BATC$ in section \ref{tcpa} are all source-dependent;
      \item The sources of the transition rules of $\textrm{BATC}^{\textrm{srt}}\textrm{I}$ contain an occurrence of $\dot{\delta}$, $\tilde{\tilde{a}}$, $\sigma^p_{\textrm{rel}}$, $\upsilon^p_{\textrm{rel}}$, $\overline{\upsilon}^p_{\textrm{rel}}$ and $\int$.
    \end{enumerate}

    So, $BATC$ is an embedding of $\textrm{BATC}^{\textrm{srt}}\textrm{I}$, as desired.
  \item It follows from the following two facts.

    \begin{enumerate}
      \item The transition rules of $BATC^{\textrm{srt}}$ are all source-dependent;
      \item The sources of the transition rules of $\textrm{BATC}^{\textrm{srt}}\textrm{I}$ contain an occurrence of $\int$.
    \end{enumerate}

    So, $BATC^{\textrm{srt}}$ is an embedding of $\textrm{BATC}^{\textrm{srt}}\textrm{I}$, as desired.
\end{enumerate}
\end{proof}

\subsubsection{Congruence}

\begin{theorem}[Congruence of $\textrm{BATC}^{\textrm{srt}}\textrm{I}$]
Truly concurrent bisimulation equivalences are all congruences with respect to $\textrm{BATC}^{\textrm{srt}}\textrm{I}$. That is,
\begin{itemize}
  \item pomset bisimulation equivalence $\sim_{p}$ is a congruence with respect to $\textrm{BATC}^{\textrm{srt}}\textrm{I}$;
  \item step bisimulation equivalence $\sim_{s}$ is a congruence with respect to $\textrm{BATC}^{\textrm{srt}}\textrm{I}$;
  \item hp-bisimulation equivalence $\sim_{hp}$ is a congruence with respect to $\textrm{BATC}^{\textrm{srt}}\textrm{I}$;
  \item hhp-bisimulation equivalence $\sim_{hhp}$ is a congruence with respect to $\textrm{BATC}^{\textrm{srt}}\textrm{I}$.
\end{itemize}
\end{theorem}

\begin{proof}
It is easy to see that $\sim_p$, $\sim_s$, $\sim_{hp}$ and $\sim_{hhp}$ are all equivalent relations on $\textrm{BATC}^{\textrm{srt}}\textrm{I}$ terms, it is only sufficient to prove that $\sim_p$, $\sim_s$, $\sim_{hp}$ and $\sim_{hhp}$ are all preserved by the operators $\int$. It is trivial and we omit it.
\end{proof}

\subsubsection{Soundness}

\begin{theorem}[Soundness of $\textrm{BATC}^{\textrm{srt}}\textrm{I}$]
The axiomatization of $\textrm{BATC}^{\textrm{srt}}\textrm{I}$ is sound modulo truly concurrent bisimulation equivalences $\sim_{p}$, $\sim_{s}$, $\sim_{hp}$ and $\sim_{hhp}$. That is,
\begin{enumerate}
  \item let $x$ and $y$ be $\textrm{BATC}^{\textrm{srt}}\textrm{I}$ terms. If $\textrm{BATC}^{\textrm{srt}}\textrm{I}\vdash x=y$, then $x\sim_{s} y$;
  \item let $x$ and $y$ be $\textrm{BATC}^{\textrm{srt}}\textrm{I}$ terms. If $\textrm{BATC}^{\textrm{srt}}\textrm{I}\vdash x=y$, then $x\sim_{p} y$;
  \item let $x$ and $y$ be $\textrm{BATC}^{\textrm{srt}}\textrm{I}$ terms. If $\textrm{BATC}^{\textrm{srt}}\textrm{I}\vdash x=y$, then $x\sim_{hp} y$;
  \item let $x$ and $y$ be $\textrm{BATC}^{\textrm{srt}}\textrm{I}$ terms. If $\textrm{BATC}^{\textrm{srt}}\textrm{I}\vdash x=y$, then $x\sim_{hhp} y$.
\end{enumerate}
\end{theorem}

\begin{proof}
Since $\sim_p$, $\sim_s$, $\sim_{hp}$ and $\sim_{hhp}$ are both equivalent and congruent relations, we only need to check if each axiom in Table \ref{AxiomsForBATCSRTI} is sound modulo $\sim_p$, $\sim_s$, $\sim_{hp}$ and $\sim_{hhp}$ respectively.

\begin{enumerate}
  \item We can check the soundness of each axiom in Table \ref{AxiomsForBATCSRTI}, by the transition rules in Table \ref{TRForBATCSRTI}, it is trivial and we omit them.
  \item From the definition of pomset bisimulation, we know that pomset bisimulation is defined by pomset transitions, which are labeled by pomsets. In a pomset transition, the events (actions) in the pomset are either within causality relations (defined by $\cdot$) or in concurrency (implicitly defined by $\cdot$ and $+$, and explicitly defined by $\between$), of course, they are pairwise consistent (without conflicts). We have already proven the case that all events are pairwise concurrent (soundness modulo step bisimulation), so, we only need to prove the case of events in causality. Without loss of generality, we take a pomset of $P=\{\tilde{\tilde{a}},\tilde{\tilde{b}}:\tilde{\tilde{a}}\cdot \tilde{\tilde{b}}\}$. Then the pomset transition labeled by the above $P$ is just composed of one single event transition labeled by $\tilde{\tilde{a}}$ succeeded by another single event transition labeled by $\tilde{\tilde{b}}$, that is, $\xrightarrow{P}=\xrightarrow{a}\xrightarrow{b}$.

        Similarly to the proof of soundness modulo step bisimulation equivalence, we can prove that each axiom in Table \ref{AxiomsForBATCSRTI} is sound modulo pomset bisimulation equivalence, we omit them.
  \item From the definition of hp-bisimulation, we know that hp-bisimulation is defined on the posetal product $(C_1,f,C_2),f:C_1\rightarrow C_2\textrm{ isomorphism}$. Two process terms $s$ related to $C_1$ and $t$ related to $C_2$, and $f:C_1\rightarrow C_2\textrm{ isomorphism}$. Initially, $(C_1,f,C_2)=(\emptyset,\emptyset,\emptyset)$, and $(\emptyset,\emptyset,\emptyset)\in\sim_{hp}$. When $s\xrightarrow{a}s'$ ($C_1\xrightarrow{a}C_1'$), there will be $t\xrightarrow{a}t'$ ($C_2\xrightarrow{a}C_2'$), and we define $f'=f[a\mapsto a]$. Then, if $(C_1,f,C_2)\in\sim_{hp}$, then $(C_1',f',C_2')\in\sim_{hp}$.

        Similarly to the proof of soundness modulo pomset bisimulation equivalence, we can prove that each axiom in Table \ref{AxiomsForBATCSRTI} is sound modulo hp-bisimulation equivalence, we just need additionally to check the above conditions on hp-bisimulation, we omit them.
  \item We just need to add downward-closed condition to the soundness modulo hp-bisimulation equivalence, we omit them.
\end{enumerate}

\end{proof}

\subsubsection{Completeness}

\begin{theorem}[Completeness of $\textrm{BATC}^{\textrm{srt}}\textrm{I}$]
The axiomatization of $\textrm{BATC}^{\textrm{srt}}\textrm{I}$ is complete modulo truly concurrent bisimulation equivalences $\sim_{p}$, $\sim_{s}$, $\sim_{hp}$ and $\sim_{hhp}$. That is,
\begin{enumerate}
  \item let $p$ and $q$ be closed $\textrm{BATC}^{\textrm{srt}}\textrm{I}$ terms, if $p\sim_{s} q$ then $p=q$;
  \item let $p$ and $q$ be closed $\textrm{BATC}^{\textrm{srt}}\textrm{I}$ terms, if $p\sim_{p} q$ then $p=q$;
  \item let $p$ and $q$ be closed $\textrm{BATC}^{\textrm{srt}}\textrm{I}$ terms, if $p\sim_{hp} q$ then $p=q$;
  \item let $p$ and $q$ be closed $\textrm{BATC}^{\textrm{srt}}\textrm{I}$ terms, if $p\sim_{hhp} q$ then $p=q$.
\end{enumerate}

\end{theorem}

\begin{proof}
\begin{enumerate}
  \item Firstly, by the elimination theorem of $\textrm{BATC}^{\textrm{srt}}\textrm{I}$, we know that for each closed $\textrm{BATC}^{\textrm{srt}}\textrm{I}$ term $p$, there exists a closed basic $\textrm{BATC}^{\textrm{srt}}\textrm{I}$ term $p'$, such that $\textrm{BATC}^{\textrm{dat}}\textrm{I}\vdash p=p'$, so, we only need to consider closed basic $\textrm{BATC}^{\textrm{srt}}\textrm{I}$ terms.

        The basic terms modulo associativity and commutativity (AC) of conflict $+$ (defined by axioms $A1$ and $A2$ in Table \ref{AxiomsForBATCSRT}), and this equivalence is denoted by $=_{AC}$. Then, each equivalence class $s$ modulo AC of $+$ has the following normal form

        $$s_1+\cdots+ s_k$$

        with each $s_i$ either an atomic event or of the form $t_1\cdot t_2$, and each $s_i$ is called the summand of $s$.

        Now, we prove that for normal forms $n$ and $n'$, if $n\sim_{s} n'$ then $n=_{AC}n'$. It is sufficient to induct on the sizes of $n$ and $n'$. We can get $n=_{AC} n'$.

        Finally, let $s$ and $t$ be basic terms, and $s\sim_s t$, there are normal forms $n$ and $n'$, such that $s=n$ and $t=n'$. The soundness theorem of $\textrm{BATC}^{\textrm{srt}}\textrm{I}$ modulo step bisimulation equivalence yields $s\sim_s n$ and $t\sim_s n'$, so $n\sim_s s\sim_s t\sim_s n'$. Since if $n\sim_s n'$ then $n=_{AC}n'$, $s=n=_{AC}n'=t$, as desired.
  \item This case can be proven similarly, just by replacement of $\sim_{s}$ by $\sim_{p}$.
  \item This case can be proven similarly, just by replacement of $\sim_{s}$ by $\sim_{hp}$.
  \item This case can be proven similarly, just by replacement of $\sim_{s}$ by $\sim_{hhp}$.
\end{enumerate}
\end{proof}

\subsection{Algebra for Parallelism in True Concurrency with Continuous Relative Timing}

In this subsection, we will introduce $\textrm{APTC}^{\textrm{srt}}$.

\subsubsection{Basic Definition}

\begin{definition}[Relative undelayable time-out]
The relative undelayable time-out $\nu_{\textrm{rel}}$ of a process $p$ behaves like the part of $p$ that starts to perform actions at the current point of time if $p$ is capable of performing actions at the current point of time; otherwise, like undelayable deadlock. And let $\nu^r_{\textrm{rel}}(t) = r \nu_{\textrm{rel}} t$.
\end{definition}

\subsubsection{The Theory $\textrm{APTC}^{\textrm{srt}}$}

\begin{definition}[Signature of $\textrm{APTC}^{\textrm{srt}}$]
The signature of $\textrm{APTC}^{\textrm{srt}}$ consists of the signature of $\textrm{BATC}^{\textrm{srt}}$, and the whole parallel composition operator $\between: \mathcal{P}_{\textrm{rel}}\times\mathcal{P}_{\textrm{rel}} \rightarrow \mathcal{P}_{\textrm{rel}}$, the parallel operator $\parallel: \mathcal{P}_{\textrm{rel}}\times\mathcal{P}_{\textrm{rel}} \rightarrow \mathcal{P}_{\textrm{rel}}$, the communication merger operator $\mid: \mathcal{P}_{\textrm{rel}}\times\mathcal{P}_{\textrm{rel}} \rightarrow \mathcal{P}_{\textrm{rel}}$, the encapsulation operator $\partial_H: \mathcal{P}_{\textrm{rel}} \rightarrow \mathcal{P}_{\textrm{rel}}$ for all $H\subseteq A$, and the relative undelayable time-out operator $\nu_{\textrm{rel}}: \mathcal{P}_{\textrm{rel}}\rightarrow \mathcal{P}_{\textrm{rel}}$.
\end{definition}

The set of axioms of $\textrm{APTC}^{\textrm{srt}}$ consists of the laws given in Table \ref{AxiomsForAPTCSRT}.

\begin{center}
    \begin{table}
        \begin{tabular}{@{}ll@{}}
            \hline No. &Axiom\\
            $P1$ & $x\between y = x\parallel y + x\mid y$\\
            $P2$ & $x\parallel y = y \parallel x$\\
            $P3$ & $(x\parallel y)\parallel z = x\parallel (y\parallel z)$\\
            $P4SR$ & $\tilde{\tilde{a}}\parallel (\tilde{\tilde{b}}\cdot y) = (\tilde{\tilde{a}}\parallel \tilde{\tilde{b}})\cdot y$\\
            $P5SR$ & $(\tilde{\tilde{a}}\cdot x)\parallel \tilde{\tilde{b}} = (\tilde{\tilde{a}}\parallel \tilde{\tilde{b}})\cdot x$\\
            $P6SR$ & $(\tilde{\tilde{a}}\cdot x)\parallel (\tilde{\tilde{b}}\cdot y) = (\tilde{\tilde{a}}\parallel \tilde{\tilde{b}})\cdot (x\between y)$\\
            $P7$ & $(x+ y)\parallel z = (x\parallel z)+ (y\parallel z)$\\
            $P8$ & $x\parallel (y+ z) = (x\parallel y)+ (x\parallel z)$\\
            $SRP9ID$ & $(\nu_{\textrm{rel}}(x)+ \tilde{\tilde{\delta}})\parallel \sigma^{r}_{\textrm{rel}}(y) = \tilde{\tilde{\delta}}$\\
            $SRP10ID$ & $\sigma^{r}_{\textrm{rel}}(x)\parallel (\nu_{\textrm{rel}}(y)+ \tilde{\tilde{\delta}}) = \tilde{\tilde{\delta}}$\\
            $SRP11$ & $\sigma^p_{\textrm{rel}}(x) \parallel \sigma^p_{\textrm{rel}}(y) = \sigma^p_{\textrm{rel}}(x\parallel y)$\\
            $PID12$ & $\dot{\delta}\parallel x = \dot{\delta}$\\
            $PID13$ & $x\parallel \dot{\delta} = \dot{\delta}$\\
            $C14SR$ & $\tilde{\tilde{a}}\mid \tilde{\tilde{b}} = \gamma(\tilde{\tilde{a}}, \tilde{\tilde{b}})$\\
            $C15SR$ & $\tilde{\tilde{a}}\mid (\tilde{\tilde{b}}\cdot y) = \gamma(\tilde{\tilde{a}}, \tilde{\tilde{b}})\cdot y$\\
            $C16SR$ & $(\tilde{\tilde{a}}\cdot x)\mid \tilde{\tilde{b}} = \gamma(\tilde{\tilde{a}}, \tilde{\tilde{b}})\cdot x$\\
            $C17SR$ & $(\tilde{\tilde{a}}\cdot x)\mid (\tilde{\tilde{b}}\cdot y) = \gamma(\tilde{\tilde{a}}, \tilde{\tilde{b}})\cdot (x\between y)$\\
            $C18$ & $(x+ y)\mid z = (x\mid z) + (y\mid z)$\\
            $C19$ & $x\mid (y+ z) = (x\mid y)+ (x\mid z)$\\
            $DRC20ID$ & $(\nu_{\textrm{rel}}(x)+ \tilde{\tilde{\delta}})\mid \sigma^{r}_{\textrm{rel}}(y) = \tilde{\tilde{\delta}}$\\
            $DRC21ID$ & $\sigma^{r}_{\textrm{rel}}(x)\mid (\nu_{\textrm{rel}}(y)+ \tilde{\tilde{\delta}}) = \tilde{\tilde{\delta}}$\\
            $DRC22$ & $\sigma^p_{\textrm{rel}}(x) \mid \sigma^p_{\textrm{rel}}(y) = \sigma^p_{\textrm{rel}}(x\mid y)$\\
            $CID23$ & $\dot{\delta}\mid x = \dot{\delta}$\\
            $CID24$ & $x\mid\dot{\delta} = \dot{\delta}$\\
            $CE25DR$ & $\Theta(\tilde{\tilde{a}}) = \tilde{\tilde{a}}$\\
            $CE26DRID$ & $\Theta(\dot{\delta}) = \dot{\delta}$\\
            $CE27$ & $\Theta(x+ y) = \Theta(x)\triangleleft y + \Theta(y)\triangleleft x$\\
            $CE28$ & $\Theta(x\cdot y)=\Theta(x)\cdot\Theta(y)$\\
            $CE29$ & $\Theta(x\parallel y) = ((\Theta(x)\triangleleft y)\parallel y)+ ((\Theta(y)\triangleleft x)\parallel x)$\\
            $CE30$ & $\Theta(x\mid y) = ((\Theta(x)\triangleleft y)\mid y)+ ((\Theta(y)\triangleleft x)\mid x)$\\
            $U31SRID$ & $(\sharp(\tilde{\tilde{a}},\tilde{\tilde{b}}))\quad \tilde{\tilde{a}}\triangleleft \tilde{\tilde{b}} = \tilde{\tilde{\tau}}$\\
            $U32SRID$ & $(\sharp(\tilde{\tilde{a}},\tilde{\tilde{b}}),\tilde{\tilde{b}}\leq \tilde{\tilde{c}})\quad \tilde{\tilde{a}}\triangleleft \tilde{\tilde{c}} = \tilde{\tilde{a}}$\\
            $U33SRID$ & $(\sharp(\tilde{\tilde{a}},\tilde{\tilde{b}}),\tilde{\tilde{b}}\leq \tilde{\tilde{c}})\quad \tilde{\tilde{c}}\triangleleft \tilde{\tilde{a}} = \tilde{\tilde{\tau}}$\\
            $U34SRID$ & $\tilde{\tilde{a}}\triangleleft \tilde{\tilde{\delta}} = \tilde{\tilde{a}}$\\
            $U35SRID$ & $\tilde{\tilde{\delta}} \triangleleft \tilde{\tilde{a}} = \tilde{\tilde{\delta}}$\\
            $U36$ & $(x+ y)\triangleleft z = (x\triangleleft z)+ (y\triangleleft z)$\\
            $U37$ & $(x\cdot y)\triangleleft z = (x\triangleleft z)\cdot (y\triangleleft z)$\\
            $U38$ & $(x\parallel y)\triangleleft z = (x\triangleleft z)\parallel (y\triangleleft z)$\\
            $U39$ & $(x\mid y)\triangleleft z = (x\triangleleft z)\mid (y\triangleleft z)$\\
            $U40$ & $x\triangleleft (y+ z) = (x\triangleleft y)\triangleleft z$\\
            $U41$ & $x\triangleleft (y\cdot z)=(x\triangleleft y)\triangleleft z$\\
            $U42$ & $x\triangleleft (y\parallel z) = (x\triangleleft y)\triangleleft z$\\
            $U43$ & $x\triangleleft (y\mid z) = (x\triangleleft y)\triangleleft z$\\
            $D1SRID$ & $\tilde{\tilde{a}}\notin H\quad\partial_H(\tilde{\tilde{a}}) = \tilde{\tilde{a}}$\\
            $D2SRID$ & $\tilde{\tilde{a}}\in H\quad \partial_H(\tilde{\tilde{a}}) = \tilde{\tilde{\delta}}$\\
            $D3SRID$ & $\partial_H(\dot{\delta}) = \dot{\delta}$\\
            $D4$ & $\partial_H(x+ y) = \partial_H(x)+\partial_H(y)$\\
            $D5$ & $\partial_H(x\cdot y) = \partial_H(x)\cdot\partial_H(y)$\\
            $D6$ & $\partial_H(x\parallel y) = \partial_H(x)\parallel\partial_H(y)$\\
            $SRD7$ & $\partial_H(\sigma^p_{\textrm{rel}}(x)) = \sigma^p_{\textrm{rel}}(\partial_H(x))$\\
            $SRU0$ & $\nu_{\textrm{rel}}(\dot{\delta}) = \dot{\delta}$\\
            $SRU1$ & $\nu_{\textrm{rel}}(\tilde{\tilde{a}}) = \tilde{\tilde{a}}$\\
            $SRU2$ & $\nu_{\textrm{rel}}(\sigma^r_{\textrm{rel}}(x)) = \tilde{\tilde{\delta}}$\\
            $SRU3$ & $\nu_{\textrm{rel}}(x+y) = \nu_{\textrm{rel}}(x)+\nu_{\textrm{rel}}(y)$\\
            $SRU4$ & $\nu_{\textrm{rel}}(x\cdot y) = \nu_{\textrm{rel}}(x)\cdot y$\\
            $SRU5$ & $\nu_{\textrm{rel}}(x\parallel y) = \nu_{\textrm{rel}}(x)\parallel \nu_{\textrm{rel}}(y)$\\
        \end{tabular}
        \caption{Axioms of $\textrm{APTC}^{\textrm{srt}}(a,b,c\in A_{\delta}, p\geq 0, r>0)$}
        \label{AxiomsForAPTCSRT}
    \end{table}
\end{center}

The operational semantics of $\textrm{APTC}^{\textrm{srt}}$ are defined by the transition rules in Table \ref{TRForAPTCSRT}.

\begin{center}
    \begin{table}
        $$\frac{x\xrightarrow{a}\surd\quad y\xrightarrow{b}\surd}{x\parallel y\xrightarrow{\{a,b\}}\surd} \quad\frac{x\xrightarrow{a}x'\quad y\xrightarrow{b}\surd}{x\parallel y\xrightarrow{\{a,b\}}x'}$$

        $$\frac{x\xrightarrow{a}\surd\quad y\xrightarrow{b}y'}{x\parallel y\xrightarrow{\{a,b\}}y'} \quad\frac{x\xrightarrow{a}x'\quad y\xrightarrow{b}y'}{x\parallel y\xrightarrow{\{a,b\}}x'\between y'}$$

        $$\frac{x\mapsto^{r}x'\quad y\mapsto^{r}y'}{x\parallel y\mapsto^{r}x'\parallel y'} \quad\frac{x\uparrow}{x\parallel y\uparrow} \quad\frac{y\uparrow}{x\parallel y\uparrow}$$

        $$\frac{x\xrightarrow{a}\surd\quad y\xrightarrow{b}\surd}{x\mid y\xrightarrow{\gamma(a,b)}\surd} \quad\frac{x\xrightarrow{a}x'\quad y\xrightarrow{b}\surd}{x\mid y\xrightarrow{\gamma(a,b)}x'}$$

        $$\frac{x\xrightarrow{a}\surd\quad y\xrightarrow{b}y'}{x\mid y\xrightarrow{\gamma(a,b)}y'} \quad\frac{x\xrightarrow{a}x'\quad y\xrightarrow{b}y'}{x\mid y\xrightarrow{\gamma(a,b)}x'\between y'}$$

        $$\frac{x\mapsto^{r}x'\quad y\mapsto^{r}y'}{x\mid y\mapsto^{r}x'\mid y'} \quad\frac{x\uparrow}{x\mid y\uparrow} \quad\frac{y\uparrow}{x\mid y\uparrow}$$

        $$\frac{x\xrightarrow{a}\surd\quad (\sharp(a,b))}{\Theta(x)\xrightarrow{a}\surd} \quad\frac{x\xrightarrow{b}\surd\quad (\sharp(a,b))}{\Theta(x)\xrightarrow{b}\surd}$$

        $$\frac{x\xrightarrow{a}x'\quad (\sharp(a,b))}{\Theta(x)\xrightarrow{a}\Theta(x')} \quad\frac{x\xrightarrow{b}x'\quad (\sharp(a,b))}{\Theta(x)\xrightarrow{b}\Theta(x')}$$

        $$\frac{x\mapsto^{r}x'}{\Theta(x)\mapsto^{r}\Theta(x')} \quad\frac{x\uparrow}{\Theta(x)\uparrow}$$

        $$\frac{x\xrightarrow{a}\surd \quad y\nrightarrow^{b}\quad (\sharp(a,b))}{x\triangleleft y\xrightarrow{\tau}\surd}
        \quad\frac{x\xrightarrow{a}x' \quad y\nrightarrow^{b}\quad (\sharp(a,b))}{x\triangleleft y\xrightarrow{\tau}x'}$$

        $$\frac{x\xrightarrow{a}\surd \quad y\nrightarrow^{c}\quad (\sharp(a,b),b\leq c)}{x\triangleleft y\xrightarrow{a}\surd}
        \quad\frac{x\xrightarrow{a}x' \quad y\nrightarrow^{c}\quad (\sharp(a,b),b\leq c)}{x\triangleleft y\xrightarrow{a}x'}$$

        $$\frac{x\xrightarrow{c}\surd \quad y\nrightarrow^{b}\quad (\sharp(a,b),a\leq c)}{x\triangleleft y\xrightarrow{\tau}\surd}
        \quad\frac{x\xrightarrow{c}x' \quad y\nrightarrow^{b}\quad (\sharp(a,b),a\leq c)}{x\triangleleft y\xrightarrow{\tau}x'}$$

        $$\frac{x\mapsto^{r}x'\quad y\mapsto^{r}y'}{x\triangleleft y\mapsto^{r}x'\triangleleft y'} \quad\frac{x\uparrow}{x\triangleleft y\uparrow}$$

        $$\frac{x\xrightarrow{a}\surd}{\partial_H(x)\xrightarrow{a}\surd}\quad (e\notin H)\quad\frac{x\xrightarrow{a}x'}{\partial_H(x)\xrightarrow{a}\partial_H(x')}\quad(e\notin H)$$

        $$\frac{x\mapsto^{r}x'}{\partial_H(x)\mapsto^{r}\partial_H(x')}\quad(e\notin H)\quad\frac{x\uparrow}{\partial_H(x)\uparrow}$$

        $$\frac{x\xrightarrow{a}x'}{\nu_{\textrm{rel}}(x)\xrightarrow{a}x'} \quad\frac{x\xrightarrow{a}\surd}{\nu_{\textrm{rel}}(x)\xrightarrow\surd}
        \quad\frac{x\uparrow}{\nu_{\textrm{rel}}(x)\uparrow}$$
    \caption{Transition rules of $\textrm{APTC}^{\textrm{srt}}(a,b,c\in A, r>0)$}
    \label{TRForAPTCSRT}
    \end{table}
\end{center}

\subsubsection{Elimination}

\begin{definition}[Basic terms of $\textrm{APTC}^{\textrm{srt}}$]
The set of basic terms of $\textrm{APTC}^{\textrm{srt}}$, $\mathcal{B}(\textrm{APTC}^{\textrm{srt}})$, is inductively defined as follows by two auxiliary sets $\mathcal{B}_0(\textrm{APTC}^{\textrm{srt}})$ and $\mathcal{B}_1(\textrm{APTC}^{\textrm{srt}})$:
\begin{enumerate}
  \item if $a\in A_{\delta}$, then $\tilde{\tilde{a}} \in \mathcal{B}_1(\textrm{APTC}^{\textrm{srt}})$;
  \item if $a\in A$ and $t\in \mathcal{B}(\textrm{APTC}^{\textrm{srt}})$, then $\tilde{\tilde{a}}\cdot t \in \mathcal{B}_1(\textrm{APTC}^{\textrm{srt}})$;
  \item if $t,t'\in \mathcal{B}_1(\textrm{APTC}^{\textrm{srt}})$, then $t+t'\in \mathcal{B}_1(\textrm{APTC}^{\textrm{srt}})$;
  \item if $t,t'\in \mathcal{B}_1(\textrm{APTC}^{\textrm{srt}})$, then $t\parallel t'\in \mathcal{B}_1(\textrm{APTC}^{\textrm{srt}})$;
  \item if $t\in \mathcal{B}_1(\textrm{APTC}^{\textrm{srt}})$, then $t\in \mathcal{B}_0(\textrm{APTC}^{\textrm{srt}})$;
  \item if $p>0$ and $t\in \mathcal{B}_0(\textrm{APTC}^{\textrm{srt}})$, then $\sigma^p_{\textrm{rel}}(t) \in \mathcal{B}_0(\textrm{APTC}^{\textrm{srt}})$;
  \item if $p>0$, $t\in \mathcal{B}_1(\textrm{APTC}^{\textrm{srt}})$ and $t'\in \mathcal{B}_0(\textrm{APTC}^{\textrm{srt}})$, then $t+\sigma^p_{\textrm{rel}}(t') \in \mathcal{B}_0(\textrm{APTC}^{\textrm{srt}})$;
  \item if $t\in \mathcal{B}_0(\textrm{APTC}^{\textrm{srt}})$, then $\nu_{\textrm{rel}}(t) \in \mathcal{B}_0(\textrm{APTC}^{\textrm{srt}})$;
  \item $\dot{\delta}\in \mathcal{B}(\textrm{APTC}^{\textrm{srt}})$;
  \item if $t\in \mathcal{B}_0(\textrm{APTC}^{\textrm{srt}})$, then $t\in \mathcal{B}(\textrm{APTC}^{\textrm{srt}})$.
\end{enumerate}
\end{definition}

\begin{theorem}[Elimination theorem]
Let $p$ be a closed $\textrm{APTC}^{\textrm{srt}}$ term. Then there is a basic $\textrm{APTC}^{\textrm{srt}}$ term $q$ such that $\textrm{APTC}^{\textrm{srt}}\vdash p=q$.
\end{theorem}

\begin{proof}
It is sufficient to induct on the structure of the closed $\textrm{APTC}^{\textrm{dat}}$ term $p$. It can be proven that $p$ combined by the constants and operators of $\textrm{APTC}^{\textrm{dat}}$ exists an equal basic term $q$, and the other operators not included in the basic terms, such as $\upsilon_{\textrm{rel}}$, $\overline{\upsilon}_{\textrm{rel}}$, $\between$, $\mid$, $\partial_H$, $\Theta$ and $\triangleleft$ can be eliminated.
\end{proof}

\subsubsection{Connections}

\begin{theorem}[Generalization of $\textrm{APTC}^{\textrm{srt}}$]
\begin{enumerate}
  \item By the definitions of $a=\tilde{\tilde{a}}$ for each $a\in A$ and $\delta=\tilde{\tilde{\delta}}$, $\textrm{APTC}^{\textrm{srt}}$ is a generalization of $APTC$.
  \item $\textrm{APTC}^{\textrm{srt}}$ is a generalization of $\textrm{BATC}^{\textrm{srt}}$¡£
\end{enumerate}

\end{theorem}

\begin{proof}
\begin{enumerate}
  \item It follows from the following two facts.

    \begin{enumerate}
      \item The transition rules of $APTC$ in section \ref{tcpa} are all source-dependent;
      \item The sources of the transition rules of $\textrm{APTC}^{\textrm{srt}}$ contain an occurrence of $\dot{\delta}$, $\tilde{\tilde{a}}$, $\sigma^p_{\textrm{rel}}$, $\upsilon^p_{\textrm{rel}}$, $\overline{\upsilon}^p_{\textrm{rel}}$, and $\nu_{\textrm{rel}}$.
    \end{enumerate}

    So, $APTC$ is an embedding of $\textrm{APTC}^{\textrm{srt}}$, as desired.
    \item It follows from the following two facts.

    \begin{enumerate}
      \item The transition rules of $\textrm{BATC}^{\textrm{srt}}$ are all source-dependent;
      \item The sources of the transition rules of $\textrm{APTC}^{\textrm{srt}}$ contain an occurrence of $\between$, $\parallel$, $\mid$, $\Theta$, $\triangleleft$, $\partial_H$ and $\nu_{\textrm{rel}}$.
    \end{enumerate}

    So, $\textrm{BATC}^{\textrm{srt}}$ is an embedding of $\textrm{APTC}^{\textrm{srt}}$, as desired.
\end{enumerate}
\end{proof}

\subsubsection{Congruence}

\begin{theorem}[Congruence of $\textrm{APTC}^{\textrm{srt}}$]
Truly concurrent bisimulation equivalences $\sim_p$, $\sim_s$ and $\sim_{hp}$ are all congruences with respect to $\textrm{APTC}^{\textrm{srt}}$. That is,
\begin{itemize}
  \item pomset bisimulation equivalence $\sim_{p}$ is a congruence with respect to $\textrm{APTC}^{\textrm{srt}}$;
  \item step bisimulation equivalence $\sim_{s}$ is a congruence with respect to $\textrm{APTC}^{\textrm{srt}}$;
  \item hp-bisimulation equivalence $\sim_{hp}$ is a congruence with respect to $\textrm{APTC}^{\textrm{srt}}$.
\end{itemize}
\end{theorem}

\begin{proof}
It is easy to see that $\sim_p$, $\sim_s$, and $\sim_{hp}$ are all equivalent relations on $\textrm{APTC}^{\textrm{srt}}$ terms, it is only sufficient to prove that $\sim_p$, $\sim_s$, and $\sim_{hp}$ are all preserved by the operators $\sigma^p_{\textrm{rel}}$, $\upsilon^p_{\textrm{rel}}$, $\overline{\upsilon}^p_{\textrm{rel}}$, and $\nu_{\textrm{rel}}$. It is trivial and we omit it.
\end{proof}

\subsubsection{Soundness}

\begin{theorem}[Soundness of $\textrm{APTC}^{\textrm{srt}}$]
The axiomatization of $\textrm{APTC}^{\textrm{srt}}$ is sound modulo truly concurrent bisimulation equivalences $\sim_{p}$, $\sim_{s}$, and $\sim_{hp}$. That is,
\begin{enumerate}
  \item let $x$ and $y$ be $\textrm{APTC}^{\textrm{srt}}$ terms. If $\textrm{APTC}^{\textrm{srt}}\vdash x=y$, then $x\sim_{s} y$;
  \item let $x$ and $y$ be $\textrm{APTC}^{\textrm{srt}}$ terms. If $\textrm{APTC}^{\textrm{srt}}\vdash x=y$, then $x\sim_{p} y$;
  \item let $x$ and $y$ be $\textrm{APTC}^{\textrm{srt}}$ terms. If $\textrm{APTC}^{\textrm{srt}}\vdash x=y$, then $x\sim_{hp} y$.
\end{enumerate}
\end{theorem}

\begin{proof}
Since $\sim_p$, $\sim_s$, and $\sim_{hp}$ are both equivalent and congruent relations, we only need to check if each axiom in Table \ref{AxiomsForAPTCSRT} is sound modulo $\sim_p$, $\sim_s$, and $\sim_{hp}$ respectively.

\begin{enumerate}
  \item We only check the soundness of the non-trivial axiom $SRP11$ modulo $\sim_s$.
        Let $p,q$ be $\textrm{APTC}^{\textrm{srt}}$ processes, and $\sigma^s_{\textrm{rel}}(p) \parallel \sigma^s_{\textrm{rel}}(q) = \sigma^s_{\textrm{rel}}(p\parallel q)$, it is sufficient to prove that $\sigma^s_{\textrm{rel}}(p) \parallel \sigma^s_{\textrm{rel}}(q) \sim_{s} \sigma^s_{\textrm{rel}}(p\parallel q)$. By the transition rules of operator $\sigma^s_{\textrm{rel}}$ and $\parallel$ in Table \ref{TRForAPTCSRT}, we get

        $$\frac{}{\sigma^s_{\textrm{rel}}(p) \parallel \sigma^s_{\textrm{rel}}(q)\mapsto^s \sigma^0_{\textrm{rel}}(p) \parallel \sigma^0_{\textrm{rel}}(q)}$$

        $$\frac{}{\sigma^s_{\textrm{rel}}(p\parallel q)\mapsto^s \sigma^0_{\textrm{rel}}(p\parallel q)}$$

        There are several cases:

        $$\frac{p\xrightarrow{a} \surd\quad q\xrightarrow{b}\surd}{\sigma^0_{\textrm{rel}}(p) \parallel \sigma^0_{\textrm{rel}}(q)\xrightarrow{\{a,b\}}\surd}$$

        $$\frac{p\xrightarrow{a} \surd\quad q\xrightarrow{b}\surd}{\sigma^0_{\textrm{rel}}(p\parallel q)\xrightarrow{\{a,b\}}\surd}$$

        $$\frac{p\xrightarrow{a} p'\quad q\xrightarrow{b}\surd}{\sigma^0_{\textrm{rel}}(p) \parallel \sigma^0_{\textrm{rel}}(q)\xrightarrow{\{a,b\}}p'}$$

        $$\frac{p\xrightarrow{a} p'\quad q\xrightarrow{b}\surd}{\sigma^0_{\textrm{rel}}(p\parallel q)\xrightarrow{\{a,b\}}p'}$$

        $$\frac{p\xrightarrow{a} \surd\quad q\xrightarrow{b}q'}{\sigma^0_{\textrm{rel}}(p) \parallel \sigma^0_{\textrm{rel}}(q)\xrightarrow{\{a,b\}}q'}$$

        $$\frac{p\xrightarrow{a} \surd\quad q\xrightarrow{b}q'}{\sigma^0_{\textrm{rel}}(p\parallel q)\xrightarrow{\{a,b\}}q'}$$

        $$\frac{p\xrightarrow{a} p'\quad q\xrightarrow{b}q'}{\sigma^0_{\textrm{rel}}(p) \parallel \sigma^0_{\textrm{rel}}(q)\xrightarrow{\{a,b\}}p'\between q'}$$

        $$\frac{p\xrightarrow{a} p'\quad q\xrightarrow{b}q'}{\sigma^0_{\textrm{rel}}(p\parallel q)\xrightarrow{\{a,b\}}p'\between q'}$$

        $$\frac{p \uparrow}{\sigma^0_{\textrm{rel}}(p) \parallel \sigma^0_{\textrm{rel}}(q)\uparrow}$$

        $$\frac{p\uparrow}{\sigma^0_{\textrm{rel}}(p\parallel q)\uparrow}$$

        $$\frac{q \uparrow}{\sigma^0_{\textrm{rel}}(p) \parallel \sigma^0_{\textrm{rel}}(q)\uparrow}$$

        $$\frac{q\uparrow}{\sigma^0_{\textrm{rel}}(p\parallel q)\uparrow}$$

        So, we see that each case leads to $\sigma^s_{\textrm{rel}}(p) \parallel \sigma^s_{\textrm{rel}}(q) \sim_{s} \sigma^s_{\textrm{rel}}(p\parallel q)$, as desired.
  \item From the definition of pomset bisimulation, we know that pomset bisimulation is defined by pomset transitions, which are labeled by pomsets. In a pomset transition, the events (actions) in the pomset are either within causality relations (defined by $\cdot$) or in concurrency (implicitly defined by $\cdot$ and $+$, and explicitly defined by $\between$), of course, they are pairwise consistent (without conflicts). We have already proven the case that all events are pairwise concurrent (soundness modulo step bisimulation), so, we only need to prove the case of events in causality. Without loss of generality, we take a pomset of $P=\{\tilde{\tilde{a}},\tilde{\tilde{b}}:\tilde{\tilde{a}}\cdot \tilde{\tilde{b}}\}$. Then the pomset transition labeled by the above $P$ is just composed of one single event transition labeled by $\tilde{\tilde{a}}$ succeeded by another single event transition labeled by $\tilde{\tilde{b}}$, that is, $\xrightarrow{P}=\xrightarrow{a}\xrightarrow{b}$.

        Similarly to the proof of soundness modulo step bisimulation equivalence, we can prove that each axiom in Table \ref{AxiomsForAPTCSRT} is sound modulo pomset bisimulation equivalence, we omit them.
  \item From the definition of hp-bisimulation, we know that hp-bisimulation is defined on the posetal product $(C_1,f,C_2),f:C_1\rightarrow C_2\textrm{ isomorphism}$. Two process terms $s$ related to $C_1$ and $t$ related to $C_2$, and $f:C_1\rightarrow C_2\textrm{ isomorphism}$. Initially, $(C_1,f,C_2)=(\emptyset,\emptyset,\emptyset)$, and $(\emptyset,\emptyset,\emptyset)\in\sim_{hp}$. When $s\xrightarrow{a}s'$ ($C_1\xrightarrow{a}C_1'$), there will be $t\xrightarrow{a}t'$ ($C_2\xrightarrow{a}C_2'$), and we define $f'=f[a\mapsto a]$. Then, if $(C_1,f,C_2)\in\sim_{hp}$, then $(C_1',f',C_2')\in\sim_{hp}$.

        Similarly to the proof of soundness modulo pomset bisimulation equivalence, we can prove that each axiom in Table \ref{AxiomsForAPTCSRT} is sound modulo hp-bisimulation equivalence, we just need additionally to check the above conditions on hp-bisimulation, we omit them.
\end{enumerate}
\end{proof}

\subsubsection{Completeness}

\begin{theorem}[Completeness of $\textrm{APTC}^{\textrm{srt}}$]
The axiomatization of $\textrm{APTC}^{\textrm{srt}}$ is complete modulo truly concurrent bisimulation equivalences $\sim_{p}$, $\sim_{s}$, and $\sim_{hp}$. That is,
\begin{enumerate}
  \item let $p$ and $q$ be closed $\textrm{APTC}^{\textrm{srt}}$ terms, if $p\sim_{s} q$ then $p=q$;
  \item let $p$ and $q$ be closed $\textrm{APTC}^{\textrm{srt}}$ terms, if $p\sim_{p} q$ then $p=q$;
  \item let $p$ and $q$ be closed $\textrm{APTC}^{\textrm{srt}}$ terms, if $p\sim_{hp} q$ then $p=q$.
\end{enumerate}

\end{theorem}

\begin{proof}
\begin{enumerate}
  \item Firstly, by the elimination theorem of $\textrm{APTC}^{\textrm{srt}}$, we know that for each closed $\textrm{APTC}^{\textrm{srt}}$ term $p$, there exists a closed basic $\textrm{APTC}^{\textrm{srt}}$ term $p'$, such that $\textrm{APTC}^{\textrm{srt}}\vdash p=p'$, so, we only need to consider closed basic $\textrm{APTC}^{\textrm{srt}}$ terms.

        The basic terms modulo associativity and commutativity (AC) of conflict $+$ (defined by axioms $A1$ and $A2$ in Table \ref{AxiomsForBATCSRT}) and associativity and commutativity (AC) of parallel $\parallel$ (defined by axioms $P2$ and $P3$ in Table \ref{AxiomsForAPTCSRT}), and these equivalences is denoted by $=_{AC}$. Then, each equivalence class $s$ modulo AC of $+$ and $\parallel$ has the following normal form

        $$s_1+\cdots+ s_k$$

        with each $s_i$ either an atomic event or of the form

        $$t_1\cdot\cdots\cdot t_m$$

        with each $t_j$ either an atomic event or of the form

        $$u_1\parallel\cdots\parallel u_n$$

        with each $u_l$ an atomic event, and each $s_i$ is called the summand of $s$.

        Now, we prove that for normal forms $n$ and $n'$, if $n\sim_{s} n'$ then $n=_{AC}n'$. It is sufficient to induct on the sizes of $n$ and $n'$. We can get $n=_{AC} n'$.

        Finally, let $s$ and $t$ be basic $\textrm{APTC}^{\textrm{srt}}$ terms, and $s\sim_s t$, there are normal forms $n$ and $n'$, such that $s=n$ and $t=n'$. The soundness theorem modulo step bisimulation equivalence yields $s\sim_s n$ and $t\sim_s n'$, so $n\sim_s s\sim_s t\sim_s n'$. Since if $n\sim_s n'$ then $n=_{AC}n'$, $s=n=_{AC}n'=t$, as desired.
  \item This case can be proven similarly, just by replacement of $\sim_{s}$ by $\sim_{p}$.
  \item This case can be proven similarly, just by replacement of $\sim_{s}$ by $\sim_{hp}$.
\end{enumerate}
\end{proof}

\subsection{$\textrm{APTC}^{\textrm{srt}}$ with Integration}

In this subsection, we will introduce the theory $\textrm{APTC}^{\textrm{srt}}$ with integration called $\textrm{APTC}^{\textrm{srt}}\textrm{I}$.

\subsubsection{The Theory $\textrm{APTC}^{\textrm{srt}}\textrm{I}$}

\begin{definition}[Signature of $\textrm{APTC}^{\textrm{srt}}\textrm{I}$]
The signature of $\textrm{APTC}^{\textrm{srt}}\textrm{I}$ consists of the signature of $\textrm{APTC}^{\textrm{srt}}$ and the integration operator $\int: \mathcal{P}(\mathbb{R}^{\geq})\times\mathbb{R}^{\geq}.\mathcal{P}_{\textrm{rel}} \rightarrow\mathcal{P}_{\textrm{rel}}$.
\end{definition}

The set of axioms of $\textrm{APTC}^{\textrm{srt}}\textrm{I}$ consists of the laws given in Table \ref{AxiomsForAPTCSRTI}.

\begin{center}
    \begin{table}
        \begin{tabular}{@{}ll@{}}
            \hline No. &Axiom\\
            $INT13$ & $\int_{v\in V}(F(v)\parallel x) = (\int_{v\in V}F(v))\parallel x$\\
            $INT14$ & $\int_{v\in V}(x \parallel F(v)) = x\parallel (\int_{v\in V}F(v))$\\
            $INT15$ & $\int_{v\in V}(F(v)\mid x) = (\int_{v\in V}F(v))\mid x$\\
            $INT16$ & $\int_{v\in V}(x \mid F(v)) = x\mid (\int_{v\in V}F(v))$\\
            $INT17$ & $\int_{v\in V}\partial_H(F(v)) = \partial_H(\int_{v\in V}F(v))$\\
            $INT18$ & $\int_{v\in V}\Theta(F(v)) = \Theta(\int_{v\in V}F(v))$\\
            $INT19$ & $\int_{v\in V}(F(v)\triangleleft x) = (\int_{v\in V}F(v))\triangleleft x$\\
            $SRU5$ & $\nu_{\textrm{rel}}(\int_{v\in V}P) = \int_{v\in V}\nu_{\textrm{rel}}(P)$\\
        \end{tabular}
        \caption{Axioms of $\textrm{APTC}^{\textrm{srt}}\textrm{I}$}
        \label{AxiomsForAPTCSRTI}
    \end{table}
\end{center}

The operational semantics of $\textrm{APTC}^{\textrm{srt}}\textrm{I}$ are defined by the transition rules in Table \ref{TRForBATCSRTI}.

\subsubsection{Elimination}

\begin{definition}[Basic terms of $\textrm{APTC}^{\textrm{srt}}\textrm{I}$]
The set of basic terms of $\textrm{APTC}^{\textrm{srt}}\textrm{I}$, $\mathcal{B}(\textrm{APTC}^{\textrm{srt}})$, is inductively defined as follows by two auxiliary sets $\mathcal{B}_0(\textrm{APTC}^{\textrm{srt}}\textrm{I})$ and $\mathcal{B}_1(\textrm{APTC}^{\textrm{srt}}\textrm{I})$:
\begin{enumerate}
  \item if $a\in A_{\delta}$, then $\tilde{\tilde{a}} \in \mathcal{B}_1(\textrm{APTC}^{\textrm{srt}}\textrm{I})$;
  \item if $a\in A$ and $t\in \mathcal{B}(\textrm{APTC}^{\textrm{srt}}\textrm{I})$, then $\tilde{\tilde{a}}\cdot t \in \mathcal{B}_1(\textrm{APTC}^{\textrm{srt}}\textrm{I})$;
  \item if $t,t'\in \mathcal{B}_1(\textrm{APTC}^{\textrm{srt}}\textrm{I})$, then $t+t'\in \mathcal{B}_1(\textrm{APTC}^{\textrm{srt}}\textrm{I})$;
  \item if $t,t'\in \mathcal{B}_1(\textrm{APTC}^{\textrm{srt}}\textrm{I})$, then $t\parallel t'\in \mathcal{B}_1(\textrm{APTC}^{\textrm{srt}}\textrm{I})$;
  \item if $t\in \mathcal{B}_1(\textrm{APTC}^{\textrm{srt}}\textrm{I})$, then $t\in \mathcal{B}_0(\textrm{APTC}^{\textrm{srt}}\textrm{I})$;
  \item if $p>0$ and $t\in \mathcal{B}_0(\textrm{APTC}^{\textrm{srt}}\textrm{I})$, then $\sigma^p_{\textrm{rel}}(t) \in \mathcal{B}_0(\textrm{APTC}^{\textrm{srt}}\textrm{I})$;
  \item if $p>0$, $t\in \mathcal{B}_1(\textrm{APTC}^{\textrm{srt}}\textrm{I})$ and $t'\in \mathcal{B}_0(\textrm{APTC}^{\textrm{srt}}\textrm{I})$, then $t+\sigma^p_{\textrm{rel}}(t') \in \mathcal{B}_0(\textrm{APTC}^{\textrm{srt}}\textrm{I})$;
  \item if $t\in \mathcal{B}_0(\textrm{APTC}^{\textrm{srt}}\textrm{I})$, then $\nu_{\textrm{rel}}(t) \in \mathcal{B}_0(\textrm{APTC}^{\textrm{srt}}\textrm{I})$;
  \item if $t\in \mathcal{B}_0(\textrm{APTC}^{\textrm{srt}}\textrm{I})$, then $\int_{v\in V}(t) \in \mathcal{B}_0(\textrm{APTC}^{\textrm{srt}}\textrm{I})$;
  \item $\dot{\delta}\in \mathcal{B}(\textrm{APTC}^{\textrm{srt}}\textrm{I})$;
  \item if $t\in \mathcal{B}_0(\textrm{APTC}^{\textrm{srt}}\textrm{I})$, then $t\in \mathcal{B}(\textrm{APTC}^{\textrm{srt}}\textrm{I})$.
\end{enumerate}
\end{definition}

\begin{theorem}[Elimination theorem]
Let $p$ be a closed $\textrm{APTC}^{\textrm{srt}}\textrm{I}$ term. Then there is a basic $\textrm{APTC}^{\textrm{srt}}\textrm{I}$ term $q$ such that $\textrm{APTC}^{\textrm{srt}}\vdash p=q$.
\end{theorem}

\begin{proof}
It is sufficient to induct on the structure of the closed $\textrm{APTC}^{\textrm{dat}}\textrm{I}$ term $p$. It can be proven that $p$ combined by the constants and operators of $\textrm{APTC}^{\textrm{dat}}\textrm{I}$ exists an equal basic term $q$, and the other operators not included in the basic terms, such as $\upsilon_{\textrm{rel}}$, $\overline{\upsilon}_{\textrm{rel}}$, $\between$, $\mid$, $\partial_H$, $\Theta$ and $\triangleleft$ can be eliminated.
\end{proof}

\subsubsection{Connections}

\begin{theorem}[Generalization of $\textrm{APTC}^{\textrm{srt}}\textrm{I}$]
\begin{enumerate}
  \item By the definitions of $a=\int_{v\in[0,\infty)}\sigma^v_{\textrm{rel}}(\tilde{\tilde{a}})$ for each $a\in A$ and $\delta=\int_{v\in[0,\infty)}\sigma^v_{\textrm{rel}}(\tilde{\tilde{\delta}})$, $\textrm{APTC}^{\textrm{srt}}\textrm{I}$ is a generalization of $APTC$.
  \item $\textrm{APTC}^{\textrm{srt}}\textrm{I}$ is a generalization of $APTC^{\textrm{srt}}$.
\end{enumerate}

\end{theorem}

\begin{proof}
\begin{enumerate}
  \item It follows from the following two facts.

    \begin{enumerate}
      \item The transition rules of $APTC$ in section \ref{tcpa} are all source-dependent;
      \item The sources of the transition rules of $\textrm{APTC}^{\textrm{srt}}\textrm{I}$ contain an occurrence of $\dot{\delta}$, $\tilde{\tilde{a}}$, $\sigma^p_{\textrm{rel}}$, $\upsilon^p_{\textrm{rel}}$, $\overline{\upsilon}^p_{\textrm{rel}}$, $\nu_{\textrm{rel}}$, and $\int$.
    \end{enumerate}

    So, $APTC$ is an embedding of $\textrm{APTC}^{\textrm{srt}}\textrm{I}$, as desired.
  \item It follows from the following two facts.

    \begin{enumerate}
      \item The transition rules of $APTC^{\textrm{srt}}$ are all source-dependent;
      \item The sources of the transition rules of $\textrm{APTC}^{\textrm{srt}}\textrm{I}$ contain an occurrence of $\int$.
    \end{enumerate}

    So, $APTC^{\textrm{srt}}$ is an embedding of $\textrm{APTC}^{\textrm{srt}}\textrm{I}$, as desired.
\end{enumerate}
\end{proof}

\subsubsection{Congruence}

\begin{theorem}[Congruence of $\textrm{APTC}^{\textrm{srt}}\textrm{I}$]
Truly concurrent bisimulation equivalences are all congruences with respect to $\textrm{APTC}^{\textrm{srt}}\textrm{I}$. That is,
\begin{itemize}
  \item pomset bisimulation equivalence $\sim_{p}$ is a congruence with respect to $\textrm{APTC}^{\textrm{srt}}\textrm{I}$;
  \item step bisimulation equivalence $\sim_{s}$ is a congruence with respect to $\textrm{APTC}^{\textrm{srt}}\textrm{I}$;
  \item hp-bisimulation equivalence $\sim_{hp}$ is a congruence with respect to $\textrm{APTC}^{\textrm{srt}}\textrm{I}$;
\end{itemize}
\end{theorem}

\begin{proof}
It is easy to see that $\sim_p$, $\sim_s$, $\sim_{hp}$ and $\sim_{hhp}$ are all equivalent relations on $\textrm{APTC}^{\textrm{srt}}\textrm{I}$ terms, it is only sufficient to prove that $\sim_p$, $\sim_s$, and $\sim_{hp}$ are all preserved by the operators $\int$. It is trivial and we omit it.
\end{proof}

\subsubsection{Soundness}

\begin{theorem}[Soundness of $\textrm{APTC}^{\textrm{srt}}\textrm{I}$]
The axiomatization of $\textrm{APTC}^{\textrm{srt}}\textrm{I}$ is sound modulo truly concurrent bisimulation equivalences $\sim_{p}$, $\sim_{s}$, $\sim_{hp}$ and $\sim_{hhp}$. That is,
\begin{enumerate}
  \item let $x$ and $y$ be $\textrm{APTC}^{\textrm{srt}}\textrm{I}$ terms. If $\textrm{APTC}^{\textrm{srt}}\textrm{I}\vdash x=y$, then $x\sim_{s} y$;
  \item let $x$ and $y$ be $\textrm{APTC}^{\textrm{srt}}\textrm{I}$ terms. If $\textrm{APTC}^{\textrm{srt}}\textrm{I}\vdash x=y$, then $x\sim_{p} y$;
  \item let $x$ and $y$ be $\textrm{APTC}^{\textrm{srt}}\textrm{I}$ terms. If $\textrm{APTC}^{\textrm{srt}}\textrm{I}\vdash x=y$, then $x\sim_{hp} y$;
\end{enumerate}
\end{theorem}

\begin{proof}
Since $\sim_p$, $\sim_s$, $\sim_{hp}$ and $\sim_{hhp}$ are both equivalent and congruent relations, we only need to check if each axiom in Table \ref{AxiomsForAPTCSRTI} is sound modulo $\sim_p$, $\sim_s$, and $\sim_{hp}$ respectively.

\begin{enumerate}
  \item We can check the soundness of each axiom in Table \ref{AxiomsForAPTCSRTI}, by the transition rules in Table \ref{TRForBATCSRTI}, it is trivial and we omit them.
  \item From the definition of pomset bisimulation, we know that pomset bisimulation is defined by pomset transitions, which are labeled by pomsets. In a pomset transition, the events (actions) in the pomset are either within causality relations (defined by $\cdot$) or in concurrency (implicitly defined by $\cdot$ and $+$, and explicitly defined by $\between$), of course, they are pairwise consistent (without conflicts). We have already proven the case that all events are pairwise concurrent (soundness modulo step bisimulation), so, we only need to prove the case of events in causality. Without loss of generality, we take a pomset of $P=\{\tilde{\tilde{a}},\tilde{\tilde{b}}:\tilde{\tilde{a}}\cdot \tilde{\tilde{b}}\}$. Then the pomset transition labeled by the above $P$ is just composed of one single event transition labeled by $\tilde{\tilde{a}}$ succeeded by another single event transition labeled by $\tilde{\tilde{b}}$, that is, $\xrightarrow{P}=\xrightarrow{a}\xrightarrow{b}$.

        Similarly to the proof of soundness modulo step bisimulation equivalence, we can prove that each axiom in Table \ref{AxiomsForAPTCSRTI} is sound modulo pomset bisimulation equivalence, we omit them.
  \item From the definition of hp-bisimulation, we know that hp-bisimulation is defined on the posetal product $(C_1,f,C_2),f:C_1\rightarrow C_2\textrm{ isomorphism}$. Two process terms $s$ related to $C_1$ and $t$ related to $C_2$, and $f:C_1\rightarrow C_2\textrm{ isomorphism}$. Initially, $(C_1,f,C_2)=(\emptyset,\emptyset,\emptyset)$, and $(\emptyset,\emptyset,\emptyset)\in\sim_{hp}$. When $s\xrightarrow{a}s'$ ($C_1\xrightarrow{a}C_1'$), there will be $t\xrightarrow{a}t'$ ($C_2\xrightarrow{a}C_2'$), and we define $f'=f[a\mapsto a]$. Then, if $(C_1,f,C_2)\in\sim_{hp}$, then $(C_1',f',C_2')\in\sim_{hp}$.

        Similarly to the proof of soundness modulo pomset bisimulation equivalence, we can prove that each axiom in Table \ref{AxiomsForAPTCSRTI} is sound modulo hp-bisimulation equivalence, we just need additionally to check the above conditions on hp-bisimulation, we omit them.
  \item We just need to add downward-closed condition to the soundness modulo hp-bisimulation equivalence, we omit them.
\end{enumerate}

\end{proof}

\subsubsection{Completeness}

\begin{theorem}[Completeness of $\textrm{APTC}^{\textrm{srt}}\textrm{I}$]
The axiomatization of $\textrm{APTC}^{\textrm{srt}}\textrm{I}$ is complete modulo truly concurrent bisimulation equivalences $\sim_{p}$, $\sim_{s}$, $\sim_{hp}$ and $\sim_{hhp}$. That is,
\begin{enumerate}
  \item let $p$ and $q$ be closed $\textrm{APTC}^{\textrm{srt}}\textrm{I}$ terms, if $p\sim_{s} q$ then $p=q$;
  \item let $p$ and $q$ be closed $\textrm{APTC}^{\textrm{srt}}\textrm{I}$ terms, if $p\sim_{p} q$ then $p=q$;
  \item let $p$ and $q$ be closed $\textrm{APTC}^{\textrm{srt}}\textrm{I}$ terms, if $p\sim_{hp} q$ then $p=q$;
\end{enumerate}

\end{theorem}

\begin{proof}
\begin{enumerate}
  \item Firstly, by the elimination theorem of $\textrm{APTC}^{\textrm{srt}}\textrm{I}$, we know that for each closed $\textrm{APTC}^{\textrm{srt}}\textrm{I}$ term $p$, there exists a closed basic $\textrm{APTC}^{\textrm{srt}}\textrm{I}$ term $p'$, such that $\textrm{APTC}^{\textrm{srt}}\textrm{I}\vdash p=p'$, so, we only need to consider closed basic $\textrm{APTC}^{\textrm{srt}}\textrm{I}$ terms.

        The basic terms modulo associativity and commutativity (AC) of conflict $+$ (defined by axioms $A1$ and $A2$ in Table \ref{AxiomsForBATCSRT}) and associativity and commutativity (AC) of parallel $\parallel$ (defined by axioms $P2$ and $P3$ in Table \ref{AxiomsForAPTCSRT}), and these equivalences is denoted by $=_{AC}$. Then, each equivalence class $s$ modulo AC of $+$ and $\parallel$ has the following normal form

        $$s_1+\cdots+ s_k$$

        with each $s_i$ either an atomic event or of the form

        $$t_1\cdot\cdots\cdot t_m$$

        with each $t_j$ either an atomic event or of the form

        $$u_1\parallel\cdots\parallel u_n$$

        with each $u_l$ an atomic event, and each $s_i$ is called the summand of $s$.

        Now, we prove that for normal forms $n$ and $n'$, if $n\sim_{s} n'$ then $n=_{AC}n'$. It is sufficient to induct on the sizes of $n$ and $n'$. We can get $n=_{AC} n'$.

        Finally, let $s$ and $t$ be basic $\textrm{APTC}^{\textrm{srt}}\textrm{I}$ terms, and $s\sim_s t$, there are normal forms $n$ and $n'$, such that $s=n$ and $t=n'$. The soundness theorem modulo step bisimulation equivalence yields $s\sim_s n$ and $t\sim_s n'$, so $n\sim_s s\sim_s t\sim_s n'$. Since if $n\sim_s n'$ then $n=_{AC}n'$, $s=n=_{AC}n'=t$, as desired.
  \item This case can be proven similarly, just by replacement of $\sim_{s}$ by $\sim_{p}$.
  \item This case can be proven similarly, just by replacement of $\sim_{s}$ by $\sim_{hp}$.
\end{enumerate}
\end{proof}

\section{Continuous Absolute Timing}{\label{sat}}

In this section, we will introduce a version of APTC with absolute timing and time measured on a continuous time scale. Measuring time on  a continuous time scale means that timing is now done with respect ro time points on a continuous time scale. While in absolute timing, all timing is counted from the start of the whole process.

Like APTC without timing, let us start with a basic algebra for true concurrency called $\textrm{BATC}^{\textrm{sat}}$ (BATC with continuous absolute timing). Then we continue with $\textrm{APTC}^{\textrm{sat}}$ (APTC with continuous absolute timing).

\subsection{Basic Definitions}

In this subsection, we will introduce some basic definitions about timing. These basic concepts come from \cite{T3}, we introduce them into the backgrounds of true concurrency.

\begin{definition}[Undelayable actions]
Undelayable actions are defined as atomic processes that perform an action and then terminate successfully. We use a constant $\tilde{a}$ to represent the undelayable action, that is, the atomic process that performs the action $a$ and then terminates successfully.
\end{definition}

\begin{definition}[Undelayable deadlock]
Undelayable deadlock $\tilde{\delta}$ is an additional process that is neither capable of performing any action nor capable of idling beyond the current point of time.
\end{definition}

\begin{definition}[Absolute delay]
The absolute delay of the process $p$ for a period of time $r$ ($r\in\mathbb{R}^{\geq}$) is the process that idles a period of time $r$ longer than $p$ and otherwise behaves like $p$. The operator $\sigma_{\textrm{abs}}$ is used to represent the absolute delay, and let $\sigma^r_{\textrm{abs}}(t) = r \sigma_{\textrm{abs}} t$.
\end{definition}

\begin{definition}[Deadlocked process]
Deadlocked process $\dot{\delta}$ is an additional process that has deadlocked before point of time 0. After a delay of a period of time, the undelayable deadlock $\tilde{\delta}$ and the deadlocked process $\dot{\delta}$ are indistinguishable from each other.
\end{definition}

\begin{definition}[Truly concurrent bisimulation equivalences with time-related capabilities]\label{TBTTC4}
The following requirement with time-related capabilities is added to truly concurrent bisimulation equivalences $\sim_{p}$, $\sim_{s}$, $\sim_{hp}$ and $\sim_{hhp}$ and Definition \ref{TBTTC3}:
\begin{itemize}
  \item in case of absolute timing, the requirements in Definition \ref{TBTTC1} apply to the capabilities in a certain period of time.
\end{itemize}
\end{definition}

\begin{definition}[Integration]
Let $f$ be a function from $\mathbb{R}^{\geq}$ to processes with continuous absolute timing and $V\subseteq \mathbb{R}^{\geq}$. The integration $\int$ of $f$ over $V$ is the process that behaves like one of the process in $\{f(r)|r\in V\}$.
\end{definition}

\begin{definition}[Absolute time-out]
The absolute time-out $\upsilon_{\textrm{abs}}$ of a process $p$ at point of time $r$ ($r\in\mathbb{R}^{\geq}$) behaves either like the part of $p$ that does not idle till point of time $r$, or like the deadlocked process after a delay of a period of time $r$ if $p$ is capable of idling till point of time $r$; otherwise, like $p$. And let $\upsilon^r_{\textrm{abs}}(t) = r \upsilon_{\textrm{abs}} t$.
\end{definition}

\begin{definition}[Absolute initialization]
The absolute initialization $\overline{\upsilon}_{\textrm{abs}}$ of a process $p$ at point of time $r$ ($r\in\mathbb{R}^{\geq}$) behaves like the part of $p$ that idles till point of time $r$ if $p$ is capable of idling till that point of time; otherwise, like the deadlocked process after a delay of a period of time $r$. And we let $\overline{\upsilon}^r_{\textrm{abs}}(t) = r \overline{\upsilon}_{\textrm{abs}} t$.
\end{definition}

\subsection{Basic Algebra for True Concurrency with Continuous Absolute Timing}

In this subsection, we will introduce the theory $\textrm{BATC}^{\textrm{sat}}$.

\subsubsection{The Theory $\textrm{BATC}^{\textrm{sat}}$}

\begin{definition}[Signature of $\textrm{BATC}^{\textrm{sat}}$]
The signature of $\textrm{BATC}^{\textrm{sat}}$ consists of the sort $\mathcal{P}_{\textrm{abs}}$ of processes with continuous absolute timing, the undelayable action constants $\tilde{a}: \rightarrow\mathcal{P}_{\textrm{abs}}$ for each $a\in A$, the undelayable deadlock constant $\tilde{\delta}: \rightarrow \mathcal{P}_{\textrm{abs}}$, the alternative composition operator $+: \mathcal{P}_{\textrm{abs}}\times\mathcal{P}_{\textrm{abs}} \rightarrow \mathcal{P}_{\textrm{abs}}$, the sequential composition operator $\cdot: \mathcal{P}_{\textrm{abs}} \times \mathcal{P}_{\textrm{abs}} \rightarrow \mathcal{P}_{\textrm{abs}}$, the absolute delay operator $\sigma_{\textrm{abs}}: \mathbb{R}^{\geq}\times \mathcal{P}_{\textrm{abs}} \rightarrow \mathcal{P}_{\textrm{abs}}$, the deadlocked process constant $\dot{\delta}: \rightarrow \mathcal{P}_{\textrm{abs}}$, the absolute time-out operator $\upsilon_{\textrm{abs}}: \mathbb{R}^{\geq}\times\mathcal{P}_{\textrm{abs}} \rightarrow\mathcal{P}_{\textrm{abs}}$ and the absolute initialization operator $\overline{\upsilon}_{\textrm{abs}}: \mathbb{R}^{\geq}\times\mathcal{P}_{\textrm{abs}} \rightarrow\mathcal{P}_{\textrm{abs}}$.
\end{definition}

The set of axioms of $\textrm{BATC}^{\textrm{sat}}$ consists of the laws given in Table \ref{AxiomsForBATCSAT}.

\begin{center}
    \begin{table}
        \begin{tabular}{@{}ll@{}}
            \hline No. &Axiom\\
            $A1$ & $x+ y = y+ x$\\
            $A2$ & $(x+ y)+ z = x+ (y+ z)$\\
            $A3$ & $x+ x = x$\\
            $A4$ & $(x+ y)\cdot z = x\cdot z + y\cdot z$\\
            $A5$ & $(x\cdot y)\cdot z = x\cdot(y\cdot z)$\\
            $A6ID$ & $x + \dot{\delta} = x$\\
            $A7ID$ & $\dot{\delta}\cdot x = \dot{\delta}$\\
            $SAT1$ & $\sigma^0_{\textrm{abs}}(x) = \overline{\upsilon}^0_{\textrm{abs}}(x)$\\
            $SAT2$ & $\sigma^q_{\textrm{abs}}( \sigma^p_{\textrm{abs}}(x)) = \sigma^{q+p}_{\textrm{abs}}(x)$\\
            $SAT3$ & $\sigma^p_{\textrm{abs}}(x) + \sigma^p_{\textrm{abs}}(y) = \sigma^p_{\textrm{abs}}(x+y)$\\
            $SAT4$ & $\sigma^p_{\textrm{abs}}(x)\cdot \upsilon^p_{\textrm{abs}}(y) = \sigma^p_{\textrm{abs}}(x\cdot \dot{\delta})$\\
            $SAT5$ & $\sigma^p_{\textrm{abs}}(x)\cdot (\upsilon^p_{\textrm{abs}}(y) +\sigma^p_{\textrm{abs}}(z)) = \sigma^p_{\textrm{abs}}(x\cdot\overline{\upsilon}^0_{\textrm{abs}}(z))$\\
            $SAT6$ & $\sigma^p_{\textrm{abs}}(\dot{\delta})\cdot x = \sigma^p_{\textrm{abs}}(\dot{\delta})$\\
            $A6SAa$ & $\tilde{a} + \tilde{\delta} = \tilde{a}$\\
            $A6SAb$ & $\sigma^r_{\textrm{abs}}(x)+\tilde{\delta} = \sigma^r_{\textrm{abs}}(x)$\\
            $A7SA$ & $\tilde{\delta}\cdot x = \tilde{\delta}$\\
            $SATO0$ & $\upsilon^p_{\textrm{abs}}(\dot{\delta}) = \dot{\delta}$\\
            $SATO1$ & $\upsilon^0_{\textrm{abs}}(x) = \dot(\delta)$\\
            $SATO2$ & $\upsilon^{r}_{\textrm{abs}}(\tilde{a}) = \tilde{a}$\\
            $SATO3$ & $\upsilon^{q+p}_{\textrm{abs}} (\sigma^p_{\textrm{abs}}(x)) = \sigma^p_{\textrm{abs}}(\upsilon^q_{\textrm{abs}}(x))$\\
            $SATO4$ & $\upsilon^p_{\textrm{abs}}(x+y) = \upsilon^p_{\textrm{abs}}(x) + \upsilon^p_{\textrm{abs}}(y)$\\
            $SATO5$ & $\upsilon^p_{\textrm{abs}}(x\cdot y) = \upsilon^p_{\textrm{abs}}(x)\cdot y$\\
            $SAI0a$ & $\overline{\upsilon}^0_{\textrm{abs}}(\dot{\delta}) = \dot{\delta}$\\
            $SAI0b$ & $\overline{\upsilon}^r_{\textrm{abs}}(\dot{\delta}) = \sigma^r_{\textrm{abs}}(\dot{\delta})$\\
            $SAI1$ & $\overline{\upsilon}^{0}_{\textrm{abs}}(\tilde{a}) = \tilde{a}$\\
            $SAI2$ & $\overline{\upsilon}^{r}_{\textrm{abs}}(\tilde{a}) = \sigma^{r}_{\textrm{abs}}(\dot{\delta})$\\
            $SAI3$ & $\overline{\upsilon}^{q+p}_{\textrm{abs}} (\sigma^p_{\textrm{abs}}(x)) = \sigma^p_{\textrm{abs}}(\overline{\upsilon}^q_{\textrm{abs}}(\overline{\upsilon}^0_{\textrm{abs}}(x)))$\\
            $SAI4$ & $\overline{\upsilon}^p_{\textrm{abs}}(x+y) = \overline{\upsilon}^p_{\textrm{abs}}(x) + \overline{\upsilon}^p_{\textrm{abs}}(y)$\\
            $SAI5$ & $\overline{\upsilon}^p_{\textrm{abs}}(x\cdot y) = \overline{\upsilon}^p_{\textrm{abs}}(x)\cdot y$\\
        \end{tabular}
        \caption{Axioms of $\textrm{BATC}^{\textrm{sat}}(a\in A_{\delta}, p,q\geq 0, r>0)$}
        \label{AxiomsForBATCSAT}
    \end{table}
\end{center}

The operational semantics of $\textrm{BATC}^{\textrm{sat}}$ are defined by the transition rules in Table \ref{TRForBATCSAT}. The transition rules are defined on $\langle t, p\rangle$, where $t$ is a term and $p\in\mathbb{R}^{\geq}$. Where $\uparrow$ is a unary deadlocked predicate, and $\langle t, p\rangle\nuparrow \triangleq \neg(\langle t, p\rangle\uparrow)$; $\langle t,p\rangle\mapsto^q \langle t',p'\rangle$ means that process $t$ is capable of first idling till the $q$th-next time slice, and then proceeding as process $t'$ and $q+p=p'$.

\begin{center}
    \begin{table}
        $$\frac{}{\langle\dot{\delta},p\rangle\uparrow}\quad\frac{}{\langle \tilde{\delta},r\rangle\uparrow}
        \quad\frac{}{\langle\tilde{a},0\rangle\xrightarrow{a}\langle\surd,0\rangle} \quad \frac{}{\langle \tilde{a},r\rangle\uparrow}$$

        $$\frac{\langle x,p\rangle\xrightarrow{a}\langle x',p\rangle}{\langle\sigma^0_{\textrm{abs}}(x),p\rangle \xrightarrow{a}\langle x',p\rangle}
        \quad\frac{\langle x,p\rangle\xrightarrow{a}\langle x',p\rangle}{\langle\sigma^r_{\textrm{abs}}(x),p+r\rangle \xrightarrow{a}\langle \sigma^r_{\textrm{abs}}(x'),p+r\rangle}$$

        $$\frac{\langle x,p\rangle\xrightarrow{a}\langle\surd,p\rangle}{\langle\sigma^{q}_{\textrm{abs}}(x),p+q\rangle \xrightarrow{a}\langle\surd,p+q\rangle}
        \quad\frac{\langle x,p\rangle\uparrow}{\langle\sigma^{q}_{\textrm{abs}}(x),p+q\rangle\uparrow}$$

        $$\frac{}{\langle\sigma^{r+q}_{\textrm{abs}}(x),p\rangle\mapsto^{r} \langle\sigma^{r+q}_{\textrm{abs}}(x),p+r\rangle}(q>p)
        \quad \frac{\langle x,0\rangle\nuparrow}{\langle\sigma^{r+q}_{\textrm{abs}}(x),q\rangle\mapsto^r \sigma^{r+q}_{\textrm{abs}}(x),r+q\rangle}$$

        $$\frac{\langle x,p\rangle\mapsto^r \langle x, p+r\rangle}{\langle\sigma^{q}_{\textrm{abs}}(x),p+q\rangle \mapsto^{r} \langle\sigma^{q}_{\textrm{abs}}(x),p+r+q\rangle}$$

        $$\frac{\langle x,p\rangle\mapsto^r \langle x, p+r\rangle}{\langle\sigma^{q}_{\textrm{abs}}(x),p\rangle \mapsto^{r+q} \langle\sigma^{q}_{\textrm{abs}}(x),p+r+q\rangle}$$

        $$\frac{\langle x,p\rangle\xrightarrow{a}\langle x',p\rangle}{\langle x+ y,p\rangle\xrightarrow{a}\langle x',p\rangle}
        \quad\frac{\langle y,p\rangle\xrightarrow{a}\langle y',p\rangle}{\langle x+ y,p\rangle\xrightarrow{a}\langle y',p\rangle}$$

        $$\frac{\langle x,p\rangle\xrightarrow{a}\langle \surd,p\rangle}{\langle x+ y,p\rangle\xrightarrow{a}\langle \surd,p\rangle}
        \quad\frac{\langle y,p\rangle\xrightarrow{a}\langle \surd,p\rangle}{\langle x+ y,p\rangle\xrightarrow{a}\langle \surd,p\rangle}$$

        $$\frac{\langle x,p\rangle\mapsto^{r}\langle x,p+r\rangle}{\langle x+ y,p\rangle\mapsto^{r}\langle x+y,p+r\rangle} \quad\frac{\langle y,p\rangle\mapsto^r \langle y,p+r\rangle}{\langle x+ y,p\rangle\mapsto^{r}\langle x+y,p+r\rangle}
        \quad\frac{\langle x,p\rangle\uparrow\quad \langle y,p\rangle\uparrow}{\langle x+ y,p\rangle\uparrow}$$

        $$\frac{\langle x,p\rangle\xrightarrow{a}\langle\surd,p\rangle}{\langle x\cdot y,p\rangle\xrightarrow{a} \langle y,p\rangle}
        \quad\frac{\langle x,p\rangle\xrightarrow{a}\langle x',p\rangle}{\langle x\cdot y,p\rangle\xrightarrow{a}\langle x'\cdot y,p\rangle}$$

        $$\frac{\langle x,p\rangle\mapsto^{r}\langle x,p+r\rangle}{\langle x\cdot y,p\rangle\mapsto^{r}\langle x\cdot y,p+r\rangle}
        \quad \frac{\langle x,p\rangle\uparrow}{\langle x\cdot y,p\rangle\uparrow}$$

        $$\frac{\langle x,p\rangle\xrightarrow{a}\langle x',p\rangle}{\langle\upsilon^{q}_{\textrm{abs}}(x),p\rangle \xrightarrow{a}\langle x',p\rangle}(q>p)
        \quad\frac{\langle x,p\rangle\xrightarrow{a}\langle\surd,p\rangle}{\langle\upsilon^{q}_{\textrm{abs}}(x),p\rangle \xrightarrow{a}\langle\surd,p\rangle}(q>p)$$

        $$\frac{\langle x,p\rangle\mapsto^r \langle x,p+r\rangle}{\langle\upsilon^{q}_{\textrm{abs}}(x),p\rangle \mapsto^r \langle\upsilon^{q}_{\textrm{abs}}(x),p+r\rangle}(q>p+r)$$

        $$\frac{}{\langle\upsilon^{q}_{\textrm{abs}}(x),p\rangle\uparrow}(q\leq p)
        \quad\frac{\langle x,p\rangle\uparrow}{\langle\upsilon^{q}_{\textrm{abs}}(x),p\rangle\uparrow}(q>p)$$

        $$\frac{\langle x,p\rangle\xrightarrow{a}\langle x',p\rangle}{\langle\overline{\upsilon}^{q}_{\textrm{abs}}(x),p\rangle \xrightarrow{a}\langle x',p\rangle}(q\leq p)
        \quad\frac{\langle x,p\rangle\xrightarrow{a}\langle\surd,p\rangle}{\langle\overline{\upsilon}^{q}_{\textrm{abs}}(x),p\rangle \xrightarrow{a}\langle\surd,p\rangle}(q\leq p)$$

        $$\frac{}{\langle\overline{\upsilon}^{r+q}_{\textrm{abs}}(x),p\rangle\mapsto^r \langle\overline{\upsilon}^{r+q}_{\textrm{abs}}(x),p+r}(q>p)$$

        $$\frac{\langle x,r+q\rangle\nuparrow}{\langle\overline{\upsilon}^{r+q}_{\textrm{abs}}(x),q\rangle \mapsto^r \langle\overline{\upsilon}^{r+q}_{\textrm{abs}}(x),r+q\rangle}$$

        $$\frac{\langle x,p\rangle\mapsto^r\langle x,p+r\rangle}{\langle\overline{\upsilon}^{q}_{\textrm{abs}}(x),p\rangle \mapsto^r \langle\overline{\upsilon}^{q}_{\textrm{abs}}(x),p+r\rangle}(q\leq p+r)
        \quad\frac{\langle x,p\rangle\uparrow}{\langle\overline{\upsilon}^{q}_{\textrm{abs}}(x),p\rangle\uparrow}(q\leq p)$$
        \caption{Transition rules of $\textrm{BATC}^{\textrm{sat}}(a\in A,p,q\geq 0, r>0)$}
        \label{TRForBATCSAT}
    \end{table}
\end{center}

\subsubsection{Elimination}

\begin{definition}[Basic terms of $\textrm{BATC}^{\textrm{sat}}$]
The set of basic terms of $\textrm{BATC}^{\textrm{sat}}$, $\mathcal{B}(\textrm{BATC}^{\textrm{sat}})$, is inductively defined as follows by two auxiliary sets $\mathcal{B}_0(\textrm{BATC}^{\textrm{sat}})$ and $\mathcal{B}_1(\textrm{BATC}^{\textrm{sat}})$:
\begin{enumerate}
  \item if $a\in A_{\delta}$, then $\tilde{a} \in \mathcal{B}_1(\textrm{BATC}^{\textrm{sat}})$;
  \item if $a\in A$ and $t\in \mathcal{B}(\textrm{BATC}^{\textrm{sat}})$, then $\tilde{a}\cdot t \in \mathcal{B}_1(\textrm{BATC}^{\textrm{sat}})$;
  \item if $t,t'\in \mathcal{B}_1(\textrm{BATC}^{\textrm{sat}})$, then $t+t'\in \mathcal{B}_1(\textrm{BATC}^{\textrm{sat}})$;
  \item if $t\in \mathcal{B}_1(\textrm{BATC}^{\textrm{sat}})$, then $t\in \mathcal{B}_0(\textrm{BATC}^{\textrm{sat}})$;
  \item if $p>0$ and $t\in \mathcal{B}_0(\textrm{BATC}^{\textrm{sat}})$, then $\sigma^p_{\textrm{abs}}(t) \in \mathcal{B}_0(\textrm{BATC}^{\textrm{sat}})$;
  \item if $p>0$, $t\in \mathcal{B}_1(\textrm{BATC}^{\textrm{sat}})$ and $t'\in \mathcal{B}_0(\textrm{BATC}^{\textrm{sat}})$, then $t+\sigma^p_{\textrm{abs}}(t') \in \mathcal{B}_0(\textrm{BATC}^{\textrm{sat}})$;
  \item $\dot{\delta}\in \mathcal{B}(\textrm{BATC}^{\textrm{sat}})$;
  \item if $t\in \mathcal{B}_0(\textrm{BATC}^{\textrm{sat}})$, then $t\in \mathcal{B}(\textrm{BATC}^{\textrm{sat}})$.
\end{enumerate}
\end{definition}

\begin{theorem}[Elimination theorem]
Let $p$ be a closed $\textrm{BATC}^{\textrm{sat}}$ term. Then there is a basic $\textrm{BATC}^{\textrm{sat}}$ term $q$ such that $\textrm{BATC}^{\textrm{sat}}\vdash p=q$.
\end{theorem}

\begin{proof}
It is sufficient to induct on the structure of the closed $\textrm{BATC}^{\textrm{sat}}$ term $p$. It can be proven that $p$ combined by the constants and operators of $\textrm{BATC}^{\textrm{sat}}$ exists an equal basic term $q$, and the other operators not included in the basic terms, such as $\upsilon_{\textrm{abs}}$ and $\overline{\upsilon}_{\textrm{abs}}$ can be eliminated.
\end{proof}

\subsubsection{Connections}

\begin{theorem}[Generalization of $\textrm{BATC}^{\textrm{sat}}$]

By the definitions of $a=\tilde{a}$ for each $a\in A$ and $\delta=\tilde{\delta}$, $\textrm{BATC}^{\textrm{sat}}$ is a generalization of $BATC$.
\end{theorem}

\begin{proof}
It follows from the following two facts.

\begin{enumerate}
  \item The transition rules of $BATC$ in section \ref{tcpa} are all source-dependent;
  \item The sources of the transition rules of $\textrm{BATC}^{\textrm{sat}}$ contain an occurrence of $\dot{\delta}$, $\tilde{a}$, $\sigma^p_{\textrm{abs}}$, $\upsilon^p_{\textrm{abs}}$ and $\overline{\upsilon}^p_{\textrm{abs}}$.
\end{enumerate}

So, $BATC$ is an embedding of $\textrm{BATC}^{\textrm{sat}}$, as desired.
\end{proof}

\subsubsection{Congruence}

\begin{theorem}[Congruence of $\textrm{BATC}^{\textrm{sat}}$]
Truly concurrent bisimulation equivalences are all congruences with respect to $\textrm{BATC}^{\textrm{sat}}$. That is,
\begin{itemize}
  \item pomset bisimulation equivalence $\sim_{p}$ is a congruence with respect to $\textrm{BATC}^{\textrm{sat}}$;
  \item step bisimulation equivalence $\sim_{s}$ is a congruence with respect to $\textrm{BATC}^{\textrm{sat}}$;
  \item hp-bisimulation equivalence $\sim_{hp}$ is a congruence with respect to $\textrm{BATC}^{\textrm{sat}}$;
  \item hhp-bisimulation equivalence $\sim_{hhp}$ is a congruence with respect to $\textrm{BATC}^{\textrm{sat}}$.
\end{itemize}
\end{theorem}

\begin{proof}
It is easy to see that $\sim_p$, $\sim_s$, $\sim_{hp}$ and $\sim_{hhp}$ are all equivalent relations on $\textrm{BATC}^{\textrm{sat}}$ terms, it is only sufficient to prove that $\sim_p$, $\sim_s$, $\sim_{hp}$ and $\sim_{hhp}$ are all preserved by the operators $\sigma^p_{\textrm{abs}}$, $\upsilon^p_{\textrm{abs}}$ and $\overline{\upsilon}^p_{\textrm{abs}}$. It is trivial and we omit it.
\end{proof}

\subsubsection{Soundness}

\begin{theorem}[Soundness of $\textrm{BATC}^{\textrm{sat}}$]
The axiomatization of $\textrm{BATC}^{\textrm{sat}}$ is sound modulo truly concurrent bisimulation equivalences $\sim_{p}$, $\sim_{s}$, $\sim_{hp}$ and $\sim_{hhp}$. That is,
\begin{enumerate}
  \item let $x$ and $y$ be $\textrm{BATC}^{\textrm{sat}}$ terms. If $\textrm{BATC}^{\textrm{sat}}\vdash x=y$, then $x\sim_{s} y$;
  \item let $x$ and $y$ be $\textrm{BATC}^{\textrm{sat}}$ terms. If $\textrm{BATC}^{\textrm{sat}}\vdash x=y$, then $x\sim_{p} y$;
  \item let $x$ and $y$ be $\textrm{BATC}^{\textrm{sat}}$ terms. If $\textrm{BATC}^{\textrm{sat}}\vdash x=y$, then $x\sim_{hp} y$;
  \item let $x$ and $y$ be $\textrm{BATC}^{\textrm{sat}}$ terms. If $\textrm{BATC}^{\textrm{sat}}\vdash x=y$, then $x\sim_{hhp} y$.
\end{enumerate}
\end{theorem}

\begin{proof}
Since $\sim_p$, $\sim_s$, $\sim_{hp}$ and $\sim_{hhp}$ are both equivalent and congruent relations, we only need to check if each axiom in Table \ref{AxiomsForBATCSAT} is sound modulo $\sim_p$, $\sim_s$, $\sim_{hp}$ and $\sim_{hhp}$ respectively.

\begin{enumerate}
  \item We only check the soundness of the non-trivial axiom $SATO3$ modulo $\sim_s$.
        Let $p$ be $\textrm{BATC}^{\textrm{sat}}$ processes, and $\upsilon^{r+s}_{\textrm{abs}} (\sigma^s_{\textrm{abs}}(p)) = \sigma^s_{\textrm{abs}}(\upsilon^r_{\textrm{abs}}(p))$, it is sufficient to prove that $\upsilon^{r+s}_{\textrm{abs}} (\sigma^s_{\textrm{abs}}(p)) \sim_{s} \sigma^s_{\textrm{abs}}(\upsilon^r_{\textrm{abs}}(p))$. By the transition rules of operator $\sigma^s_{\textrm{abs}}$ and $\upsilon^r_{\textrm{abs}}$ in Table \ref{TRForBATCSAT}, we get

        $$\frac{\langle p,0\rangle\nuparrow}{\langle\upsilon^{r+s}_{\textrm{abs}}(\sigma^s_{\textrm{abs}}(p)),s'\rangle\mapsto^s \langle\upsilon^{r}_{\textrm{abs}} (\sigma^s_{\textrm{abs}}(p)),s'+s\rangle}$$

        $$\frac{\langle p,0\rangle\nuparrow}{\langle\sigma^s_{\textrm{abs}}(\upsilon^r_{\textrm{abs}}(p)),s'\rangle\mapsto^s \langle\sigma^s_{\textrm{abs}}(\upsilon^r_{\textrm{abs}}(p)),s'+s\rangle}$$

        There are several cases:

        $$\frac{\langle p,s'\rangle\xrightarrow{a} \langle\surd,s'\rangle}{\langle\upsilon^{r}_{\textrm{abs}} (\sigma^s_{\textrm{abs}}(p)),s'+s\rangle\xrightarrow{a}\langle\surd,s'+s\rangle}$$

        $$\frac{\langle p,s'\rangle\xrightarrow{a} \langle\surd,s'\rangle}{\langle\sigma^s_{\textrm{abs}}(\upsilon^r_{\textrm{abs}}(p)),s'+s\rangle\xrightarrow{a}\langle\surd,s'+s\rangle}$$

        $$\frac{\langle p,s'\rangle\xrightarrow{a} \langle p',s'\rangle}{\langle\upsilon^{r}_{\textrm{abs}} (\sigma^s_{\textrm{abs}}(p)),s'+s\rangle\xrightarrow{a}\langle\sigma^s_{\textrm{abs}}(p'),s'+s\rangle}$$

        $$\frac{\langle p,s'\rangle\xrightarrow{a} \langle p',s'\rangle}{\langle\sigma^s_{\textrm{abs}}(\upsilon^r_{\textrm{abs}}(p)),s'+s\rangle\xrightarrow{a}\langle\sigma^s_{\textrm{abs}}(p'),s'+s\rangle}$$

        $$\frac{\langle p,s'\rangle\uparrow}{\langle\upsilon^{r}_{\textrm{abs}} (\sigma^s_{\textrm{abs}}(p)),s'+s\rangle\uparrow}$$

        $$\frac{\langle p,s'\rangle\uparrow}{\langle\sigma^s_{\textrm{abs}}(\upsilon^r_{\textrm{abs}}(p)),s'+s\rangle\uparrow}$$

        So, we see that each case leads to $\upsilon^{r+s}_{\textrm{abs}} (\sigma^s_{\textrm{abs}}(p)) \sim_{s} \sigma^s_{\textrm{abs}}(\upsilon^r_{\textrm{abs}}(p))$, as desired.
  \item From the definition of pomset bisimulation, we know that pomset bisimulation is defined by pomset transitions, which are labeled by pomsets. In a pomset transition, the events (actions) in the pomset are either within causality relations (defined by $\cdot$) or in concurrency (implicitly defined by $\cdot$ and $+$, and explicitly defined by $\between$), of course, they are pairwise consistent (without conflicts). We have already proven the case that all events are pairwise concurrent (soundness modulo step bisimulation), so, we only need to prove the case of events in causality. Without loss of generality, we take a pomset of $P=\{\tilde{a},\tilde{b}:\tilde{a}\cdot \tilde{b}\}$. Then the pomset transition labeled by the above $P$ is just composed of one single event transition labeled by $\tilde{a}$ succeeded by another single event transition labeled by $\tilde{b}$, that is, $\xrightarrow{P}=\xrightarrow{a}\xrightarrow{b}$.

        Similarly to the proof of soundness modulo step bisimulation equivalence, we can prove that each axiom in Table \ref{AxiomsForBATCSAT} is sound modulo pomset bisimulation equivalence, we omit them.
  \item From the definition of hp-bisimulation, we know that hp-bisimulation is defined on the posetal product $(C_1,f,C_2),f:C_1\rightarrow C_2\textrm{ isomorphism}$. Two process terms $s$ related to $C_1$ and $t$ related to $C_2$, and $f:C_1\rightarrow C_2\textrm{ isomorphism}$. Initially, $(C_1,f,C_2)=(\emptyset,\emptyset,\emptyset)$, and $(\emptyset,\emptyset,\emptyset)\in\sim_{hp}$. When $s\xrightarrow{a}s'$ ($C_1\xrightarrow{a}C_1'$), there will be $t\xrightarrow{a}t'$ ($C_2\xrightarrow{a}C_2'$), and we define $f'=f[a\mapsto a]$. Then, if $(C_1,f,C_2)\in\sim_{hp}$, then $(C_1',f',C_2')\in\sim_{hp}$.

        Similarly to the proof of soundness modulo pomset bisimulation equivalence, we can prove that each axiom in Table \ref{AxiomsForBATCSAT} is sound modulo hp-bisimulation equivalence, we just need additionally to check the above conditions on hp-bisimulation, we omit them.
  \item We just need to add downward-closed condition to the soundness modulo hp-bisimulation equivalence, we omit them.
\end{enumerate}

\end{proof}

\subsubsection{Completeness}

\begin{theorem}[Completeness of $\textrm{BATC}^{\textrm{sat}}$]
The axiomatization of $\textrm{BATC}^{\textrm{sat}}$ is complete modulo truly concurrent bisimulation equivalences $\sim_{p}$, $\sim_{s}$, $\sim_{hp}$ and $\sim_{hhp}$. That is,
\begin{enumerate}
  \item let $p$ and $q$ be closed $\textrm{BATC}^{\textrm{sat}}$ terms, if $p\sim_{s} q$ then $p=q$;
  \item let $p$ and $q$ be closed $\textrm{BATC}^{\textrm{sat}}$ terms, if $p\sim_{p} q$ then $p=q$;
  \item let $p$ and $q$ be closed $\textrm{BATC}^{\textrm{sat}}$ terms, if $p\sim_{hp} q$ then $p=q$;
  \item let $p$ and $q$ be closed $\textrm{BATC}^{\textrm{sat}}$ terms, if $p\sim_{hhp} q$ then $p=q$.
\end{enumerate}

\end{theorem}

\begin{proof}
\begin{enumerate}
  \item Firstly, by the elimination theorem of $\textrm{BATC}^{\textrm{sat}}$, we know that for each closed $\textrm{BATC}^{\textrm{sat}}$ term $p$, there exists a closed basic $\textrm{BATC}^{\textrm{sat}}$ term $p'$, such that $\textrm{BATC}^{\textrm{sat}}\vdash p=p'$, so, we only need to consider closed basic $\textrm{BATC}^{\textrm{sat}}$ terms.

        The basic terms modulo associativity and commutativity (AC) of conflict $+$ (defined by axioms $A1$ and $A2$ in Table \ref{AxiomsForBATCSAT}), and this equivalence is denoted by $=_{AC}$. Then, each equivalence class $s$ modulo AC of $+$ has the following normal form

        $$s_1+\cdots+ s_k$$

        with each $s_i$ either an atomic event or of the form $t_1\cdot t_2$, and each $s_i$ is called the summand of $s$.

        Now, we prove that for normal forms $n$ and $n'$, if $n\sim_{s} n'$ then $n=_{AC}n'$. It is sufficient to induct on the sizes of $n$ and $n'$. We can get $n=_{AC} n'$.

        Finally, let $s$ and $t$ be basic terms, and $s\sim_s t$, there are normal forms $n$ and $n'$, such that $s=n$ and $t=n'$. The soundness theorem of $\textrm{BATC}^{\textrm{sat}}$ modulo step bisimulation equivalence yields $s\sim_s n$ and $t\sim_s n'$, so $n\sim_s s\sim_s t\sim_s n'$. Since if $n\sim_s n'$ then $n=_{AC}n'$, $s=n=_{AC}n'=t$, as desired.
  \item This case can be proven similarly, just by replacement of $\sim_{s}$ by $\sim_{p}$.
  \item This case can be proven similarly, just by replacement of $\sim_{s}$ by $\sim_{hp}$.
  \item This case can be proven similarly, just by replacement of $\sim_{s}$ by $\sim_{hhp}$.
\end{enumerate}
\end{proof}

\subsection{$\textrm{BATC}^{\textrm{sat}}$ with Integration}

In this subsection, we will introduce the theory $\textrm{BATC}^{\textrm{sat}}$ with integration called $\textrm{BATC}^{\textrm{sat}}\textrm{I}$.

\subsubsection{The Theory $\textrm{BATC}^{\textrm{sat}}\textrm{I}$}

\begin{definition}[Signature of $\textrm{BATC}^{\textrm{sat}}\textrm{I}$]
The signature of $\textrm{BATC}^{\textrm{sat}}\textrm{I}$ consists of the signature of $\textrm{BATC}^{\textrm{sat}}$ and the integration operator $\int: \mathcal{P}(\mathbb{R}^{\geq})\times\mathbb{R}^{\geq}.\mathcal{P}_{\textrm{abs}} \rightarrow\mathcal{P}_{\textrm{abs}}$.
\end{definition}

The set of axioms of $\textrm{BATC}^{\textrm{sat}}\textrm{I}$ consists of the laws given in Table \ref{AxiomsForBATCSATI}.

\begin{center}
    \begin{table}
        \begin{tabular}{@{}ll@{}}
            \hline No. &Axiom\\
            $INT1$ & $\int_{v\in V}F(v) = \int_{w\in V}F(w)$\\
            $INT2$ & $\int_{v\in\emptyset}F(v) = \dot{\delta}$\\
            $INT3$ & $\int_{v\in\{p\}}F(v) = F(p)$\\
            $INT4$ & $\int_{v\in V\cup W}F(v) = \int_{v\in V}F(v) + \int_{v\in W}F(v)$\\
            $INT5$ & $V\neq\emptyset\Rightarrow\int_{v\in V}x = x$\\
            $INT6$ & $(\forall v\in V.F(v)=G(v))\Rightarrow \int_{v\in V}F(v) = \int_{v\in V}G(v)$\\
            $INT7SAa$ & $\textrm{sup }V = p\Rightarrow\int_{v\in V}\sigma^v_{\textrm{abs}}(\dot{\delta}) = \sigma^p_{\textrm{abs}}(\dot{\delta})$\\
            $INT7SAb$ & $V,W\textrm{ unbounded }\Rightarrow \int_{v\in V}\sigma^v_{\textrm{abs}}(\dot{\delta}) = \int_{v\in W}\sigma^v_{\textrm{abs}}(\dot{\delta})$\\
            $INT8SAa$ & $\textrm{sup }V = p,p\notin V\Rightarrow\int_{v\in V}\sigma^v_{\textrm{abs}}(\tilde{\delta}) = \sigma^p_{\textrm{abs}}(\dot{\delta})$\\
            $INT8SAb$ & $V,W\textrm{ unbounded }\Rightarrow \int_{v\in V}\sigma^v_{\textrm{abs}}(\tilde{\delta}) = \int_{v\in W}\sigma^v_{\textrm{abs}}(\dot{\delta})$\\
            $INT9SA$ & $\textrm{sup }V = p,p\in V\Rightarrow\int_{v\in V}\sigma^v_{\textrm{abs}}(\tilde{\delta}) = \sigma^p_{\textrm{abs}}(\tilde{\delta})$\\
            $INT10SA$ & $\int_{v\in V}\sigma^p_{\textrm{abs}}(F(v)) = \sigma^p_{\textrm{abs}}(\int_{v\in V}F(v))$\\
            $INT11$ & $\int_{v\in V}(F(v)+G(v)) = \int_{v\in V}F(v) + \int_{v\in V}G(v)$\\
            $INT12$ & $\int_{v\in V}(F(v)\cdot x) = (\int_{v\in V}F(v))\cdot x$\\
            $SATO6$ & $\upsilon^p_{\textrm{abs}}(\int_{v\in V}F(v)) = \int_{v\in V}\upsilon^p_{\textrm{abs}}(F(v))$\\
            $SAI6$ & $\overline{\upsilon}^p_{\textrm{abs}}(\int_{v\in V}F(v)) = \int_{v\in V}\overline{\upsilon}^p_{\textrm{abs}}(F(v))$\\
        \end{tabular}
        \caption{Axioms of $\textrm{BATC}^{\textrm{SAT}}\textrm{I}(p\geq 0)$}
        \label{AxiomsForBATCSATI}
    \end{table}
\end{center}

The operational semantics of $\textrm{BATC}^{\textrm{sat}}\textrm{I}$ are defined by the transition rules in Table \ref{TRForBATCSATI}.

\begin{center}
    \begin{table}
        $$\frac{\langle F(q),p\rangle\xrightarrow{a}\langle x',p\rangle}{\langle\int_{v\in V}F(v),p\rangle\xrightarrow{a}\langle x',p\rangle}(q\in V)
        \quad\frac{\langle F(q),p\rangle\xrightarrow{a}\langle\surd,p\rangle}{\langle\int_{v\in V}F(v),p\rangle\xrightarrow{a}\langle\surd,p\rangle}(q\in V)$$

        $$\frac{\langle F(q),p\rangle\mapsto^r \langle F(q),p+r\rangle}{\langle\int_{v\in V}F(v),p\rangle\mapsto^r\langle \int_{v\in V}F(v),p+r\rangle}(q\in V)$$

        $$\frac{\{\langle F(q),p\rangle\uparrow|q\in V\}}{\langle\int_{v\in V}F(v),p\rangle\uparrow}$$
        \caption{Transition rules of $\textrm{BATC}^{\textrm{sat}}\textrm{I}(a\in A, p,q\geq 0, r>0)$}
        \label{TRForBATCSATI}
    \end{table}
\end{center}

\subsubsection{Elimination}

\begin{definition}[Basic terms of $\textrm{BATC}^{\textrm{sat}}\textrm{I}$]
The set of basic terms of $\textrm{BATC}^{\textrm{sat}}\textrm{I}$, $\mathcal{B}(\textrm{BATC}^{\textrm{sat}})$, is inductively defined as follows by two auxiliary sets $\mathcal{B}_0(\textrm{BATC}^{\textrm{sat}}\textrm{I})$ and $\mathcal{B}_1(\textrm{BATC}^{\textrm{sat}}\textrm{I})$:
\begin{enumerate}
  \item if $a\in A_{\delta}$, then $\tilde{a} \in \mathcal{B}_1(\textrm{BATC}^{\textrm{sat}}\textrm{I})$;
  \item if $a\in A$ and $t\in \mathcal{B}(\textrm{BATC}^{\textrm{sat}}\textrm{I})$, then $\tilde{a}\cdot t \in \mathcal{B}_1(\textrm{BATC}^{\textrm{sat}}\textrm{I})$;
  \item if $t,t'\in \mathcal{B}_1(\textrm{BATC}^{\textrm{sat}}\textrm{I})$, then $t+t'\in \mathcal{B}_1(\textrm{BATC}^{\textrm{sat}}\textrm{I})$;
  \item if $t\in \mathcal{B}_1(\textrm{BATC}^{\textrm{sat}}\textrm{I})$, then $t\in \mathcal{B}_0(\textrm{BATC}^{\textrm{sat}}\textrm{I})$;
  \item if $p>0$ and $t\in \mathcal{B}_0(\textrm{BATC}^{\textrm{sat}}\textrm{I})$, then $\sigma^p_{\textrm{abs}}(t) \in \mathcal{B}_0(\textrm{BATC}^{\textrm{sat}}\textrm{I})$;
  \item if $p>0$, $t\in \mathcal{B}_1(\textrm{BATC}^{\textrm{sat}}\textrm{I})$ and $t'\in \mathcal{B}_0(\textrm{BATC}^{\textrm{sat}}\textrm{I})$, then $t+\sigma^p_{\textrm{abs}}(t') \in \mathcal{B}_0(\textrm{BATC}^{\textrm{sat}}\textrm{I})$;
  \item if $t\in \mathcal{B}_0(\textrm{BATC}^{\textrm{sat}}\textrm{I})$, then $\int_{v\in V}(t) \in \mathcal{B}_0(\textrm{BATC}^{\textrm{sat}}\textrm{I})$;
  \item $\dot{\delta}\in \mathcal{B}(\textrm{BATC}^{\textrm{sat}}\textrm{I})$;
  \item if $t\in \mathcal{B}_0(\textrm{BATC}^{\textrm{sat}}\textrm{I})$, then $t\in \mathcal{B}(\textrm{BATC}^{\textrm{sat}}\textrm{I})$.
\end{enumerate}
\end{definition}

\begin{theorem}[Elimination theorem]
Let $p$ be a closed $\textrm{BATC}^{\textrm{sat}}\textrm{I}$ term. Then there is a basic $\textrm{BATC}^{\textrm{sat}}\textrm{I}$ term $q$ such that $\textrm{BATC}^{\textrm{sat}}\vdash p=q$.
\end{theorem}

\begin{proof}
It is sufficient to induct on the structure of the closed $\textrm{BATC}^{\textrm{sat}}\textrm{I}$ term $p$. It can be proven that $p$ combined by the constants and operators of $\textrm{BATC}^{\textrm{sat}}\textrm{I}$ exists an equal basic term $q$, and the other operators not included in the basic terms, such as $\upsilon_{\textrm{abs}}$ and $\overline{\upsilon}_{\textrm{abs}}$ can be eliminated.
\end{proof}

\subsubsection{Connections}

\begin{theorem}[Generalization of $\textrm{BATC}^{\textrm{sat}}\textrm{I}$]
\begin{enumerate}
  \item By the definitions of $a=\int_{v\in[0,\infty)}\sigma^v_{\textrm{abs}}(\tilde{a})$ for each $a\in A$ and $\delta=\int_{v\in[0,\infty)}\sigma^v_{\textrm{abs}}(\tilde{\delta})$, $\textrm{BATC}^{\textrm{sat}}\textrm{I}$ is a generalization of $BATC$.
  \item $\textrm{BATC}^{\textrm{sat}}\textrm{I}$ is a generalization of $BATC^{\textrm{sat}}$.
\end{enumerate}

\end{theorem}

\begin{proof}
\begin{enumerate}
  \item It follows from the following two facts.

    \begin{enumerate}
      \item The transition rules of $BATC$ in section \ref{tcpa} are all source-dependent;
      \item The sources of the transition rules of $\textrm{BATC}^{\textrm{sat}}\textrm{I}$ contain an occurrence of $\dot{\delta}$, $\tilde{a}$, $\sigma^p_{\textrm{abs}}$, $\upsilon^p_{\textrm{abs}}$, $\overline{\upsilon}^p_{\textrm{abs}}$ and $\int$.
    \end{enumerate}

    So, $BATC$ is an embedding of $\textrm{BATC}^{\textrm{sat}}\textrm{I}$, as desired.
  \item It follows from the following two facts.

    \begin{enumerate}
      \item The transition rules of $BATC^{\textrm{sat}}$ are all source-dependent;
      \item The sources of the transition rules of $\textrm{BATC}^{\textrm{sat}}\textrm{I}$ contain an occurrence of $\int$.
    \end{enumerate}

    So, $BATC^{\textrm{sat}}$ is an embedding of $\textrm{BATC}^{\textrm{sat}}\textrm{I}$, as desired.
\end{enumerate}
\end{proof}

\subsubsection{Congruence}

\begin{theorem}[Congruence of $\textrm{BATC}^{\textrm{sat}}\textrm{I}$]
Truly concurrent bisimulation equivalences are all congruences with respect to $\textrm{BATC}^{\textrm{sat}}\textrm{I}$. That is,
\begin{itemize}
  \item pomset bisimulation equivalence $\sim_{p}$ is a congruence with respect to $\textrm{BATC}^{\textrm{sat}}\textrm{I}$;
  \item step bisimulation equivalence $\sim_{s}$ is a congruence with respect to $\textrm{BATC}^{\textrm{sat}}\textrm{I}$;
  \item hp-bisimulation equivalence $\sim_{hp}$ is a congruence with respect to $\textrm{BATC}^{\textrm{sat}}\textrm{I}$;
  \item hhp-bisimulation equivalence $\sim_{hhp}$ is a congruence with respect to $\textrm{BATC}^{\textrm{sat}}\textrm{I}$.
\end{itemize}
\end{theorem}

\begin{proof}
It is easy to see that $\sim_p$, $\sim_s$, $\sim_{hp}$ and $\sim_{hhp}$ are all equivalent relations on $\textrm{BATC}^{\textrm{sat}}\textrm{I}$ terms, it is only sufficient to prove that $\sim_p$, $\sim_s$, $\sim_{hp}$ and $\sim_{hhp}$ are all preserved by the operators $\int$. It is trivial and we omit it.
\end{proof}

\subsubsection{Soundness}

\begin{theorem}[Soundness of $\textrm{BATC}^{\textrm{sat}}\textrm{I}$]
The axiomatization of $\textrm{BATC}^{\textrm{sat}}\textrm{I}$ is sound modulo truly concurrent bisimulation equivalences $\sim_{p}$, $\sim_{s}$, $\sim_{hp}$ and $\sim_{hhp}$. That is,
\begin{enumerate}
  \item let $x$ and $y$ be $\textrm{BATC}^{\textrm{sat}}\textrm{I}$ terms. If $\textrm{BATC}^{\textrm{sat}}\textrm{I}\vdash x=y$, then $x\sim_{s} y$;
  \item let $x$ and $y$ be $\textrm{BATC}^{\textrm{sat}}\textrm{I}$ terms. If $\textrm{BATC}^{\textrm{sat}}\textrm{I}\vdash x=y$, then $x\sim_{p} y$;
  \item let $x$ and $y$ be $\textrm{BATC}^{\textrm{sat}}\textrm{I}$ terms. If $\textrm{BATC}^{\textrm{sat}}\textrm{I}\vdash x=y$, then $x\sim_{hp} y$;
  \item let $x$ and $y$ be $\textrm{BATC}^{\textrm{sat}}\textrm{I}$ terms. If $\textrm{BATC}^{\textrm{sat}}\textrm{I}\vdash x=y$, then $x\sim_{hhp} y$.
\end{enumerate}
\end{theorem}

\begin{proof}
Since $\sim_p$, $\sim_s$, $\sim_{hp}$ and $\sim_{hhp}$ are both equivalent and congruent relations, we only need to check if each axiom in Table \ref{AxiomsForBATCSATI} is sound modulo $\sim_p$, $\sim_s$, $\sim_{hp}$ and $\sim_{hhp}$ respectively.

\begin{enumerate}
  \item We can check the soundness of each axiom in Table \ref{AxiomsForBATCSATI}, by the transition rules in Table \ref{TRForBATCSATI}, it is trivial and we omit them.
  \item From the definition of pomset bisimulation, we know that pomset bisimulation is defined by pomset transitions, which are labeled by pomsets. In a pomset transition, the events (actions) in the pomset are either within causality relations (defined by $\cdot$) or in concurrency (implicitly defined by $\cdot$ and $+$, and explicitly defined by $\between$), of course, they are pairwise consistent (without conflicts). We have already proven the case that all events are pairwise concurrent (soundness modulo step bisimulation), so, we only need to prove the case of events in causality. Without loss of generality, we take a pomset of $P=\{\tilde{a},\tilde{b}:\tilde{a}\cdot \tilde{b}\}$. Then the pomset transition labeled by the above $P$ is just composed of one single event transition labeled by $\tilde{a}$ succeeded by another single event transition labeled by $\tilde{b}$, that is, $\xrightarrow{P}=\xrightarrow{a}\xrightarrow{b}$.

        Similarly to the proof of soundness modulo step bisimulation equivalence, we can prove that each axiom in Table \ref{AxiomsForBATCSATI} is sound modulo pomset bisimulation equivalence, we omit them.
  \item From the definition of hp-bisimulation, we know that hp-bisimulation is defined on the posetal product $(C_1,f,C_2),f:C_1\rightarrow C_2\textrm{ isomorphism}$. Two process terms $s$ related to $C_1$ and $t$ related to $C_2$, and $f:C_1\rightarrow C_2\textrm{ isomorphism}$. Initially, $(C_1,f,C_2)=(\emptyset,\emptyset,\emptyset)$, and $(\emptyset,\emptyset,\emptyset)\in\sim_{hp}$. When $s\xrightarrow{a}s'$ ($C_1\xrightarrow{a}C_1'$), there will be $t\xrightarrow{a}t'$ ($C_2\xrightarrow{a}C_2'$), and we define $f'=f[a\mapsto a]$. Then, if $(C_1,f,C_2)\in\sim_{hp}$, then $(C_1',f',C_2')\in\sim_{hp}$.

        Similarly to the proof of soundness modulo pomset bisimulation equivalence, we can prove that each axiom in Table \ref{AxiomsForBATCSATI} is sound modulo hp-bisimulation equivalence, we just need additionally to check the above conditions on hp-bisimulation, we omit them.
  \item We just need to add downward-closed condition to the soundness modulo hp-bisimulation equivalence, we omit them.
\end{enumerate}

\end{proof}

\subsubsection{Completeness}

\begin{theorem}[Completeness of $\textrm{BATC}^{\textrm{sat}}\textrm{I}$]
The axiomatization of $\textrm{BATC}^{\textrm{sat}}\textrm{I}$ is complete modulo truly concurrent bisimulation equivalences $\sim_{p}$, $\sim_{s}$, $\sim_{hp}$ and $\sim_{hhp}$. That is,
\begin{enumerate}
  \item let $p$ and $q$ be closed $\textrm{BATC}^{\textrm{sat}}\textrm{I}$ terms, if $p\sim_{s} q$ then $p=q$;
  \item let $p$ and $q$ be closed $\textrm{BATC}^{\textrm{sat}}\textrm{I}$ terms, if $p\sim_{p} q$ then $p=q$;
  \item let $p$ and $q$ be closed $\textrm{BATC}^{\textrm{sat}}\textrm{I}$ terms, if $p\sim_{hp} q$ then $p=q$;
  \item let $p$ and $q$ be closed $\textrm{BATC}^{\textrm{sat}}\textrm{I}$ terms, if $p\sim_{hhp} q$ then $p=q$.
\end{enumerate}

\end{theorem}

\begin{proof}
\begin{enumerate}
  \item Firstly, by the elimination theorem of $\textrm{BATC}^{\textrm{sat}}\textrm{I}$, we know that for each closed $\textrm{BATC}^{\textrm{sat}}\textrm{I}$ term $p$, there exists a closed basic $\textrm{BATC}^{\textrm{sat}}\textrm{I}$ term $p'$, such that $\textrm{BATC}^{\textrm{sat}}\textrm{I}\vdash p=p'$, so, we only need to consider closed basic $\textrm{BATC}^{\textrm{sat}}\textrm{I}$ terms.

        The basic terms modulo associativity and commutativity (AC) of conflict $+$ (defined by axioms $A1$ and $A2$ in Table \ref{AxiomsForBATCSAT}), and this equivalence is denoted by $=_{AC}$. Then, each equivalence class $s$ modulo AC of $+$ has the following normal form

        $$s_1+\cdots+ s_k$$

        with each $s_i$ either an atomic event or of the form $t_1\cdot t_2$, and each $s_i$ is called the summand of $s$.

        Now, we prove that for normal forms $n$ and $n'$, if $n\sim_{s} n'$ then $n=_{AC}n'$. It is sufficient to induct on the sizes of $n$ and $n'$. We can get $n=_{AC} n'$.

        Finally, let $s$ and $t$ be basic terms, and $s\sim_s t$, there are normal forms $n$ and $n'$, such that $s=n$ and $t=n'$. The soundness theorem of $\textrm{BATC}^{\textrm{sat}}\textrm{I}$ modulo step bisimulation equivalence yields $s\sim_s n$ and $t\sim_s n'$, so $n\sim_s s\sim_s t\sim_s n'$. Since if $n\sim_s n'$ then $n=_{AC}n'$, $s=n=_{AC}n'=t$, as desired.
  \item This case can be proven similarly, just by replacement of $\sim_{s}$ by $\sim_{p}$.
  \item This case can be proven similarly, just by replacement of $\sim_{s}$ by $\sim_{hp}$.
  \item This case can be proven similarly, just by replacement of $\sim_{s}$ by $\sim_{hhp}$.
\end{enumerate}
\end{proof}

\subsection{Algebra for Parallelism in True Concurrency with Continuous Absolute Timing}

In this subsection, we will introduce $\textrm{APTC}^{\textrm{sat}}$.

\subsubsection{Basic Definition}

\begin{definition}[Absolute undelayable time-out]
The relative undelayable time-out $\nu_{\textrm{abs}}$ of a process $p$ behaves like the part of $p$ that starts to perform actions at the point of time 0 if $p$ is capable of performing actions at point of time 0; otherwise, like undelayable deadlock. And let $\nu^r_{\textrm{abs}}(t) = r \nu_{\textrm{abs}} t$.
\end{definition}

\subsubsection{The Theory $\textrm{APTC}^{\textrm{sat}}$}

\begin{definition}[Signature of $\textrm{APTC}^{\textrm{sat}}$]
The signature of $\textrm{APTC}^{\textrm{sat}}$ consists of the signature of $\textrm{BATC}^{\textrm{sat}}$, and the whole parallel composition operator $\between: \mathcal{P}_{\textrm{abs}}\times\mathcal{P}_{\textrm{abs}} \rightarrow \mathcal{P}_{\textrm{abs}}$, the parallel operator $\parallel: \mathcal{P}_{\textrm{abs}}\times\mathcal{P}_{\textrm{abs}} \rightarrow \mathcal{P}_{\textrm{abs}}$, the communication merger operator $\mid: \mathcal{P}_{\textrm{abs}}\times\mathcal{P}_{\textrm{abs}} \rightarrow \mathcal{P}_{\textrm{abs}}$, the encapsulation operator $\partial_H: \mathcal{P}_{\textrm{abs}} \rightarrow \mathcal{P}_{\textrm{abs}}$ for all $H\subseteq A$, and the absolute undelayable time-out operator $\nu_{\textrm{abs}}: \mathcal{P}_{\textrm{abs}}\rightarrow\mathcal{P}_{\textrm{abs}}$.
\end{definition}

The set of axioms of $\textrm{APTC}^{\textrm{sat}}$ consists of the laws given in Table \ref{AxiomsForAPTCSAT}.

\begin{center}
    \begin{table}
        \begin{tabular}{@{}ll@{}}
            \hline No. &Axiom\\
            $P1$ & $x\between y = x\parallel y + x\mid y$\\
            $P2$ & $x\parallel y = y \parallel x$\\
            $P3$ & $(x\parallel y)\parallel z = x\parallel (y\parallel z)$\\
            $P4SA$ & $\tilde{a}\parallel (\tilde{b}\cdot y) = (\tilde{a}\parallel \tilde{b})\cdot y$\\
            $P5SA$ & $(\tilde{a}\cdot x)\parallel \tilde{b} = (\tilde{a}\parallel \tilde{b})\cdot x$\\
            $P6SA$ & $(\tilde{a}\cdot x)\parallel (\tilde{b}\cdot y) = (\tilde{a}\parallel \tilde{b})\cdot (x\between y)$\\
            $P7$ & $(x+ y)\parallel z = (x\parallel z)+ (y\parallel z)$\\
            $P8$ & $x\parallel (y+ z) = (x\parallel y)+ (x\parallel z)$\\
            $SAP9ID$ & $(\nu_{\textrm{abs}}(x)+ \tilde{\delta})\parallel \sigma^{r}_{\textrm{abs}}(y) = \tilde{\delta}$\\
            $SAP10ID$ & $\sigma^{r}_{\textrm{abs}}(x)\parallel (\nu_{\textrm{abs}}(y)+ \tilde{\delta}) = \tilde{\delta}$\\
            $SAP11$ & $\sigma^p_{\textrm{abs}}(x) \parallel \sigma^p_{\textrm{abs}}(y) = \sigma^p_{\textrm{abs}}(x\parallel y)$\\
            $PID12$ & $\dot{\delta}\parallel x = \dot{\delta}$\\
            $PID13$ & $x\parallel \dot{\delta} = \dot{\delta}$\\
            $C14SA$ & $\tilde{a}\mid \tilde{b} = \gamma(\tilde{a}, \tilde{b})$\\
            $C15SA$ & $\tilde{a}\mid (\tilde{b}\cdot y) = \gamma(\tilde{a}, \tilde{b})\cdot y$\\
            $C16SA$ & $(\tilde{a}\cdot x)\mid \tilde{b} = \gamma(\tilde{a}, \tilde{b})\cdot x$\\
            $C17SA$ & $(\tilde{a}\cdot x)\mid (\tilde{b}\cdot y) = \gamma(\tilde{a}, \tilde{b})\cdot (x\between y)$\\
            $C18$ & $(x+ y)\mid z = (x\mid z) + (y\mid z)$\\
            $C19$ & $x\mid (y+ z) = (x\mid y)+ (x\mid z)$\\
            $SAC20ID$ & $(\nu_{\textrm{abs}}(x)+ \tilde{\delta})\mid \sigma^{r}_{\textrm{abs}}(y) = \tilde{\delta}$\\
            $SAC21ID$ & $\sigma^{r}_{\textrm{abs}}(x)\mid (\nu_{\textrm{abs}}(y)+ \tilde{\delta}) = \tilde{\delta}$\\
            $SAC22$ & $\sigma^p_{\textrm{abs}}(x) \mid \sigma^p_{\textrm{abs}}(y) = \sigma^p_{\textrm{abs}}(x\mid y)$\\
            $CID23$ & $\dot{\delta}\mid x = \dot{\delta}$\\
            $CID24$ & $x\mid\dot{\delta} = \dot{\delta}$\\
            $CE25SA$ & $\Theta(\tilde{a}) = \tilde{a}$\\
            $CE26SAID$ & $\Theta(\dot{\delta}) = \dot{\delta}$\\
            $CE27$ & $\Theta(x+ y) = \Theta(x)\triangleleft y + \Theta(y)\triangleleft x$\\
            $CE28$ & $\Theta(x\cdot y)=\Theta(x)\cdot\Theta(y)$\\
            $CE29$ & $\Theta(x\parallel y) = ((\Theta(x)\triangleleft y)\parallel y)+ ((\Theta(y)\triangleleft x)\parallel x)$\\
            $CE30$ & $\Theta(x\mid y) = ((\Theta(x)\triangleleft y)\mid y)+ ((\Theta(y)\triangleleft x)\mid x)$\\
            $U31SAID$ & $(\sharp(\tilde{a},\tilde{b}))\quad \tilde{a}\triangleleft \tilde{b} = \tilde{\tau}$\\
            $U32SAID$ & $(\sharp(\tilde{a},\tilde{b}),\tilde{b}\leq \tilde{c})\quad \tilde{a}\triangleleft \tilde{c} = \tilde{a}$\\
            $U33SAID$ & $(\sharp(\tilde{a},\tilde{b}),\tilde{b}\leq \tilde{c})\quad \tilde{c}\triangleleft \tilde{a} = \tilde{\tau}$\\
            $U34SAID$ & $\tilde{a}\triangleleft \tilde{\delta} = \tilde{a}$\\
            $U35SAID$ & $\tilde{\delta} \triangleleft \tilde{a} = \tilde{\delta}$\\
            $U36$ & $(x+ y)\triangleleft z = (x\triangleleft z)+ (y\triangleleft z)$\\
            $U37$ & $(x\cdot y)\triangleleft z = (x\triangleleft z)\cdot (y\triangleleft z)$\\
            $U38$ & $(x\parallel y)\triangleleft z = (x\triangleleft z)\parallel (y\triangleleft z)$\\
            $U39$ & $(x\mid y)\triangleleft z = (x\triangleleft z)\mid (y\triangleleft z)$\\
            $U40$ & $x\triangleleft (y+ z) = (x\triangleleft y)\triangleleft z$\\
            $U41$ & $x\triangleleft (y\cdot z)=(x\triangleleft y)\triangleleft z$\\
            $U42$ & $x\triangleleft (y\parallel z) = (x\triangleleft y)\triangleleft z$\\
            $U43$ & $x\triangleleft (y\mid z) = (x\triangleleft y)\triangleleft z$\\
            $D1SAID$ & $\tilde{a}\notin H\quad\partial_H(\tilde{a}) = \tilde{a}$\\
            $D2SAID$ & $\tilde{a}\in H\quad \partial_H(\tilde{a}) = \tilde{\delta}$\\
            $D3SAID$ & $\partial_H(\dot{\delta}) = \dot{\delta}$\\
            $4$ & $\partial_H(x+ y) = \partial_H(x)+\partial_H(y)$\\
            $D5$ & $\partial_H(x\cdot y) = \partial_H(x)\cdot\partial_H(y)$\\
            $D6$ & $\partial_H(x\parallel y) = \partial_H(x)\parallel\partial_H(y)$\\
            $SAD7$ & $\partial_H(\sigma^p_{\textrm{abs}}(x)) = \sigma^p_{\textrm{abs}}(\partial_H(x))$\\
            $SAU0$ & $\nu_{\textrm{abs}}(\dot{\delta}) = \dot{\delta}$\\
            $SAU1$ & $\nu_{\textrm{abs}}(\tilde{a}) = \tilde{a}$\\
            $SAU2$ & $\nu_{\textrm{abs}}(\sigma^r_{\textrm{abs}}(x)) = \tilde{\delta}$\\
            $SAU3$ & $\nu_{\textrm{abs}}(x+y) = \nu_{\textrm{abs}}(x)+\nu_{\textrm{abs}}(y)$\\
            $SAU4$ & $\nu_{\textrm{abs}}(x\cdot y) = \nu_{\textrm{abs}}(x)\cdot y$\\
            $SAU5$ & $\nu_{\textrm{abs}}(x\parallel y) = \nu_{\textrm{abs}}(x)\parallel \nu_{\textrm{abs}}(y)$\\
        \end{tabular}
        \caption{Axioms of $\textrm{APTC}^{\textrm{sat}}(a,b,c\in A_{\delta}, p\geq 0, r>0)$}
        \label{AxiomsForAPTCSAT}
    \end{table}
\end{center}

The operational semantics of $\textrm{APTC}^{\textrm{sat}}$ are defined by the transition rules in Table \ref{TRForAPTCSAT}.

\begin{center}
    \begin{table}
        $$\frac{\langle x,p\rangle\xrightarrow{a}\langle\surd,p\rangle\quad \langle y,p\rangle\xrightarrow{b}\langle\surd,p\rangle}{\langle x\parallel y,p\rangle\xrightarrow{\{a,b\}}\langle\surd,p\rangle} \quad\frac{\langle x,p\rangle\xrightarrow{a}\langle x',p\rangle\quad \langle y,p\rangle\xrightarrow{b}\langle\surd,p\rangle}{\langle x\parallel y,p\rangle\xrightarrow{\{a,b\}}\langle x',p\rangle}$$

        $$\frac{\langle x,p\rangle\xrightarrow{a}\langle\surd,p\rangle\quad \langle y,p\rangle\xrightarrow{b}\langle y',p\rangle}{\langle x\parallel y,p\rangle\xrightarrow{\{a,b\}}\langle y',p\rangle} \quad\frac{\langle x,p\rangle\xrightarrow{a}\langle x',p\rangle\quad \langle y,p\rangle\xrightarrow{b}\langle y',p\rangle}{\langle x\parallel y,p\rangle\xrightarrow{\{a,b\}}\langle x'\between y',p\rangle}$$

        $$\frac{\langle x,p\rangle\mapsto^{r}\langle x,p+r\rangle\quad \langle y,p\rangle\mapsto^{r}\langle y,p+r\rangle}{\langle x\parallel y,p\rangle\mapsto^{r}\langle x\parallel y,p+r\rangle} \quad\frac{\langle x,p\rangle\uparrow}{\langle x\parallel y,p\rangle\uparrow} \quad\frac{\langle y,p\rangle\uparrow}{\langle x\parallel y,p\rangle\uparrow}$$

        $$\frac{\langle x,p\rangle\xrightarrow{a}\langle \surd,p\rangle\quad \langle y,p\rangle\xrightarrow{b}\langle\surd,p\rangle}{\langle x\mid y,p\rangle\xrightarrow{\gamma(a,b)}\langle\surd,p\rangle} \quad\frac{\langle x,p\rangle\xrightarrow{a}\langle x',p\rangle\quad \langle y,p\rangle\xrightarrow{b}\langle \surd,p\rangle}{\langle x\mid y,p\rangle\xrightarrow{\gamma(a,b)}\langle x',p\rangle}$$

        $$\frac{\langle x,p\rangle\xrightarrow{a}\langle\surd,p\rangle\quad \langle y,p\rangle\xrightarrow{b}\langle y',p\rangle}{\langle x\mid y,p\rangle\xrightarrow{\gamma(a,b)}\langle y',p\rangle} \quad\frac{\langle x,p\rangle\xrightarrow{a}\langle x',p\rangle\quad \langle y,p\rangle\xrightarrow{b}\langle y',p\rangle}{\langle x\mid y,p\rangle\xrightarrow{\gamma(a,b)}\langle x'\between y',p\rangle}$$

        $$\frac{\langle x,p\rangle\mapsto^{r}\langle x,p+r\rangle\quad \langle y,p\rangle\mapsto^{r}\langle y,p+r\rangle}{\langle x\mid y,p\rangle\mapsto^{r}\langle x\mid y,p+r\rangle} \quad\frac{\langle x,p\rangle\uparrow}{\langle x\mid y,p\rangle\uparrow} \quad\frac{\langle y,p\rangle\uparrow}{\langle x\mid y,p\rangle\uparrow}$$

        $$\frac{\langle x,p\rangle\xrightarrow{a}\langle\surd,p\rangle\quad (\sharp(a,b))}{\langle\Theta(x),p\rangle\xrightarrow{a}\langle\surd,p\rangle} \quad\frac{\langle x,p\rangle\xrightarrow{b}\langle\surd,p\rangle\quad (\sharp(a,b))}{\langle\Theta(x),p\rangle\xrightarrow{b}\langle\surd,p\rangle}$$

        $$\frac{\langle x,p\rangle\xrightarrow{a}\langle x',p\rangle\quad (\sharp(a,b))}{\langle\Theta(x),p\rangle\xrightarrow{a}\langle\Theta(x'),p\rangle} \quad\frac{\langle x,p\rangle\xrightarrow{b}\langle x',p\rangle\quad (\sharp(a,b))}{\langle\Theta(x),p\rangle\xrightarrow{b}\langle\Theta(x'),p\rangle}$$

        $$\frac{\langle x,p\rangle\mapsto^{r}\langle x,p+r\rangle}{\langle\Theta(x),p\rangle\mapsto^{r}\langle\Theta(x),p+r\rangle} \quad\frac{\langle x,p\rangle\uparrow}{\langle\Theta(x),p\rangle\uparrow}$$

        $$\frac{\langle x,p\rangle\xrightarrow{a}\langle\surd,p\rangle \quad \langle y,p\rangle\nrightarrow^{b}\quad (\sharp(a,b))}{\langle x\triangleleft y,p\rangle\xrightarrow{\tau}\langle\surd,p\rangle}
        \quad\frac{\langle x,p\rangle\xrightarrow{a}\langle x',p\rangle \quad \langle y,p\rangle\nrightarrow^{b}\quad (\sharp(a,b))}{\langle x\triangleleft y,p\rangle\xrightarrow{\tau}\langle x',p\rangle}$$

        $$\frac{\langle x,p\rangle\xrightarrow{a}\langle\surd,p\rangle \quad \langle y,p\rangle\nrightarrow^{c}\quad (\sharp(a,b),b\leq c)}{\langle x\triangleleft y,p\rangle\xrightarrow{a}\langle \surd,p\rangle}
        \quad\frac{x\xrightarrow{a}x' \quad y\nrightarrow^{c}\quad (\sharp(a,b),b\leq c)}{x\triangleleft y\xrightarrow{a}x'}$$

        $$\frac{\langle x,p\rangle\xrightarrow{c}\langle\surd,p\rangle \quad \langle y,p\rangle\nrightarrow^{b}\quad (\sharp(a,b),a\leq c)}{\langle x\triangleleft y,p\rangle\xrightarrow{\tau}\langle\surd,p\rangle}
        \quad\frac{\langle x,p\rangle\xrightarrow{c}\langle x',p\rangle \quad \langle y,p\rangle\nrightarrow^{b}\quad (\sharp(a,b),a\leq c)}{\langle x\triangleleft y,p\rangle\xrightarrow{\tau}\langle x',p\rangle}$$

        $$\frac{\langle x,p\rangle\mapsto^{r}\langle x,p+r\rangle\quad \langle y,p\rangle\mapsto^{r}\langle y,p+r\rangle}{\langle x\triangleleft y,p\rangle\mapsto^{r}\langle x\triangleleft y,p+r\rangle} \quad\frac{\langle x,p\rangle\uparrow}{\langle x\triangleleft y,p\rangle\uparrow}$$

        $$\frac{\langle x,p\rangle\xrightarrow{a}\langle\surd,p\rangle}{\langle\partial_H(x),p\rangle\xrightarrow{a}\langle\surd,p\rangle}\quad (e\notin H)\quad\frac{\langle x,p\rangle\xrightarrow{a}\langle x',p\rangle}{\langle\partial_H(x),p\rangle\xrightarrow{a}\langle\partial_H(x'),p\rangle}\quad(e\notin H)$$

        $$\frac{\langle x,p\rangle\mapsto^{r}\langle x,p+r\rangle}{\langle\partial_H(x),p\rangle\mapsto^{r}\langle\partial_H(x'),p+r\rangle}\quad(e\notin H)\quad\frac{\langle x,p\rangle\uparrow}{\langle\partial_H(x),p\rangle\uparrow}$$

        $$\frac{\langle x,0\rangle\xrightarrow{a}\langle x',0\rangle}{\langle\nu_{\textrm{abs}}(x),0\rangle\xrightarrow{a}\langle x',0\rangle} \quad\frac{\langle x,0\rangle\xrightarrow{a}\langle\surd,0\rangle}{\langle\nu_{\textrm{abs}}(x),0\rangle\xrightarrow\langle\surd,0\rangle}
        \quad\frac{\langle x,0\rangle\uparrow}{\langle\nu_{\textrm{abs}}(x),0\rangle\uparrow}
        \quad\frac{}{\langle\nu_{\textrm{abs}}(x),r\rangle}\uparrow$$
    \caption{Transition rules of $\textrm{APTC}^{\textrm{sat}}(a,b,c\in a, p\geq 0, r>0)$}
    \label{TRForAPTCSAT}
    \end{table}
\end{center}

\subsubsection{Elimination}

\begin{definition}[Basic terms of $\textrm{APTC}^{\textrm{sat}}$]
The set of basic terms of $\textrm{APTC}^{\textrm{sat}}$, $\mathcal{B}(\textrm{APTC}^{\textrm{sat}})$, is inductively defined as follows by two auxiliary sets $\mathcal{B}_0(\textrm{APTC}^{\textrm{sat}})$ and $\mathcal{B}_1(\textrm{APTC}^{\textrm{sat}})$:
\begin{enumerate}
  \item if $a\in A_{\delta}$, then $\tilde{a} \in \mathcal{B}_1(\textrm{APTC}^{\textrm{sat}})$;
  \item if $a\in A$ and $t\in \mathcal{B}(\textrm{APTC}^{\textrm{sat}})$, then $\tilde{a}\cdot t \in \mathcal{B}_1(\textrm{APTC}^{\textrm{sat}})$;
  \item if $t,t'\in \mathcal{B}_1(\textrm{APTC}^{\textrm{sat}})$, then $t+t'\in \mathcal{B}_1(\textrm{APTC}^{\textrm{sat}})$;
  \item if $t,t'\in \mathcal{B}_1(\textrm{APTC}^{\textrm{sat}})$, then $t\parallel t'\in \mathcal{B}_1(\textrm{APTC}^{\textrm{sat}})$;
  \item if $t\in \mathcal{B}_1(\textrm{APTC}^{\textrm{sat}})$, then $t\in \mathcal{B}_0(\textrm{APTC}^{\textrm{sat}})$;
  \item if $p>0$ and $t\in \mathcal{B}_0(\textrm{APTC}^{\textrm{sat}})$, then $\sigma^p_{\textrm{abs}}(t) \in \mathcal{B}_0(\textrm{APTC}^{\textrm{sat}})$;
  \item if $p>0$, $t\in \mathcal{B}_1(\textrm{APTC}^{\textrm{sat}})$ and $t'\in \mathcal{B}_0(\textrm{APTC}^{\textrm{sat}})$, then $t+\sigma^p_{\textrm{abs}}(t') \in \mathcal{B}_0(\textrm{APTC}^{\textrm{sat}})$;
  \item if $t\in \mathcal{B}_0(\textrm{APTC}^{\textrm{sat}})$, then $\nu_{\textrm{abs}}(t) \in \mathcal{B}_0(\textrm{APTC}^{\textrm{sat}})$;
  \item $\dot{\delta}\in \mathcal{B}(\textrm{APTC}^{\textrm{sat}})$;
  \item if $t\in \mathcal{B}_0(\textrm{APTC}^{\textrm{sat}})$, then $t\in \mathcal{B}(\textrm{APTC}^{\textrm{sat}})$.
\end{enumerate}
\end{definition}

\begin{theorem}[Elimination theorem]
Let $p$ be a closed $\textrm{APTC}^{\textrm{sat}}$ term. Then there is a basic $\textrm{APTC}^{\textrm{sat}}$ term $q$ such that $\textrm{APTC}^{\textrm{sat}}\vdash p=q$.
\end{theorem}

\begin{proof}
It is sufficient to induct on the structure of the closed $\textrm{APTC}^{\textrm{sat}}$ term $p$. It can be proven that $p$ combined by the constants and operators of $\textrm{APTC}^{\textrm{sat}}$ exists an equal basic term $q$, and the other operators not included in the basic terms, such as $\upsilon_{\textrm{abs}}$, $\overline{\upsilon}_{\textrm{abs}}$, $\between$, $\mid$, $\partial_H$, $\Theta$ and $\triangleleft$ can be eliminated.
\end{proof}

\subsubsection{Connections}

\begin{theorem}[Generalization of $\textrm{APTC}^{\textrm{sat}}$]
\begin{enumerate}
  \item By the definitions of $a=\tilde{a}$ for each $a\in A$ and $\delta=\tilde{\delta}$, $\textrm{APTC}^{\textrm{sat}}$ is a generalization of $APTC$.
  \item $\textrm{APTC}^{\textrm{sat}}$ is a generalization of $\textrm{BATC}^{\textrm{sat}}$¡£
\end{enumerate}

\end{theorem}

\begin{proof}
\begin{enumerate}
  \item It follows from the following two facts.

    \begin{enumerate}
      \item The transition rules of $APTC$ in section \ref{tcpa} are all source-dependent;
      \item The sources of the transition rules of $\textrm{APTC}^{\textrm{sat}}$ contain an occurrence of $\dot{\delta}$, $\tilde{a}$, $\sigma^p_{\textrm{abs}}$, $\upsilon^p_{\textrm{abs}}$, $\overline{\upsilon}^p_{\textrm{abs}}$, and $\nu_{\textrm{abs}}$.
    \end{enumerate}

    So, $APTC$ is an embedding of $\textrm{APTC}^{\textrm{sat}}$, as desired.
    \item It follows from the following two facts.

    \begin{enumerate}
      \item The transition rules of $\textrm{BATC}^{\textrm{sat}}$ are all source-dependent;
      \item The sources of the transition rules of $\textrm{APTC}^{\textrm{sat}}$ contain an occurrence of $\between$, $\parallel$, $\mid$, $\Theta$, $\triangleleft$, $\partial_H$ and $\nu_{\textrm{abs}}$.
    \end{enumerate}

    So, $\textrm{BATC}^{\textrm{sat}}$ is an embedding of $\textrm{APTC}^{\textrm{sat}}$, as desired.
\end{enumerate}
\end{proof}

\subsubsection{Congruence}

\begin{theorem}[Congruence of $\textrm{APTC}^{\textrm{sat}}$]
Truly concurrent bisimulation equivalences $\sim_p$, $\sim_s$ and $\sim_{hp}$ are all congruences with respect to $\textrm{APTC}^{\textrm{sat}}$. That is,
\begin{itemize}
  \item pomset bisimulation equivalence $\sim_{p}$ is a congruence with respect to $\textrm{APTC}^{\textrm{sat}}$;
  \item step bisimulation equivalence $\sim_{s}$ is a congruence with respect to $\textrm{APTC}^{\textrm{sat}}$;
  \item hp-bisimulation equivalence $\sim_{hp}$ is a congruence with respect to $\textrm{APTC}^{\textrm{sat}}$.
\end{itemize}
\end{theorem}

\begin{proof}
It is easy to see that $\sim_p$, $\sim_s$, and $\sim_{hp}$ are all equivalent relations on $\textrm{APTC}^{\textrm{sat}}$ terms, it is only sufficient to prove that $\sim_p$, $\sim_s$, and $\sim_{hp}$ are all preserved by the operators $\sigma^p_{\textrm{abs}}$, $\upsilon^p_{\textrm{abs}}$, $\overline{\upsilon}^p_{\textrm{abs}}$, and $\nu_{\textrm{abs}}$. It is trivial and we omit it.
\end{proof}

\subsubsection{Soundness}

\begin{theorem}[Soundness of $\textrm{APTC}^{\textrm{sat}}$]
The axiomatization of $\textrm{APTC}^{\textrm{sat}}$ is sound modulo truly concurrent bisimulation equivalences $\sim_{p}$, $\sim_{s}$, and $\sim_{hp}$. That is,
\begin{enumerate}
  \item let $x$ and $y$ be $\textrm{APTC}^{\textrm{sat}}$ terms. If $\textrm{APTC}^{\textrm{sat}}\vdash x=y$, then $x\sim_{s} y$;
  \item let $x$ and $y$ be $\textrm{APTC}^{\textrm{sat}}$ terms. If $\textrm{APTC}^{\textrm{sat}}\vdash x=y$, then $x\sim_{p} y$;
  \item let $x$ and $y$ be $\textrm{APTC}^{\textrm{sat}}$ terms. If $\textrm{APTC}^{\textrm{sat}}\vdash x=y$, then $x\sim_{hp} y$.
\end{enumerate}
\end{theorem}

\begin{proof}
Since $\sim_p$, $\sim_s$, and $\sim_{hp}$ are both equivalent and congruent relations, we only need to check if each axiom in Table \ref{AxiomsForAPTCSAT} is sound modulo $\sim_p$, $\sim_s$, and $\sim_{hp}$ respectively.

\begin{enumerate}
  \item We only check the soundness of the non-trivial axiom $SAP11$ modulo $\sim_s$.
        Let $p,q$ be $\textrm{APTC}^{\textrm{dat}}$ processes, and $\sigma^s_{\textrm{abs}}(p) \parallel \sigma^s_{\textrm{abs}}(q) = \sigma^s_{\textrm{abs}}(p\parallel q)$, it is sufficient to prove that $\sigma^s_{\textrm{abs}}(p) \parallel \sigma^s_{\textrm{abs}}(q) \sim_{s} \sigma^s_{\textrm{abs}}(p\parallel q)$. By the transition rules of operator $\sigma^s_{\textrm{abs}}$ and $\parallel$ in Table \ref{TRForAPTCSAT}, we get

        $$\frac{\langle p,0\rangle\nuparrow}{\langle\sigma^s_{\textrm{abs}}(p) \parallel \sigma^s_{\textrm{abs}}(q),s'\rangle\mapsto^s \langle\sigma^s_{\textrm{abs}}(p) \parallel \sigma^s_{\textrm{abs}}(q),s'+s\rangle}$$

        $$\frac{\langle p,0\rangle\nuparrow}{\langle\sigma^s_{\textrm{abs}}(p\parallel q),s'\rangle\mapsto^s \langle\sigma^s_{\textrm{abs}}(p\parallel q),s'+s\rangle}$$

        There are several cases:

        $$\frac{\langle p,s'\rangle\xrightarrow{a} \langle\surd,s'\rangle\quad \langle q,s'\rangle\xrightarrow{b}\langle\surd,s'\rangle}{\langle\sigma^s_{\textrm{abs}}(p) \parallel \sigma^s_{\textrm{abs}}(q),s'+s\rangle\xrightarrow{\{a,b\}}\langle\surd,s'+s\rangle}$$

        $$\frac{\langle p,s'\rangle\xrightarrow{a} \langle\surd,s'\rangle\quad \langle q,s'\rangle\xrightarrow{b}\langle\surd,s'\rangle}{\langle\sigma^s_{\textrm{abs}}(p\parallel q),s'+s\rangle\xrightarrow{\{a,b\}}\langle\surd,s'+s\rangle}$$

        $$\frac{\langle p,s'\rangle\xrightarrow{a} \langle p',s'\rangle\quad \langle q,s'\rangle\xrightarrow{b}\langle\surd,s'\rangle}{\langle\sigma^s_{\textrm{abs}}(p) \parallel \sigma^s_{\textrm{abs}}(q),s'+s\rangle\xrightarrow{\{a,b\}}\langle\sigma^s_{\textrm{abs}}(p'),s'+s\rangle}$$

        $$\frac{\langle p,s'\rangle\xrightarrow{a} \langle p',s'\rangle\quad \langle q,s'\rangle\xrightarrow{b}\langle\surd,s'\rangle}{\langle\sigma^s_{\textrm{abs}}(p\parallel q),s'+s\rangle\xrightarrow{\{a,b\}}\langle\sigma^s_{\textrm{abs}}(p'),s'+s\rangle}$$

        $$\frac{\langle p,s'\rangle\xrightarrow{a} \langle\surd,s'\rangle\quad \langle q,s'\rangle\xrightarrow{b}\langle q',s'\rangle}{\langle\sigma^s_{\textrm{abs}}(p) \parallel \sigma^s_{\textrm{abs}}(q),s'+s\rangle\xrightarrow{\{a,b\}}\langle\sigma^s_{\textrm{abs}}(q'),s'+s\rangle}$$

        $$\frac{\langle p,s'\rangle\xrightarrow{a} \langle\surd,s'\rangle\quad \langle q,s'\rangle\xrightarrow{b}\langle q',s'\rangle}{\langle\sigma^s_{\textrm{abs}}(p\parallel q),s'+s\rangle\xrightarrow{\{a,b\}}\langle\sigma^s_{\textrm{abs}}(q'),s'+s\rangle}$$

        $$\frac{\langle p,s'\rangle\xrightarrow{a} \langle p',s'\rangle\quad \langle q,s'\rangle\xrightarrow{b}\langle q',s'\rangle}{\langle\sigma^s_{\textrm{abs}}(p) \parallel \sigma^s_{\textrm{abs}}(q),s'+s\rangle\xrightarrow{\{a,b\}}\langle\sigma^s_{\textrm{abs}}(p')\between \sigma^s_{\textrm{abs}}(q'),s'+s\rangle}$$

        $$\frac{\langle p,s'\rangle\xrightarrow{a} \langle p',s'\rangle\quad \langle q,s'\rangle\xrightarrow{b}\langle q',s'\rangle}{\langle\sigma^s_{\textrm{abs}}(p\parallel q),s'+s\rangle\xrightarrow{\{a,b\}}\langle\sigma^s_{\textrm{abs}}(p'\between q'),s'+s\rangle}$$

        $$\frac{\langle p,s'\rangle \uparrow}{\langle\sigma^s_{\textrm{abs}}(p) \parallel \sigma^s_{\textrm{abs}}(q),s'+s\rangle\uparrow}$$

        $$\frac{\langle p,s'\rangle\uparrow}{\langle\sigma^s_{\textrm{abs}}(p\parallel q),s'+s\rangle\uparrow}$$

        $$\frac{\langle q,s'\rangle \uparrow}{\langle\sigma^s_{\textrm{abs}}(p) \parallel \sigma^s_{\textrm{abs}}(q),s'+s\rangle\uparrow}$$

        $$\frac{\langle q,s'\rangle\uparrow}{\langle\sigma^s_{\textrm{abs}}(p\parallel q),s'+s\rangle\uparrow}$$

        So, we see that each case leads to $\sigma^s_{\textrm{abs}}(p) \parallel \sigma^s_{\textrm{abs}}(q) \sim_{s} \sigma^s_{\textrm{abs}}(p\parallel q)$, as desired.
  \item From the definition of pomset bisimulation, we know that pomset bisimulation is defined by pomset transitions, which are labeled by pomsets. In a pomset transition, the events (actions) in the pomset are either within causality relations (defined by $\cdot$) or in concurrency (implicitly defined by $\cdot$ and $+$, and explicitly defined by $\between$), of course, they are pairwise consistent (without conflicts). We have already proven the case that all events are pairwise concurrent (soundness modulo step bisimulation), so, we only need to prove the case of events in causality. Without loss of generality, we take a pomset of $P=\{\tilde{a},\tilde{b}:\tilde{a}\cdot \tilde{b}\}$. Then the pomset transition labeled by the above $P$ is just composed of one single event transition labeled by $\tilde{a}$ succeeded by another single event transition labeled by $\tilde{b}$, that is, $\xrightarrow{P}=\xrightarrow{a}\xrightarrow{b}$.

        Similarly to the proof of soundness modulo step bisimulation equivalence, we can prove that each axiom in Table \ref{AxiomsForAPTCSAT} is sound modulo pomset bisimulation equivalence, we omit them.
  \item From the definition of hp-bisimulation, we know that hp-bisimulation is defined on the posetal product $(C_1,f,C_2),f:C_1\rightarrow C_2\textrm{ isomorphism}$. Two process terms $s$ related to $C_1$ and $t$ related to $C_2$, and $f:C_1\rightarrow C_2\textrm{ isomorphism}$. Initially, $(C_1,f,C_2)=(\emptyset,\emptyset,\emptyset)$, and $(\emptyset,\emptyset,\emptyset)\in\sim_{hp}$. When $s\xrightarrow{a}s'$ ($C_1\xrightarrow{a}C_1'$), there will be $t\xrightarrow{a}t'$ ($C_2\xrightarrow{a}C_2'$), and we define $f'=f[a\mapsto a]$. Then, if $(C_1,f,C_2)\in\sim_{hp}$, then $(C_1',f',C_2')\in\sim_{hp}$.

        Similarly to the proof of soundness modulo pomset bisimulation equivalence, we can prove that each axiom in Table \ref{AxiomsForAPTCSAT} is sound modulo hp-bisimulation equivalence, we just need additionally to check the above conditions on hp-bisimulation, we omit them.
\end{enumerate}

\end{proof}

\subsubsection{Completeness}

\begin{theorem}[Completeness of $\textrm{APTC}^{\textrm{sat}}$]
The axiomatization of $\textrm{APTC}^{\textrm{sat}}$ is complete modulo truly concurrent bisimulation equivalences $\sim_{p}$, $\sim_{s}$, and $\sim_{hp}$. That is,
\begin{enumerate}
  \item let $p$ and $q$ be closed $\textrm{APTC}^{\textrm{sat}}$ terms, if $p\sim_{s} q$ then $p=q$;
  \item let $p$ and $q$ be closed $\textrm{APTC}^{\textrm{sat}}$ terms, if $p\sim_{p} q$ then $p=q$;
  \item let $p$ and $q$ be closed $\textrm{APTC}^{\textrm{sat}}$ terms, if $p\sim_{hp} q$ then $p=q$.
\end{enumerate}

\end{theorem}

\begin{proof}
\begin{enumerate}
  \item Firstly, by the elimination theorem of $\textrm{APTC}^{\textrm{sat}}$, we know that for each closed $\textrm{APTC}^{\textrm{sat}}$ term $p$, there exists a closed basic $\textrm{APTC}^{\textrm{sat}}$ term $p'$, such that $\textrm{APTC}^{\textrm{sat}}\vdash p=p'$, so, we only need to consider closed basic $\textrm{APTC}^{\textrm{sat}}$ terms.

        The basic terms modulo associativity and commutativity (AC) of conflict $+$ (defined by axioms $A1$ and $A2$ in Table \ref{AxiomsForBATCSAT}) and associativity and commutativity (AC) of parallel $\parallel$ (defined by axioms $P2$ and $P3$ in Table \ref{AxiomsForAPTCSAT}), and these equivalences is denoted by $=_{AC}$. Then, each equivalence class $s$ modulo AC of $+$ and $\parallel$ has the following normal form

        $$s_1+\cdots+ s_k$$

        with each $s_i$ either an atomic event or of the form

        $$t_1\cdot\cdots\cdot t_m$$

        with each $t_j$ either an atomic event or of the form

        $$u_1\parallel\cdots\parallel u_n$$

        with each $u_l$ an atomic event, and each $s_i$ is called the summand of $s$.

        Now, we prove that for normal forms $n$ and $n'$, if $n\sim_{s} n'$ then $n=_{AC}n'$. It is sufficient to induct on the sizes of $n$ and $n'$. We can get $n=_{AC} n'$.

        Finally, let $s$ and $t$ be basic $\textrm{APTC}^{\textrm{sat}}$ terms, and $s\sim_s t$, there are normal forms $n$ and $n'$, such that $s=n$ and $t=n'$. The soundness theorem modulo step bisimulation equivalence yields $s\sim_s n$ and $t\sim_s n'$, so $n\sim_s s\sim_s t\sim_s n'$. Since if $n\sim_s n'$ then $n=_{AC}n'$, $s=n=_{AC}n'=t$, as desired.
  \item This case can be proven similarly, just by replacement of $\sim_{s}$ by $\sim_{p}$.
  \item This case can be proven similarly, just by replacement of $\sim_{s}$ by $\sim_{hp}$.
\end{enumerate}
\end{proof}

\subsection{$\textrm{APTC}^{\textrm{sat}}$ with Integration}

In this subsection, we will introduce the theory $\textrm{APTC}^{\textrm{sat}}$ with integration called $\textrm{APTC}^{\textrm{sat}}\textrm{I}$.

\subsubsection{The Theory $\textrm{APTC}^{\textrm{sat}}\textrm{I}$}

\begin{definition}[Signature of $\textrm{APTC}^{\textrm{sat}}\textrm{I}$]
The signature of $\textrm{APTC}^{\textrm{sat}}\textrm{I}$ consists of the signature of $\textrm{APTC}^{\textrm{sat}}$ and the integration operator $\int: \mathcal{P}(\mathbb{R}^{\geq})\times\mathbb{R}^{\geq}.\mathcal{P}_{\textrm{abs}} \rightarrow\mathcal{P}_{\textrm{abs}}$.
\end{definition}

The set of axioms of $\textrm{APTC}^{\textrm{sat}}\textrm{I}$ consists of the laws given in Table \ref{AxiomsForAPTCSATI}.

\begin{center}
    \begin{table}
        \begin{tabular}{@{}ll@{}}
            \hline No. &Axiom\\
            $INT13$ & $\int_{v\in V}(F(v)\parallel x) = (\int_{v\in V}F(v))\parallel x$\\
            $INT14$ & $\int_{v\in V}(x \parallel F(v)) = x\parallel (\int_{v\in V}F(v))$\\
            $INT15$ & $\int_{v\in V}(F(v)\mid x) = (\int_{v\in V}F(v))\mid x$\\
            $INT16$ & $\int_{v\in V}(x \mid F(v)) = x\mid (\int_{v\in V}F(v))$\\
            $INT17$ & $\int_{v\in V}\partial_H(F(v)) = \partial_H(\int_{v\in V}F(v))$\\
            $INT18$ & $\int_{v\in V}\Theta(F(v)) = \Theta(\int_{v\in V}F(v))$\\
            $INT19$ & $\int_{v\in V}(F(v)\triangleleft x) = (\int_{v\in V}F(v))\triangleleft x$\\
            $SAU5$ & $\nu_{\textrm{abs}}(\int_{v\in V}F(v)) = \int_{v\in V}\nu_{\textrm{abs}}(F(v))$\\
        \end{tabular}
        \caption{Axioms of $\textrm{APTC}^{\textrm{sat}}\textrm{I}$}
        \label{AxiomsForAPTCSATI}
    \end{table}
\end{center}

The operational semantics of $\textrm{APTC}^{\textrm{sat}}\textrm{I}$ are defined by the transition rules in Table \ref{TRForBATCSATI}.

\subsubsection{Elimination}

\begin{definition}[Basic terms of $\textrm{APTC}^{\textrm{sat}}\textrm{I}$]
The set of basic terms of $\textrm{APTC}^{\textrm{sat}}\textrm{I}$, $\mathcal{B}(\textrm{APTC}^{\textrm{sat}})$, is inductively defined as follows by two auxiliary sets $\mathcal{B}_0(\textrm{APTC}^{\textrm{sat}}\textrm{I})$ and $\mathcal{B}_1(\textrm{APTC}^{\textrm{sat}}\textrm{I})$:
\begin{enumerate}
  \item if $a\in A_{\delta}$, then $\tilde{a} \in \mathcal{B}_1(\textrm{APTC}^{\textrm{sat}}\textrm{I})$;
  \item if $a\in A$ and $t\in \mathcal{B}(\textrm{APTC}^{\textrm{sat}}\textrm{I})$, then $\tilde{a}\cdot t \in \mathcal{B}_1(\textrm{APTC}^{\textrm{sat}}\textrm{I})$;
  \item if $t,t'\in \mathcal{B}_1(\textrm{APTC}^{\textrm{sat}}\textrm{I})$, then $t+t'\in \mathcal{B}_1(\textrm{APTC}^{\textrm{sat}}\textrm{I})$;
  \item if $t,t'\in \mathcal{B}_1(\textrm{APTC}^{\textrm{sat}}\textrm{I})$, then $t\parallel t'\in \mathcal{B}_1(\textrm{APTC}^{\textrm{sat}}\textrm{I})$;
  \item if $t\in \mathcal{B}_1(\textrm{APTC}^{\textrm{sat}}\textrm{I})$, then $t\in \mathcal{B}_0(\textrm{APTC}^{\textrm{sat}}\textrm{I})$;
  \item if $p>0$ and $t\in \mathcal{B}_0(\textrm{APTC}^{\textrm{sat}}\textrm{I})$, then $\sigma^p_{\textrm{abs}}(t) \in \mathcal{B}_0(\textrm{APTC}^{\textrm{sat}}\textrm{I})$;
  \item if $p>0$, $t\in \mathcal{B}_1(\textrm{APTC}^{\textrm{sat}}\textrm{I})$ and $t'\in \mathcal{B}_0(\textrm{APTC}^{\textrm{sat}}\textrm{I})$, then $t+\sigma^p_{\textrm{abs}}(t') \in \mathcal{B}_0(\textrm{APTC}^{\textrm{sat}}\textrm{I})$;
  \item if $t\in \mathcal{B}_0(\textrm{APTC}^{\textrm{sat}}\textrm{I})$, then $\nu_{\textrm{abs}}(t) \in \mathcal{B}_0(\textrm{APTC}^{\textrm{sat}}\textrm{I})$;
  \item if $t\in \mathcal{B}_0(\textrm{APTC}^{\textrm{sat}}\textrm{I})$, then $\int_{v\in V}(t) \in \mathcal{B}_0(\textrm{APTC}^{\textrm{sat}}\textrm{I})$;
  \item $\dot{\delta}\in \mathcal{B}(\textrm{APTC}^{\textrm{sat}}\textrm{I})$;
  \item if $t\in \mathcal{B}_0(\textrm{APTC}^{\textrm{sat}}\textrm{I})$, then $t\in \mathcal{B}(\textrm{APTC}^{\textrm{sat}}\textrm{I})$.
\end{enumerate}
\end{definition}

\begin{theorem}[Elimination theorem]
Let $p$ be a closed $\textrm{APTC}^{\textrm{sat}}\textrm{I}$ term. Then there is a basic $\textrm{APTC}^{\textrm{sat}}\textrm{I}$ term $q$ such that $\textrm{APTC}^{\textrm{sat}}\vdash p=q$.
\end{theorem}

\begin{proof}
It is sufficient to induct on the structure of the closed $\textrm{APTC}^{\textrm{sat}}\textrm{I}$ term $p$. It can be proven that $p$ combined by the constants and operators of $\textrm{APTC}^{\textrm{sat}}\textrm{I}$ exists an equal basic term $q$, and the other operators not included in the basic terms, such as $\upsilon_{\textrm{abs}}$, $\overline{\upsilon}_{\textrm{abs}}$, $\between$, $\mid$, $\partial_H$, $\Theta$ and $\triangleleft$ can be eliminated.
\end{proof}

\subsubsection{Connections}

\begin{theorem}[Generalization of $\textrm{APTC}^{\textrm{sat}}\textrm{I}$]
\begin{enumerate}
  \item By the definitions of $a=\int_{v\in[0,\infty)}\sigma^v_{\textrm{abs}}(\tilde{a})$ for each $a\in A$ and $\delta=\int_{v\in[0,\infty)}\sigma^v_{\textrm{abs}}(\tilde{\delta})$, $\textrm{APTC}^{\textrm{sat}}\textrm{I}$ is a generalization of $APTC$.
  \item $\textrm{APTC}^{\textrm{sat}}\textrm{I}$ is a generalization of $APTC^{\textrm{sat}}$.
\end{enumerate}

\end{theorem}

\begin{proof}
\begin{enumerate}
  \item It follows from the following two facts.

    \begin{enumerate}
      \item The transition rules of $APTC$ in section \ref{tcpa} are all source-dependent;
      \item The sources of the transition rules of $\textrm{APTC}^{\textrm{sat}}\textrm{I}$ contain an occurrence of $\dot{\delta}$, $\tilde{a}$, $\sigma^p_{\textrm{abs}}$, $\upsilon^p_{\textrm{abs}}$, $\overline{\upsilon}^p_{\textrm{abs}}$, $\nu_{\textrm{abs}}$, and $\int$.
    \end{enumerate}

    So, $APTC$ is an embedding of $\textrm{APTC}^{\textrm{sat}}\textrm{I}$, as desired.
  \item It follows from the following two facts.

    \begin{enumerate}
      \item The transition rules of $APTC^{\textrm{sat}}$ are all source-dependent;
      \item The sources of the transition rules of $\textrm{APTC}^{\textrm{sat}}\textrm{I}$ contain an occurrence of $\int$.
    \end{enumerate}

    So, $APTC^{\textrm{sat}}$ is an embedding of $\textrm{APTC}^{\textrm{sat}}\textrm{I}$, as desired.
\end{enumerate}
\end{proof}

\subsubsection{Congruence}

\begin{theorem}[Congruence of $\textrm{APTC}^{\textrm{sat}}\textrm{I}$]
Truly concurrent bisimulation equivalences are all congruences with respect to $\textrm{APTC}^{\textrm{sat}}\textrm{I}$. That is,
\begin{itemize}
  \item pomset bisimulation equivalence $\sim_{p}$ is a congruence with respect to $\textrm{APTC}^{\textrm{sat}}\textrm{I}$;
  \item step bisimulation equivalence $\sim_{s}$ is a congruence with respect to $\textrm{APTC}^{\textrm{sat}}\textrm{I}$;
  \item hp-bisimulation equivalence $\sim_{hp}$ is a congruence with respect to $\textrm{APTC}^{\textrm{sat}}\textrm{I}$;
\end{itemize}
\end{theorem}

\begin{proof}
It is easy to see that $\sim_p$, $\sim_s$, $\sim_{hp}$ and $\sim_{hhp}$ are all equivalent relations on $\textrm{APTC}^{\textrm{sat}}\textrm{I}$ terms, it is only sufficient to prove that $\sim_p$, $\sim_s$, and $\sim_{hp}$ are all preserved by the operators $\int$. It is trivial and we omit it.
\end{proof}

\subsubsection{Soundness}

\begin{theorem}[Soundness of $\textrm{APTC}^{\textrm{sat}}\textrm{I}$]
The axiomatization of $\textrm{APTC}^{\textrm{sat}}\textrm{I}$ is sound modulo truly concurrent bisimulation equivalences $\sim_{p}$, $\sim_{s}$, $\sim_{hp}$ and $\sim_{hhp}$. That is,
\begin{enumerate}
  \item let $x$ and $y$ be $\textrm{APTC}^{\textrm{sat}}\textrm{I}$ terms. If $\textrm{APTC}^{\textrm{sat}}\textrm{I}\vdash x=y$, then $x\sim_{s} y$;
  \item let $x$ and $y$ be $\textrm{APTC}^{\textrm{sat}}\textrm{I}$ terms. If $\textrm{APTC}^{\textrm{sat}}\textrm{I}\vdash x=y$, then $x\sim_{p} y$;
  \item let $x$ and $y$ be $\textrm{APTC}^{\textrm{sat}}\textrm{I}$ terms. If $\textrm{APTC}^{\textrm{sat}}\textrm{I}\vdash x=y$, then $x\sim_{hp} y$;
\end{enumerate}
\end{theorem}

\begin{proof}
Since $\sim_p$, $\sim_s$, $\sim_{hp}$ and $\sim_{hhp}$ are both equivalent and congruent relations, we only need to check if each axiom in Table \ref{AxiomsForAPTCSATI} is sound modulo $\sim_p$, $\sim_s$, and $\sim_{hp}$ respectively.

\begin{enumerate}
  \item We can check the soundness of each axiom in Table \ref{AxiomsForAPTCSATI}, by the transition rules in Table \ref{TRForBATCSRTI}, it is trivial and we omit them.
  \item From the definition of pomset bisimulation, we know that pomset bisimulation is defined by pomset transitions, which are labeled by pomsets. In a pomset transition, the events (actions) in the pomset are either within causality relations (defined by $\cdot$) or in concurrency (implicitly defined by $\cdot$ and $+$, and explicitly defined by $\between$), of course, they are pairwise consistent (without conflicts). We have already proven the case that all events are pairwise concurrent (soundness modulo step bisimulation), so, we only need to prove the case of events in causality. Without loss of generality, we take a pomset of $P=\{\tilde{a},\tilde{b}:\tilde{a}\cdot \tilde{b}\}$. Then the pomset transition labeled by the above $P$ is just composed of one single event transition labeled by $\tilde{a}$ succeeded by another single event transition labeled by $\tilde{b}$, that is, $\xrightarrow{P}=\xrightarrow{a}\xrightarrow{b}$.

        Similarly to the proof of soundness modulo step bisimulation equivalence, we can prove that each axiom in Table \ref{AxiomsForAPTCSATI} is sound modulo pomset bisimulation equivalence, we omit them.
  \item From the definition of hp-bisimulation, we know that hp-bisimulation is defined on the posetal product $(C_1,f,C_2),f:C_1\rightarrow C_2\textrm{ isomorphism}$. Two process terms $s$ related to $C_1$ and $t$ related to $C_2$, and $f:C_1\rightarrow C_2\textrm{ isomorphism}$. Initially, $(C_1,f,C_2)=(\emptyset,\emptyset,\emptyset)$, and $(\emptyset,\emptyset,\emptyset)\in\sim_{hp}$. When $s\xrightarrow{a}s'$ ($C_1\xrightarrow{a}C_1'$), there will be $t\xrightarrow{a}t'$ ($C_2\xrightarrow{a}C_2'$), and we define $f'=f[a\mapsto a]$. Then, if $(C_1,f,C_2)\in\sim_{hp}$, then $(C_1',f',C_2')\in\sim_{hp}$.

        Similarly to the proof of soundness modulo pomset bisimulation equivalence, we can prove that each axiom in Table \ref{AxiomsForAPTCSATI} is sound modulo hp-bisimulation equivalence, we just need additionally to check the above conditions on hp-bisimulation, we omit them.
  \item We just need to add downward-closed condition to the soundness modulo hp-bisimulation equivalence, we omit them.
\end{enumerate}

\end{proof}

\subsubsection{Completeness}

\begin{theorem}[Completeness of $\textrm{APTC}^{\textrm{sat}}\textrm{I}$]
The axiomatization of $\textrm{APTC}^{\textrm{sat}}\textrm{I}$ is complete modulo truly concurrent bisimulation equivalences $\sim_{p}$, $\sim_{s}$, $\sim_{hp}$ and $\sim_{hhp}$. That is,
\begin{enumerate}
  \item let $p$ and $q$ be closed $\textrm{APTC}^{\textrm{sat}}\textrm{I}$ terms, if $p\sim_{s} q$ then $p=q$;
  \item let $p$ and $q$ be closed $\textrm{APTC}^{\textrm{sat}}\textrm{I}$ terms, if $p\sim_{p} q$ then $p=q$;
  \item let $p$ and $q$ be closed $\textrm{APTC}^{\textrm{sat}}\textrm{I}$ terms, if $p\sim_{hp} q$ then $p=q$;
\end{enumerate}

\end{theorem}

\begin{proof}
\begin{enumerate}
  \item Firstly, by the elimination theorem of $\textrm{APTC}^{\textrm{sat}}\textrm{I}$, we know that for each closed $\textrm{APTC}^{\textrm{sat}}\textrm{I}$ term $p$, there exists a closed basic $\textrm{APTC}^{\textrm{sat}}\textrm{I}$ term $p'$, such that $\textrm{APTC}^{\textrm{sat}}\textrm{I}\vdash p=p'$, so, we only need to consider closed basic $\textrm{APTC}^{\textrm{sat}}\textrm{I}$ terms.

        The basic terms modulo associativity and commutativity (AC) of conflict $+$ (defined by axioms $A1$ and $A2$ in Table \ref{AxiomsForBATCSAT}) and associativity and commutativity (AC) of parallel $\parallel$ (defined by axioms $P2$ and $P3$ in Table \ref{AxiomsForAPTCSAT}), and these equivalences is denoted by $=_{AC}$. Then, each equivalence class $s$ modulo AC of $+$ and $\parallel$ has the following normal form

        $$s_1+\cdots+ s_k$$

        with each $s_i$ either an atomic event or of the form

        $$t_1\cdot\cdots\cdot t_m$$

        with each $t_j$ either an atomic event or of the form

        $$u_1\parallel\cdots\parallel u_n$$

        with each $u_l$ an atomic event, and each $s_i$ is called the summand of $s$.

        Now, we prove that for normal forms $n$ and $n'$, if $n\sim_{s} n'$ then $n=_{AC}n'$. It is sufficient to induct on the sizes of $n$ and $n'$. We can get $n=_{AC} n'$.

        Finally, let $s$ and $t$ be basic $\textrm{APTC}^{\textrm{sat}}\textrm{I}$ terms, and $s\sim_s t$, there are normal forms $n$ and $n'$, such that $s=n$ and $t=n'$. The soundness theorem modulo step bisimulation equivalence yields $s\sim_s n$ and $t\sim_s n'$, so $n\sim_s s\sim_s t\sim_s n'$. Since if $n\sim_s n'$ then $n=_{AC}n'$, $s=n=_{AC}n'=t$, as desired.
  \item This case can be proven similarly, just by replacement of $\sim_{s}$ by $\sim_{p}$.
  \item This case can be proven similarly, just by replacement of $\sim_{s}$ by $\sim_{hp}$.
\end{enumerate}
\end{proof}

\subsection{Standard Initial Abstraction}

In this subsection, we will introduce $\textrm{APTC}^{\textrm{sat}}$ with standard initial abstraction called $\textrm{APTC}^{\textrm{sat}}\surd$.

\subsubsection{Basic Definition}

\begin{definition}[Standard initial abstraction]
Standard initial abstraction $\surd_s$ is an abstraction mechanism to form functions from non-negative real numbers to processes with absolute timing, that map each number $r$ to a process initialized at time $r$.
\end{definition}

\subsubsection{The Theory $\textrm{APTC}^{\textrm{sat}}\textrm{I}\surd$}

\begin{definition}[Signature of $\textrm{APTC}^{\textrm{sat}}\textrm{I}\surd$]
The signature of $\textrm{APTC}^{\textrm{sat}}\textrm{I}\surd$ consists of the signature of $\textrm{APTC}^{\textrm{sat}}\textrm{I}$, and the standard initial abstraction operator $\surd_s: \mathbb{R}^{\geq}.\mathcal{P}^*_{\textrm{abs}}\rightarrow\mathcal{P}^*_{\textrm{abs}}$. Where $\mathcal{P}^*_{\textrm{abs}}$ is the sorts with standard initial abstraction.
\end{definition}

The set of axioms of $\textrm{APTC}^{\textrm{sat}}\textrm{I}\surd$ consists of the laws given in Table \ref{AxiomsForAPTCSATDIA}. Where $v,w,\cdots$ are variables of sort $\mathbb{R}^{\geq}$, $F,G,\cdots$ are variables of sort $\mathbb{R}^{\geq}.\mathcal{P}^*_{\textrm{abs}}$, $K,L,\cdots$ are variables of sort $\mathbb{R}^{\geq},\mathbb{R}^{\geq}.\mathcal{P}^*_{\textrm{abs}}$, and we write $\surd_s v.t$ for $\surd_s(v.t)$.

\begin{center}
    \begin{table}
        \begin{tabular}{@{}ll@{}}
            \hline No. &Axiom\\
            $SIA1$ & $\surd_s v.F(v) = \surd_s w.F(w)$\\
            $SIA2$ & $\overline{\upsilon}^p_{\textrm{abs}}(\surd_s v.F(v)) = \overline{\upsilon}^p_{\textrm{abs}}(F(p))$\\
            $SIA3$ & $\surd_s v.(\surd_s w.K(v,w)) = \surd_s v.K(v,v)$\\
            $SIA4$ & $x = \surd_s v.x$\\
            $SIA5$ & $(\forall v\in\mathbb{R}^{\geq}.\overline{\upsilon}^v_{\textrm{abs}}(x) = \overline{\upsilon}^v_{\textrm{abs}}(y))\Rightarrow x=y$\\
            $SIA6$ & $\sigma^p_{\textrm{abs}}(\tilde{a})\cdot x = \sigma^p_{\textrm{abs}}(\tilde{a})\cdot\overline{\upsilon}^p_{\textrm{abs}}(x)$\\
            $SIA7$ & $\sigma^p_{\textrm{abs}}(\surd_s v.F(v)) = \sigma^p_{\textrm{abs}}(F(0))$\\
            $SIA8$ & $(\surd_s v.F(v)) + x = \surd_s v.(F(v) + \overline{\upsilon}^v_{\textrm{abs}}(x))$\\
            $SIA9$ & $(\surd_s v.F(v)) \cdot x = \surd_s v.(F(v) \cdot x)$\\
            $SIA10$ & $\upsilon^p_{\textrm{abs}}(\surd_s v.F(v)) = \surd_s v.\upsilon^p_{\textrm{abs}}(F(v))$\\
            $SIA11$ & $(\surd_s v.F(v)) \parallel x = \surd_s v.(F(v) \parallel\overline{\upsilon}^v_{\textrm{abs}}(x))$\\
            $SIA12$ & $x\parallel(\surd_s v.F(v)) = \surd_s v.(\overline{\upsilon}^v_{\textrm{abs}}(x)\parallel F(v))$\\
            $SIA13$ & $(\surd_s v.F(v)) \mid x = \surd_s v.(F(v) \mid\overline{\upsilon}^v_{\textrm{abs}}(x))$\\
            $SIA14$ & $x\mid(\surd_s v.F(v)) = \surd_s v.(\overline{\upsilon}^v_{\textrm{abs}}(x)\mid F(v))$\\
            $SIA15$ & $\Theta(\surd_s v.F(v)) = \surd_s v.\Theta(F(v))$\\
            $SIA16$ & $(\surd_s v.F(v)) \triangleleft x = \surd_s v.(F(v) \triangleleft x)$\\
            $SIA17$ & $\partial_H(\surd_s v.F(v)) = \surd_s v.\partial_H(F(v))$\\
            $SIA18$ & $\nu_{\textrm{abs}}(\surd_s v.F(v)) = \surd_s v.\nu_{\textrm{abs}}(F(v))$\\
            $SIA19$ & $\int_{v\in V}(\surd_s w.F(w)) = \surd_s w.(\int_{v\in V}K(v,w))(v\neq w)$\\
        \end{tabular}
        \caption{Axioms of $\textrm{APTC}^{\textrm{sat}}\textrm{I}\surd(p\geq 0)$}
        \label{AxiomsForAPTCSATDIA}
    \end{table}
\end{center}

It sufficient to extend bisimulations $\mathcal{CI}/\sim$ of $APTC^{\textrm{sat}}$ to

$$(\mathcal{CI}/\sim)^* = \{f:\mathbb{R}^{\geq}\rightarrow \mathcal{CI}/\sim | \forall v\in\mathbb{R}^{\geq}.f(v) = \overline{\upsilon}^v_{\textrm{abs}}(f(v))\}$$

and define the constants and operators of $APTC^{\textrm{sat}}\surd$ on $(\mathcal{CI}/\sim)^*$ as in Table \ref{DefsForAPTCSATDIA}.

\begin{center}
    \begin{table}
        \begin{tabular}{@{}ll@{}}
            \hline
            $\dot{\delta} = \lambda w.\dot{\delta}$ & $\overline{\upsilon}^v_{\textrm{abs}}(f) = f(v)$\\
            $\tilde{a} = \lambda w.\overline{\upsilon}^w_{\textrm{abs}}(\tilde{a}) (a\in A_{\delta})$ & $f\between g = \lambda w.(f(w)\between g(w))$\\
            $\sigma^v_{\textrm{abs}}(f) = \lambda_w.\overline{\upsilon}^w_{\textrm{abs}}(\sigma^v_{\textrm{abs}}(f(0)))$ & $f\parallel g = \lambda w.(f(w)\parallel g(w))$\\
            $f+g = \lambda w.(f(w)+g(w))$ & $f\mid g = \lambda w.(f(w)\mid g(w))$\\
            $f\cdot g = \lambda w.(f(w)*g)$ & $\partial_H(f) = \lambda w.\partial_H(f(w))$\\
            $\Theta(f) = \lambda w.\Theta(f(w))$ & $f\triangleleft g = \lambda w. (f(w)\triangleleft g(w))$\\
            $\upsilon^v_{\textrm{abs}}(f) = \lambda w.\overline{\upsilon}^w_{\textrm{abs}}(\upsilon^v_{\textrm{abs}}f(w)))$ & $\surd_s(\varphi) = \lambda w.\overline{\upsilon}^w_{\textrm{abs}}(\varphi(w))$\\
            $\nu_{\textrm{abs}}(f) = \lambda w.\overline{\upsilon}^w_{\textrm{abs}}(\nu_{\textrm{abs}}(f(w)))$ & $\int(V,\varphi) = \lambda w.\int(V,\lambda w'.\varphi(w')(w))$\\
        \end{tabular}
        \caption{Definitions of $\textrm{APTC}^{\textrm{sat}}\textrm{I}$ on $(\mathcal{CI}/\sim)^*$}
        \label{DefsForAPTCSATDIA}
    \end{table}
\end{center}

\subsubsection{Connections}

\begin{center}
    \begin{table}
        \begin{tabular}{@{}ll@{}}
            \hline
            $\tilde{\tilde{a}} = \surd_s w.\sigma^w_{\textrm{abs}}(\tilde{a})(a\in A)$\\
            $\tilde{\tilde{\delta}} = \surd_s w.\sigma^w_{\textrm{abs}}(\tilde{\delta})$\\
            $\sigma^v_{\textrm{abs}}(x) = \surd_s w.\overline{\upsilon}^{v+w}_{\textrm{abs}}(x)$\\
            $\upsilon^v_{\textrm{abs}}(x) = \surd_s w.\upsilon^{v+w}_{\textrm{abs}}(\overline{\upsilon}^w_{\textrm{abs}}(x))$\\
            $\overline{\upsilon}^v_{\textrm{abs}}(x) = \surd_s w.\overline{\upsilon}^{v+w}_{\textrm{abs}}(\overline{\upsilon}^w_{\textrm{abs}}(x))$\\
            $\nu_{\textrm{abs}}(x) = \surd_s w.\sigma^w_{\textrm{abs}}(\nu_{\textrm{abs}}(x))$\\
        \end{tabular}
        \caption{Definitions of constants and operators of $ACTC^{\textrm{srt}}\textrm{I}$ in $\textrm{APTC}^{\textrm{sat}}\textrm{I}\surd$}
        \label{DefsForAPTCSATDIASRT}
    \end{table}
\end{center}

\begin{center}
    \begin{table}
        \begin{tabular}{@{}ll@{}}
            \hline
            $\tilde{a} = \int_{v\in[0,1)}\sigma^v_{\textrm{abs}}(\tilde{a})(a\in A)$\\
            $\tilde{\delta} = \int_{v\in[0,1)}\sigma^v_{\textrm{abs}}(\tilde{\delta})$\\
            $\sigma^i_{\textrm{abs}}(x) = \sigma^i_{\textrm{abs}}(x)$\\
            $\upsilon^i_{\textrm{abs}}(x) = \upsilon^{i}_{\textrm{abs}}(x)$\\
            $\overline{\upsilon}^i_{\textrm{abs}}(x) = \overline{\upsilon}^{i}_{\textrm{abs}}(x)$\\
            $\surd_d.F(i) = \surd_s v.F(\llcorner v\lrcorner)$\\
        \end{tabular}
        \caption{Definitions of constants and operators of $ACTC^{\textrm{dat}}\surd$ in $\textrm{APTC}^{\textrm{sat}}\textrm{I}\surd$}
        \label{DefsForAPTCSATDIADAT}
    \end{table}
\end{center}

\begin{center}
    \begin{table}
        \begin{tabular}{@{}ll@{}}
            \hline
            $\mathcal{D}(\dot{\delta}) = \dot{\delta}$\\
            $\mathcal{D}(\tilde{a}) = \tilde{a}$\\
            $\mathcal{D}(\sigma^{p}_{\textrm{abs}}(x)) = \sigma^{\llcorner p\lrcorner}_{\textrm{abs}}(\mathcal{D}(x))$\\
            $\mathcal{D}(x+y) = \mathcal{D}(x) + \mathcal{D}(y)$\\
            $\mu(x\cdot y) = \mu(x)\cdot\mu(y)$\\
            $\mathcal{D}(x\parallel y) = \mathcal{D}(x)\parallel\mathcal{D}(y)$\\
            $\mathcal{D}(\int_{v\in V}F(v)) = \int_{v\in V}\mathcal{D}(F(v))$\\
            $\mathcal{D}(\surd_s v.F(v)) = \surd_s v.\mathcal{D}(F(v))$\\
        \end{tabular}
        \caption{Axioms for discretization $(a\in A_{\delta}, p\geq 0)$}
        \label{MUForAPTCDATDIADAT}
    \end{table}
\end{center}

\begin{center}
    \begin{table}
        $$\frac{\langle x, p\rangle \xrightarrow{a}\langle x',p\rangle}{\langle \mathcal{D}(x),q\rangle\xrightarrow{a}\langle\mathcal{D}(x'),q\rangle}(q\in [\llcorner p\lrcorner,\llcorner p\lrcorner+1))
        \quad\frac{\langle x, p\rangle \xrightarrow{a}\langle \surd,p\rangle}{\langle \mathcal{D}(x),q\rangle\xrightarrow{a}\langle\surd,q\rangle}(q\in [\llcorner p\lrcorner,\llcorner p\lrcorner+1))$$

        $$\frac{\langle x, p\rangle \mapsto^{r}\langle x,p+r\rangle}{\langle \mathcal{D}(x),p\rangle\mapsto^{r'}\langle\mathcal{D}(x),p+r'\rangle}(p+r'\in [p+r,\llcorner p+r\lrcorner+1))
        \quad\frac{\langle x, p\rangle\uparrow}{\langle\mathcal{D}(x),p\rangle\uparrow}\quad\frac{\langle x,p\rangle\nuparrow}{\langle\mathcal{D}(x),p\rangle\mapsto^r \langle\mathcal{D}(x),p+r\rangle}(p+r\in[p,\llcorner p\lrcorner +1))$$
        \caption{Transition rules of discretization $a\in A, p,q\geq 0,r,r'>0$}
        \label{TRMUForAPTCDATDIADAT}
    \end{table}
\end{center}

\begin{center}
    \begin{table}
        \begin{tabular}{@{}ll@{}}
            \hline
            $\mathcal{D}(f) = \lambda k.\mathcal{D}(f(k))$\\
        \end{tabular}
        \caption{Definitions of of discretization on $(\mathcal{CI}/\sim)^*$}
        \label{DEFMUForAPTCDATDIADAT}
    \end{table}
\end{center}

\begin{theorem}[Generalization of $\textrm{APTC}^{\textrm{sat}}\textrm{I}\surd$]
\begin{enumerate}
  \item By the definitions of constants and operators of $ACTC^{\textrm{srt}}\textrm{I}$ in $\textrm{APTC}^{\textrm{sat}}\textrm{I}\surd$ in Table \ref{DefsForAPTCSATDIASRT}, $\textrm{APTC}^{\textrm{sat}}\textrm{I}\surd$ is a generalization of $APTC^{\textrm{srt}}\textrm{I}$.
  \item $\textrm{APTC}^{\textrm{sat}}\textrm{I}\surd$ is a generalization of $\textrm{BATC}^{\textrm{sat}}\textrm{I}$.
  \item By the definitions of constants and operators of $ACTC^{\textrm{dat}}\surd$ in $\textrm{APTC}^{\textrm{sat}}\textrm{I}\surd$ in Table \ref{DefsForAPTCSATDIADAT}, a discretization operator $\mathcal{D}: \mathcal{P}^*_{\textrm{abs}}\rightarrow \mathcal{P}^*_{\textrm{abs}}$ in Table \ref{MUForAPTCDATDIADAT}, Table \ref{TRMUForAPTCDATDIADAT} and Table \ref{DEFMUForAPTCDATDIADAT}, $\textrm{APTC}^{\textrm{sat}}\textrm{I}\surd$ is a generalization of $APTC^{\textrm{dat}}\surd$.
\end{enumerate}

\end{theorem}

\begin{proof}
\begin{enumerate}
  \item It follows from the following two facts. By the definitions of constants and operators of $ACTC^{\textrm{srt}}\textrm{I}$ in $\textrm{APTC}^{\textrm{sat}}\textrm{I}\surd$ in Table \ref{DefsForAPTCSATDIASRT},

    \begin{enumerate}
      \item the transition rules of $ACTC^{\textrm{srt}}\textrm{I}$ are all source-dependent;
      \item the sources of the transition rules of $\textrm{APTC}^{\textrm{sat}}\textrm{I}\surd$ contain an occurrence of $\surd_s$, $\int$ and $\nu_{\textrm{abs}}$.
    \end{enumerate}

    So, $ACTC^{\textrm{srt}}\textrm{I}$ is an embedding of $\textrm{APTC}^{\textrm{sat}}\textrm{I}\surd$, as desired.
    \item It follows from the following two facts.

    \begin{enumerate}
      \item The transition rules of $\textrm{APTC}^{\textrm{sat}}\textrm{I}$ are all source-dependent;
      \item The sources of the transition rules of $\textrm{APTC}^{\textrm{sat}}\textrm{I}\surd$ contain an occurrence of $\surd_s$.
    \end{enumerate}

    So, $\textrm{APTC}^{\textrm{sat}}\textrm{I}$ is an embedding of $\textrm{APTC}^{\textrm{sat}}\textrm{I}\surd$, as desired.
    \item It follows from the following two facts. By the definitions of constants and operators of $ACTC^{\textrm{dat}}\surd$ in $\textrm{APTC}^{\textrm{sat}}\textrm{I}\surd$ in Table \ref{DefsForAPTCSATDIADAT}, a discretization operator $\mathcal{D}: \mathcal{P}^*_{\textrm{abs}}\rightarrow \mathcal{P}^*_{\textrm{abs}}$ in Table \ref{MUForAPTCDATDIADAT}, Table \ref{TRMUForAPTCDATDIADAT} and Table \ref{DEFMUForAPTCDATDIADAT},

    \begin{enumerate}
      \item The transition rules of $\textrm{APTC}^{\textrm{dat}}\surd$ are all source-dependent;
      \item The sources of the transition rules of $\textrm{APTC}^{\textrm{sat}}\textrm{I}\surd$ contain an occurrence of $\surd_s$, $\int$ and $\nu_{\textrm{abs}}$.
    \end{enumerate}

    So, $\textrm{APTC}^{\textrm{sat}}\textrm{I}$ is an embedding of $\textrm{APTC}^{\textrm{sat}}\textrm{I}\surd$, as desired.

\end{enumerate}
\end{proof}

\subsubsection{Congruence}

\begin{theorem}[Congruence of $\textrm{APTC}^{\textrm{sat}}\textrm{I}\surd$]
Truly concurrent bisimulation equivalences $\sim_p$, $\sim_s$ and $\sim_{hp}$ are all congruences with respect to $\textrm{APTC}^{\textrm{sat}}\textrm{I}\surd$. That is,
\begin{itemize}
  \item pomset bisimulation equivalence $\sim_{p}$ is a congruence with respect to $\textrm{APTC}^{\textrm{sat}}\textrm{I}\surd$;
  \item step bisimulation equivalence $\sim_{s}$ is a congruence with respect to $\textrm{APTC}^{\textrm{sat}}\textrm{I}\surd$;
  \item hp-bisimulation equivalence $\sim_{hp}$ is a congruence with respect to $\textrm{APTC}^{\textrm{sat}}\textrm{I}\surd$.
\end{itemize}
\end{theorem}

\begin{proof}
It is easy to see that $\sim_p$, $\sim_s$, and $\sim_{hp}$ are all equivalent relations on $\textrm{APTC}^{\textrm{sat}}\textrm{I}\surd$ terms, it is only sufficient to prove that $\sim_p$, $\sim_s$, and $\sim_{hp}$ are all preserved by the operators $\sigma^p_{\textrm{abs}}$, $\upsilon^p_{\textrm{abs}}$ and $\overline{\upsilon}^p_{\textrm{abs}}$. It is trivial and we omit it.
\end{proof}

\subsubsection{Soundness}

\begin{theorem}[Soundness of $\textrm{APTC}^{\textrm{sat}}\textrm{I}\surd$]
The axiomatization of $\textrm{APTC}^{\textrm{sat}}\textrm{I}\surd$ is sound modulo truly concurrent bisimulation equivalences $\sim_{p}$, $\sim_{s}$, and $\sim_{hp}$. That is,
\begin{enumerate}
  \item let $x$ and $y$ be $\textrm{APTC}^{\textrm{sat}}\textrm{I}\surd$ terms. If $\textrm{APTC}^{\textrm{sat}}\textrm{I}\surd\vdash x=y$, then $x\sim_{s} y$;
  \item let $x$ and $y$ be $\textrm{APTC}^{\textrm{sat}}\textrm{I}\surd$ terms. If $\textrm{APTC}^{\textrm{sat}}\textrm{I}\surd\vdash x=y$, then $x\sim_{p} y$;
  \item let $x$ and $y$ be $\textrm{APTC}^{\textrm{sat}}\textrm{I}\surd$ terms. If $\textrm{APTC}^{\textrm{sat}}\textrm{I}\surd\vdash x=y$, then $x\sim_{hp} y$.
\end{enumerate}
\end{theorem}

\begin{proof}
Since $\sim_p$, $\sim_s$, and $\sim_{hp}$ are both equivalent and congruent relations, we only need to check if each axiom in Table \ref{AxiomsForAPTCSATDIA} is sound modulo $\sim_p$, $\sim_s$, and $\sim_{hp}$ respectively.

\begin{enumerate}
  \item Each axiom in Table \ref{AxiomsForAPTCSATDIA} can be checked that it is sound modulo step bisimulation equivalence, by $\lambda$-definitions in Table \ref{DefsForAPTCSATDIA}. We omit them.
  \item From the definition of pomset bisimulation, we know that pomset bisimulation is defined by pomset transitions, which are labeled by pomsets. In a pomset transition, the events (actions) in the pomset are either within causality relations (defined by $\cdot$) or in concurrency (implicitly defined by $\cdot$ and $+$, and explicitly defined by $\between$), of course, they are pairwise consistent (without conflicts). We have already proven the case that all events are pairwise concurrent (soundness modulo step bisimulation), so, we only need to prove the case of events in causality. Without loss of generality, we take a pomset of $P=\{\tilde{a},\tilde{b}:\tilde{a}\cdot \tilde{b}\}$. Then the pomset transition labeled by the above $P$ is just composed of one single event transition labeled by $\tilde{a}$ succeeded by another single event transition labeled by $\tilde{b}$, that is, $\xrightarrow{P}=\xrightarrow{a}\xrightarrow{b}$.

        Similarly to the proof of soundness modulo step bisimulation equivalence, we can prove that each axiom in Table \ref{AxiomsForAPTCSATDIA} is sound modulo pomset bisimulation equivalence, we omit them.
  \item From the definition of hp-bisimulation, we know that hp-bisimulation is defined on the posetal product $(C_1,f,C_2),f:C_1\rightarrow C_2\textrm{ isomorphism}$. Two process terms $s$ related to $C_1$ and $t$ related to $C_2$, and $f:C_1\rightarrow C_2\textrm{ isomorphism}$. Initially, $(C_1,f,C_2)=(\emptyset,\emptyset,\emptyset)$, and $(\emptyset,\emptyset,\emptyset)\in\sim_{hp}$. When $s\xrightarrow{a}s'$ ($C_1\xrightarrow{a}C_1'$), there will be $t\xrightarrow{a}t'$ ($C_2\xrightarrow{a}C_2'$), and we define $f'=f[a\mapsto a]$. Then, if $(C_1,f,C_2)\in\sim_{hp}$, then $(C_1',f',C_2')\in\sim_{hp}$.

        Similarly to the proof of soundness modulo pomset bisimulation equivalence, we can prove that each axiom in Table \ref{AxiomsForAPTCSATDIA} is sound modulo hp-bisimulation equivalence, we just need additionally to check the above conditions on hp-bisimulation, we omit them.
\end{enumerate}

\end{proof}

\subsection{Time-Dependent Conditions}

In this subsection, we will introduce $\textrm{APTC}^{\textrm{sat}}\textrm{I}\surd$ with time-dependent conditions called $\textrm{APTC}^{\textrm{sat}}\textrm{I}\surd\textrm{C}$.

\subsubsection{Basic Definition}

\begin{definition}[Time-dependent conditions]
The basic kinds of time-dependent conditions are at-time-point and at-time-point-greater-than. At-time-point $p$ $(p\in\mathbb{R}^{\geq})$ is the condition that holds only at point of $p$ and at-time-point-greater-than $p$ $(p\in\mathbb{R}^{\geq})$ is the condition that holds in all point of time greater than $p$. $\mathbf{t}$ is as the truth and $\mathbf{f}$ is as falsity.
\end{definition}

\subsubsection{The Theory $\textrm{APTC}^{\textrm{sat}}\textrm{I}\surd\textrm{C}$}

\begin{definition}[Signature of $\textrm{APTC}^{\textrm{sat}}\textrm{I}\surd\textrm{C}$]
The signature of $\textrm{APTC}^{\textrm{sat}}\textrm{I}\surd\textrm{C}$ consists of the signature of $\textrm{APTC}^{\textrm{sat}}\textrm{I}\surd$, and the at-time-point operator $\mathbf{pt}: \mathbb{R}^{\geq}\rightarrow\mathbb{B}^*$, the at-time-point-greater-than operator $\mathbf{pt}_>: \mathbb{R}^{\geq}\rightarrow\mathbb{B}^*$, the logical constants and operators $\mathbf{t}:\rightarrow\mathbb{B}^*$, $\mathbf{f}:\rightarrow\mathbb{B}^*$, $\neg: \mathbb{B}^*\rightarrow \mathbb{B}^*$, $\vee: \mathbb{B}^*\times\mathbb{B}^*\rightarrow \mathbb{B}^*$, $\wedge: \mathbb{B}^*\times\mathbb{B}^*\rightarrow \mathbb{B}^*$, the absolute initialization operator $\overline{\upsilon}_{\textrm{abs}}: \mathbb{R}^{\geq}\times\mathbb{B}^*\rightarrow\mathbb{B}^*$, the standard initial abstraction operator $\surd_s: \mathbb{R}^{\geq}.\mathbb{B}^*\rightarrow \mathbb{B}^*$, and the conditional operator $::\rightarrow: \mathbb{B}^*\times\mathcal{P}^*_{\textrm{abs}}\rightarrow\mathcal{P}^*_{\textrm{abs}}$. Where $\mathbb{B}^*$ is the sort of time-dependent conditions.
\end{definition}

The set of axioms of $\textrm{APTC}^{\textrm{sat}}\textrm{I}\surd\textrm{C}$ consists of the laws given in Table \ref{LOForAPTCSATDIAC}, Table \ref{LOForAPTCSATDIAC1} and Table \ref{LOForAPTCSATDIAC2}. Where $b$ is a condition.

\begin{center}
    \begin{table}
        \begin{tabular}{@{}ll@{}}
            \hline No. &Axiom\\
            $BOOL1$ & $\neg\mathbf{t} = \mathbf{f}$\\
            $BOOL2$ & $\neg\mathbf{f} = \mathbf{t}$\\
            $BOOL3$ & $\neg\neg b = b$\\
            $BOOL4$ & $\mathbf{t}\vee b = \mathbf{t}$\\
            $BOOL5$ & $\mathbf{f}\vee b = b$\\
            $BOOL6$ & $b\vee b'=b'\vee b$\\
            $BOOL7$ & $b\wedge b' = \neg(\neg b\vee\neg b')$\\
        \end{tabular}
        \caption{Axioms of logical operators}
        \label{LOForAPTCSATDIAC}
    \end{table}
\end{center}

\begin{center}
    \begin{table}
        \begin{tabular}{@{}ll@{}}
            \hline No. &Axiom\\
            $CSAI1$ & $\overline{\upsilon}^p_{\textrm{abs}}(\mathbf{t}) = \mathbf{t}$\\
            $CSAI2$ & $\overline{\upsilon}^p_{\textrm{abs}}(\mathbf{f}) = \mathbf{f}$\\
            $CSAI3$ & $\overline{\upsilon}^p_{\textrm{abs}}(\mathbf{pt}(p)) = \mathbf{t}$\\
            $CSAI4$ & $\overline{\upsilon}^{p+r}_{\textrm{abs}}(\mathbf{pt}(p)) = \mathbf{f}$\\
            $CSAI5$ & $\overline{\upsilon}^p_{\textrm{abs}}(\mathbf{pt}(p+r)) = \mathbf{f}$\\
            $CSAI6$ & $\overline{\upsilon}^{p+r}_{\textrm{abs}}(\mathbf{pt}_>(p)) = \mathbf{t}$\\
            $CSAI7$ & $\overline{\upsilon}^{p}_{\textrm{abs}}(\mathbf{pt}_>(p+q)) = \mathbf{f}$\\
            $CSAI8$ & $\overline{\upsilon}^p_{\textrm{abs}}(\neg b) = \neg\overline{\upsilon}^p_{\textrm{abs}}(b)$\\
            $CSAI9$ & $\overline{\upsilon}^p_{\textrm{abs}}(b \wedge b') = \overline{\upsilon}^p_{\textrm{abs}}(b)\wedge \overline{\upsilon}^p_{\textrm{abs}}(b')$\\
            $CSAI10$ & $\overline{\upsilon}^p_{\textrm{abs}}(b \vee b') = \overline{\upsilon}^p_{\textrm{abs}}(b)\vee \overline{\upsilon}^p_{\textrm{abs}}(b')$\\

            $CSIA1$ & $\surd_s v.C(v) = \surd_s w.C(w)$\\
            $CSIA2$ & $\overline{\upsilon}^p_{\textrm{abs}}(\surd_s v.C(v)) = \overline{\upsilon}^p_{\textrm{abs}}(C(p))$\\
            $CSIA3$ & $\surd_s v.(\surd_s w.E(v,w)) = \surd_s v.E(v,v)$\\
            $CSIA4$ & $b = \surd_s v.b$\\
            $CSIA5$ & $(\forall v\in\mathbb{R}^{\geq}.\overline{\upsilon}^v_{\textrm{asb}}(b)=\overline{\upsilon}^v_{\textrm{abs}}(b'))\Rightarrow b = b'$\\
            $CSIA6$ & $\neg(\surd_s v.C(v)) = \surd_s v.\neg C(v)$\\
            $CSIA7$ & $(\surd_s v.C(v))\wedge b = \surd_s v.(C(v)\wedge \overline{\upsilon}^v_{\textrm{abs}}(b))$\\
            $CSIA8$ & $(\surd_s v.C(v))\vee b = \surd_s v.(C(v)\vee \overline{\upsilon}^v_{\textrm{abs}}(b))$\\
        \end{tabular}
        \caption{Axioms of conditions$(p,q\geq 0, r>0)$}
        \label{LOForAPTCSATDIAC1}
    \end{table}
\end{center}

\begin{center}
    \begin{table}
        \begin{tabular}{@{}ll@{}}
            \hline No. &Axiom\\
            $SGC1$ & $\mathbf{t}::\rightarrow x = x$\\
            $SGC2ID$ & $\mathbf{f}::\rightarrow x = \dot{\delta}$\\
            $SASGC1$ & $\overline{\upsilon}^p_{\textrm{abs}}(b::\rightarrow x) = \overline{\upsilon}^p_{\textrm{abs}}(b)::\rightarrow\overline{\upsilon}^p_{\textrm{abs}} (x) + \sigma^p_{\textrm{abs}}(\dot{\delta})$\\
            $SASGC2$ & $x = \sum_{k\in[0,p]}(\mathbf{pt}(k+1)::\rightarrow\overline{\upsilon}^k_{\textrm{abs}}(x)) + \mathbf{pt}_>(p+1)::\rightarrow x$\\
            $SGC3ID$ & $b::\rightarrow\dot{\delta} = \dot{\delta}$\\
            $SASGC3$ & $b::\rightarrow \sigma^p_{\textrm{abs}}(x) + \sigma^p_{\textrm{abs}}(\dot{\delta}) = \surd_s v.\sigma^p_{\textrm{abs}}(\overline{\upsilon}^v_{\textrm{abs}}(b)::\rightarrow x)$\\
            $SGC4$ & $b::\rightarrow(x+y)=b::\rightarrow x + b::\rightarrow y$\\
            $SGC5$ & $b::\rightarrow x\cdot y = (b::\rightarrow x)\cdot y$\\
            $SGC6$ & $(b\vee b')::\rightarrow x = b::\rightarrow x + b'::\rightarrow x$\\
            $SGC7$ & $b::\rightarrow(b'::\rightarrow x) = (b\wedge b')::\rightarrow x$\\
            $SASGC4$ & $b::\rightarrow \upsilon^p_{\textrm{abs}}(x) = \upsilon^p_{\textrm{abs}}(b::\rightarrow x)$\\
            $SASGC5$ & $b::\rightarrow(x\parallel y) = (b::\rightarrow x)\parallel (b::\rightarrow y)$\\
            $SASGC6$ & $b::\rightarrow(x\mid y) = (b::\rightarrow x)\mid (b::\rightarrow y)$\\
            $SASGC7$ & $b::\rightarrow\Theta(x) = \Theta(b::\rightarrow x)$\\
            $SASGC8$ & $b::\rightarrow(x\triangleleft y) = (b::\rightarrow x)\triangleleft (b::\rightarrow y)$\\
            $SASGC9$ & $b::\rightarrow\partial_H(x) = \partial_H(b::\rightarrow x)$\\
            $SASGC10$ & $b::\rightarrow(\surd_s v.F(v)) = \surd_s v.(\overline{\upsilon}^v_{\textrm{abs}}(b)::\rightarrow F(v))$\\
            $SASGC11$ & $(\surd_s v.C(v))::\rightarrow x = \surd_s v.(C(v)::\rightarrow \overline{\upsilon}^v_{\textrm{abs}}(x))$\\
            $SASGC12$ & $b::\rightarrow(\int_{v\in V}F(v)) = \int_{v\in V}(b::\rightarrow F(v))$\\
            $SASGC13$ & $b::\rightarrow\nu_{\textrm{abs}}(x) = \nu_{\textrm{abs}}(b::\rightarrow x)$\\
        \end{tabular}
        \caption{Axioms of conditionals $(p\geq 0)$}
        \label{LOForAPTCSATDIAC2}
    \end{table}
\end{center}

The operational semantics of $\textrm{APTC}^{\textrm{sat}}\textrm{I}\surd\textrm{C}$ are defined by the transition rules in Table \ref{TRForAPTCSATDIAC} and Table \ref{DefsForAPTCSATDIAC}.

\begin{center}
    \begin{table}
        $$\frac{\langle x,p\rangle\xrightarrow{a}\langle x',p\rangle}{\langle \mathbf{t}::\rightarrow x, p\rangle\xrightarrow{a}\langle x',p\rangle}
        \quad\frac{\langle x,p\rangle\xrightarrow{a}\langle \surd,p\rangle}{\langle \mathbf{t}::\rightarrow x, p\rangle\xrightarrow{a}\langle \surd,p\rangle}$$

        $$\frac{\langle x,p\rangle\mapsto^{r}\langle x,p+r\rangle}{\langle \mathbf{t}::\rightarrow x, p\rangle\mapsto^{r}\langle \mathbf{t}::\rightarrow x,p+r\rangle}
        \quad\frac{\langle x, p\rangle\uparrow}{\langle \mathbf{t}::\rightarrow x,p\rangle\uparrow}
        \quad\frac{}{\langle \mathbf{f}::\rightarrow x,p\rangle\uparrow}$$
    \caption{Transition rules of $\textrm{APTC}^{\textrm{sat}}\textrm{I}\surd\textrm{C}$}
    \label{TRForAPTCSATDIAC}
    \end{table}
\end{center}

\begin{center}
    \begin{table}
        \begin{tabular}{@{}ll@{}}
            \hline
            $c::\rightarrow f = \lambda w.(c(w)::\rightarrow f(w))$ & $\neg c = \lambda w.\neg(c(w))$\\
            $\mathbf{t} = \lambda w.\mathbf{t}$ & $c\wedge d = \lambda w.(c(w)\wedge d(w))$\\
            $\mathbf{f} = \lambda w.\mathbf{f}$ & $c\vee d = \lambda w.(c(w)\vee d(w))$\\
            $\mathbf{pt}(v) = \lambda w.(\textrm{if }w+1=v \textrm{ then }t \textrm{ else } f)$ & $\overline{\upsilon}^v_{\textrm{abs}}(c) = c(v)$\\
            $\mathbf{pt}_>(v) = \lambda w.(\textrm{if }w+1>v \textrm{ then }t \textrm{ else } f)$ & $\surd_s^*(\gamma) = \lambda w.\overline{\upsilon}^w_{\textrm{abs}}(\gamma(w))$\\
        \end{tabular}
        \caption{Definitions of conditional operator on $(\mathcal{CI}/\sim)^*$}
        \label{DefsForAPTCSATDIAC}
    \end{table}
\end{center}

\subsubsection{Elimination}

\begin{definition}[Basic terms of $\textrm{APTC}^{\textrm{sat}}\textrm{I}\surd\textrm{C}$]
The set of basic terms of $\textrm{APTC}^{\textrm{sat}}\textrm{I}\surd\textrm{C}$, $\mathcal{B}(\textrm{APTC}^{\textrm{sat}}\textrm{I}\surd\textrm{C})$, is inductively defined as follows by two auxiliary sets $\mathcal{B}_0(\textrm{APTC}^{\textrm{sat}}\textrm{I}\surd\textrm{C})$ and $\mathcal{B}_1(\textrm{APTC}^{\textrm{sat}}\textrm{I}\surd\textrm{C})$:
\begin{enumerate}
  \item if $a\in A_{\delta}$, then $\tilde{a} \in \mathcal{B}_1(\textrm{APTC}^{\textrm{sat}}\textrm{I}\surd\textrm{C})$;
  \item if $a\in A$ and $t\in \mathcal{B}(\textrm{APTC}^{\textrm{sat}}\textrm{I}\surd\textrm{C})$, then $\tilde{a}\cdot t \in \mathcal{B}_1(\textrm{APTC}^{\textrm{sat}}\textrm{I}\surd\textrm{C})$;
  \item if $t,t'\in \mathcal{B}_1(\textrm{APTC}^{\textrm{sat}}\textrm{I}\surd\textrm{C})$, then $t+t'\in \mathcal{B}_1(\textrm{APTC}^{\textrm{sat}}\textrm{I}\surd\textrm{C})$;
  \item if $t,t'\in \mathcal{B}_1(\textrm{APTC}^{\textrm{sat}}\textrm{I}\surd\textrm{C})$, then $t\parallel t'\in \mathcal{B}_1(\textrm{APTC}^{\textrm{sat}}\textrm{I}\surd\textrm{C})$;
  \item if $t\in \mathcal{B}_1(\textrm{APTC}^{\textrm{sat}}\textrm{I}\surd\textrm{C})$, then $t\in \mathcal{B}_0(\textrm{APTC}^{\textrm{sat}}\textrm{I}\surd\textrm{C})$;
  \item if $p>0$ and $t\in \mathcal{B}_0(\textrm{APTC}^{\textrm{sat}}\textrm{I}\surd\textrm{C})$, then $\sigma^p_{\textrm{abs}}(t) \in \mathcal{B}_0(\textrm{APTC}^{\textrm{sat}}\textrm{I}\surd\textrm{C})$;
  \item if $p>0$, $t\in \mathcal{B}_1(\textrm{APTC}^{\textrm{sat}}\textrm{I}\surd\textrm{C})$ and $t'\in \mathcal{B}_0(\textrm{APTC}^{\textrm{sat}}\textrm{I}\surd\textrm{C})$, then $t+\sigma^p_{\textrm{abs}}(t') \in \mathcal{B}_0(\textrm{APTC}^{\textrm{sat}}\textrm{I}\surd\textrm{C})$;
  \item if $t\in \mathcal{B}_0(\textrm{APTC}^{\textrm{sat}}\textrm{I}\surd\textrm{C})$, then $\nu_{\textrm{abs}}(t) \in \mathcal{B}_0(\textrm{APTC}^{\textrm{sat}}\textrm{I}\surd\textrm{C})$;
  \item if $t\in \mathcal{B}_0(\textrm{APTC}^{\textrm{sat}}\textrm{I}\surd\textrm{C})$, then $\int_{v\in V}(t) \in \mathcal{B}_0(\textrm{APTC}^{\textrm{sat}}\textrm{I}\surd\textrm{C})$;
  \item if $s>0$ and $t\in \mathcal{B}_0(\textrm{APTC}^{\textrm{sat}}\textrm{I}\surd\textrm{C})$, then $\surd_s s.t(s) \in \mathcal{B}_0(\textrm{APTC}^{\textrm{sat}}\textrm{I}\surd\textrm{C})$;
  \item $\dot{\delta}\in \mathcal{B}(\textrm{APTC}^{\textrm{sat}}\textrm{I}\surd\textrm{C})$;
  \item if $t\in \mathcal{B}_0(\textrm{APTC}^{\textrm{sat}}\textrm{I}\surd\textrm{C})$, then $t\in \mathcal{B}(\textrm{APTC}^{\textrm{sat}}\textrm{I}\surd\textrm{C})$.
\end{enumerate}
\end{definition}

\begin{theorem}[Elimination theorem]
Let $p$ be a closed $\textrm{APTC}^{\textrm{sat}}\textrm{I}\surd\textrm{C}$ term. Then there is a basic $\textrm{APTC}^{\textrm{sat}}\textrm{I}\surd\textrm{C}$ term $q$ such that $\textrm{APTC}^{\textrm{sat}}\textrm{I}\surd\textrm{C}\vdash p=q$.
\end{theorem}

\begin{proof}
It is sufficient to induct on the structure of the closed $\textrm{APTC}^{\textrm{sat}}\textrm{I}\surd\textrm{C}$ term $p$. It can be proven that $p$ combined by the constants and operators of $\textrm{APTC}^{\textrm{sat}}\textrm{I}\surd\textrm{C}$ exists an equal basic term $q$, and the other operators not included in the basic terms, such as $\upsilon_{\textrm{abs}}$, $\overline{\upsilon}_{\textrm{abs}}$, $\between$, $\mid$, $\partial_H$, $\Theta$, $\triangleleft$, and the constants and operators related to conditions can be eliminated.
\end{proof}

\subsubsection{Congruence}

\begin{theorem}[Congruence of $\textrm{APTC}^{\textrm{sat}}\textrm{I}\surd\textrm{C}$]
Truly concurrent bisimulation equivalences $\sim_p$, $\sim_s$ and $\sim_{hp}$ are all congruences with respect to $\textrm{APTC}^{\textrm{sat}}\textrm{I}\surd\textrm{C}$. That is,
\begin{itemize}
  \item pomset bisimulation equivalence $\sim_{p}$ is a congruence with respect to $\textrm{APTC}^{\textrm{sat}}\textrm{I}\surd\textrm{C}$;
  \item step bisimulation equivalence $\sim_{s}$ is a congruence with respect to $\textrm{APTC}^{\textrm{sat}}\textrm{I}\surd\textrm{C}$;
  \item hp-bisimulation equivalence $\sim_{hp}$ is a congruence with respect to $\textrm{APTC}^{\textrm{sat}}\textrm{I}\surd\textrm{C}$.
\end{itemize}
\end{theorem}

\begin{proof}
It is easy to see that $\sim_p$, $\sim_s$, and $\sim_{hp}$ are all equivalent relations on $\textrm{APTC}^{\textrm{sat}}\textrm{I}\surd\textrm{C}$ terms, it is only sufficient to prove that $\sim_p$, $\sim_s$, and $\sim_{hp}$ are all preserved by the operators $\sigma^p_{\textrm{abs}}$, $\upsilon^p_{\textrm{abs}}$ and $\overline{\upsilon}^p_{\textrm{abs}}$. It is trivial and we omit it.
\end{proof}

\subsubsection{Soundness}

\begin{theorem}[Soundness of $\textrm{APTC}^{\textrm{sat}}\textrm{I}\surd\textrm{C}$]
The axiomatization of $\textrm{APTC}^{\textrm{sat}}\textrm{I}\surd\textrm{C}$ is sound modulo truly concurrent bisimulation equivalences $\sim_{p}$, $\sim_{s}$, and $\sim_{hp}$. That is,
\begin{enumerate}
  \item let $x$ and $y$ be $\textrm{APTC}^{\textrm{sat}}\textrm{I}\surd\textrm{C}$ terms. If $\textrm{APTC}^{\textrm{sat}}\textrm{I}\surd\textrm{C}\vdash x=y$, then $x\sim_{s} y$;
  \item let $x$ and $y$ be $\textrm{APTC}^{\textrm{sat}}\textrm{I}\surd\textrm{C}$ terms. If $\textrm{APTC}^{\textrm{sat}}\textrm{I}\surd\textrm{C}\vdash x=y$, then $x\sim_{p} y$;
  \item let $x$ and $y$ be $\textrm{APTC}^{\textrm{sat}}\textrm{I}\surd\textrm{C}$ terms. If $\textrm{APTC}^{\textrm{sat}}\textrm{I}\surd\textrm{C}\vdash x=y$, then $x\sim_{hp} y$.
\end{enumerate}
\end{theorem}

\begin{proof}
Since $\sim_p$, $\sim_s$, and $\sim_{hp}$ are both equivalent and congruent relations, we only need to check if each axiom in Table \ref{LOForAPTCSATDIAC1} and Table \ref{LOForAPTCSATDIAC2} is sound modulo $\sim_p$, $\sim_s$, and $\sim_{hp}$ respectively.

\begin{enumerate}
  \item Each axiom in Table \ref{LOForAPTCSATDIAC1}, Table \ref{LOForAPTCSATDIAC2} and Table \ref{LOForAPTCDATDIAC2} can be checked that it is sound modulo step bisimulation equivalence, by transition rules of conditionals in Table \ref{TRForAPTCSATDIAC}. We omit them.
  \item From the definition of pomset bisimulation, we know that pomset bisimulation is defined by pomset transitions, which are labeled by pomsets. In a pomset transition, the events (actions) in the pomset are either within causality relations (defined by $\cdot$) or in concurrency (implicitly defined by $\cdot$ and $+$, and explicitly defined by $\between$), of course, they are pairwise consistent (without conflicts). We have already proven the case that all events are pairwise concurrent (soundness modulo step bisimulation), so, we only need to prove the case of events in causality. Without loss of generality, we take a pomset of $P=\{\tilde{a},\tilde{b}:\tilde{a}\cdot \tilde{b}\}$. Then the pomset transition labeled by the above $P$ is just composed of one single event transition labeled by $\tilde{a}$ succeeded by another single event transition labeled by $\tilde{b}$, that is, $\xrightarrow{P}=\xrightarrow{a}\xrightarrow{b}$.

        Similarly to the proof of soundness modulo step bisimulation equivalence, we can prove that each axiom in Table \ref{LOForAPTCSATDIAC1} and Table \ref{LOForAPTCSATDIAC2} is sound modulo pomset bisimulation equivalence, we omit them.
  \item From the definition of hp-bisimulation, we know that hp-bisimulation is defined on the posetal product $(C_1,f,C_2),f:C_1\rightarrow C_2\textrm{ isomorphism}$. Two process terms $s$ related to $C_1$ and $t$ related to $C_2$, and $f:C_1\rightarrow C_2\textrm{ isomorphism}$. Initially, $(C_1,f,C_2)=(\emptyset,\emptyset,\emptyset)$, and $(\emptyset,\emptyset,\emptyset)\in\sim_{hp}$. When $s\xrightarrow{a}s'$ ($C_1\xrightarrow{a}C_1'$), there will be $t\xrightarrow{a}t'$ ($C_2\xrightarrow{a}C_2'$), and we define $f'=f[a\mapsto a]$. Then, if $(C_1,f,C_2)\in\sim_{hp}$, then $(C_1',f',C_2')\in\sim_{hp}$.

        Similarly to the proof of soundness modulo pomset bisimulation equivalence, we can prove that each axiom in Table \ref{LOForAPTCSATDIAC1} and Table \ref{LOForAPTCSATDIAC2} is sound modulo hp-bisimulation equivalence, we just need additionally to check the above conditions on hp-bisimulation, we omit them.
\end{enumerate}

\end{proof}

\subsubsection{Completeness}

\begin{theorem}[Completeness of $\textrm{APTC}^{\textrm{sat}}\textrm{I}\surd\textrm{C}$]
The axiomatization of $\textrm{APTC}^{\textrm{sat}}\textrm{I}\surd\textrm{C}$ is complete modulo truly concurrent bisimulation equivalences $\sim_{p}$, $\sim_{s}$, and $\sim_{hp}$. That is,
\begin{enumerate}
  \item let $p$ and $q$ be closed $\textrm{APTC}^{\textrm{sat}}\textrm{I}\surd\textrm{C}$ terms, if $p\sim_{s} q$ then $p=q$;
  \item let $p$ and $q$ be closed $\textrm{APTC}^{\textrm{sat}}\textrm{I}\surd\textrm{C}$ terms, if $p\sim_{p} q$ then $p=q$;
  \item let $p$ and $q$ be closed $\textrm{APTC}^{\textrm{sat}}\textrm{I}\surd\textrm{C}$ terms, if $p\sim_{hp} q$ then $p=q$.
\end{enumerate}

\end{theorem}

\begin{proof}
\begin{enumerate}
  \item Firstly, by the elimination theorem of $\textrm{APTC}^{\textrm{sat}}\textrm{I}\surd\textrm{C}$, we know that for each closed $\textrm{APTC}^{\textrm{sat}}\textrm{I}\surd\textrm{C}$ term $p$, there exists a closed basic $\textrm{APTC}^{\textrm{sat}}\textrm{I}\surd\textrm{C}$ term $p'$, such that $\textrm{APTC}^{\textrm{sat}}\textrm{I}\surd\textrm{C}\vdash p=p'$, so, we only need to consider closed basic $\textrm{APTC}^{\textrm{sat}}\textrm{I}\surd\textrm{C}$ terms.

        The basic terms modulo associativity and commutativity (AC) of conflict $+$ (defined by axioms $A1$ and $A2$ in Table \ref{AxiomsForBATCSAT}) and associativity and commutativity (AC) of parallel $\parallel$ (defined by axioms $P2$ and $P3$ in Table \ref{AxiomsForAPTCSAT}), and these equivalences is denoted by $=_{AC}$. Then, each equivalence class $s$ modulo AC of $+$ and $\parallel$ has the following normal form

        $$s_1+\cdots+ s_k$$

        with each $s_i$ either an atomic event or of the form

        $$t_1\cdot\cdots\cdot t_m$$

        with each $t_j$ either an atomic event or of the form

        $$u_1\parallel\cdots\parallel u_n$$

        with each $u_l$ an atomic event, and each $s_i$ is called the summand of $s$.

        Now, we prove that for normal forms $n$ and $n'$, if $n\sim_{s} n'$ then $n=_{AC}n'$. It is sufficient to induct on the sizes of $n$ and $n'$. We can get $n=_{AC} n'$.

        Finally, let $s$ and $t$ be basic $\textrm{APTC}^{\textrm{sat}}\textrm{I}\surd\textrm{C}$ terms, and $s\sim_s t$, there are normal forms $n$ and $n'$, such that $s=n$ and $t=n'$. The soundness theorem modulo step bisimulation equivalence yields $s\sim_s n$ and $t\sim_s n'$, so $n\sim_s s\sim_s t\sim_s n'$. Since if $n\sim_s n'$ then $n=_{AC}n'$, $s=n=_{AC}n'=t$, as desired.
  \item This case can be proven similarly, just by replacement of $\sim_{s}$ by $\sim_{p}$.
  \item This case can be proven similarly, just by replacement of $\sim_{s}$ by $\sim_{hp}$.
\end{enumerate}
\end{proof}

\section{Recursion}\label{rec}

In this section, we will introduce recursion for APTC with timing, including recursion for $\textrm{APTC}^{\textrm{drt}}$, $\textrm{APTC}^{\textrm{dat}}$, $\textrm{APTC}^{\textrm{srt}}\textrm{I}$ and $\textrm{APTC}^{\textrm{sat}}\textrm{I}$.

\subsection{Discrete Relative Timing}

As recursion for APTC in section \ref{tcpa}, we also need the concept of guardedness in recursion for APTC with timing. With the capabilities related timing, guardedness means that $X$ is always preceded by an action or delayed for at least one time slice.

\begin{definition}[Guarded recursive specification of $\textrm{APTC}^{\textrm{drt}}$+Rec]
Let $t$ be a term of $\textrm{APTC}^{\textrm{drt}}$ containing a variable $X$, an occurrence of $X$ in $t$ is guarded if $t$ has a subterm of the form $(\underline{\underline{a_1}}\parallel\cdots\parallel\underline{\underline{a_k}})\cdot t'(a_1,\cdots,a_k\in A, k\in\mathbb{N})$ or $\sigma^n_{\textrm{rel}}(t')(n>0)$ and $t'$ is a $\textrm{APTC}^{\textrm{drt}}$ term containing this occurrence of $X$.

A recursive specification over $\textrm{APTC}^{\textrm{drt}}$ is called guarded if all occurrences of variables in the right-hand sides of its equations are guarded, or it can be rewritten to such a recursive specification using the axioms of $\textrm{APTC}^{\textrm{drt}}$ and the equations of the recursive specification.
\end{definition}

\begin{definition}[Signature of $\textrm{APTC}^{\textrm{drt}}$+Rec]
The signature of $\textrm{APTC}^{\textrm{drt}}$+Rec contains the signature of $\textrm{APTC}^{\textrm{drt}}$ extended with a constant $\langle X|E\rangle:\rightarrow\mathcal{P}_{\textrm{rel}}$ for each guarded recursive specification $E$ and $X\in V(E)$.
\end{definition}

The axioms of $\textrm{APTC}^{\textrm{drt}}$+Rec consists of the axioms of $\textrm{APTC}^{\textrm{drt}}$, and RDP and RSP in Table \ref{RDPRSP}.

The additional transition rules of $\textrm{APTC}^{\textrm{drt}}$+Rec is shown in Table \ref{TRForAPTCDRTRec}.

\begin{center}
    \begin{table}
        $$\frac{t_i(\langle X_1|E\rangle,\cdots,\langle X_j|E\rangle)\xrightarrow{\{a_1,\cdots,a_k\}}\surd}{\langle X_i|E\rangle\xrightarrow{\{a_1,\cdots,a_k\}}\surd}$$
        $$\frac{t_i(\langle X_1|E\rangle,\cdots,\langle X_j|E\rangle)\xrightarrow{\{a_1,\cdots,a_k\}} y}{\langle X_i|E\rangle\xrightarrow{\{a_1,\cdots,a_k\}} y}$$
        $$\frac{t_i(\langle X_1|E\rangle,\cdots,\langle X_j|E\rangle)\mapsto^{m}y}{\langle X_i|E\rangle\mapsto^{m}y}
        \quad\frac{t_i(\langle X_1|E\rangle,\cdots,\langle X_j|E\rangle)\uparrow}{\langle X_i|E\rangle\uparrow}$$
        \caption{Transition rules of $\textrm{APTC}^{\textrm{drt}}$+Rec $(m>0,n\geq 0)$}
        \label{TRForAPTCDRTRec}
    \end{table}
\end{center}

\begin{theorem}[Generalization of $\textrm{APTC}^{\textrm{drt}}$+Rec]
By the definitions of $a_1\parallel\cdots\parallel a_k=\langle X|X=\underline{\underline{a_1}}\parallel\cdots\parallel\underline{\underline{a_k}}+\sigma^1_{\textrm{rel}}(X)\rangle$ for each $a_1,\cdots,a_k\in A, k\in\mathbb{N}$ and $\delta=\langle X|X=\underline{\underline{\delta}}+\sigma^1_{\textrm{rel}}(X)\rangle$ is a generalization of $APTC$+Rec.
\end{theorem}

\begin{proof}
It follows from the following two facts.

\begin{enumerate}
      \item The transition rules of $APTC$+Rec in section \ref{tcpa} are all source-dependent;
      \item The sources of the transition rules of $\textrm{APTC}^{\textrm{drt}}$+Rec contain an occurrence of $\dot{\delta}$, $\underline{\underline{a}}$, $\sigma^n_{\textrm{rel}}$, $\upsilon^n_{\textrm{rel}}$ and $\overline{\upsilon}^n_{\textrm{rel}}$.
\end{enumerate}

So, $APTC$+Rec is an embedding of $\textrm{APTC}^{\textrm{drt}}$+Rec, as desired.
\end{proof}

\begin{theorem}[Soundness of $\textrm{APTC}^{\textrm{drt}}$+Rec]
Let $x$ and $y$ be $\textrm{APTC}^{\textrm{drt}}$+Rec terms. If $\textrm{APTC}^{\textrm{drt}}\textrm{+Rec}\vdash x=y$, then
\begin{enumerate}
  \item $x\sim_{s} y$;
  \item $x\sim_{p} y$;
  \item $x\sim_{hp} y$.
\end{enumerate}
\end{theorem}

\begin{proof}
Since $\sim_p$, $\sim_s$, and $\sim_{hp}$ are both equivalent and congruent relations, we only need to check if each axiom in Table \ref{RDPRSP} is sound modulo $\sim_p$, $\sim_s$, and $\sim_{hp}$ respectively.

\begin{enumerate}
  \item Each axiom in Table \ref{RDPRSP} can be checked that it is sound modulo step bisimulation equivalence, by transition rules in Table \ref{TRForAPTCDRTRec}. We omit them.
  \item From the definition of pomset bisimulation, we know that pomset bisimulation is defined by pomset transitions, which are labeled by pomsets. In a pomset transition, the events (actions) in the pomset are either within causality relations (defined by $\cdot$) or in concurrency (implicitly defined by $\cdot$ and $+$, and explicitly defined by $\between$), of course, they are pairwise consistent (without conflicts). We have already proven the case that all events are pairwise concurrent (soundness modulo step bisimulation), so, we only need to prove the case of events in causality. Without loss of generality, we take a pomset of $P=\{\underline{\underline{a}},\underline{\underline{b}}:\underline{\underline{a}}\cdot \underline{\underline{b}}\}$. Then the pomset transition labeled by the above $P$ is just composed of one single event transition labeled by $\underline{\underline{a}}$ succeeded by another single event transition labeled by $\underline{\underline{b}}$, that is, $\xrightarrow{P}=\xrightarrow{a}\xrightarrow{b}$.

        Similarly to the proof of soundness modulo step bisimulation equivalence, we can prove that each axiom in Table \ref{RDPRSP} is sound modulo pomset bisimulation equivalence, we omit them.
  \item From the definition of hp-bisimulation, we know that hp-bisimulation is defined on the posetal product $(C_1,f,C_2),f:C_1\rightarrow C_2\textrm{ isomorphism}$. Two process terms $s$ related to $C_1$ and $t$ related to $C_2$, and $f:C_1\rightarrow C_2\textrm{ isomorphism}$. Initially, $(C_1,f,C_2)=(\emptyset,\emptyset,\emptyset)$, and $(\emptyset,\emptyset,\emptyset)\in\sim_{hp}$. When $s\xrightarrow{a}s'$ ($C_1\xrightarrow{a}C_1'$), there will be $t\xrightarrow{a}t'$ ($C_2\xrightarrow{a}C_2'$), and we define $f'=f[a\mapsto a]$. Then, if $(C_1,f,C_2)\in\sim_{hp}$, then $(C_1',f',C_2')\in\sim_{hp}$.

        Similarly to the proof of soundness modulo pomset bisimulation equivalence, we can prove that each axiom in Table \ref{RDPRSP} is sound modulo hp-bisimulation equivalence, we just need additionally to check the above conditions on hp-bisimulation, we omit them.
\end{enumerate}
\end{proof}

\begin{theorem}[Completeness of $\textrm{APTC}^{\textrm{drt}}$+linear Rec]
Let $p$ and $q$ be closed $\textrm{APTC}^{\textrm{drt}}$+linear Rec terms, then,
\begin{enumerate}
  \item if $p\sim_{s} q$ then $p=q$;
  \item if $p\sim_{p} q$ then $p=q$;
  \item if $p\sim_{hp} q$ then $p=q$.
\end{enumerate}
\end{theorem}

\begin{proof}
Firstly, we know that each process term in $\textrm{APTC}^{\textrm{drt}}$ with linear recursion is equal to a process term $\langle X_1|E\rangle$ with $E$ a linear recursive specification.

It remains to prove the following cases.

\begin{enumerate}
  \item If $\langle X_1|E_1\rangle \sim_s \langle Y_1|E_2\rangle$ for linear recursive specification $E_1$ and $E_2$, then $\langle X_1|E_1\rangle = \langle Y_1|E_2\rangle$.

        It can be proven similarly to the completeness of APTC + linear Rec, see \cite{ATC}.
        
  \item If $\langle X_1|E_1\rangle \sim_p \langle Y_1|E_2\rangle$ for linear recursive specification $E_1$ and $E_2$, then $\langle X_1|E_1\rangle = \langle Y_1|E_2\rangle$.

        It can be proven similarly, just by replacement of $\sim_{s}$ by $\sim_{p}$, we omit it.
  \item If $\langle X_1|E_1\rangle \sim_{hp} \langle Y_1|E_2\rangle$ for linear recursive specification $E_1$ and $E_2$, then $\langle X_1|E_1\rangle = \langle Y_1|E_2\rangle$.

        It can be proven similarly, just by replacement of $\sim_{s}$ by $\sim_{hp}$, we omit it.
\end{enumerate}
\end{proof}

\subsection{Discrete Absolute Timing}

\begin{definition}[Guarded recursive specification of $\textrm{APTC}^{\textrm{dat}}$+Rec]
Let $t$ be a term of $\textrm{APTC}^{\textrm{dat}}$ containing a variable $X$, an occurrence of $X$ in $t$ is guarded if $t$ has a subterm of the form $(\underline{a_1}\parallel\cdots\parallel\underline{a_k})\cdot t'(a_1,\cdots,a_k\in A, k\in\mathbb{N})$, $\sigma^n_{\textrm{abs}}(t')$ or $\sigma^n_{\textrm{abs}}(s)\cdot t'(n>0)$ and $s,t'$ is a $\textrm{APTC}^{\textrm{dat}}$ term, with $t'$ containing this occurrence of $X$.

A recursive specification over $\textrm{APTC}^{\textrm{dat}}$ is called guarded if all occurrences of variables in the right-hand sides of its equations are guarded, or it can be rewritten to such a recursive specification using the axioms of $\textrm{APTC}^{\textrm{dat}}$ and the equations of the recursive specification.
\end{definition}

\begin{definition}[Signature of $\textrm{APTC}^{\textrm{dat}}$+Rec]
The signature of $\textrm{APTC}^{\textrm{dat}}$+Rec contains the signature of $\textrm{APTC}^{\textrm{dat}}$ extended with a constant $\langle X|E\rangle:\rightarrow\mathcal{P}_{\textrm{abs}}$ for each guarded recursive specification $E$ and $X\in V(E)$.
\end{definition}

The axioms of $\textrm{APTC}^{\textrm{dat}}$+Rec consists of the axioms of $\textrm{APTC}^{\textrm{dat}}$, and RDP and RSP in Table \ref{RDPRSP}.

The additional transition rules of $\textrm{APTC}^{\textrm{dat}}$+Rec is shown in Table \ref{TRForAPTCDATRec}.

\begin{center}
    \begin{table}
        $$\frac{\langle t_i(\langle X_1|E\rangle,\cdots,\langle X_j|E\rangle),n\rangle\xrightarrow{\{a_1,\cdots,a_k\}}\langle\surd,n\rangle}{\langle\langle X_i|E,n\rangle\rangle\xrightarrow{\{a_1,\cdots,a_k\}}\langle\surd,n\rangle}$$
        $$\frac{\langle t_i(\langle X_1|E\rangle,\cdots,\langle X_j|E\rangle),n\rangle\xrightarrow{\{a_1,\cdots,a_k\}} \langle y,n\rangle}{\langle\langle X_i|E\rangle,n\rangle\xrightarrow{\{a_1,\cdots,a_k\}} \langle y,n\rangle}$$
        $$\frac{\langle t_i(\langle X_1|E\rangle,\cdots,\langle X_j|E\rangle),n\rangle\mapsto^{m}\langle t_i(\langle X_1|E\rangle,\cdots,\langle X_j|E\rangle),n+m\rangle}{\langle\langle X_i|E\rangle,n\rangle\mapsto^{m}\langle\langle X_i|E\rangle,n+m\rangle}
        \quad\frac{\langle t_i(\langle X_1|E\rangle,\cdots,\langle X_j|E\rangle),n\rangle\uparrow}{\langle\langle X_i|E\rangle,n\rangle\uparrow}$$
        \caption{Transition rules of $\textrm{APTC}^{\textrm{dat}}$+Rec $(m>0,n\geq 0)$}
        \label{TRForAPTCDATRec}
    \end{table}
\end{center}

\begin{theorem}[Generalization of $\textrm{APTC}^{\textrm{dat}}$+Rec]
By the definitions of $a_1\parallel\cdots\parallel a_k=\langle X|X=\underline{a_1}\parallel\cdots\parallel\underline{a_k}+\sigma^1_{\textrm{abs}}(X)\rangle$ for each $a_1,\cdots,a_k\in A, k\in\mathbb{N}$ and $\delta=\langle X|X=\underline{\delta}+\sigma^1_{\textrm{abs}}(X)\rangle$ is a generalization of $APTC$+Rec.

\end{theorem}

\begin{proof}
It follows from the following two facts.

\begin{enumerate}
      \item The transition rules of $APTC$+Rec in section \ref{tcpa} are all source-dependent;
      \item The sources of the transition rules of $\textrm{APTC}^{\textrm{dat}}$+Rec contain an occurrence of $\dot{\delta}$, $\underline{a}$, $\sigma^n_{\textrm{abs}}$, $\upsilon^n_{\textrm{abs}}$ and $\overline{\upsilon}^n_{\textrm{abs}}$.
\end{enumerate}

So, $APTC$+Rec is an embedding of $\textrm{APTC}^{\textrm{dat}}$+Rec, as desired.
\end{proof}

\begin{theorem}[Soundness of $\textrm{APTC}^{\textrm{dat}}$+Rec]
Let $x$ and $y$ be $\textrm{APTC}^{\textrm{dat}}$+Rec terms. If $\textrm{APTC}^{\textrm{dat}}\textrm{+Rec}\vdash x=y$, then
\begin{enumerate}
  \item $x\sim_{s} y$;
  \item $x\sim_{p} y$;
  \item $x\sim_{hp} y$.
\end{enumerate}
\end{theorem}

\begin{proof}
Since $\sim_p$, $\sim_s$, and $\sim_{hp}$ are both equivalent and congruent relations, we only need to check if each axiom in Table \ref{RDPRSP} is sound modulo $\sim_p$, $\sim_s$, and $\sim_{hp}$ respectively.

\begin{enumerate}
  \item Each axiom in Table \ref{RDPRSP} can be checked that it is sound modulo step bisimulation equivalence, by transition rules in Table \ref{TRForAPTCDATRec}. We omit them.
  \item From the definition of pomset bisimulation, we know that pomset bisimulation is defined by pomset transitions, which are labeled by pomsets. In a pomset transition, the events (actions) in the pomset are either within causality relations (defined by $\cdot$) or in concurrency (implicitly defined by $\cdot$ and $+$, and explicitly defined by $\between$), of course, they are pairwise consistent (without conflicts). We have already proven the case that all events are pairwise concurrent (soundness modulo step bisimulation), so, we only need to prove the case of events in causality. Without loss of generality, we take a pomset of $P=\{\underline{a},\underline{b}:\underline{a}\cdot \underline{b}\}$. Then the pomset transition labeled by the above $P$ is just composed of one single event transition labeled by $\underline{a}$ succeeded by another single event transition labeled by $\underline{b}$, that is, $\xrightarrow{P}=\xrightarrow{a}\xrightarrow{b}$.

        Similarly to the proof of soundness modulo step bisimulation equivalence, we can prove that each axiom in Table \ref{RDPRSP} is sound modulo pomset bisimulation equivalence, we omit them.
  \item From the definition of hp-bisimulation, we know that hp-bisimulation is defined on the posetal product $(C_1,f,C_2),f:C_1\rightarrow C_2\textrm{ isomorphism}$. Two process terms $s$ related to $C_1$ and $t$ related to $C_2$, and $f:C_1\rightarrow C_2\textrm{ isomorphism}$. Initially, $(C_1,f,C_2)=(\emptyset,\emptyset,\emptyset)$, and $(\emptyset,\emptyset,\emptyset)\in\sim_{hp}$. When $s\xrightarrow{a}s'$ ($C_1\xrightarrow{a}C_1'$), there will be $t\xrightarrow{a}t'$ ($C_2\xrightarrow{a}C_2'$), and we define $f'=f[a\mapsto a]$. Then, if $(C_1,f,C_2)\in\sim_{hp}$, then $(C_1',f',C_2')\in\sim_{hp}$.

        Similarly to the proof of soundness modulo pomset bisimulation equivalence, we can prove that each axiom in Table \ref{RDPRSP} is sound modulo hp-bisimulation equivalence, we just need additionally to check the above conditions on hp-bisimulation, we omit them.
\end{enumerate}
\end{proof}

\begin{theorem}[Completeness of $\textrm{APTC}^{\textrm{dat}}$+linear Rec]
Let $p$ and $q$ be closed $\textrm{APTC}^{\textrm{dat}}$+linear Rec terms, then,
\begin{enumerate}
  \item if $p\sim_{s} q$ then $p=q$;
  \item if $p\sim_{p} q$ then $p=q$;
  \item if $p\sim_{hp} q$ then $p=q$.
\end{enumerate}
\end{theorem}

\begin{proof}
Firstly, we know that each process term in $\textrm{APTC}^{\textrm{dat}}$ with linear recursion is equal to a process term $\langle X_1|E\rangle$ with $E$ a linear recursive specification.

It remains to prove the following cases.

\begin{enumerate}
  \item If $\langle X_1|E_1\rangle \sim_s \langle Y_1|E_2\rangle$ for linear recursive specification $E_1$ and $E_2$, then $\langle X_1|E_1\rangle = \langle Y_1|E_2\rangle$.

        It can be proven similarly to the completeness of APTC + linear Rec, see \cite{ATC}.

  \item If $\langle X_1|E_1\rangle \sim_p \langle Y_1|E_2\rangle$ for linear recursive specification $E_1$ and $E_2$, then $\langle X_1|E_1\rangle = \langle Y_1|E_2\rangle$.

        It can be proven similarly, just by replacement of $\sim_{s}$ by $\sim_{p}$, we omit it.
  \item If $\langle X_1|E_1\rangle \sim_{hp} \langle Y_1|E_2\rangle$ for linear recursive specification $E_1$ and $E_2$, then $\langle X_1|E_1\rangle = \langle Y_1|E_2\rangle$.

        It can be proven similarly, just by replacement of $\sim_{s}$ by $\sim_{hp}$, we omit it.
\end{enumerate}
\end{proof}

\subsection{Continuous Relative Timing}

As recursion for APTC in section \ref{tcpa}, we also need the concept of guardedness in recursion for APTC with timing. With the capabilities related timing, guardedness means that $X$ is always preceded by an action or delayed for a period of time greater than 0.

\begin{definition}[Guarded recursive specification of $\textrm{APTC}^{\textrm{srt}}\textrm{I}$+Rec]
Let $t$ be a term of $\textrm{APTC}^{\textrm{srt}}\textrm{I}$ containing a variable $X$, an occurrence of $X$ in $t$ is guarded if $t$ has a subterm of the form $(\tilde{\tilde{a_1}}\parallel\cdots\parallel\tilde{\tilde{a_k}})\cdot t'(a_1,\cdots,a_k\in A, k\in\mathbb{N})$ or $\sigma^p_{\textrm{rel}}(t')(p>0)$ and $t'$ is a $\textrm{APTC}^{\textrm{srt}}\textrm{I}$ term containing this occurrence of $X$.

A recursive specification over $\textrm{APTC}^{\textrm{srt}}\textrm{I}$ is called guarded if all occurrences of variables in the right-hand sides of its equations are guarded, or it can be rewritten to such a recursive specification using the axioms of $\textrm{APTC}^{\textrm{srt}}\textrm{I}$ and the equations of the recursive specification.
\end{definition}

\begin{definition}[Signature of $\textrm{APTC}^{\textrm{srt}}\textrm{I}$+Rec]
The signature of $\textrm{APTC}^{\textrm{srt}}\textrm{I}$+Rec contains the signature of $\textrm{APTC}^{\textrm{srt}}\textrm{I}$ extended with a constant $\langle X|E\rangle:\rightarrow\mathcal{P}_{\textrm{rel}}$ for each guarded recursive specification $E$ and $X\in V(E)$.
\end{definition}

The axioms of $\textrm{APTC}^{\textrm{srt}}\textrm{I}$+Rec consists of the axioms of $\textrm{APTC}^{\textrm{srt}}\textrm{I}$, and RDP and RSP in Table \ref{RDPRSP}.

The additional transition rules of $\textrm{APTC}^{\textrm{srt}}\textrm{I}$+Rec is shown in Table \ref{TRForAPTCSRTRec}.

\begin{center}
    \begin{table}
        $$\frac{t_i(\langle X_1|E\rangle,\cdots,\langle X_j|E\rangle)\xrightarrow{\{a_1,\cdots,a_k\}}\surd}{\langle X_i|E\rangle\xrightarrow{\{a_1,\cdots,a_k\}}\surd}$$
        $$\frac{t_i(\langle X_1|E\rangle,\cdots,\langle X_j|E\rangle)\xrightarrow{\{a_1,\cdots,a_k\}} y}{\langle X_i|E\rangle\xrightarrow{\{a_1,\cdots,a_k\}} y}$$
        $$\frac{t_i(\langle X_1|E\rangle,\cdots,\langle X_j|E\rangle)\mapsto^{r}y}{\langle X_i|E\rangle\mapsto^{r}y}
        \quad\frac{t_i(\langle X_1|E\rangle,\cdots,\langle X_j|E\rangle)\uparrow}{\langle X_i|E\rangle\uparrow}$$
        \caption{Transition rules of $\textrm{APTC}^{\textrm{srt}}\textrm{I}$+Rec $(r>0)$}
        \label{TRForAPTCSRTRec}
    \end{table}
\end{center}

\begin{theorem}[Generalization of $\textrm{APTC}^{\textrm{srt}}\textrm{I}$+Rec]
By the definitions of $a_1\parallel\cdots\parallel a_k=\langle X|X=\int_{v\in[0,\infty)}\sigma^v_{\textrm{rel}}(\tilde{\tilde{a_1}})\parallel\cdots\parallel\int_{v\in[0,\infty)}\sigma^v_{\textrm{rel}}(\tilde{\tilde{a_k}})+\sigma^r_{\textrm{rel}}(X)\rangle$ for each $a_1,\cdots,a_k\in A, k\in\mathbb{N}$, $r>0$, and $\delta=\langle X|X=\int_{v\in[0,\infty)}\sigma^v_{\textrm{rel}}(\tilde{\tilde{\delta}})+\sigma^r_{\textrm{rel}}(X)\rangle$ ($r>0$) is a generalization of $APTC$+Rec.

\end{theorem}

\begin{proof}
It follows from the following two facts.

\begin{enumerate}
      \item The transition rules of $APTC$+Rec in section \ref{tcpa} are all source-dependent;
      \item The sources of the transition rules of $\textrm{APTC}^{\textrm{srt}}\textrm{I}$+Rec contain an occurrence of $\dot{\delta}$, $\tilde{\tilde{a}}$, $\sigma^p_{\textrm{rel}}$, $\upsilon^p_{\textrm{rel}}$, $\overline{\upsilon}^p_{\textrm{rel}}$ and $\int$.
\end{enumerate}

So, $APTC$+Rec is an embedding of $\textrm{APTC}^{\textrm{srt}}\textrm{I}$+Rec, as desired.
\end{proof}

\begin{theorem}[Soundness of $\textrm{APTC}^{\textrm{srt}}\textrm{I}$+Rec]
Let $x$ and $y$ be $\textrm{APTC}^{\textrm{srt}}\textrm{I}$+Rec terms. If $\textrm{APTC}^{\textrm{srt}}\textrm{I}\textrm{+Rec}\vdash x=y$, then
\begin{enumerate}
  \item $x\sim_{s} y$;
  \item $x\sim_{p} y$;
  \item $x\sim_{hp} y$.
\end{enumerate}
\end{theorem}

\begin{proof}
Since $\sim_p$, $\sim_s$, and $\sim_{hp}$ are both equivalent and congruent relations, we only need to check if each axiom in Table \ref{RDPRSP} is sound modulo $\sim_p$, $\sim_s$, and $\sim_{hp}$ respectively.

\begin{enumerate}
  \item Each axiom in Table \ref{RDPRSP} can be checked that it is sound modulo step bisimulation equivalence, by transition rules in Table \ref{TRForAPTCSRTRec}. We omit them.
  \item From the definition of pomset bisimulation, we know that pomset bisimulation is defined by pomset transitions, which are labeled by pomsets. In a pomset transition, the events (actions) in the pomset are either within causality relations (defined by $\cdot$) or in concurrency (implicitly defined by $\cdot$ and $+$, and explicitly defined by $\between$), of course, they are pairwise consistent (without conflicts). We have already proven the case that all events are pairwise concurrent (soundness modulo step bisimulation), so, we only need to prove the case of events in causality. Without loss of generality, we take a pomset of $P=\{\tilde{\tilde{a}},\tilde{\tilde{b}}:\tilde{\tilde{a}}\cdot \tilde{\tilde{b}}\}$. Then the pomset transition labeled by the above $P$ is just composed of one single event transition labeled by $\tilde{\tilde{a}}$ succeeded by another single event transition labeled by $\tilde{\tilde{b}}$, that is, $\xrightarrow{P}=\xrightarrow{a}\xrightarrow{b}$.

        Similarly to the proof of soundness modulo step bisimulation equivalence, we can prove that each axiom in Table \ref{RDPRSP} is sound modulo pomset bisimulation equivalence, we omit them.
  \item From the definition of hp-bisimulation, we know that hp-bisimulation is defined on the posetal product $(C_1,f,C_2),f:C_1\rightarrow C_2\textrm{ isomorphism}$. Two process terms $s$ related to $C_1$ and $t$ related to $C_2$, and $f:C_1\rightarrow C_2\textrm{ isomorphism}$. Initially, $(C_1,f,C_2)=(\emptyset,\emptyset,\emptyset)$, and $(\emptyset,\emptyset,\emptyset)\in\sim_{hp}$. When $s\xrightarrow{a}s'$ ($C_1\xrightarrow{a}C_1'$), there will be $t\xrightarrow{a}t'$ ($C_2\xrightarrow{a}C_2'$), and we define $f'=f[a\mapsto a]$. Then, if $(C_1,f,C_2)\in\sim_{hp}$, then $(C_1',f',C_2')\in\sim_{hp}$.

        Similarly to the proof of soundness modulo pomset bisimulation equivalence, we can prove that each axiom in Table \ref{RDPRSP} is sound modulo hp-bisimulation equivalence, we just need additionally to check the above conditions on hp-bisimulation, we omit them.
\end{enumerate}
\end{proof}

\begin{theorem}[Completeness of $\textrm{APTC}^{\textrm{srt}}\textrm{I}$+linear Rec]
Let $p$ and $q$ be closed $\textrm{APTC}^{\textrm{srt}}\textrm{I}$+linear Rec terms, then,
\begin{enumerate}
  \item if $p\sim_{s} q$ then $p=q$;
  \item if $p\sim_{p} q$ then $p=q$;
  \item if $p\sim_{hp} q$ then $p=q$.
\end{enumerate}
\end{theorem}

\begin{proof}
Firstly, we know that each process term in $\textrm{APTC}^{\textrm{srt}}$ with linear recursion is equal to a process term $\langle X_1|E\rangle$ with $E$ a linear recursive specification.

It remains to prove the following cases.

\begin{enumerate}
  \item If $\langle X_1|E_1\rangle \sim_s \langle Y_1|E_2\rangle$ for linear recursive specification $E_1$ and $E_2$, then $\langle X_1|E_1\rangle = \langle Y_1|E_2\rangle$.

        It can be proven similarly to the completeness of APTC + linear Rec, see \cite{ATC}.

  \item If $\langle X_1|E_1\rangle \sim_p \langle Y_1|E_2\rangle$ for linear recursive specification $E_1$ and $E_2$, then $\langle X_1|E_1\rangle = \langle Y_1|E_2\rangle$.

        It can be proven similarly, just by replacement of $\sim_{s}$ by $\sim_{p}$, we omit it.
  \item If $\langle X_1|E_1\rangle \sim_{hp} \langle Y_1|E_2\rangle$ for linear recursive specification $E_1$ and $E_2$, then $\langle X_1|E_1\rangle = \langle Y_1|E_2\rangle$.

        It can be proven similarly, just by replacement of $\sim_{s}$ by $\sim_{hp}$, we omit it.
\end{enumerate}
\end{proof}

\subsection{Continuous Absolute Timing}

\begin{definition}[Guarded recursive specification of $\textrm{APTC}^{\textrm{sat}}\textrm{I}$+Rec]
Let $t$ be a term of $\textrm{APTC}^{\textrm{sat}}\textrm{I}$ containing a variable $X$, an occurrence of $X$ in $t$ is guarded if $t$ has a subterm of the form $(\tilde{a_1}\parallel\cdots\parallel\tilde{a_k})\cdot t'(a_1,\cdots,a_k\in A, k\in\mathbb{N})$, $\sigma^p_{\textrm{abs}}(t')$ or $\sigma^p_{\textrm{abs}}(s)\cdot t'(p>0)$ and $s,t'$ is a $\textrm{APTC}^{\textrm{sat}}\textrm{I}$ term, with $t'$ containing this occurrence of $X$.

A recursive specification over $\textrm{APTC}^{\textrm{sat}}\textrm{I}$ is called guarded if all occurrences of variables in the right-hand sides of its equations are guarded, or it can be rewritten to such a recursive specification using the axioms of $\textrm{APTC}^{\textrm{sat}}\textrm{I}$ and the equations of the recursive specification.
\end{definition}

\begin{definition}[Signature of $\textrm{APTC}^{\textrm{sat}}\textrm{I}$+Rec]
The signature of $\textrm{APTC}^{\textrm{sat}}\textrm{I}$+Rec contains the signature of $\textrm{APTC}^{\textrm{sat}}\textrm{I}$ extended with a constant $\langle X|E\rangle:\rightarrow\mathcal{P}_{\textrm{abs}}$ for each guarded recursive specification $E$ and $X\in V(E)$.
\end{definition}

The axioms of $\textrm{APTC}^{\textrm{sat}}\textrm{I}$+Rec consists of the axioms of $\textrm{APTC}^{\textrm{sat}}\textrm{I}$, and RDP and RSP in Table \ref{RDPRSP}.

The additional transition rules of $\textrm{APTC}^{\textrm{sat}}\textrm{I}$+Rec is shown in Table \ref{TRForAPTCSATRec}.

\begin{center}
    \begin{table}
        $$\frac{\langle t_i(\langle X_1|E\rangle,\cdots,\langle X_j|E\rangle),p\rangle\xrightarrow{\{a_1,\cdots,a_k\}}\langle\surd,p\rangle}{\langle\langle X_i|E,p\rangle\rangle\xrightarrow{\{a_1,\cdots,a_k\}}\langle\surd,p\rangle}$$
        $$\frac{\langle t_i(\langle X_1|E\rangle,\cdots,\langle X_j|E\rangle),p\rangle\xrightarrow{\{a_1,\cdots,a_k\}} \langle y,p\rangle}{\langle\langle X_i|E\rangle,p\rangle\xrightarrow{\{a_1,\cdots,a_k\}} \langle y,p\rangle}$$
        $$\frac{\langle t_i(\langle X_1|E\rangle,\cdots,\langle X_j|E\rangle),p\rangle\mapsto^{r}\langle t_i(\langle X_1|E\rangle,\cdots,\langle X_j|E\rangle),p+r\rangle}{\langle\langle X_i|E\rangle,p\rangle\mapsto^{r}\langle\langle X_i|E\rangle,p+r\rangle}
        \quad\frac{\langle t_i(\langle X_1|E\rangle,\cdots,\langle X_j|E\rangle),p\rangle\uparrow}{\langle\langle X_i|E\rangle,p\rangle\uparrow}$$
        \caption{Transition rules of $\textrm{APTC}^{\textrm{sat}}\textrm{I}$+Rec $(r>0,p\geq 0)$}
        \label{TRForAPTCSATRec}
    \end{table}
\end{center}

\begin{theorem}[Generalization of $\textrm{APTC}^{\textrm{sat}}\textrm{I}$+Rec]
By the definitions of $a_1\parallel\cdots\parallel a_k=\langle X|X=\int_{v\in[0,\infty)}\sigma^v_{\textrm{abs}}(\tilde{a_1})\parallel\cdots\parallel\int_{v\in[0,\infty)}\sigma^v_{\textrm{abs}}(\tilde{a_k})+\sigma^r_{\textrm{abs}}(X)\rangle$ for each $a_1,\cdots,a_k\in A, k\in\mathbb{N}, r>0$ and $\delta=\langle X|X=\int_{v\in[0,\infty)}\sigma^v_{\textrm{abs}}(\tilde{\delta})+\sigma^v_{\textrm{abs}}(X)\rangle$ ($r>0$) is a generalization of $APTC$+Rec.
\end{theorem}

\begin{proof}
It follows from the following two facts.

\begin{enumerate}
      \item The transition rules of $APTC$+Rec in section \ref{tcpa} are all source-dependent;
      \item The sources of the transition rules of $\textrm{APTC}^{\textrm{sat}}\textrm{I}$+Rec contain an occurrence of $\dot{\delta}$, $\tilde{a}$, $\sigma^p_{\textrm{abs}}$, $\upsilon^p_{\textrm{abs}}$, $\overline{\upsilon}^p_{\textrm{abs}}$ and $\int$.
\end{enumerate}

So, $APTC$+Rec is an embedding of $\textrm{APTC}^{\textrm{sat}}\textrm{I}$+Rec, as desired.
\end{proof}

\begin{theorem}[Soundness of $\textrm{APTC}^{\textrm{sat}}\textrm{I}$+Rec]
Let $x$ and $y$ be $\textrm{APTC}^{\textrm{sat}}\textrm{I}$+Rec terms. If $\textrm{APTC}^{\textrm{sat}}\textrm{I}\textrm{+Rec}\vdash x=y$, then
\begin{enumerate}
  \item $x\sim_{s} y$;
  \item $x\sim_{p} y$;
  \item $x\sim_{hp} y$.
\end{enumerate}
\end{theorem}

\begin{proof}
Since $\sim_p$, $\sim_s$, and $\sim_{hp}$ are both equivalent and congruent relations, we only need to check if each axiom in Table \ref{RDPRSP} is sound modulo $\sim_p$, $\sim_s$, and $\sim_{hp}$ respectively.

\begin{enumerate}
  \item Each axiom in Table \ref{RDPRSP} can be checked that it is sound modulo step bisimulation equivalence, by transition rules in Table \ref{TRForAPTCSATRec}. We omit them.
  \item From the definition of pomset bisimulation, we know that pomset bisimulation is defined by pomset transitions, which are labeled by pomsets. In a pomset transition, the events (actions) in the pomset are either within causality relations (defined by $\cdot$) or in concurrency (implicitly defined by $\cdot$ and $+$, and explicitly defined by $\between$), of course, they are pairwise consistent (without conflicts). We have already proven the case that all events are pairwise concurrent (soundness modulo step bisimulation), so, we only need to prove the case of events in causality. Without loss of generality, we take a pomset of $P=\{\tilde{a},\tilde{b}:\tilde{a}\cdot \tilde{b}\}$. Then the pomset transition labeled by the above $P$ is just composed of one single event transition labeled by $\tilde{a}$ succeeded by another single event transition labeled by $\tilde{b}$, that is, $\xrightarrow{P}=\xrightarrow{a}\xrightarrow{b}$.

        Similarly to the proof of soundness modulo step bisimulation equivalence, we can prove that each axiom in Table \ref{RDPRSP} is sound modulo pomset bisimulation equivalence, we omit them.
  \item From the definition of hp-bisimulation, we know that hp-bisimulation is defined on the posetal product $(C_1,f,C_2),f:C_1\rightarrow C_2\textrm{ isomorphism}$. Two process terms $s$ related to $C_1$ and $t$ related to $C_2$, and $f:C_1\rightarrow C_2\textrm{ isomorphism}$. Initially, $(C_1,f,C_2)=(\emptyset,\emptyset,\emptyset)$, and $(\emptyset,\emptyset,\emptyset)\in\sim_{hp}$. When $s\xrightarrow{a}s'$ ($C_1\xrightarrow{a}C_1'$), there will be $t\xrightarrow{a}t'$ ($C_2\xrightarrow{a}C_2'$), and we define $f'=f[a\mapsto a]$. Then, if $(C_1,f,C_2)\in\sim_{hp}$, then $(C_1',f',C_2')\in\sim_{hp}$.

        Similarly to the proof of soundness modulo pomset bisimulation equivalence, we can prove that each axiom in Table \ref{RDPRSP} is sound modulo hp-bisimulation equivalence, we just need additionally to check the above conditions on hp-bisimulation, we omit them.
\end{enumerate}

\end{proof}

\begin{theorem}[Completeness of $\textrm{APTC}^{\textrm{sat}}\textrm{I}$+linear Rec]
Let $p$ and $q$ be closed $\textrm{APTC}^{\textrm{sat}}\textrm{I}$+linear Rec terms, then,
\begin{enumerate}
  \item if $p\sim_{s} q$ then $p=q$;
  \item if $p\sim_{p} q$ then $p=q$;
  \item if $p\sim_{hp} q$ then $p=q$.
\end{enumerate}
\end{theorem}

\begin{proof}
Firstly, we know that each process term in $\textrm{APTC}^{\textrm{sat}}$ with linear recursion is equal to a process term $\langle X_1|E\rangle$ with $E$ a linear recursive specification.

It remains to prove the following cases.

\begin{enumerate}
  \item If $\langle X_1|E_1\rangle \sim_s \langle Y_1|E_2\rangle$ for linear recursive specification $E_1$ and $E_2$, then $\langle X_1|E_1\rangle = \langle Y_1|E_2\rangle$.

        It can be proven similarly to the completeness of APTC + linear Rec, see \cite{ATC}.

  \item If $\langle X_1|E_1\rangle \sim_p \langle Y_1|E_2\rangle$ for linear recursive specification $E_1$ and $E_2$, then $\langle X_1|E_1\rangle = \langle Y_1|E_2\rangle$.

        It can be proven similarly, just by replacement of $\sim_{s}$ by $\sim_{p}$, we omit it.
  \item If $\langle X_1|E_1\rangle \sim_{hp} \langle Y_1|E_2\rangle$ for linear recursive specification $E_1$ and $E_2$, then $\langle X_1|E_1\rangle = \langle Y_1|E_2\rangle$.

        It can be proven similarly, just by replacement of $\sim_{s}$ by $\sim_{hp}$, we omit it.
\end{enumerate}
\end{proof}

\section{Abstraction}\label{abs}

In this section, we will introduce something about silent step $\tau$ and abstraction $\tau_I$. The version of abstraction of APTC without timing, please refer to section \ref{tcpa}. We will introduce $\textrm{APTC}^{\textrm{drt}}$ with abstraction called $\textrm{APTC}^{\textrm{drt}}_{\tau}$, $\textrm{APTC}^{\textrm{dat}}$ with abstraction called $\textrm{APTC}^{\textrm{dat}}_{\tau}$, $\textrm{APTC}^{\textrm{srt}}$ with abstraction called $\textrm{APTC}^{\textrm{srt}}_{\tau}$, and $\textrm{APTC}^{\textrm{sat}}$ with abstraction called $\textrm{APTC}^{\textrm{sat}}_{\tau}$, respectively.

\subsection{Discrete Relative Timing}

\begin{definition}[Rooted branching truly concurrent bisimulations]
The following two conditions related timing should be added into the concepts of branching truly concurrent bisimulations in section \ref{tcpa}:
\begin{enumerate}
  \item if $C_1\mapsto^1 C_2$, then there are $C_1^*,C_2'$ such that $C_1'\Rightarrow C_1^*\mapsto^1 C_2'$, and $(C_1,C_1^*)\in R$ and $(C_2,C_2')\in R$, or $(C_1,f[\emptyset\mapsto\emptyset],C_1^*)\in R$ and $(C_2,f[\emptyset\mapsto\emptyset],C_2')\in R$;
  \item if $C_1\uparrow$, then $C_1'\uparrow$.
\end{enumerate}

And the following root conditions related timing should be added into the concepts of rooted branching truly concurrent bisimulations in section \ref{tcpa}:
\begin{enumerate}
  \item if $C_1\mapsto^m C_2(m>0)$, then there is $C_2'$ such that $C_1'\mapsto^m C_2'$, and $(C_2,C_2')\in R$, or $(C_2,f[\emptyset\mapsto\emptyset],C_2')\in R$.
\end{enumerate}
\end{definition}

\begin{definition}[Signature of $\textrm{APTC}^{\textrm{drt}}_{\tau}$]
The signature of $\textrm{APTC}^{\textrm{drt}}_{\tau}$ consists of the signature of $\textrm{APTC}^{\textrm{drt}}$, and the undelayable silent step constant $\underline{\underline{\tau}}: \rightarrow\mathcal{P}_{\textrm{rel}}$, and the abstraction operator $\tau_I: \mathcal{P}_{\textrm{rel}}\rightarrow\mathcal{P}_{\textrm{rel}}$ for $I\subseteq A$.
\end{definition}

The axioms of $\textrm{APTC}^{\textrm{drt}}_{\tau}$ include the laws in Table \ref{AxiomsForAPTCDRT} covering the case that $a\equiv\tau$ and $b\equiv\tau$, and the axioms in Table \ref{AxiomsForAPTCDRTTau}.

\begin{center}
\begin{table}
  \begin{tabular}{@{}ll@{}}
\hline No. &Axiom\\
  $DRTB1$ & $x\cdot(\underline{\underline{\tau}}\cdot(\upsilon^1_{\textrm{rel}}(y)+z+\underline{\underline{\delta}})+\upsilon^1_{\textrm{rel}}(y)) = x\cdot(\upsilon^1_{\textrm{rel}}(y)+z+\underline{\underline{\delta}})$\\
  $DRTB2$ & $x\cdot(\underline{\underline{\tau}}\cdot(\upsilon^1_{\textrm{rel}}(y)+z+\underline{\underline{\delta}})+z) = x\cdot(\upsilon^1_{\textrm{rel}}(y)+z+\underline{\underline{\delta}})$\\
  $DRTB3$ & $x\cdot(\sigma^1_{
  \textrm{rel}}(\underline{\underline{\tau}}\cdot (y+\underline{\underline{\delta}})+\upsilon^1_{\textrm{rel}}(z)) = x\cdot(\sigma^1_{\textrm{rel}}(y+\underline{\underline{\delta}})+\upsilon^1_{\textrm{rel}}(z))$\\
  $B3$ & $x\parallel\underline{\underline{\tau}}=x$\\
  $TI0$ & $\tau_I(\dot{\delta}) = \dot{\delta}$\\
  $TI1$ & $a\notin I\quad \tau_I(\underline{\underline{a}})=\underline{\underline{a}}$\\
  $TI2$ & $a\in I\quad \tau_I(\underline{\underline{a}})=\underline{\underline{\tau}}$\\
  $DRTI$ & $\tau_I(\sigma^n_{\textrm{rel}}(x)) = \sigma^n_{\textrm{rel}}(\tau_I(x))$\\
  $TI4$ & $\tau_I(x+y)=\tau_I(x)+\tau_I(y)$\\
  $TI5$ & $\tau_I(x\cdot y)=\tau_I(x)\cdot\tau_I(y)$\\
  $TI6$ & $\tau_I(x\parallel y)=\tau_I(x)\parallel\tau_I(y)$\\
\end{tabular}
\caption{Additional axioms of $\textrm{APTC}^{drt}_{\tau}(a\in A_{\tau\delta},n\geq 0)$}
\label{AxiomsForAPTCDRTTau}
\end{table}
\end{center}

The additional transition rules of $\textrm{APTC}^{drt}_{\tau}$ is shown in Table \ref{TRForAPTCDRTTau}.

\begin{center}
    \begin{table}
        $$\frac{x\xrightarrow{a}\surd}{\tau_I(x)\xrightarrow{a}\surd}\quad a\notin I
        \quad\quad\frac{x\xrightarrow{a}x'}{\tau_I(x)\xrightarrow{a}\tau_I(x')}\quad a\notin I$$

        $$\frac{x\xrightarrow{a}\surd}{\tau_I(x)\xrightarrow{\tau}\surd}\quad a\in I
        \quad\quad\frac{x\xrightarrow{a}x'}{\tau_I(x)\xrightarrow{\tau}\tau_I(x')}\quad a\in I$$

        $$\frac{x\mapsto^m x'}{\tau_I(x)\mapsto^m\tau_I(x')}\quad\frac{x\uparrow}{\tau_I(x)\uparrow}$$
        \caption{Transition rule of $\textrm{APTC}^{\textrm{drt}}_{\tau}(a\in A_{\tau},m>0,n\geq 0)$}
        \label{TRForAPTCDRTTau}
    \end{table}
\end{center}

\begin{definition}[Basic terms of $\textrm{APTC}^{\textrm{drt}}_{\tau}$]
The set of basic terms of $\textrm{APTC}^{\textrm{drt}}_{\tau}$, $\mathcal{B}(\textrm{APTC}^{\textrm{drt}}_{\tau})$, is inductively defined as follows by two auxiliary sets $\mathcal{B}_0(\textrm{APTC}^{\textrm{drt}}_{\tau})$ and $\mathcal{B}_1(\textrm{APTC}^{\textrm{drt}}_{\tau})$:
\begin{enumerate}
  \item if $a\in A_{\delta\tau}$, then $\underline{\underline{a}} \in \mathcal{B}_1(\textrm{APTC}^{\textrm{drt}}_{\tau})$;
  \item if $a\in A_{\tau}$ and $t\in \mathcal{B}(\textrm{APTC}^{\textrm{drt}}_{\tau})$, then $\underline{\underline{a}}\cdot t \in \mathcal{B}_1(\textrm{APTC}^{\textrm{drt}}_{\tau})$;
  \item if $t,t'\in \mathcal{B}_1(\textrm{APTC}^{\textrm{drt}}_{\tau})$, then $t+t'\in \mathcal{B}_1(\textrm{APTC}^{\textrm{drt}}_{\tau})$;
  \item if $t,t'\in \mathcal{B}_1(\textrm{APTC}^{\textrm{drt}}_{\tau})$, then $t\parallel t'\in \mathcal{B}_1(\textrm{APTC}^{\textrm{drt}}_{\tau})$;
  \item if $t\in \mathcal{B}_1(\textrm{APTC}^{\textrm{drt}}_{\tau})$, then $t\in \mathcal{B}_0(\textrm{APTC}^{\textrm{drt}}_{\tau})$;
  \item if $n>0$ and $t\in \mathcal{B}_0(\textrm{APTC}^{\textrm{drt}}_{\tau})$, then $\sigma^n_{\textrm{rel}}(t) \in \mathcal{B}_0(\textrm{APTC}^{\textrm{drt}}_{\tau})$;
  \item if $n>0$, $t\in \mathcal{B}_1(\textrm{APTC}^{\textrm{drt}}_{\tau})$ and $t'\in \mathcal{B}_0(\textrm{APTC}^{\textrm{drt}}_{\tau})$, then $t+\sigma^n_{\textrm{rel}}(t') \in \mathcal{B}_0(\textrm{APTC}^{\textrm{drt}}_{\tau})$;
  \item $\dot{\delta}\in \mathcal{B}(\textrm{APTC}^{\textrm{drt}}_{\tau})$;
  \item if $t\in \mathcal{B}_0(\textrm{APTC}^{\textrm{drt}}_{\tau})$, then $t\in \mathcal{B}(\textrm{APTC}^{\textrm{drt}}_{\tau})$.
\end{enumerate}
\end{definition}

\begin{theorem}[Elimination theorem]
Let $p$ be a closed $\textrm{APTC}^{\textrm{drt}}_{\tau}$ term. Then there is a basic $\textrm{APTC}^{\textrm{drt}}_{\tau}$ term $q$ such that $\textrm{APTC}^{\textrm{drt}}_{\tau}\vdash p=q$.
\end{theorem}

\begin{proof}
It is sufficient to induct on the structure of the closed $\textrm{APTC}^{\textrm{drt}}_{\tau}$ term $p$. It can be proven that $p$ combined by the constants and operators of $\textrm{APTC}^{\textrm{drt}}_{\tau}$ exists an equal basic term $q$, and the other operators not included in the basic terms, such as $\upsilon_{\textrm{rel}}$, $\overline{\upsilon}_{\textrm{rel}}$, $\between$, $\mid$, $\partial_H$, $\Theta$, $\triangleleft$ and $\tau_I$ can be eliminated.
\end{proof}

\subsubsection{Connections}

\begin{theorem}[Conservativity of $\textrm{APTC}^{\textrm{drt}}_{\tau}$]
$\textrm{APTC}^{\textrm{drt}}_{\tau}$ is a conservative extension of $\textrm{APTC}^{\textrm{drt}}$.
\end{theorem}

\begin{proof}
It follows from the following two facts.

    \begin{enumerate}
      \item The transition rules of $\textrm{APTC}^{\textrm{drt}}$ are all source-dependent;
      \item The sources of the transition rules of $\textrm{APTC}^{\textrm{drt}}_{\tau}$ contain an occurrence of $\underline{\underline{\tau}}$, and $\tau_I$.
    \end{enumerate}

So, $\textrm{APTC}^{\textrm{drt}}_{\tau}$ is a conservative extension of $\textrm{APTC}^{\textrm{drt}}$, as desired.
\end{proof}

\subsubsection{Congruence}

\begin{theorem}[Congruence of $\textrm{APTC}^{\textrm{drt}}_{\tau}$]
Rooted branching truly concurrent bisimulation equivalences $\approx_{rbp}$, $\approx_{rbs}$ and $\approx_{rbhp}$ are all congruences with respect to $\textrm{APTC}^{\textrm{drt}}_{\tau}$. That is,
\begin{itemize}
  \item rooted branching pomset bisimulation equivalence $\approx_{rbp}$ is a congruence with respect to $\textrm{APTC}^{\textrm{drt}}_{\tau}$;
  \item rooted branching step bisimulation equivalence $\approx_{rbs}$ is a congruence with respect to $\textrm{APTC}^{\textrm{drt}}_{\tau}$;
  \item rooted branching hp-bisimulation equivalence $\approx_{rbhp}$ is a congruence with respect to $\textrm{APTC}^{\textrm{drt}}_{\tau}$.
\end{itemize}
\end{theorem}

\begin{proof}
It is easy to see that $\approx_{rbp}$, $\approx_{rbs}$, and $\approx_{rbhp}$ are all equivalent relations on $\textrm{APTC}^{\textrm{drt}}_{\tau}$ terms, it is only sufficient to prove that $\approx_{rbp}$, $\approx_{rbs}$, and $\approx_{rbhp}$ are all preserved by the operators $\tau_I$. It is trivial and we omit it.
\end{proof}

\subsubsection{Soundness}

\begin{theorem}[Soundness of $\textrm{APTC}^{\textrm{drt}}_{\tau}$]
The axiomatization of $\textrm{APTC}^{\textrm{drt}}_{\tau}$ is sound modulo rooted branching truly concurrent bisimulation equivalences $\approx_{rbp}$, $\approx_{rbs}$, and $\approx_{rbhp}$. That is,
\begin{enumerate}
  \item let $x$ and $y$ be $\textrm{APTC}^{\textrm{drt}}_{\tau}$ terms. If $\textrm{APTC}^{\textrm{drt}}_{\tau}\vdash x=y$, then $x\approx_{rbs} y$;
  \item let $x$ and $y$ be $\textrm{APTC}^{\textrm{drt}}_{\tau}$ terms. If $\textrm{APTC}^{\textrm{drt}}_{\tau}\vdash x=y$, then $x\approx_{rbp} y$;
  \item let $x$ and $y$ be $\textrm{APTC}^{\textrm{drt}}_{\tau}$ terms. If $\textrm{APTC}^{\textrm{drt}}_{\tau}\vdash x=y$, then $x\approx_{rbhp} y$.
\end{enumerate}
\end{theorem}

\begin{proof}
Since $\approx_{rbp}$, $\approx_{rbs}$, and $\approx_{rbhp}$ are both equivalent and congruent relations, we only need to check if each axiom in Table \ref{AxiomsForAPTCDRTTau} is sound modulo $\approx_{rbp}$, $\approx_{rbs}$, and $\approx_{rbhp}$ respectively.

\begin{enumerate}
  \item Each axiom in Table \ref{AxiomsForAPTCDRTTau} can be checked that it is sound modulo rooted branching step bisimulation equivalence, by transition rules in Table \ref{TRForAPTCDRTTau}. We omit them.
  \item From the definition of rooted branching pomset bisimulation $\approx_{rbp}$, we know that rooted branching pomset bisimulation $\approx_{rbp}$ is defined by weak pomset transitions, which are labeled by pomsets with $\underline{\underline{\tau}}$. In a weak pomset transition, the events in the pomset are either within causality relations (defined by $\cdot$) or in concurrency (implicitly defined by $\cdot$ and $+$, and explicitly defined by $\between$), of course, they are pairwise consistent (without conflicts). We have already proven the case that all events are pairwise concurrent, so, we only need to prove the case of events in causality. Without loss of generality, we take a pomset of $P=\{\underline{\underline{a}},\underline{\underline{b}}:\underline{\underline{a}}\cdot \underline{\underline{b}}\}$. Then the weak pomset transition labeled by the above $P$ is just composed of one single event transition labeled by $\underline{\underline{a}}$ succeeded by another single event transition labeled by $\underline{\underline{b}}$, that is, $\xRightarrow{P}=\xRightarrow{a}\xRightarrow{b}$.

        Similarly to the proof of soundness modulo rooted branching step bisimulation $\approx_{rbs}$, we can prove that each axiom in Table \ref{AxiomsForAPTCDRTTau} is sound modulo rooted branching pomset bisimulation $\approx_{rbp}$, we omit them.

  \item From the definition of rooted branching hp-bisimulation $\approx_{rbhp}$, we know that rooted branching hp-bisimulation $\approx_{rbhp}$ is defined on the weakly posetal product $(C_1,f,C_2),f:\hat{C_1}\rightarrow \hat{C_2}\textrm{ isomorphism}$. Two process terms $s$ related to $C_1$ and $t$ related to $C_2$, and $f:\hat{C_1}\rightarrow \hat{C_2}\textrm{ isomorphism}$. Initially, $(C_1,f,C_2)=(\emptyset,\emptyset,\emptyset)$, and $(\emptyset,\emptyset,\emptyset)\in\approx_{rbhp}$. When $s\xrightarrow{a}s'$ ($C_1\xrightarrow{a}C_1'$), there will be $t\xRightarrow{a}t'$ ($C_2\xRightarrow{a}C_2'$), and we define $f'=f[a\mapsto a]$. Then, if $(C_1,f,C_2)\in\approx_{rbhp}$, then $(C_1',f',C_2')\in\approx_{rbhp}$.

        Similarly to the proof of soundness modulo rooted branching pomset bisimulation equivalence, we can prove that each axiom in Table \ref{AxiomsForAPTCDRTTau} is sound modulo rooted branching hp-bisimulation equivalence, we just need additionally to check the above conditions on rooted branching hp-bisimulation, we omit them.
\end{enumerate}

\end{proof}

\subsubsection{Completeness}

For $\textrm{APTC}^{\textrm{drt}}_{\tau}$ + Rec, it is similar to $\textrm{APTC}^{\textrm{drt}}$ + Rec, except that $\underline{\underline{\tau}}\cdot X$ is forbidden in recursive specifications for the sake of fairness. Like APTC, the proof of completeness need the help of $CFAR$ (see section \ref{tcpa}).

\begin{theorem}[Completeness of $\textrm{APTC}^{\textrm{drt}}_{\tau}$ + CFAR + guarded linear Rec]
The axiomatization of $\textrm{APTC}^{\textrm{drt}}_{\tau}$ + CFAR + guarded linear Rec is complete modulo rooted branching truly concurrent bisimulation equivalences $\approx_{rbs}$, $\approx_{rbp}$, and $\approx_{rbhp}$. That is,
\begin{enumerate}
  \item let $p$ and $q$ be closed $\textrm{APTC}^{\textrm{drt}}_{\tau}$ + CFAR + guarded linear Rec terms, if $p\approx_{rbs} q$ then $p=q$;
  \item let $p$ and $q$ be closed $\textrm{APTC}^{\textrm{drt}}_{\tau}$ + CFAR + guarded linear Rec terms, if $p\approx_{rbp} q$ then $p=q$;
  \item let $p$ and $q$ be closed $\textrm{APTC}^{\textrm{drt}}_{\tau}$ + CFAR + guarded linear Rec terms, if $p\approx_{rbhp} q$ then $p=q$.
\end{enumerate}

\end{theorem}

\begin{proof}
Firstly, we know that each process term in $\textrm{APTC}^{\textrm{drt}}_{\tau}$  + CFAR + guarded linear Rec is equal to a process term $\langle X_1|E\rangle$ with $E$ a linear recursive specification.

It remains to prove the following cases.

\begin{enumerate}
  \item If $\langle X_1|E_1\rangle \approx_{rbs} \langle Y_1|E_2\rangle$ for linear recursive specification $E_1$ and $E_2$, then $\langle X_1|E_1\rangle = \langle Y_1|E_2\rangle$.

        It can be proven similarly to the completeness of $\textrm{APTC}_{\tau}$ + CFAR + linear Rec, see \cite{ATC}.

  \item If $\langle X_1|E_1\rangle \approx_{rbp} \langle Y_1|E_2\rangle$ for linear recursive specification $E_1$ and $E_2$, then $\langle X_1|E_1\rangle = \langle Y_1|E_2\rangle$.

        It can be proven similarly, just by replacement of $\approx_{rbs}$ by $\approx_{rbp}$, we omit it.
  \item If $\langle X_1|E_1\rangle \approx_{rbhp} \langle Y_1|E_2\rangle$ for linear recursive specification $E_1$ and $E_2$, then $\langle X_1|E_1\rangle = \langle Y_1|E_2\rangle$.

        It can be proven similarly, just by replacement of $\approx_{rbs}$ by $\approx_{rbhp}$, we omit it.
\end{enumerate}
\end{proof}

\subsection{Discrete Absolute Timing}

\begin{definition}[Rooted branching truly concurrent bisimulations]
The following two conditions related timing should be added into the concepts of branching truly concurrent bisimulations in section \ref{tcpa}:
\begin{enumerate}
  \item if $C_1\mapsto^1 C_2$, then there are $C_1^*,C_2'$ such that $C_1'\Rightarrow C_1^*\mapsto^1 C_2'$, and $(C_1,C_1^*)\in R$ and $(C_2,C_2')\in R$, or $(C_1,f[\emptyset\mapsto\emptyset],C_1^*)\in R$ and $(C_2,f[\emptyset\mapsto\emptyset],C_2')\in R$;
  \item if $C_1\uparrow$, then $C_1'\uparrow$.
\end{enumerate}

And the following root conditions related timing should be added into the concepts of rooted branching truly concurrent bisimulations in section \ref{tcpa}:
\begin{enumerate}
  \item if $C_1\mapsto^m C_2(m>0)$, then there is $C_2'$ such that $C_1'\mapsto^m C_2'$, and $(C_2,C_2')\in R$, or $(C_2,f[\emptyset\mapsto\emptyset],C_2')\in R$.
\end{enumerate}
\end{definition}

\begin{definition}[Signature of $\textrm{APTC}^{\textrm{dat}}_{\tau}$]
The signature of $\textrm{APTC}^{\textrm{dat}}_{\tau}$ consists of the signature of $\textrm{APTC}^{\textrm{dat}}$, and the undelayable silent step constant $\underline{\tau}: \rightarrow\mathcal{P}_{\textrm{abs}}$, and the abstraction operator $\tau_I: \mathcal{P}_{\textrm{abs}}\rightarrow\mathcal{P}_{\textrm{abs}}$ for $I\subseteq A$.
\end{definition}

The axioms of $\textrm{APTC}^{\textrm{dat}}_{\tau}$ include the laws in Table \ref{AxiomsForAPTCDAT} covering the case that $a\equiv\tau$ and $b\equiv\tau$, and the axioms in Table \ref{AxiomsForAPTCDATTau}.

\begin{center}
\begin{table}
  \begin{tabular}{@{}ll@{}}
\hline No. &Axiom\\
  $DATB1$ & $x\cdot(\underline{\tau}\cdot(\upsilon^1_{\textrm{abs}}(y)+z+\underline{\delta})+\upsilon^1_{\textrm{abs}}(y)) = x\cdot(\upsilon^1_{\textrm{abs}}(y)+z+\underline{\delta})$\\
  $DATB2$ & $x\cdot(\underline{\tau}\cdot(\upsilon^1_{\textrm{abs}}(y)+z+\underline{\delta})+z) = x\cdot(\upsilon^1_{\textrm{abs}}(y)+z+\underline{\delta})$\\
  $DATB3$ & $x\cdot(\sigma^1_{
  \textrm{abs}}(\underline{\tau}\cdot (y+\underline{\delta})+ \upsilon^1_{\textrm{abs}}(z)) = x\cdot(\sigma^1_{\textrm{abs}}(y+\underline{\delta})+\upsilon^1_{\textrm{abs}}(z))$\\
  $B3$ & $x\parallel\underline{\tau}=x$\\
  $TI0$ & $\tau_I(\dot{\delta}) = \dot{\delta}$\\
  $TI1$ & $a\notin I\quad \tau_I(\underline{a})=\underline{a}$\\
  $TI2$ & $a\in I\quad \tau_I(\underline{a})=\underline{\tau}$\\
  $DATI$ & $\tau_I(\sigma^n_{\textrm{abs}}(x)) = \sigma^n_{\textrm{abs}}(\tau_I(x))$\\
  $TI4$ & $\tau_I(x+y)=\tau_I(x)+\tau_I(y)$\\
  $TI5$ & $\tau_I(x\cdot y)=\tau_I(x)\cdot\tau_I(y)$\\
  $TI6$ & $\tau_I(x\parallel y)=\tau_I(x)\parallel\tau_I(y)$\\
\end{tabular}
\caption{Additional axioms of $\textrm{APTC}^{dat}_{\tau}(a\in A_{\tau\delta},n\geq 0)$}
\label{AxiomsForAPTCDATTau}
\end{table}
\end{center}

The additional transition rules of $\textrm{APTC}^{dat}_{\tau}$ is shown in Table \ref{TRForAPTCDATTau}.

\begin{center}
    \begin{table}
        $$\frac{\langle x,n\rangle\xrightarrow{a}\langle\surd,n\rangle}{\langle\tau_I(x),n\rangle\xrightarrow{a}\langle\surd,n\rangle}\quad a\notin I
        \quad\quad\frac{\langle x,n\rangle\xrightarrow{a}\langle x',n\rangle}{\langle\tau_I(x),n\rangle\xrightarrow{a}\langle\tau_I(x'),n\rangle}\quad a\notin I$$

        $$\frac{\langle x,n\rangle\xrightarrow{a}\langle\surd,n\rangle}{\langle\tau_I(x),n\rangle\xrightarrow{\tau}\langle\surd,n\rangle}\quad a\in I
        \quad\quad\frac{\langle x,n\rangle\xrightarrow{a}\langle x',n\rangle}{\langle\tau_I(x),n\rangle\xrightarrow{\tau}\langle\tau_I(x'),n\rangle}\quad a\in I$$

        $$\frac{\langle x,n\rangle\mapsto^m \langle x,n+m\rangle}{\langle\tau_I(x),n\rangle\mapsto^m\langle\tau_I(x),n+m\rangle}\quad\frac{\langle x,n\rangle\uparrow}{\langle\tau_I(x),n\rangle\uparrow}$$
        \caption{Transition rule of $\textrm{APTC}^{\textrm{dat}}_{\tau}(a\in A_{\tau},m>0,n\geq 0)$}
        \label{TRForAPTCDATTau}
    \end{table}
\end{center}

\begin{definition}[Basic terms of $\textrm{APTC}^{\textrm{dat}}_{\tau}$]
The set of basic terms of $\textrm{APTC}^{\textrm{dat}}_{\tau}$, $\mathcal{B}(\textrm{APTC}^{\textrm{dat}}_{\tau})$, is inductively defined as follows by two auxiliary sets $\mathcal{B}_0(\textrm{APTC}^{\textrm{dat}}_{\tau})$ and $\mathcal{B}_1(\textrm{APTC}^{\textrm{dat}}_{\tau})$:
\begin{enumerate}
  \item if $a\in A_{\delta\tau}$, then $\underline{a} \in \mathcal{B}_1(\textrm{APTC}^{\textrm{dat}}_{\tau})$;
  \item if $a\in A_{\tau}$ and $t\in \mathcal{B}(\textrm{APTC}^{\textrm{dat}}_{\tau})$, then $\underline{a}\cdot t \in \mathcal{B}_1(\textrm{APTC}^{\textrm{dat}}_{\tau})$;
  \item if $t,t'\in \mathcal{B}_1(\textrm{APTC}^{\textrm{dat}}_{\tau})$, then $t+t'\in \mathcal{B}_1(\textrm{APTC}^{\textrm{dat}}_{\tau})$;
  \item if $t,t'\in \mathcal{B}_1(\textrm{APTC}^{\textrm{dat}}_{\tau})$, then $t\parallel t'\in \mathcal{B}_1(\textrm{APTC}^{\textrm{dat}}_{\tau})$;
  \item if $t\in \mathcal{B}_1(\textrm{APTC}^{\textrm{dat}}_{\tau})$, then $t\in \mathcal{B}_0(\textrm{APTC}^{\textrm{dat}}_{\tau})$;
  \item if $n>0$ and $t\in \mathcal{B}_0(\textrm{APTC}^{\textrm{dat}}_{\tau})$, then $\sigma^n_{\textrm{abs}}(t) \in \mathcal{B}_0(\textrm{APTC}^{\textrm{dat}}_{\tau})$;
  \item if $n>0$, $t\in \mathcal{B}_1(\textrm{APTC}^{\textrm{dat}}_{\tau})$ and $t'\in \mathcal{B}_0(\textrm{APTC}^{\textrm{dat}}_{\tau})$, then $t+\sigma^n_{\textrm{abs}}(t') \in \mathcal{B}_0(\textrm{APTC}^{\textrm{dat}}_{\tau})$;
  \item $\dot{\delta}\in \mathcal{B}(\textrm{APTC}^{\textrm{dat}}_{\tau})$;
  \item if $t\in \mathcal{B}_0(\textrm{APTC}^{\textrm{dat}}_{\tau})$, then $t\in \mathcal{B}(\textrm{APTC}^{\textrm{dat}}_{\tau})$.
\end{enumerate}
\end{definition}

\begin{theorem}[Elimination theorem]
Let $p$ be a closed $\textrm{APTC}^{\textrm{dat}}_{\tau}$ term. Then there is a basic $\textrm{APTC}^{\textrm{dat}}_{\tau}$ term $q$ such that $\textrm{APTC}^{\textrm{dat}}_{\tau}\vdash p=q$.
\end{theorem}

\begin{proof}
It is sufficient to induct on the structure of the closed $\textrm{APTC}^{\textrm{dat}}_{\tau}$ term $p$. It can be proven that $p$ combined by the constants and operators of $\textrm{APTC}^{\textrm{dat}}_{\tau}$ exists an equal basic term $q$, and the other operators not included in the basic terms, such as $\upsilon_{\textrm{abs}}$, $\overline{\upsilon}_{\textrm{abs}}$, $\between$, $\mid$, $\partial_H$, $\Theta$, $\triangleleft$ and $\tau_I$ can be eliminated.
\end{proof}

\subsubsection{Connections}

\begin{theorem}[Conservativity of $\textrm{APTC}^{\textrm{dat}}_{\tau}$]
$\textrm{APTC}^{\textrm{dat}}_{\tau}$ is a conservative extension of $\textrm{APTC}^{\textrm{dat}}$.
\end{theorem}

\begin{proof}
It follows from the following two facts.

    \begin{enumerate}
      \item The transition rules of $\textrm{APTC}^{\textrm{dat}}$ are all source-dependent;
      \item The sources of the transition rules of $\textrm{APTC}^{\textrm{dat}}_{\tau}$ contain an occurrence of $\underline{\tau}$, and $\tau_I$.
    \end{enumerate}

So, $\textrm{APTC}^{\textrm{dat}}_{\tau}$ is a conservative extension of $\textrm{APTC}^{\textrm{dat}}$, as desired.
\end{proof}

\subsubsection{Congruence}

\begin{theorem}[Congruence of $\textrm{APTC}^{\textrm{dat}}_{\tau}$]
Rooted branching truly concurrent bisimulation equivalences $\approx_{rbp}$, $\approx_{rbs}$ and $\approx_{rbhp}$ are all congruences with respect to $\textrm{APTC}^{\textrm{dat}}_{\tau}$. That is,
\begin{itemize}
  \item rooted branching pomset bisimulation equivalence $\approx_{rbp}$ is a congruence with respect to $\textrm{APTC}^{\textrm{dat}}_{\tau}$;
  \item rooted branching step bisimulation equivalence $\approx_{rbs}$ is a congruence with respect to $\textrm{APTC}^{\textrm{dat}}_{\tau}$;
  \item rooted branching hp-bisimulation equivalence $\approx_{rbhp}$ is a congruence with respect to $\textrm{APTC}^{\textrm{dat}}_{\tau}$.
\end{itemize}
\end{theorem}

\begin{proof}
It is easy to see that $\approx_{rbp}$, $\approx_{rbs}$, and $\approx_{rbhp}$ are all equivalent relations on $\textrm{APTC}^{\textrm{dat}}_{\tau}$ terms, it is only sufficient to prove that $\approx_{rbp}$, $\approx_{rbs}$, and $\approx_{rbhp}$ are all preserved by the operators $\tau_I$. It is trivial and we omit it.
\end{proof}

\subsubsection{Soundness}

\begin{theorem}[Soundness of $\textrm{APTC}^{\textrm{dat}}_{\tau}$]
The axiomatization of $\textrm{APTC}^{\textrm{dat}}_{\tau}$ is sound modulo rooted branching truly concurrent bisimulation equivalences $\approx_{rbp}$, $\approx_{rbs}$, and $\approx_{rbhp}$. That is,
\begin{enumerate}
  \item let $x$ and $y$ be $\textrm{APTC}^{\textrm{dat}}_{\tau}$ terms. If $\textrm{APTC}^{\textrm{dat}}_{\tau}\vdash x=y$, then $x\approx_{rbs} y$;
  \item let $x$ and $y$ be $\textrm{APTC}^{\textrm{dat}}_{\tau}$ terms. If $\textrm{APTC}^{\textrm{dat}}_{\tau}\vdash x=y$, then $x\approx_{rbp} y$;
  \item let $x$ and $y$ be $\textrm{APTC}^{\textrm{dat}}_{\tau}$ terms. If $\textrm{APTC}^{\textrm{dat}}_{\tau}\vdash x=y$, then $x\approx_{rbhp} y$.
\end{enumerate}
\end{theorem}

\begin{proof}
Since $\approx_{rbp}$, $\approx_{rbs}$, and $\approx_{rbhp}$ are both equivalent and congruent relations, we only need to check if each axiom in Table \ref{AxiomsForAPTCDATTau} is sound modulo $\approx_{rbp}$, $\approx_{rbs}$, and $\approx_{rbhp}$ respectively.

\begin{enumerate}
  \item Each axiom in Table \ref{AxiomsForAPTCDATTau} can be checked that it is sound modulo rooted branching step bisimulation equivalence, by transition rules in Table \ref{TRForAPTCDATTau}. We omit them.
  \item From the definition of rooted branching pomset bisimulation $\approx_{rbp}$, we know that rooted branching pomset bisimulation $\approx_{rbp}$ is defined by weak pomset transitions, which are labeled by pomsets with $\underline{\tau}$. In a weak pomset transition, the events in the pomset are either within causality relations (defined by $\cdot$) or in concurrency (implicitly defined by $\cdot$ and $+$, and explicitly defined by $\between$), of course, they are pairwise consistent (without conflicts). We have already proven the case that all events are pairwise concurrent, so, we only need to prove the case of events in causality. Without loss of generality, we take a pomset of $P=\{\underline{a},\underline{b}:\underline{a}\cdot \underline{b}\}$. Then the weak pomset transition labeled by the above $P$ is just composed of one single event transition labeled by $\underline{a}$ succeeded by another single event transition labeled by $\underline{b}$, that is, $\xRightarrow{P}=\xRightarrow{a}\xRightarrow{b}$.

        Similarly to the proof of soundness modulo rooted branching step bisimulation $\approx_{rbs}$, we can prove that each axiom in Table \ref{AxiomsForAPTCDATTau} is sound modulo rooted branching pomset bisimulation $\approx_{rbp}$, we omit them.

  \item From the definition of rooted branching hp-bisimulation $\approx_{rbhp}$, we know that rooted branching hp-bisimulation $\approx_{rbhp}$ is defined on the weakly posetal product $(C_1,f,C_2),f:\hat{C_1}\rightarrow \hat{C_2}\textrm{ isomorphism}$. Two process terms $s$ related to $C_1$ and $t$ related to $C_2$, and $f:\hat{C_1}\rightarrow \hat{C_2}\textrm{ isomorphism}$. Initially, $(C_1,f,C_2)=(\emptyset,\emptyset,\emptyset)$, and $(\emptyset,\emptyset,\emptyset)\in\approx_{rbhp}$. When $s\xrightarrow{a}s'$ ($C_1\xrightarrow{a}C_1'$), there will be $t\xRightarrow{a}t'$ ($C_2\xRightarrow{a}C_2'$), and we define $f'=f[a\mapsto a]$. Then, if $(C_1,f,C_2)\in\approx_{rbhp}$, then $(C_1',f',C_2')\in\approx_{rbhp}$.

        Similarly to the proof of soundness modulo rooted branching pomset bisimulation equivalence, we can prove that each axiom in Table \ref{AxiomsForAPTCDATTau} is sound modulo rooted branching hp-bisimulation equivalence, we just need additionally to check the above conditions on rooted branching hp-bisimulation, we omit them.
\end{enumerate}

\end{proof}

\subsubsection{Completeness}

For $\textrm{APTC}^{\textrm{dat}}_{\tau}$ + Rec, it is similar to $\textrm{APTC}^{\textrm{dat}}$ + Rec, except that $\underline{\tau}\cdot X$ is forbidden in recursive specifications for the sake of fairness. Like APTC, the proof of completeness need the help of $CFAR$ (see section \ref{tcpa}).

\begin{theorem}[Completeness of $\textrm{APTC}^{\textrm{dat}}_{\tau}$ + CFAR + guarded linear Rec]
The axiomatization of $\textrm{APTC}^{\textrm{dat}}_{\tau}$ + CFAR + guarded linear Rec is complete modulo rooted branching truly concurrent bisimulation equivalences $\approx_{rbs}$, $\approx_{rbp}$, and $\approx_{rbhp}$. That is,
\begin{enumerate}
  \item let $p$ and $q$ be closed $\textrm{APTC}^{\textrm{dat}}_{\tau}$ + CFAR + guarded linear Rec terms, if $p\approx_{rbs} q$ then $p=q$;
  \item let $p$ and $q$ be closed $\textrm{APTC}^{\textrm{dat}}_{\tau}$ + CFAR + guarded linear Rec terms, if $p\approx_{rbp} q$ then $p=q$;
  \item let $p$ and $q$ be closed $\textrm{APTC}^{\textrm{dat}}_{\tau}$ + CFAR + guarded linear Rec terms, if $p\approx_{rbhp} q$ then $p=q$.
\end{enumerate}

\end{theorem}

\begin{proof}
Firstly, we know that each process term in $\textrm{APTC}^{\textrm{dat}}_{\tau}$  + CFAR + guarded linear Rec is equal to a process term $\langle X_1|E\rangle$ with $E$ a linear recursive specification.

It remains to prove the following cases.

\begin{enumerate}
  \item If $\langle X_1|E_1\rangle \approx_{rbs} \langle Y_1|E_2\rangle$ for linear recursive specification $E_1$ and $E_2$, then $\langle X_1|E_1\rangle = \langle Y_1|E_2\rangle$.

        It can be proven similarly to the completeness of $\textrm{APTC}_{\tau}$ + CFAR + linear Rec, see \cite{ATC}.

  \item If $\langle X_1|E_1\rangle \approx_{rbp} \langle Y_1|E_2\rangle$ for linear recursive specification $E_1$ and $E_2$, then $\langle X_1|E_1\rangle = \langle Y_1|E_2\rangle$.

        It can be proven similarly, just by replacement of $\approx_{rbs}$ by $\approx_{rbp}$, we omit it.
  \item If $\langle X_1|E_1\rangle \approx_{rbhp} \langle Y_1|E_2\rangle$ for linear recursive specification $E_1$ and $E_2$, then $\langle X_1|E_1\rangle = \langle Y_1|E_2\rangle$.

        It can be proven similarly, just by replacement of $\approx_{rbs}$ by $\approx_{rbhp}$, we omit it.
\end{enumerate}
\end{proof}

\subsection{Continuous Relative Timing}

\begin{definition}[Rooted branching truly concurrent bisimulations]
The following two conditions related timing should be added into the concepts of branching truly concurrent bisimulations in section \ref{tcpa}:
\begin{enumerate}
  \item if $C_1\mapsto^r C_2(r>0)$, then either there are $C_1^*,C_2',C_2''$ and $r':0<r'<r$ such that $C_1'\Rightarrow C_1^*\mapsto^{r'} C_2'$ and $C_2''\mapsto^{r-r'}C_2'$, and $(C_1,C_1^*)\in R$ and $(C_2,C_2')\in R$, or $(C_1,f[\emptyset\mapsto\emptyset],C_1^*)\in R$ and $(C_2,f[\emptyset\mapsto\emptyset],C_2')\in R$; or there are $C_1^*,C_2'$ such that $C_1'\Rightarrow C_1^*\mapsto^r C_2'$, and $(C_1,C_1^*)\in R$ and $(C_2,C_2')\in R$, or $(C_1,f[\emptyset\mapsto\emptyset],C_1^*)\in R$ and $(C_2,f[\emptyset\mapsto\emptyset],C_2')\in R$;
  \item if $C_1\uparrow$, then $C_1'\uparrow$.
\end{enumerate}

And the following root conditions related timing should be added into the concepts of rooted branching truly concurrent bisimulations in section \ref{tcpa}:
\begin{enumerate}
  \item if $C_1\mapsto^r C_2(r>0)$, then there is $C_2'$ such that $C_1'\mapsto^r C_2'$, and $(C_2,C_2')\in R$, or $(C_2,f[\emptyset\mapsto\emptyset],C_2')\in R$.
\end{enumerate}
\end{definition}

\begin{definition}[Signature of $\textrm{APTC}^{\textrm{srt}}_{\tau}$]
The signature of $\textrm{APTC}^{\textrm{srt}}_{\tau}$ consists of the signature of $\textrm{APTC}^{\textrm{srt}}$, and the undelayable silent step constant $\tilde{\tilde{\tau}}: \rightarrow\mathcal{P}_{\textrm{rel}}$, and the abstraction operator $\tau_I: \mathcal{P}_{\textrm{rel}}\rightarrow\mathcal{P}_{\textrm{rel}}$ for $I\subseteq A$.
\end{definition}

The axioms of $\textrm{APTC}^{\textrm{srt}}_{\tau}$ include the laws in Table \ref{AxiomsForAPTCSRT} covering the case that $a\equiv\tau$ and $b\equiv\tau$, and the axioms in Table \ref{AxiomsForAPTCSRTTau}.

\begin{center}
\begin{table}
  \begin{tabular}{@{}ll@{}}
\hline No. &Axiom\\
  $SRTB1$ & $x\cdot(\tilde{\tilde{\tau}}\cdot(\nu_{\textrm{rel}}(y)+z+\tilde{\tilde{\delta}})+\nu_{\textrm{rel}}(y)) = x\cdot(\nu_{\textrm{rel}}(y)+z+\tilde{\tilde{\delta}})$\\
  $SRTB2$ & $x\cdot(\tilde{\tilde{\tau}}\cdot(\nu_{\textrm{rel}}(y)+z+\tilde{\tilde{\delta}})+z) = x\cdot(\nu_{\textrm{rel}}(y)+z+\tilde{\tilde{\delta}})$\\
  $SRTB3$ & $x\cdot(\sigma^r_{
  \textrm{rel}}(\tilde{\tilde{\tau}}\cdot (y+\tilde{\tilde{\delta}})+\upsilon^r_{\textrm{rel}}(z)) = x\cdot(\sigma^r_{\textrm{rel}}(y+\tilde{\tilde{\delta}})+\upsilon^r_{\textrm{rel}}(z))$\\
  $B3$ & $x\parallel\tilde{\tilde{\tau}}=x$\\
  $TI0$ & $\tau_I(\dot{\delta}) = \dot{\delta}$\\
  $TI1$ & $a\notin I\quad \tau_I(\tilde{\tilde{a}})=\tilde{\tilde{a}}$\\
  $TI2$ & $a\in I\quad \tau_I(\tilde{\tilde{a}})=\tilde{\tilde{\tau}}$\\
  $SRTI$ & $\tau_I(\sigma^p_{\textrm{rel}}(x)) = \sigma^p_{\textrm{rel}}(\tau_I(x))$\\
  $TI4$ & $\tau_I(x+y)=\tau_I(x)+\tau_I(y)$\\
  $TI5$ & $\tau_I(x\cdot y)=\tau_I(x)\cdot\tau_I(y)$\\
  $TI6$ & $\tau_I(x\parallel y)=\tau_I(x)\parallel\tau_I(y)$\\
\end{tabular}
\caption{Additional axioms of $\textrm{APTC}^{srt}_{\tau}(a\in A_{\tau\delta},p\geq 0,r>0)$}
\label{AxiomsForAPTCSRTTau}
\end{table}
\end{center}

The additional transition rules of $\textrm{APTC}^{srt}_{\tau}$ is shown in Table \ref{TRForAPTCSRTTau}.

\begin{center}
    \begin{table}
        $$\frac{x\xrightarrow{a}\surd}{\tau_I(x)\xrightarrow{a}\surd}\quad a\notin I
        \quad\quad\frac{x\xrightarrow{a}x'}{\tau_I(x)\xrightarrow{a}\tau_I(x')}\quad a\notin I$$

        $$\frac{x\xrightarrow{a}\surd}{\tau_I(x)\xrightarrow{\tau}\surd}\quad a\in I
        \quad\quad\frac{x\xrightarrow{a}x'}{\tau_I(x)\xrightarrow{\tau}\tau_I(x')}\quad a\in I$$

        $$\frac{x\mapsto^r x'}{\tau_I(x)\mapsto^r\tau_I(x')}\quad\frac{x\uparrow}{\tau_I(x)\uparrow}$$
        \caption{Transition rule of $\textrm{APTC}^{\textrm{srt}}_{\tau}(a\in A_{\tau},r>0,p\geq 0)$}
        \label{TRForAPTCSRTTau}
    \end{table}
\end{center}

\begin{definition}[Basic terms of $\textrm{APTC}^{\textrm{srt}}_{\tau}$]
The set of basic terms of $\textrm{APTC}^{\textrm{srt}}_{\tau}$, $\mathcal{B}(\textrm{APTC}^{\textrm{srt}}_{\tau})$, is inductively defined as follows by two auxiliary sets $\mathcal{B}_0(\textrm{APTC}^{\textrm{srt}}_{\tau})$ and $\mathcal{B}_1(\textrm{APTC}^{\textrm{srt}}_{\tau})$:
\begin{enumerate}
  \item if $a\in A_{\delta\tau}$, then $\tilde{\tilde{a}} \in \mathcal{B}_1(\textrm{APTC}^{\textrm{srt}}_{\tau})$;
  \item if $a\in A_{\tau}$ and $t\in \mathcal{B}(\textrm{APTC}^{\textrm{srt}}_{\tau})$, then $\tilde{\tilde{a}}\cdot t \in \mathcal{B}_1(\textrm{APTC}^{\textrm{srt}}_{\tau})$;
  \item if $t,t'\in \mathcal{B}_1(\textrm{APTC}^{\textrm{srt}}_{\tau})$, then $t+t'\in \mathcal{B}_1(\textrm{APTC}^{\textrm{srt}}_{\tau})$;
  \item if $t,t'\in \mathcal{B}_1(\textrm{APTC}^{\textrm{srt}}_{\tau})$, then $t\parallel t'\in \mathcal{B}_1(\textrm{APTC}^{\textrm{srt}}_{\tau})$;
  \item if $t\in \mathcal{B}_1(\textrm{APTC}^{\textrm{srt}}_{\tau})$, then $t\in \mathcal{B}_0(\textrm{APTC}^{\textrm{srt}}_{\tau})$;
  \item if $p>0$ and $t\in \mathcal{B}_0(\textrm{APTC}^{\textrm{srt}}_{\tau})$, then $\sigma^p_{\textrm{rel}}(t) \in \mathcal{B}_0(\textrm{APTC}^{\textrm{srt}}_{\tau})$;
  \item if $p>0$, $t\in \mathcal{B}_1(\textrm{APTC}^{\textrm{srt}}_{\tau})$ and $t'\in \mathcal{B}_0(\textrm{APTC}^{\textrm{srt}}_{\tau})$, then $t+\sigma^p_{\textrm{rel}}(t') \in \mathcal{B}_0(\textrm{APTC}^{\textrm{srt}}_{\tau})$;
  \item if $t\in \mathcal{B}_0(\textrm{APTC}^{\textrm{srt}}_{\tau})$, then $\nu_{\textrm{rel}}(t) \in \mathcal{B}_0(\textrm{APTC}^{\textrm{srt}}_{\tau})$;
  \item $\dot{\delta}\in \mathcal{B}(\textrm{APTC}^{\textrm{srt}}_{\tau})$;
  \item if $t\in \mathcal{B}_0(\textrm{APTC}^{\textrm{srt}}_{\tau})$, then $t\in \mathcal{B}(\textrm{APTC}^{\textrm{srt}}_{\tau})$.
\end{enumerate}
\end{definition}

\begin{theorem}[Elimination theorem]
Let $p$ be a closed $\textrm{APTC}^{\textrm{srt}}_{\tau}$ term. Then there is a basic $\textrm{APTC}^{\textrm{srt}}_{\tau}$ term $q$ such that $\textrm{APTC}^{\textrm{srt}}_{\tau}\vdash p=q$.
\end{theorem}

\begin{proof}
It is sufficient to induct on the structure of the closed $\textrm{APTC}^{\textrm{srt}}_{\tau}$ term $p$. It can be proven that $p$ combined by the constants and operators of $\textrm{APTC}^{\textrm{srt}}_{\tau}$ exists an equal basic term $q$, and the other operators not included in the basic terms, such as $\upsilon_{\textrm{rel}}$, $\overline{\upsilon}_{\textrm{rel}}$, $\between$, $\mid$, $\partial_H$, $\Theta$, $\triangleleft$ and $\tau_I$ can be eliminated.
\end{proof}

\subsubsection{Connections}

\begin{theorem}[Conservativity of $\textrm{APTC}^{\textrm{srt}}_{\tau}$]
$\textrm{APTC}^{\textrm{srt}}_{\tau}$ is a conservative extension of $\textrm{APTC}^{\textrm{srt}}$.
\end{theorem}

\begin{proof}
It follows from the following two facts.

    \begin{enumerate}
      \item The transition rules of $\textrm{APTC}^{\textrm{srt}}$ are all source-dependent;
      \item The sources of the transition rules of $\textrm{APTC}^{\textrm{srt}}_{\tau}$ contain an occurrence of $\tilde{\tilde{\tau}}$, and $\tau_I$.
    \end{enumerate}

So, $\textrm{APTC}^{\textrm{srt}}_{\tau}$ is a conservative extension of $\textrm{APTC}^{\textrm{srt}}$, as desired.
\end{proof}

\subsubsection{Congruence}

\begin{theorem}[Congruence of $\textrm{APTC}^{\textrm{srt}}_{\tau}$]
Rooted branching truly concurrent bisimulation equivalences $\approx_{rbp}$, $\approx_{rbs}$ and $\approx_{rbhp}$ are all congruences with respect to $\textrm{APTC}^{\textrm{srt}}_{\tau}$. That is,
\begin{itemize}
  \item rooted branching pomset bisimulation equivalence $\approx_{rbp}$ is a congruence with respect to $\textrm{APTC}^{\textrm{srt}}_{\tau}$;
  \item rooted branching step bisimulation equivalence $\approx_{rbs}$ is a congruence with respect to $\textrm{APTC}^{\textrm{srt}}_{\tau}$;
  \item rooted branching hp-bisimulation equivalence $\approx_{rbhp}$ is a congruence with respect to $\textrm{APTC}^{\textrm{srt}}_{\tau}$.
\end{itemize}
\end{theorem}

\begin{proof}
It is easy to see that $\approx_{rbp}$, $\approx_{rbs}$, and $\approx_{rbhp}$ are all equivalent relations on $\textrm{APTC}^{\textrm{srt}}_{\tau}$ terms, it is only sufficient to prove that $\approx_{rbp}$, $\approx_{rbs}$, and $\approx_{rbhp}$ are all preserved by the operators $\tau_I$. It is trivial and we omit it.
\end{proof}

\subsubsection{Soundness}

\begin{theorem}[Soundness of $\textrm{APTC}^{\textrm{srt}}_{\tau}$]
The axiomatization of $\textrm{APTC}^{\textrm{srt}}_{\tau}$ is sound modulo rooted branching truly concurrent bisimulation equivalences $\approx_{rbp}$, $\approx_{rbs}$, and $\approx_{rbhp}$. That is,
\begin{enumerate}
  \item let $x$ and $y$ be $\textrm{APTC}^{\textrm{srt}}_{\tau}$ terms. If $\textrm{APTC}^{\textrm{srt}}_{\tau}\vdash x=y$, then $x\approx_{rbs} y$;
  \item let $x$ and $y$ be $\textrm{APTC}^{\textrm{srt}}_{\tau}$ terms. If $\textrm{APTC}^{\textrm{srt}}_{\tau}\vdash x=y$, then $x\approx_{rbp} y$;
  \item let $x$ and $y$ be $\textrm{APTC}^{\textrm{srt}}_{\tau}$ terms. If $\textrm{APTC}^{\textrm{srt}}_{\tau}\vdash x=y$, then $x\approx_{rbhp} y$.
\end{enumerate}
\end{theorem}

\begin{proof}
Since $\approx_{rbp}$, $\approx_{rbs}$, and $\approx_{rbhp}$ are both equivalent and congruent relations, we only need to check if each axiom in Table \ref{AxiomsForAPTCSRTTau} is sound modulo $\approx_{rbp}$, $\approx_{rbs}$, and $\approx_{rbhp}$ respectively.

\begin{enumerate}
  \item Each axiom in Table \ref{AxiomsForAPTCSRTTau} can be checked that it is sound modulo rooted branching step bisimulation equivalence, by transition rules in Table \ref{TRForAPTCSRTTau}. We omit them.
  \item From the definition of rooted branching pomset bisimulation $\approx_{rbp}$, we know that rooted branching pomset bisimulation $\approx_{rbp}$ is defined by weak pomset transitions, which are labeled by pomsets with $\tilde{\tilde{\tau}}$. In a weak pomset transition, the events in the pomset are either within causality relations (defined by $\cdot$) or in concurrency (implicitly defined by $\cdot$ and $+$, and explicitly defined by $\between$), of course, they are pairwise consistent (without conflicts). We have already proven the case that all events are pairwise concurrent, so, we only need to prove the case of events in causality. Without loss of generality, we take a pomset of $P=\{\tilde{\tilde{a}},\tilde{\tilde{b}}:\tilde{\tilde{a}}\cdot \tilde{\tilde{b}}\}$. Then the weak pomset transition labeled by the above $P$ is just composed of one single event transition labeled by $\tilde{\tilde{a}}$ succeeded by another single event transition labeled by $\tilde{\tilde{b}}$, that is, $\xRightarrow{P}=\xRightarrow{a}\xRightarrow{b}$.

        Similarly to the proof of soundness modulo rooted branching step bisimulation $\approx_{rbs}$, we can prove that each axiom in Table \ref{AxiomsForAPTCSRTTau} is sound modulo rooted branching pomset bisimulation $\approx_{rbp}$, we omit them.

  \item From the definition of rooted branching hp-bisimulation $\approx_{rbhp}$, we know that rooted branching hp-bisimulation $\approx_{rbhp}$ is defined on the weakly posetal product $(C_1,f,C_2),f:\hat{C_1}\rightarrow \hat{C_2}\textrm{ isomorphism}$. Two process terms $s$ related to $C_1$ and $t$ related to $C_2$, and $f:\hat{C_1}\rightarrow \hat{C_2}\textrm{ isomorphism}$. Initially, $(C_1,f,C_2)=(\emptyset,\emptyset,\emptyset)$, and $(\emptyset,\emptyset,\emptyset)\in\approx_{rbhp}$. When $s\xrightarrow{a}s'$ ($C_1\xrightarrow{a}C_1'$), there will be $t\xRightarrow{a}t'$ ($C_2\xRightarrow{a}C_2'$), and we define $f'=f[a\mapsto a]$. Then, if $(C_1,f,C_2)\in\approx_{rbhp}$, then $(C_1',f',C_2')\in\approx_{rbhp}$.

        Similarly to the proof of soundness modulo rooted branching pomset bisimulation equivalence, we can prove that each axiom in Table \ref{AxiomsForAPTCSRTTau} is sound modulo rooted branching hp-bisimulation equivalence, we just need additionally to check the above conditions on rooted branching hp-bisimulation, we omit them.
\end{enumerate}

\end{proof}

\subsubsection{Completeness}

For $\textrm{APTC}^{\textrm{srt}}_{\tau}$ + Rec, it is similar to $\textrm{APTC}^{\textrm{srt}}$ + Rec, except that $\tilde{\tilde{\tau}}\cdot X$ is forbidden in recursive specifications for the sake of fairness. Like APTC, the proof of completeness need the help of $CFAR$ (see section \ref{tcpa}).

\begin{theorem}[Completeness of $\textrm{APTC}^{\textrm{srt}}_{\tau}$ + CFAR + guarded linear Rec]
The axiomatization of $\textrm{APTC}^{\textrm{dat}}_{\tau}$ + CFAR + guarded linear Rec is complete modulo rooted branching truly concurrent bisimulation equivalences $\approx_{rbs}$, $\approx_{rbp}$, and $\approx_{rbhp}$. That is,
\begin{enumerate}
  \item let $p$ and $q$ be closed $\textrm{APTC}^{\textrm{srt}}_{\tau}$ + CFAR + guarded linear Rec terms, if $p\approx_{rbs} q$ then $p=q$;
  \item let $p$ and $q$ be closed $\textrm{APTC}^{\textrm{srt}}_{\tau}$ + CFAR + guarded linear Rec terms, if $p\approx_{rbp} q$ then $p=q$;
  \item let $p$ and $q$ be closed $\textrm{APTC}^{\textrm{srt}}_{\tau}$ + CFAR + guarded linear Rec terms, if $p\approx_{rbhp} q$ then $p=q$.
\end{enumerate}

\end{theorem}

\begin{proof}
Firstly, we know that each process term in $\textrm{APTC}^{\textrm{srt}}_{\tau}$  + CFAR + guarded linear Rec is equal to a process term $\langle X_1|E\rangle$ with $E$ a linear recursive specification.

It remains to prove the following cases.

\begin{enumerate}
  \item If $\langle X_1|E_1\rangle \approx_{rbs} \langle Y_1|E_2\rangle$ for linear recursive specification $E_1$ and $E_2$, then $\langle X_1|E_1\rangle = \langle Y_1|E_2\rangle$.

        It can be proven similarly to the completeness of $\textrm{APTC}_{\tau}$ + CFAR + linear Rec, see \cite{ATC}.

  \item If $\langle X_1|E_1\rangle \approx_{rbp} \langle Y_1|E_2\rangle$ for linear recursive specification $E_1$ and $E_2$, then $\langle X_1|E_1\rangle = \langle Y_1|E_2\rangle$.

        It can be proven similarly, just by replacement of $\approx_{rbs}$ by $\approx_{rbp}$, we omit it.
  \item If $\langle X_1|E_1\rangle \approx_{rbhp} \langle Y_1|E_2\rangle$ for linear recursive specification $E_1$ and $E_2$, then $\langle X_1|E_1\rangle = \langle Y_1|E_2\rangle$.

        It can be proven similarly, just by replacement of $\approx_{rbs}$ by $\approx_{rbhp}$, we omit it.
\end{enumerate}
\end{proof}

\subsection{Continuous Absolute Timing}

\begin{definition}[Rooted branching truly concurrent bisimulations]
The following two conditions related timing should be added into the concepts of branching truly concurrent bisimulations in section \ref{tcpa}:
\begin{enumerate}
  \item if $C_1\mapsto^r C_2(r>0)$, then either there are $C_1^*,C_2',C_2''$ and $r':0<r'<r$ such that $C_1'\Rightarrow C_1^*\mapsto^{r'} C_2'$ and $C_2''\mapsto^{r-r'}C_2'$, and $(C_1,C_1^*)\in R$ and $(C_2,C_2')\in R$, or $(C_1,f[\emptyset\mapsto\emptyset],C_1^*)\in R$ and $(C_2,f[\emptyset\mapsto\emptyset],C_2')\in R$; or there are $C_1^*,C_2'$ such that $C_1'\Rightarrow C_1^*\mapsto^r C_2'$, and $(C_1,C_1^*)\in R$ and $(C_2,C_2')\in R$, or $(C_1,f[\emptyset\mapsto\emptyset],C_1^*)\in R$ and $(C_2,f[\emptyset\mapsto\emptyset],C_2')\in R$;
  \item if $C_1\uparrow$, then $C_1'\uparrow$.
\end{enumerate}

And the following root conditions related timing should be added into the concepts of rooted branching truly concurrent bisimulations in section \ref{tcpa}:
\begin{enumerate}
  \item if $C_1\mapsto^r C_2(r>0)$, then there is $C_2'$ such that $C_1'\mapsto^r C_2'$, and $(C_2,C_2')\in R$, or $(C_2,f[\emptyset\mapsto\emptyset],C_2')\in R$.
\end{enumerate}
\end{definition}

\begin{definition}[Signature of $\textrm{APTC}^{\textrm{sat}}_{\tau}$]
The signature of $\textrm{APTC}^{\textrm{sat}}_{\tau}$ consists of the signature of $\textrm{APTC}^{\textrm{sat}}$, and the undelayable silent step constant $\tilde{\tau}: \rightarrow\mathcal{P}_{\textrm{abs}}$, and the abstraction operator $\tau_I: \mathcal{P}_{\textrm{abs}}\rightarrow\mathcal{P}_{\textrm{abs}}$ for $I\subseteq A$.
\end{definition}

The axioms of $\textrm{APTC}^{\textrm{sat}}_{\tau}$ include the laws in Table \ref{AxiomsForAPTCSAT} covering the case that $a\equiv\tau$ and $b\equiv\tau$, and the axioms in Table \ref{AxiomsForAPTCSATTau}.

\begin{center}
\begin{table}
  \begin{tabular}{@{}ll@{}}
\hline No. &Axiom\\
  $SATB1$ & $x\cdot(\tilde{\tau}\cdot(\nu_{\textrm{abs}}(y)+z+\tilde{\delta})+\nu_{\textrm{abs}}(y)) = x\cdot(\nu_{\textrm{abs}}(y)+z+\tilde{\delta})$\\
  $SATB2$ & $x\cdot(\tilde{\tau}\cdot(\nu_{\textrm{abs}}(y)+z+\tilde{\delta})+z) = x\cdot(\nu_{\textrm{abs}}(y)+z+\tilde{\delta})$\\
  $SATB3$ & $x\cdot(\sigma^r_{
  \textrm{abs}}(\tilde{\tau}\cdot (y+\tilde{\delta})+ \upsilon^r_{\textrm{abs}}(z)) = x\cdot(\sigma^r_{\textrm{abs}}(y+\tilde{\delta})+\upsilon^r_{\textrm{abs}}(z))$\\
  $B3$ & $x\parallel\tilde{\tau}=x$\\
  $TI0$ & $\tau_I(\dot{\delta}) = \dot{\delta}$\\
  $TI1$ & $a\notin I\quad \tau_I(\tilde{a})=\tilde{a}$\\
  $TI2$ & $a\in I\quad \tau_I(\tilde{a})=\tilde{\tau}$\\
  $SATI$ & $\tau_I(\sigma^p_{\textrm{abs}}(x)) = \sigma^p_{\textrm{abs}}(\tau_I(x))$\\
  $TI4$ & $\tau_I(x+y)=\tau_I(x)+\tau_I(y)$\\
  $TI5$ & $\tau_I(x\cdot y)=\tau_I(x)\cdot\tau_I(y)$\\
  $TI6$ & $\tau_I(x\parallel y)=\tau_I(x)\parallel\tau_I(y)$\\
\end{tabular}
\caption{Additional axioms of $\textrm{APTC}^{sat}_{\tau}(a\in A_{\tau\delta},p\geq 0,r>0)$}
\label{AxiomsForAPTCSATTau}
\end{table}
\end{center}

The additional transition rules of $\textrm{APTC}^{sat}_{\tau}$ is shown in Table \ref{TRForAPTCSATTau}.

\begin{center}
    \begin{table}
        $$\frac{\langle x,p\rangle\xrightarrow{a}\langle\surd,p\rangle}{\langle\tau_I(x),p\rangle\xrightarrow{a}\langle\surd,p\rangle}\quad a\notin I
        \quad\quad\frac{\langle x,p\rangle\xrightarrow{a}\langle x',p\rangle}{\langle\tau_I(x),p\rangle\xrightarrow{a}\langle\tau_I(x'),p\rangle}\quad a\notin I$$

        $$\frac{\langle x,p\rangle\xrightarrow{a}\langle\surd,p\rangle}{\langle\tau_I(x),p\rangle\xrightarrow{\tau}\langle\surd,p\rangle}\quad a\in I
        \quad\quad\frac{\langle x,p\rangle\xrightarrow{a}\langle x',p\rangle}{\langle\tau_I(x),p\rangle\xrightarrow{\tau}\langle\tau_I(x'),p\rangle}\quad a\in I$$

        $$\frac{\langle x,p\rangle\mapsto^r \langle x,p+r\rangle}{\langle\tau_I(x),p\rangle\mapsto^r\langle\tau_I(x),p+r\rangle}\quad\frac{\langle x,p\rangle\uparrow}{\langle\tau_I(x),p\rangle\uparrow}$$
        \caption{Transition rule of $\textrm{APTC}^{\textrm{sat}}_{\tau}(a\in A_{\tau},r>0,p\geq 0)$}
        \label{TRForAPTCSATTau}
    \end{table}
\end{center}

\begin{definition}[Basic terms of $\textrm{APTC}^{\textrm{sat}}_{\tau}$]
The set of basic terms of $\textrm{APTC}^{\textrm{sat}}_{\tau}$, $\mathcal{B}(\textrm{APTC}^{\textrm{sat}}_{\tau})$, is inductively defined as follows by two auxiliary sets $\mathcal{B}_0(\textrm{APTC}^{\textrm{sat}}_{\tau})$ and $\mathcal{B}_1(\textrm{APTC}^{\textrm{sat}}_{\tau})$:
\begin{enumerate}
  \item if $a\in A_{\delta\tau}$, then $\tilde{a} \in \mathcal{B}_1(\textrm{APTC}^{\textrm{sat}}_{\tau})$;
  \item if $a\in A_{\tau}$ and $t\in \mathcal{B}(\textrm{APTC}^{\textrm{sat}}_{\tau})$, then $\tilde{a}\cdot t \in \mathcal{B}_1(\textrm{APTC}^{\textrm{sat}}_{\tau})$;
  \item if $t,t'\in \mathcal{B}_1(\textrm{APTC}^{\textrm{sat}}_{\tau})$, then $t+t'\in \mathcal{B}_1(\textrm{APTC}^{\textrm{sat}}_{\tau})$;
  \item if $t,t'\in \mathcal{B}_1(\textrm{APTC}^{\textrm{sat}}_{\tau})$, then $t\parallel t'\in \mathcal{B}_1(\textrm{APTC}^{\textrm{sat}}_{\tau})$;
  \item if $t\in \mathcal{B}_1(\textrm{APTC}^{\textrm{sat}}_{\tau})$, then $t\in \mathcal{B}_0(\textrm{APTC}^{\textrm{sat}}_{\tau})$;
  \item if $p>0$ and $t\in \mathcal{B}_0(\textrm{APTC}^{\textrm{sat}}_{\tau})$, then $\sigma^p_{\textrm{abs}}(t) \in \mathcal{B}_0(\textrm{APTC}^{\textrm{sat}}_{\tau})$;
  \item if $p>0$, $t\in \mathcal{B}_1(\textrm{APTC}^{\textrm{sat}}_{\tau})$ and $t'\in \mathcal{B}_0(\textrm{APTC}^{\textrm{sat}}_{\tau})$, then $t+\sigma^p_{\textrm{abs}}(t') \in \mathcal{B}_0(\textrm{APTC}^{\textrm{sat}}_{\tau})$;
  \item if $t\in \mathcal{B}_0(\textrm{APTC}^{\textrm{sat}}_{\tau})$, then $\nu_{\textrm{abs}}(t) \in \mathcal{B}_0(\textrm{APTC}^{\textrm{sat}}_{\tau})$;
  \item $\dot{\delta}\in \mathcal{B}(\textrm{APTC}^{\textrm{sat}}_{\tau})$;
  \item if $t\in \mathcal{B}_0(\textrm{APTC}^{\textrm{sat}}_{\tau})$, then $t\in \mathcal{B}(\textrm{APTC}^{\textrm{sat}}_{\tau})$.
\end{enumerate}
\end{definition}

\begin{theorem}[Elimination theorem]
Let $p$ be a closed $\textrm{APTC}^{\textrm{sat}}_{\tau}$ term. Then there is a basic $\textrm{APTC}^{\textrm{sat}}_{\tau}$ term $q$ such that $\textrm{APTC}^{\textrm{sat}}_{\tau}\vdash p=q$.
\end{theorem}

\begin{proof}
It is sufficient to induct on the structure of the closed $\textrm{APTC}^{\textrm{sat}}_{\tau}$ term $p$. It can be proven that $p$ combined by the constants and operators of $\textrm{APTC}^{\textrm{sat}}_{\tau}$ exists an equal basic term $q$, and the other operators not included in the basic terms, such as $\upsilon_{\textrm{abs}}$, $\overline{\upsilon}_{\textrm{abs}}$, $\between$, $\mid$, $\partial_H$, $\Theta$, $\triangleleft$ and $\tau_I$ can be eliminated.
\end{proof}

\subsubsection{Connections}

\begin{theorem}[Conservativity of $\textrm{APTC}^{\textrm{sat}}_{\tau}$]
$\textrm{APTC}^{\textrm{sat}}_{\tau}$ is a conservative extension of $\textrm{APTC}^{\textrm{sat}}$.
\end{theorem}

\begin{proof}
It follows from the following two facts.

    \begin{enumerate}
      \item The transition rules of $\textrm{APTC}^{\textrm{sat}}$ are all source-dependent;
      \item The sources of the transition rules of $\textrm{APTC}^{\textrm{sat}}_{\tau}$ contain an occurrence of $\tilde{\tau}$, and $\tau_I$.
    \end{enumerate}

So, $\textrm{APTC}^{\textrm{sat}}_{\tau}$ is a conservative extension of $\textrm{APTC}^{\textrm{sat}}$, as desired.
\end{proof}

\subsubsection{Congruence}

\begin{theorem}[Congruence of $\textrm{APTC}^{\textrm{sat}}_{\tau}$]
Rooted branching truly concurrent bisimulation equivalences $\approx_{rbp}$, $\approx_{rbs}$ and $\approx_{rbhp}$ are all congruences with respect to $\textrm{APTC}^{\textrm{sat}}_{\tau}$. That is,
\begin{itemize}
  \item rooted branching pomset bisimulation equivalence $\approx_{rbp}$ is a congruence with respect to $\textrm{APTC}^{\textrm{sat}}_{\tau}$;
  \item rooted branching step bisimulation equivalence $\approx_{rbs}$ is a congruence with respect to $\textrm{APTC}^{\textrm{sat}}_{\tau}$;
  \item rooted branching hp-bisimulation equivalence $\approx_{rbhp}$ is a congruence with respect to $\textrm{APTC}^{\textrm{sat}}_{\tau}$.
\end{itemize}
\end{theorem}

\begin{proof}
It is easy to see that $\approx_{rbp}$, $\approx_{rbs}$, and $\approx_{rbhp}$ are all equivalent relations on $\textrm{APTC}^{\textrm{sat}}_{\tau}$ terms, it is only sufficient to prove that $\approx_{rbp}$, $\approx_{rbs}$, and $\approx_{rbhp}$ are all preserved by the operators $\tau_I$. It is trivial and we omit it.
\end{proof}

\subsubsection{Soundness}

\begin{theorem}[Soundness of $\textrm{APTC}^{\textrm{sat}}_{\tau}$]
The axiomatization of $\textrm{APTC}^{\textrm{sat}}_{\tau}$ is sound modulo rooted branching truly concurrent bisimulation equivalences $\approx_{rbp}$, $\approx_{rbs}$, and $\approx_{rbhp}$. That is,
\begin{enumerate}
  \item let $x$ and $y$ be $\textrm{APTC}^{\textrm{sat}}_{\tau}$ terms. If $\textrm{APTC}^{\textrm{sat}}_{\tau}\vdash x=y$, then $x\approx_{rbs} y$;
  \item let $x$ and $y$ be $\textrm{APTC}^{\textrm{sat}}_{\tau}$ terms. If $\textrm{APTC}^{\textrm{sat}}_{\tau}\vdash x=y$, then $x\approx_{rbp} y$;
  \item let $x$ and $y$ be $\textrm{APTC}^{\textrm{sat}}_{\tau}$ terms. If $\textrm{APTC}^{\textrm{sat}}_{\tau}\vdash x=y$, then $x\approx_{rbhp} y$.
\end{enumerate}
\end{theorem}

\begin{proof}
Since $\approx_{rbp}$, $\approx_{rbs}$, and $\approx_{rbhp}$ are both equivalent and congruent relations, we only need to check if each axiom in Table \ref{AxiomsForAPTCSATTau} is sound modulo $\approx_{rbp}$, $\approx_{rbs}$, and $\approx_{rbhp}$ respectively.

\begin{enumerate}
  \item Each axiom in Table \ref{AxiomsForAPTCSATTau} can be checked that it is sound modulo rooted branching step bisimulation equivalence, by transition rules in Table \ref{TRForAPTCSATTau}. We omit them.
  \item From the definition of rooted branching pomset bisimulation $\approx_{rbp}$, we know that rooted branching pomset bisimulation $\approx_{rbp}$ is defined by weak pomset transitions, which are labeled by pomsets with $\tilde{\tau}$. In a weak pomset transition, the events in the pomset are either within causality relations (defined by $\cdot$) or in concurrency (implicitly defined by $\cdot$ and $+$, and explicitly defined by $\between$), of course, they are pairwise consistent (without conflicts). We have already proven the case that all events are pairwise concurrent, so, we only need to prove the case of events in causality. Without loss of generality, we take a pomset of $P=\{\tilde{a},\tilde{b}:\tilde{a}\cdot \tilde{b}\}$. Then the weak pomset transition labeled by the above $P$ is just composed of one single event transition labeled by $\tilde{a}$ succeeded by another single event transition labeled by $\tilde{b}$, that is, $\xRightarrow{P}=\xRightarrow{a}\xRightarrow{b}$.

        Similarly to the proof of soundness modulo rooted branching step bisimulation $\approx_{rbs}$, we can prove that each axiom in Table \ref{AxiomsForAPTCSATTau} is sound modulo rooted branching pomset bisimulation $\approx_{rbp}$, we omit them.

  \item From the definition of rooted branching hp-bisimulation $\approx_{rbhp}$, we know that rooted branching hp-bisimulation $\approx_{rbhp}$ is defined on the weakly posetal product $(C_1,f,C_2),f:\hat{C_1}\rightarrow \hat{C_2}\textrm{ isomorphism}$. Two process terms $s$ related to $C_1$ and $t$ related to $C_2$, and $f:\hat{C_1}\rightarrow \hat{C_2}\textrm{ isomorphism}$. Initially, $(C_1,f,C_2)=(\emptyset,\emptyset,\emptyset)$, and $(\emptyset,\emptyset,\emptyset)\in\approx_{rbhp}$. When $s\xrightarrow{a}s'$ ($C_1\xrightarrow{a}C_1'$), there will be $t\xRightarrow{a}t'$ ($C_2\xRightarrow{a}C_2'$), and we define $f'=f[a\mapsto a]$. Then, if $(C_1,f,C_2)\in\approx_{rbhp}$, then $(C_1',f',C_2')\in\approx_{rbhp}$.

        Similarly to the proof of soundness modulo rooted branching pomset bisimulation equivalence, we can prove that each axiom in Table \ref{AxiomsForAPTCSATTau} is sound modulo rooted branching hp-bisimulation equivalence, we just need additionally to check the above conditions on rooted branching hp-bisimulation, we omit them.
\end{enumerate}

\end{proof}

\subsubsection{Completeness}

For $\textrm{APTC}^{\textrm{sat}}_{\tau}$ + Rec, it is similar to $\textrm{APTC}^{\textrm{sat}}$ + Rec, except that $\tilde{\tau}\cdot X$ is forbidden in recursive specifications for the sake of fairness. Like APTC, the proof of completeness need the help of $CFAR$ (see section \ref{tcpa}).

\begin{theorem}[Completeness of $\textrm{APTC}^{\textrm{sat}}_{\tau}$ + CFAR + guarded linear Rec]
The axiomatization of $\textrm{APTC}^{\textrm{dat}}_{\tau}$ + CFAR + guarded linear Rec is complete modulo rooted branching truly concurrent bisimulation equivalences $\approx_{rbs}$, $\approx_{rbp}$, and $\approx_{rbhp}$. That is,
\begin{enumerate}
  \item let $p$ and $q$ be closed $\textrm{APTC}^{\textrm{sat}}_{\tau}$ + CFAR + guarded linear Rec terms, if $p\approx_{rbs} q$ then $p=q$;
  \item let $p$ and $q$ be closed $\textrm{APTC}^{\textrm{sat}}_{\tau}$ + CFAR + guarded linear Rec terms, if $p\approx_{rbp} q$ then $p=q$;
  \item let $p$ and $q$ be closed $\textrm{APTC}^{\textrm{sat}}_{\tau}$ + CFAR + guarded linear Rec terms, if $p\approx_{rbhp} q$ then $p=q$.
\end{enumerate}

\end{theorem}

\begin{proof}
Firstly, we know that each process term in $\textrm{APTC}^{\textrm{sat}}_{\tau}$  + CFAR + guarded linear Rec is equal to a process term $\langle X_1|E\rangle$ with $E$ a linear recursive specification.

It remains to prove the following cases.

\begin{enumerate}
  \item If $\langle X_1|E_1\rangle \approx_{rbs} \langle Y_1|E_2\rangle$ for linear recursive specification $E_1$ and $E_2$, then $\langle X_1|E_1\rangle = \langle Y_1|E_2\rangle$.

        It can be proven similarly to the completeness of $\textrm{APTC}_{\tau}$ + CFAR + linear Rec, see \cite{ATC}.

  \item If $\langle X_1|E_1\rangle \approx_{rbp} \langle Y_1|E_2\rangle$ for linear recursive specification $E_1$ and $E_2$, then $\langle X_1|E_1\rangle = \langle Y_1|E_2\rangle$.

        It can be proven similarly, just by replacement of $\approx_{rbs}$ by $\approx_{rbp}$, we omit it.
  \item If $\langle X_1|E_1\rangle \approx_{rbhp} \langle Y_1|E_2\rangle$ for linear recursive specification $E_1$ and $E_2$, then $\langle X_1|E_1\rangle = \langle Y_1|E_2\rangle$.

        It can be proven similarly, just by replacement of $\approx_{rbs}$ by $\approx_{rbhp}$, we omit it.
\end{enumerate}
\end{proof}

\section{Applications}{\label{app}}

APTC with timing provides a formal framework based on truly concurrent behavioral semantics, which can be used to verify the correctness of system behaviors with timing. In this section, we choose one protocol verified by APTC \cite{ATC} -- alternating bit protocol (ABP) \cite{ABP}.

The ABP protocol is used to ensure successful transmission of data through a corrupted channel. This success is based on the assumption that data can be resent an unlimited number of times, which is illustrated in Fig.\ref{ABP}, we alter it into the true concurrency situation.

\begin{enumerate}
  \item Data elements $d_1,d_2,d_3,\cdots$ from a finite set $\Delta$ are communicated between a Sender and a Receiver.
  \item If the Sender reads a datum from channel $A_1$, then this datum is sent to the Receiver in parallel through channel $A_2$.
  \item The Sender processes the data in $\Delta$, formes new data, and sends them to the Receiver through channel $B$.
  \item And the Receiver sends the datum into channel $C$.
  \item If channel $B$ is corrupted, the message communicated through $B$ can be turn into an error message $\bot$.
  \item Every time the Receiver receives a message via channel $B$, it sends an acknowledgement to the Sender via channel $D$, which is also corrupted.
  \item Finally, then Sender and the Receiver send out their outputs in parallel through channels $C_1$ and $C_2$.
\end{enumerate}

\begin{figure}
    \centering
    \includegraphics{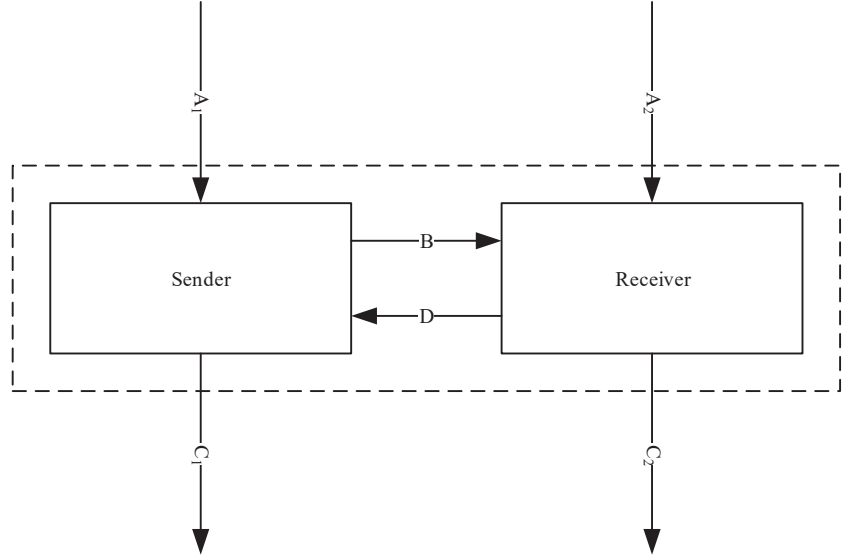}
    \caption{Alternating bit protocol}
    \label{ABP}
\end{figure}

In the truly concurrent ABP, the Sender sends its data to the Receiver; and the Receiver can also send its data to the Sender, for simplicity and without loss of generality, we assume that only the Sender sends its data and the Receiver only receives the data from the Sender. The Sender attaches a bit 0 to data elements $d_{2k-1}$ and a bit 1 to data elements $d_{2k}$, when they are sent into channel $B$. When the Receiver reads a datum, it sends back the attached bit via channel $D$. If the Receiver receives a corrupted message, then it sends back the previous acknowledgement to the Sender.

\subsection{Discrete Relative Timing}

Time is divided into slices, $t_1,t_2$ are the times that it takes the different processes to send, receive, etc. $t_1'$ is the time-out time of the sender and $t_2'$ is the time-out of the receiver.

Then the state transition of the Sender can be described by $\textrm{APTC}^{\textrm{drt}}_{\tau}$ + Rec as follows.

\begin{eqnarray}
&&S_b=\sum_{d\in\Delta}\underline{\underline{r_{A_1}(d)}}\cdot T_{db}\nonumber\\
&&T_{db}=(\sum_{d'\in\Delta}(\underline{\underline{s_B(d',b)}}\cdot \underline{\underline{s_{C_1}(d')}})+\underline{\underline{s_B(\bot)}})\cdot \sigma^{t_1}_{\textrm{rel}}(U_{db})+\sigma^1_{\textrm{rel}}(T_{db})\nonumber\\
&&U_{db}=\sum_{k<t_1'}\sigma^k_{\textrm{rel}}(\underline{\underline{r_D(b)}})\cdot S_{1-b}+\sum_{k<t_1'}\sigma^k_{\textrm{rel}}((\underline{\underline{r_D(1-b)}}+\underline{\underline{r_D(\bot)}}))\cdot \sigma^{t_1'}_{\textrm{rel}}(T_{db})\nonumber
\end{eqnarray}

where $s_B$ denotes sending data through channel $B$, $r_D$ denotes receiving data through channel $D$, similarly, $r_{A_1}$ means receiving data via channel $A_1$, $s_{C_1}$ denotes sending data via channel $C_1$, and $b\in\{0,1\}$.

And the state transition of the Receiver can be described by $\textrm{APTC}^{\textrm{drt}}_{\tau}$ + Rec as follows.

\begin{eqnarray}
&&R_b=\sum_{d\in\Delta}\underline{\underline{r_{A_2}(d)}}\cdot R_b'\nonumber\\
&&R_b'=\sum_{d'\in\Delta}\{\underline{\underline{r_B(d',b)}}\cdot \sigma^{t_2}_{\textrm{rel}}(\underline{\underline{s_{C_2}(d')}})\cdot Q_b+\underline{\underline{r_B(d',1-b)}}\cdot Q_{1-b}\}+\underline{\underline{r_B(\bot)}}\cdot Q_{1-b} + \sigma^1_{\textrm{rel}}(R_b')\nonumber\\
&&Q_b=\sigma^{t_2'}_{\textrm{rel}}(\underline{\underline{s_D(b)}}+\underline{\underline{s_D(\bot)}})\cdot R_{1-b}\nonumber
\end{eqnarray}

where $r_{A_2}$ denotes receiving data via channel $A_2$, $r_B$ denotes receiving data via channel $B$, $s_{C_2}$ denotes sending data via channel $C_2$, $s_D$ denotes sending data via channel $D$, and $b\in\{0,1\}$.

The send action and receive action of the same data through the same channel can communicate each other, otherwise, a deadlock $\delta$ will be caused. We define the following communication functions.

\begin{eqnarray}
&&\gamma(\underline{\underline{s_B(d',b)}},\underline{\underline{r_B(d',b)}})\triangleq \underline{\underline{c_B(d',b)}}\nonumber\\
&&\gamma(\underline{\underline{s_B(\bot)}},\underline{\underline{r_B(\bot)}})\triangleq \underline{\underline{c_B(\bot)}}\nonumber\\
&&\gamma(\underline{\underline{s_D(b)}},\underline{\underline{r_D(b)}})\triangleq \underline{\underline{c_D(b)}}\nonumber\\
&&\gamma(\underline{\underline{s_D(\bot)}},\underline{\underline{r_D(\bot)}})\triangleq \underline{\underline{c_D(\bot)}}\nonumber
\end{eqnarray}

Let $R_0$ and $S_0$ be in parallel, then the system $R_0S_0$ can be represented by the following process term.

$$\tau_I(\partial_H(\Theta(R_0\between S_0)))=\tau_I(\partial_H(R_0\between S_0))$$

where $H=\{\underline{\underline{s_B(d',b)}},\underline{\underline{r_B(d',b)}},\underline{\underline{s_D(b)}},\underline{\underline{r_D(b)}}|d'\in\Delta,b\in\{0,1\}\}\\
\{\underline{\underline{s_B(\bot)}},\underline{\underline{r_B(\bot)}},\underline{\underline{s_D(\bot)}},\underline{\underline{r_D(\bot)}}\}$

$I=\{\underline{\underline{c_B(d',b)}},\underline{\underline{c_D(b)}}|d'\in\Delta,b\in\{0,1\}\}\cup\{\underline{\underline{c_B(\bot)}},\underline{\underline{c_D(\bot)}}\}$.

Then we get the following conclusion.

\begin{theorem}[Correctness of the ABP protocol with discrete relative timing]
The ABP protocol $\tau_I(\partial_H(R_0\between S_0))$ exhibits desired external behaviors with discrete relative timing.
\end{theorem}

\begin{proof}
We get $\tau_I(\partial_H(R_0\between S_0))=\sum_{d,d'\in \Delta}(\underline{\underline{r_{A_1}(d)}}\parallel \underline{\underline{r_{A_2}(d)}})\cdot (\underline{\underline{s_{C_1}(d')}}\parallel \underline{\underline{s_{C_2}(d')}})\cdot \tau_I(\partial_H(R_0\between S_0))$. So, the ABP protocol $\tau_I(\partial_H(R_0\between S_0))$ exhibits desired external behaviors with discrete relative timing.
\end{proof}

\subsection{Discrete Absolute Timing}

Time is divided into slices, $t_1,t_2$ are the times that it takes the different processes to send, receive, etc. $t_1'$ is the time-out time of the sender and $t_2'$ is the time-out of the receiver.

Then the state transition of the Sender can be described by $\textrm{APTC}^{\textrm{dat}}_{\tau}$ + Rec as follows.

\begin{eqnarray}
&&S_b=\sum_{d\in\Delta}\underline{r_{A_1}(d)}\cdot T_{db}\nonumber\\
&&T_{db}=(\sum_{d'\in\Delta}(\underline{s_B(d',b)}\cdot \underline{s_{C_1}(d')})+\underline{s_B(\bot)})\cdot \sigma^{t_1}_{\textrm{abs}}(U_{db})+\sigma^1_{\textrm{abs}}(T_{db})\nonumber\\
&&U_{db}=\sum_{k<t_1'}\sigma^k_{\textrm{abs}}(\underline{r_D(b)})\cdot S_{1-b}+\sum_{k<t_1'}\sigma^k_{\textrm{abs}}((\underline{r_D(1-b)}+\underline{r_D(\bot)}))\cdot \sigma^{t_1'}_{\textrm{abs}}(T_{db})\nonumber
\end{eqnarray}

where $s_B$ denotes sending data through channel $B$, $r_D$ denotes receiving data through channel $D$, similarly, $r_{A_1}$ means receiving data via channel $A_1$, $s_{C_1}$ denotes sending data via channel $C_1$, and $b\in\{0,1\}$.

And the state transition of the Receiver can be described by $\textrm{APTC}^{\textrm{dat}}_{\tau}$ + Rec as follows.

\begin{eqnarray}
&&R_b=\sum_{d\in\Delta}\underline{r_{A_2}(d)}\cdot R_b'\nonumber\\
&&R_b'=\sum_{d'\in\Delta}\{\underline{r_B(d',b)}\cdot \sigma^{t_2}_{\textrm{abs}}(\underline{s_{C_2}(d')})\cdot Q_b+\underline{r_B(d',1-b)}\cdot Q_{1-b}\}+\underline{r_B(\bot)}\cdot Q_{1-b} + \sigma^1_{\textrm{abs}}(R_b')\nonumber\\
&&Q_b=\sigma^{t_2'}_{\textrm{abs}}(\underline{s_D(b)}+\underline{s_D(\bot)})\cdot R_{1-b}\nonumber
\end{eqnarray}

where $r_{A_2}$ denotes receiving data via channel $A_2$, $r_B$ denotes receiving data via channel $B$, $s_{C_2}$ denotes sending data via channel $C_2$, $s_D$ denotes sending data via channel $D$, and $b\in\{0,1\}$.

The send action and receive action of the same data through the same channel can communicate each other, otherwise, a deadlock $\delta$ will be caused. We define the following communication functions.

\begin{eqnarray}
&&\gamma(\underline{s_B(d',b)},\underline{r_B(d',b)})\triangleq \underline{c_B(d',b)}\nonumber\\
&&\gamma(\underline{s_B(\bot)},\underline{r_B(\bot)})\triangleq \underline{c_B(\bot)}\nonumber\\
&&\gamma(\underline{s_D(b)},\underline{r_D(b)})\triangleq \underline{c_D(b)}\nonumber\\
&&\gamma(\underline{s_D(\bot)},\underline{r_D(\bot)})\triangleq \underline{c_D(\bot)}\nonumber
\end{eqnarray}

Let $R_0$ and $S_0$ be in parallel, then the system $R_0S_0$ can be represented by the following process term.

$$\tau_I(\partial_H(\Theta(R_0\between S_0)))=\tau_I(\partial_H(R_0\between S_0))$$

where $H=\{\underline{s_B(d',b)},\underline{r_B(d',b)},\underline{s_D(b)},\underline{r_D(b)}|d'\in\Delta,b\in\{0,1\}\}\\
\{\underline{s_B(\bot)},\underline{r_B(\bot)},\underline{s_D(\bot)},\underline{r_D(\bot)}\}$

$I=\{\underline{c_B(d',b)},\underline{c_D(b)}|d'\in\Delta,b\in\{0,1\}\}\cup\{\underline{c_B(\bot)},\underline{c_D(\bot)}\}$.

Then we get the following conclusion.

\begin{theorem}[Correctness of the ABP protocol with discrete absolute timing]
The ABP protocol $\tau_I(\partial_H(R_0\between S_0))$ exhibits desired external behaviors with discrete absolute timing.
\end{theorem}

\begin{proof}
We get $\tau_I(\partial_H(R_0\between S_0))=\sum_{d,d'\in \Delta}(\underline{r_{A_1}(d)}\parallel \underline{r_{A_2}(d)})\cdot (\underline{s_{C_1}(d')}\parallel \underline{s_{C_2}(d')})\cdot \tau_I(\partial_H(R_0\between S_0))$. So, the ABP protocol $\tau_I(\partial_H(R_0\between S_0))$ exhibits desired external behaviors with discrete absolute timing.
\end{proof}

\subsection{Continuous Relative Timing}

Time is denoted by time point, $t_1,t_2$ are the time points that it takes the different processes to send, receive, etc. $t_1'$ is the time-out time of the sender and $t_2'$ is the time-out of the receiver.

Then the state transition of the Sender can be described by $\textrm{APTC}^{\textrm{srt}}_{\tau}$ + Rec as follows.

\begin{eqnarray}
&&S_b=\sum_{d\in\Delta}\widetilde{\widetilde{r_{A_1}(d)}}\cdot T_{db}\nonumber\\
&&T_{db}=(\sum_{d'\in\Delta}(\widetilde{\widetilde{s_B(d',b)}}\cdot \widetilde{\widetilde{s_{C_1}(d')}})+\widetilde{\widetilde{s_B(\bot)}})\cdot \sigma^{t_1}_{\textrm{rel}}(U_{db})+\sigma^r_{\textrm{rel}}(T_{db})\nonumber\\
&&U_{db}=\sum_{k<t_1'}\sigma^k_{\textrm{rel}}(\widetilde{\widetilde{r_D(b)}})\cdot S_{1-b}+\sum_{k<t_1'}\sigma^k_{\textrm{rel}}((\widetilde{\widetilde{r_D(1-b)}}+\widetilde{\widetilde{r_D(\bot)}}))\cdot \sigma^{t_1'}_{\textrm{rel}}(T_{db})\nonumber
\end{eqnarray}

where $s_B$ denotes sending data through channel $B$, $r_D$ denotes receiving data through channel $D$, similarly, $r_{A_1}$ means receiving data via channel $A_1$, $s_{C_1}$ denotes sending data via channel $C_1$, and $b\in\{0,1\}$.

And the state transition of the Receiver can be described by $\textrm{APTC}^{\textrm{srt}}_{\tau}$ + Rec as follows.

\begin{eqnarray}
&&R_b=\sum_{d\in\Delta}\widetilde{\widetilde{r_{A_2}(d)}}\cdot R_b'\nonumber\\
&&R_b'=\sum_{d'\in\Delta}\{\widetilde{\widetilde{r_B(d',b)}}\cdot \sigma^{t_2}_{\textrm{rel}}(\widetilde{\widetilde{s_{C_2}(d')}})\cdot Q_b+\widetilde{\widetilde{r_B(d',1-b)}}\cdot Q_{1-b}\}+\widetilde{\widetilde{r_B(\bot)}}\cdot Q_{1-b} + \sigma^r_{\textrm{rel}}(R_b')\nonumber\\
&&Q_b=\sigma^{t_2'}_{\textrm{rel}}(\widetilde{\widetilde{s_D(b)}}+\widetilde{\widetilde{s_D(\bot)}})\cdot R_{1-b}\nonumber
\end{eqnarray}

where $r_{A_2}$ denotes receiving data via channel $A_2$, $r_B$ denotes receiving data via channel $B$, $s_{C_2}$ denotes sending data via channel $C_2$, $s_D$ denotes sending data via channel $D$, and $b\in\{0,1\}$.

The send action and receive action of the same data through the same channel can communicate each other, otherwise, a deadlock $\delta$ will be caused. We define the following communication functions.

\begin{eqnarray}
&&\gamma(\widetilde{\widetilde{s_B(d',b)}},\widetilde{\widetilde{r_B(d',b)}})\triangleq \widetilde{\widetilde{c_B(d',b)}}\nonumber\\
&&\gamma(\widetilde{\widetilde{s_B(\bot)}},\widetilde{\widetilde{r_B(\bot)}})\triangleq \widetilde{\widetilde{c_B(\bot)}}\nonumber\\
&&\gamma(\widetilde{\widetilde{s_D(b)}},\widetilde{\widetilde{r_D(b)}})\triangleq \widetilde{\widetilde{c_D(b)}}\nonumber\\
&&\gamma(\widetilde{\widetilde{s_D(\bot)}},\widetilde{\widetilde{r_D(\bot)}})\triangleq \widetilde{\widetilde{c_D(\bot)}}\nonumber
\end{eqnarray}

Let $R_0$ and $S_0$ be in parallel, then the system $R_0S_0$ can be represented by the following process term.

$$\tau_I(\partial_H(\Theta(R_0\between S_0)))=\tau_I(\partial_H(R_0\between S_0))$$

where $H=\{\widetilde{\widetilde{s_B(d',b)}},\widetilde{\widetilde{r_B(d',b)}},\widetilde{\widetilde{s_D(b)}},\widetilde{\widetilde{r_D(b)}}|d'\in\Delta,b\in\{0,1\}\}\\
\{\widetilde{\widetilde{s_B(\bot)}},\widetilde{\widetilde{r_B(\bot)}},\widetilde{\widetilde{s_D(\bot)}},\widetilde{\widetilde{r_D(\bot)}}\}$

$I=\{\widetilde{\widetilde{c_B(d',b)}},\widetilde{\widetilde{c_D(b)}}|d'\in\Delta,b\in\{0,1\}\}\cup\{\widetilde{\widetilde{c_B(\bot)}},\widetilde{\widetilde{c_D(\bot)}}\}$.

Then we get the following conclusion.

\begin{theorem}[Correctness of the ABP protocol with continuous relative timing]
The ABP protocol $\tau_I(\partial_H(R_0\between S_0))$ exhibits desired external behaviors with continuous relative timing.
\end{theorem}

\begin{proof}
We get $\tau_I(\partial_H(R_0\between S_0))=\sum_{d,d'\in \Delta}(\widetilde{\widetilde{r_{A_1}(d)}}\parallel \widetilde{\widetilde{r_{A_2}(d)}})\cdot (\widetilde{\widetilde{s_{C_1}(d')}}\parallel \widetilde{\widetilde{s_{C_2}(d')}})\cdot \tau_I(\partial_H(R_0\between S_0))$. So, the ABP protocol $\tau_I(\partial_H(R_0\between S_0))$ exhibits desired external behaviors with continuous relative timing.
\end{proof}

\subsection{Continuous Absolute Timing}

Time is divided into slices, $t_1,t_2$ are the times that it takes the different processes to send, receive, etc. $t_1'$ is the time-out time of the sender and $t_2'$ is the time-out of the receiver.

Then the state transition of the Sender can be described by $\textrm{APTC}^{\textrm{dat}}_{\tau}$ + Rec as follows.

\begin{eqnarray}
&&S_b=\sum_{d\in\Delta}\widetilde{r_{A_1}(d)}\cdot T_{db}\nonumber\\
&&T_{db}=(\sum_{d'\in\Delta}(\widetilde{s_B(d',b)}\cdot \widetilde{s_{C_1}(d')})+\widetilde{s_B(\bot)})\cdot \sigma^{t_1}_{\textrm{abs}}(U_{db})+\sigma^r_{\textrm{abs}}(T_{db})\nonumber\\
&&U_{db}=\sum_{k<t_1'}\sigma^k_{\textrm{abs}}(\widetilde{r_D(b)})\cdot S_{1-b}+\sum_{k<t_1'}\sigma^k_{\textrm{abs}}((\widetilde{r_D(1-b)}+\widetilde{r_D(\bot)}))\cdot \sigma^{t_1'}_{\textrm{abs}}(T_{db})\nonumber
\end{eqnarray}

where $s_B$ denotes sending data through channel $B$, $r_D$ denotes receiving data through channel $D$, similarly, $r_{A_1}$ means receiving data via channel $A_1$, $s_{C_1}$ denotes sending data via channel $C_1$, and $b\in\{0,1\}$.

And the state transition of the Receiver can be described by $\textrm{APTC}^{\textrm{dat}}_{\tau}$ + Rec as follows.

\begin{eqnarray}
&&R_b=\sum_{d\in\Delta}\widetilde{r_{A_2}(d)}\cdot R_b'\nonumber\\
&&R_b'=\sum_{d'\in\Delta}\{\widetilde{r_B(d',b)}\cdot \sigma^{t_2}_{\textrm{abs}}(\widetilde{s_{C_2}(d')})\cdot Q_b+\widetilde{r_B(d',1-b)}\cdot Q_{1-b}\}+\widetilde{r_B(\bot)}\cdot Q_{1-b} + \sigma^r_{\textrm{abs}}(R_b')\nonumber\\
&&Q_b=\sigma^{t_2'}_{\textrm{abs}}(\widetilde{s_D(b)}+\widetilde{s_D(\bot)})\cdot R_{1-b}\nonumber
\end{eqnarray}

where $r_{A_2}$ denotes receiving data via channel $A_2$, $r_B$ denotes receiving data via channel $B$, $s_{C_2}$ denotes sending data via channel $C_2$, $s_D$ denotes sending data via channel $D$, and $b\in\{0,1\}$.

The send action and receive action of the same data through the same channel can communicate each other, otherwise, a deadlock $\delta$ will be caused. We define the following communication functions.

\begin{eqnarray}
&&\gamma(\widetilde{s_B(d',b)},\widetilde{r_B(d',b)})\triangleq \widetilde{c_B(d',b)}\nonumber\\
&&\gamma(\widetilde{s_B(\bot)},\widetilde{r_B(\bot)})\triangleq \widetilde{c_B(\bot)}\nonumber\\
&&\gamma(\widetilde{s_D(b)},\widetilde{r_D(b)})\triangleq \widetilde{c_D(b)}\nonumber\\
&&\gamma(\widetilde{s_D(\bot)},\widetilde{r_D(\bot)})\triangleq \widetilde{c_D(\bot)}\nonumber
\end{eqnarray}

Let $R_0$ and $S_0$ be in parallel, then the system $R_0S_0$ can be represented by the following process term.

$$\tau_I(\partial_H(\Theta(R_0\between S_0)))=\tau_I(\partial_H(R_0\between S_0))$$

where $H=\{\widetilde{s_B(d',b)},\widetilde{r_B(d',b)},\widetilde{s_D(b)},\widetilde{r_D(b)}|d'\in\Delta,b\in\{0,1\}\}\\
\{\widetilde{s_B(\bot)},\widetilde{r_B(\bot)},\widetilde{s_D(\bot)},\widetilde{r_D(\bot)}\}$

$I=\{\widetilde{c_B(d',b)},\widetilde{c_D(b)}|d'\in\Delta,b\in\{0,1\}\}\cup\{\widetilde{c_B(\bot)},\widetilde{c_D(\bot)}\}$.

Then we get the following conclusion.

\begin{theorem}[Correctness of the ABP protocol with continuous absolute timing]
The ABP protocol $\tau_I(\partial_H(R_0\between S_0))$ exhibits desired external behaviors with continuous absolute timing.
\end{theorem}

\begin{proof}
We get $\tau_I(\partial_H(R_0\between S_0))=\sum_{d,d'\in \Delta}(\widetilde{r_{A_1}(d)}\parallel \widetilde{r_{A_2}(d)})\cdot (\widetilde{s_{C_1}(d')}\parallel \widetilde{s_{C_2}(d')})\cdot \tau_I(\partial_H(R_0\between S_0))$. So, the ABP protocol $\tau_I(\partial_H(R_0\between S_0))$ exhibits desired external behaviors with continuous absolute timing.
\end{proof}

\section{Extensions}{\label{ext}}

In the above sections, we have already seen the modular structure of APTC with timing by use of the concepts of conservative extension and generalization, just like APTC \cite{ATC} and ACP \cite{ACP}. New computational properties can be extended elegantly based on the modular structure.

In this section, we show the extension mechanism of APTC with timing by extending a new renaming property. We will introduce $\textrm{APTC}^{\textrm{drt}}_{\tau}$ + Rec with renaming called $\textrm{APTC}^{\textrm{drt}}_{\tau}$ + Rec + renaming, $\textrm{APTC}^{\textrm{dat}}_{\tau}$ + Rec with renaming called $\textrm{APTC}^{\textrm{dat}}_{\tau}$ + Rec + renaming, $\textrm{APTC}^{\textrm{srt}}_{\tau}$ +Rec with renaming called $\textrm{APTC}^{\textrm{srt}}_{\tau}$ + Rec + renaming, and $\textrm{APTC}^{\textrm{sat}}_{\tau}$ +Rec with renaming called $\textrm{APTC}^{\textrm{sat}}_{\tau}$ + Rec + renaming, respectively.

\subsection{Discrete Relative Timing}

\begin{definition}[Signature of $\textrm{APTC}^{\textrm{drt}}_{\tau}$ + Rec + renaming]
The signature of $\textrm{APTC}^{\textrm{drt}}_{\tau}$ + Rec + renaming consists of the signature of $\textrm{APTC}^{\textrm{drt}}_{\tau}$ + Rec, and the renaming operator $\rho_f: \mathcal{P}_{\textrm{rel}}\rightarrow\mathcal{P}_{\textrm{rel}}$.
\end{definition}

The axioms of $\textrm{APTC}^{\textrm{drt}}_{\tau}$ + Rec + renaming include the laws of $\textrm{APTC}^{\textrm{drt}}_{\tau}$ + Rec, and the axioms of renaming operator in Table \ref{AxiomsForAPTCDRTTauRen}.

\begin{center}
\begin{table}
  \begin{tabular}{@{}ll@{}}
\hline No. &Axiom\\
  $DRTRN1$ & $\rho_f(\underline{\underline{a}})=f(\underline{\underline{a}})$\\
  $DRTRN2$ & $\rho_f(\dot{\delta})=\dot{\delta}$\\
  $DRTRN3$ & $\rho_f(\underline{\underline{\tau}})=\underline{\underline{\tau}}$\\
  $DRTRN$ & $\rho_f(\sigma^n_{\textrm{rel}}(x)) = \sigma^n_{\textrm{rel}}(\rho_f(x))$\\
  $RN3$ & $\rho_f(x+y)=\rho_f(x)+\rho_f(y)$\\
  $RN4$ & $\rho_f(x\cdot y)=\rho_f(x)\cdot\rho_f(y)$\\
  $RN5$ & $\rho_f(x\parallel y)=\rho_f(x)\parallel\rho_f(y)$\\
\end{tabular}
\caption{Additional axioms of renaming operator $(a\in A_{\tau\delta},n\geq 0)$}
\label{AxiomsForAPTCDRTTauRen}
\end{table}
\end{center}

The additional transition rules of renaming operator is shown in Table \ref{TRForAPTCDRTTauRen}.

\begin{center}
    \begin{table}
        $$\frac{x\xrightarrow{a}\surd}{\rho_f(x)\xrightarrow{f(a)}\surd}
        \quad\frac{x\xrightarrow{a}x'}{\rho_f(x)\xrightarrow{f(a)}\rho_f(x')}$$

        $$\frac{x\mapsto^m x'}{\rho_f(x)\mapsto^m\rho_f(x')}\quad\frac{x\uparrow}{\rho_f(x)\uparrow}$$
        \caption{Transition rule of renaming operator $(a\in A_{\tau},m>0,n\geq 0)$}
        \label{TRForAPTCDRTTauRen}
    \end{table}
\end{center}

\begin{theorem}[Elimination theorem]
Let $p$ be a closed $\textrm{APTC}^{\textrm{drt}}_{\tau}$ + Rec + renaming term. Then there is a basic $\textrm{APTC}^{\textrm{drt}}_{\tau}$ + Rec term $q$ such that $\textrm{APTC}^{\textrm{drt}}_{\tau}\textrm{ + renaming}\vdash p=q$.
\end{theorem}

\begin{proof}
It is sufficient to induct on the structure of the closed $\textrm{APTC}^{\textrm{drt}}_{\tau}$ + Rec + renaming term $p$. It can be proven that $p$ combined by the constants and operators of $\textrm{APTC}^{\textrm{drt}}_{\tau}$ + Rec + renaming exists an equal basic $\textrm{APTC}^{\textrm{drt}}_{\tau}$ + Rec term $q$, and the other operators not included in the basic terms, such as $\rho_f$ can be eliminated.
\end{proof}

 \subsubsection{Connections}

\begin{theorem}[Conservativity of $\textrm{APTC}^{\textrm{drt}}_{\tau}$ + Rec + renaming]
$\textrm{APTC}^{\textrm{drt}}_{\tau}$ + Rec + renaming is a conservative extension of $\textrm{APTC}^{\textrm{drt}}_{\tau}$ + Rec.
\end{theorem}

\begin{proof}
It follows from the following two facts.

    \begin{enumerate}
      \item The transition rules of $\textrm{APTC}^{\textrm{drt}}_{\tau}$ + Rec are all source-dependent;
      \item The sources of the transition rules of $\textrm{APTC}^{\textrm{drt}}_{\tau}$ + Rec + renaming contain an occurrence of $\rho_f$.
    \end{enumerate}

So, $\textrm{APTC}^{\textrm{drt}}_{\tau}$ + Rec + renaming is a conservative extension of $\textrm{APTC}^{\textrm{drt}}_{\tau}$ + Rec, as desired.
\end{proof}

\subsubsection{Congruence}

\begin{theorem}[Congruence of $\textrm{APTC}^{\textrm{drt}}_{\tau}$ + Rec + renaming]
Rooted branching truly concurrent bisimulation equivalences $\approx_{rbp}$, $\approx_{rbs}$ and $\approx_{rbhp}$ are all congruences with respect to $\textrm{APTC}^{\textrm{drt}}_{\tau}$ + Rec + renaming. That is,
\begin{itemize}
  \item rooted branching pomset bisimulation equivalence $\approx_{rbp}$ is a congruence with respect to $\textrm{APTC}^{\textrm{drt}}_{\tau}$ + Rec + renaming;
  \item rooted branching step bisimulation equivalence $\approx_{rbs}$ is a congruence with respect to $\textrm{APTC}^{\textrm{drt}}_{\tau}$ + Rec + renaming;
  \item rooted branching hp-bisimulation equivalence $\approx_{rbhp}$ is a congruence with respect to $\textrm{APTC}^{\textrm{drt}}_{\tau}$ + Rec + renaming.
\end{itemize}
\end{theorem}

\begin{proof}
It is easy to see that $\approx_{rbp}$, $\approx_{rbs}$, and $\approx_{rbhp}$ are all equivalent relations on $\textrm{APTC}^{\textrm{drt}}_{\tau}$ + Rec + renaming terms, it is only sufficient to prove that $\approx_{rbp}$, $\approx_{rbs}$, and $\approx_{rbhp}$ are all preserved by the operators $\rho_f$. It is trivial and we omit it.
\end{proof}

\subsubsection{Soundness}

\begin{theorem}[Soundness of $\textrm{APTC}^{\textrm{drt}}_{\tau}$ + Rec + renaming]
The axiomatization of $\textrm{APTC}^{\textrm{drt}}_{\tau}$ + Rec + renaming is sound modulo rooted branching truly concurrent bisimulation equivalences $\approx_{rbp}$, $\approx_{rbs}$, and $\approx_{rbhp}$. That is,
\begin{enumerate}
  \item let $x$ and $y$ be $\textrm{APTC}^{\textrm{drt}}_{\tau}$ + Rec + renaming terms. If $\textrm{APTC}^{\textrm{drt}}_{\tau}\textrm{ + Rec + renaming}\vdash x=y$, then $x\approx_{rbs} y$;
  \item let $x$ and $y$ be $\textrm{APTC}^{\textrm{drt}}_{\tau}$ + Rec + renaming terms. If $\textrm{APTC}^{\textrm{drt}}_{\tau}\textrm{ + Rec + renaming}\vdash x=y$, then $x\approx_{rbp} y$;
  \item let $x$ and $y$ be $\textrm{APTC}^{\textrm{drt}}_{\tau}$ + Rec + renaming terms. If $\textrm{APTC}^{\textrm{drt}}_{\tau}\textrm{ + Rec +renaming}\vdash x=y$, then $x\approx_{rbhp} y$.
\end{enumerate}
\end{theorem}

\begin{proof}
Since $\approx_{rbp}$, $\approx_{rbs}$, and $\approx_{rbhp}$ are both equivalent and congruent relations, we only need to check if each axiom in Table \ref{AxiomsForAPTCDRTTauRen} is sound modulo $\approx_{rbp}$, $\approx_{rbs}$, and $\approx_{rbhp}$ respectively.

\begin{enumerate}
  \item Each axiom in Table \ref{AxiomsForAPTCDRTTauRen} can be checked that it is sound modulo rooted branching step bisimulation equivalence, by transition rules in Table \ref{TRForAPTCDRTTauRen}. We omit them.
  \item From the definition of rooted branching pomset bisimulation $\approx_{rbp}$, we know that rooted branching pomset bisimulation $\approx_{rbp}$ is defined by weak pomset transitions, which are labeled by pomsets with $\underline{\underline{\tau}}$. In a weak pomset transition, the events in the pomset are either within causality relations (defined by $\cdot$) or in concurrency (implicitly defined by $\cdot$ and $+$, and explicitly defined by $\between$), of course, they are pairwise consistent (without conflicts). We have already proven the case that all events are pairwise concurrent, so, we only need to prove the case of events in causality. Without loss of generality, we take a pomset of $P=\{\underline{\underline{a}},\underline{\underline{b}}:\underline{\underline{a}}\cdot \underline{\underline{b}}\}$. Then the weak pomset transition labeled by the above $P$ is just composed of one single event transition labeled by $\underline{\underline{a}}$ succeeded by another single event transition labeled by $\underline{\underline{b}}$, that is, $\xRightarrow{P}=\xRightarrow{a}\xRightarrow{b}$.

        Similarly to the proof of soundness modulo rooted branching step bisimulation $\approx_{rbs}$, we can prove that each axiom in Table \ref{AxiomsForAPTCDRTTauRen} is sound modulo rooted branching pomset bisimulation $\approx_{rbp}$, we omit them.

  \item From the definition of rooted branching hp-bisimulation $\approx_{rbhp}$, we know that rooted branching hp-bisimulation $\approx_{rbhp}$ is defined on the weakly posetal product $(C_1,f,C_2),f:\hat{C_1}\rightarrow \hat{C_2}\textrm{ isomorphism}$. Two process terms $s$ related to $C_1$ and $t$ related to $C_2$, and $f:\hat{C_1}\rightarrow \hat{C_2}\textrm{ isomorphism}$. Initially, $(C_1,f,C_2)=(\emptyset,\emptyset,\emptyset)$, and $(\emptyset,\emptyset,\emptyset)\in\approx_{rbhp}$. When $s\xrightarrow{a}s'$ ($C_1\xrightarrow{a}C_1'$), there will be $t\xRightarrow{a}t'$ ($C_2\xRightarrow{a}C_2'$), and we define $f'=f[a\mapsto a]$. Then, if $(C_1,f,C_2)\in\approx_{rbhp}$, then $(C_1',f',C_2')\in\approx_{rbhp}$.

        Similarly to the proof of soundness modulo rooted branching pomset bisimulation equivalence, we can prove that each axiom in Table \ref{AxiomsForAPTCDRTTauRen} is sound modulo rooted branching hp-bisimulation equivalence, we just need additionally to check the above conditions on rooted branching hp-bisimulation, we omit them.
\end{enumerate}
\end{proof}

\subsubsection{Completeness}

\begin{theorem}[Completeness of $\textrm{APTC}^{\textrm{drt}}_{\tau}$ + CFAR + guarded linear Rec + renaming]
The axiomatization of $\textrm{APTC}^{\textrm{drt}}_{\tau}$ + CFAR + guarded linear Rec + renaming is complete modulo rooted branching truly concurrent bisimulation equivalences $\approx_{rbs}$, $\approx_{rbp}$, and $\approx_{rbhp}$. That is,
\begin{enumerate}
  \item let $p$ and $q$ be closed $\textrm{APTC}^{\textrm{drt}}_{\tau}$ + CFAR + guarded linear Rec + renaming terms, if $p\approx_{rbs} q$ then $p=q$;
  \item let $p$ and $q$ be closed $\textrm{APTC}^{\textrm{drt}}_{\tau}$ + CFAR + guarded linear Rec + renaming terms, if $p\approx_{rbp} q$ then $p=q$;
  \item let $p$ and $q$ be closed $\textrm{APTC}^{\textrm{drt}}_{\tau}$ + CFAR + guarded linear Rec + renaming terms, if $p\approx_{rbhp} q$ then $p=q$.
\end{enumerate}

\end{theorem}

\begin{proof}
Firstly, we know that each process term in $\textrm{APTC}^{\textrm{drt}}_{\tau}$  + CFAR + guarded linear Rec + renaming is equal to a process term $\langle X_1|E\rangle$ with $E$ a linear recursive specification.

It remains to prove the following cases.

\begin{enumerate}
  \item If $\langle X_1|E_1\rangle \approx_{rbs} \langle Y_1|E_2\rangle$ for linear recursive specification $E_1$ and $E_2$, then $\langle X_1|E_1\rangle = \langle Y_1|E_2\rangle$.

        It can be proven similarly to the completeness of $\textrm{APTC}_{\tau}$ + CFAR + linear Rec + renaming, see \cite{ATC}.

  \item If $\langle X_1|E_1\rangle \approx_{rbp} \langle Y_1|E_2\rangle$ for linear recursive specification $E_1$ and $E_2$, then $\langle X_1|E_1\rangle = \langle Y_1|E_2\rangle$.

        It can be proven similarly, just by replacement of $\approx_{rbs}$ by $\approx_{rbp}$, we omit it.
  \item If $\langle X_1|E_1\rangle \approx_{rbhp} \langle Y_1|E_2\rangle$ for linear recursive specification $E_1$ and $E_2$, then $\langle X_1|E_1\rangle = \langle Y_1|E_2\rangle$.

        It can be proven similarly, just by replacement of $\approx_{rbs}$ by $\approx_{rbhp}$, we omit it.
\end{enumerate}
\end{proof}

\subsection{Discrete Absolute Timing}

\begin{definition}[Signature of $\textrm{APTC}^{\textrm{dat}}_{\tau}$ + Rec + renaming]
The signature of $\textrm{APTC}^{\textrm{dat}}_{\tau}$ + Rec + renaming consists of the signature of $\textrm{APTC}^{\textrm{dat}}_{\tau}$ + Rec, and the renaming operator $\rho_f: \mathcal{P}_{\textrm{abs}}\rightarrow\mathcal{P}_{\textrm{abs}}$.
\end{definition}

The axioms of $\textrm{APTC}^{\textrm{dat}}_{\tau}$ + Rec + renaming include the laws of $\textrm{APTC}^{\textrm{dat}}_{\tau}$ + Rec, and the axioms of renaming operator in Table \ref{AxiomsForAPTCDATTauRen}.

\begin{center}
\begin{table}
  \begin{tabular}{@{}ll@{}}
\hline No. &Axiom\\
  $DATRN1$ & $\rho_f(\underline{a})=f(\underline{a})$\\
  $DATRN2$ & $\rho_f(\dot{\delta})=\dot{\delta}$\\
  $DATRN3$ & $\rho_f(\underline{\tau})=\underline{\tau}$\\
  $DATRN$ & $\rho_f(\sigma^n_{\textrm{abs}}(x)) = \sigma^n_{\textrm{abs}}(\rho_f(x))$\\
  $RN3$ & $\rho_f(x+y)=\rho_f(x)+\rho_f(y)$\\
  $RN4$ & $\rho_f(x\cdot y)=\rho_f(x)\cdot\rho_f(y)$\\
  $RN5$ & $\rho_f(x\parallel y)=\rho_f(x)\parallel\rho_f(y)$\\
\end{tabular}
\caption{Additional axioms of renaming operator $(a\in A_{\tau\delta},n\geq 0)$}
\label{AxiomsForAPTCDATTauRen}
\end{table}
\end{center}

The additional transition rules of renaming operator is shown in Table \ref{TRForAPTCDATTauRen}.

\begin{center}
    \begin{table}
        $$\frac{\langle x,n\rangle\xrightarrow{a}\langle\surd,n\rangle}{\langle\rho_f(x),n\rangle\xrightarrow{f(a)}\langle\surd,n\rangle}
        \quad\frac{\langle x,n\rangle\xrightarrow{a}\langle x',n\rangle}{\langle \rho_f(x),n\rangle\xrightarrow{f(a)}\langle\rho_f(x'),n\rangle}$$

        $$\frac{\langle x,n\rangle\mapsto^m \langle x,n+m\rangle}{\langle \rho_f(x),n\rangle\mapsto^m\langle\rho_f(x),n+m\rangle}\quad\frac{\langle x,n\rangle\uparrow}{\langle\rho_f(x),n\rangle\uparrow}$$
        \caption{Transition rule of renaming operator $(a\in A_{\tau},m>0,n\geq 0)$}
        \label{TRForAPTCDATTauRen}
    \end{table}
\end{center}

\begin{theorem}[Elimination theorem]
Let $p$ be a closed $\textrm{APTC}^{\textrm{dat}}_{\tau}$ + Rec + renaming term. Then there is a basic $\textrm{APTC}^{\textrm{dat}}_{\tau}$ + Rec term $q$ such that $\textrm{APTC}^{\textrm{dat}}_{\tau}\textrm{ + renaming}\vdash p=q$.
\end{theorem}

\begin{proof}
It is sufficient to induct on the structure of the closed $\textrm{APTC}^{\textrm{dat}}_{\tau}$ + Rec + renaming term $p$. It can be proven that $p$ combined by the constants and operators of $\textrm{APTC}^{\textrm{dat}}_{\tau}$ + Rec + renaming exists an equal basic $\textrm{APTC}^{\textrm{dat}}_{\tau}$ + Rec term $q$, and the other operators not included in the basic terms, such as $\rho_f$ can be eliminated.
\end{proof}

 \subsubsection{Connections}

\begin{theorem}[Conservativity of $\textrm{APTC}^{\textrm{dat}}_{\tau}$ + Rec + renaming]
$\textrm{APTC}^{\textrm{dat}}_{\tau}$ + Rec + renaming is a conservative extension of $\textrm{APTC}^{\textrm{dat}}_{\tau}$ + Rec.
\end{theorem}

\begin{proof}
It follows from the following two facts.

    \begin{enumerate}
      \item The transition rules of $\textrm{APTC}^{\textrm{dat}}_{\tau}$ + Rec are all source-dependent;
      \item The sources of the transition rules of $\textrm{APTC}^{\textrm{dat}}_{\tau}$ + Rec + renaming contain an occurrence of $\rho_f$.
    \end{enumerate}

So, $\textrm{APTC}^{\textrm{dat}}_{\tau}$ + Rec + renaming is a conservative extension of $\textrm{APTC}^{\textrm{dat}}_{\tau}$ + Rec, as desired.
\end{proof}

\subsubsection{Congruence}

\begin{theorem}[Congruence of $\textrm{APTC}^{\textrm{dat}}_{\tau}$ + Rec + renaming]
Rooted branching truly concurrent bisimulation equivalences $\approx_{rbp}$, $\approx_{rbs}$ and $\approx_{rbhp}$ are all congruences with respect to $\textrm{APTC}^{\textrm{dat}}_{\tau}$ + Rec + renaming. That is,
\begin{itemize}
  \item rooted branching pomset bisimulation equivalence $\approx_{rbp}$ is a congruence with respect to $\textrm{APTC}^{\textrm{dat}}_{\tau}$ + Rec + renaming;
  \item rooted branching step bisimulation equivalence $\approx_{rbs}$ is a congruence with respect to $\textrm{APTC}^{\textrm{dat}}_{\tau}$ + Rec + renaming;
  \item rooted branching hp-bisimulation equivalence $\approx_{rbhp}$ is a congruence with respect to $\textrm{APTC}^{\textrm{dat}}_{\tau}$ + Rec + renaming.
\end{itemize}
\end{theorem}

\begin{proof}
It is easy to see that $\approx_{rbp}$, $\approx_{rbs}$, and $\approx_{rbhp}$ are all equivalent relations on $\textrm{APTC}^{\textrm{dat}}_{\tau}$ + Rec + renaming terms, it is only sufficient to prove that $\approx_{rbp}$, $\approx_{rbs}$, and $\approx_{rbhp}$ are all preserved by the operators $\rho_f$. It is trivial and we omit it.
\end{proof}

\subsubsection{Soundness}

\begin{theorem}[Soundness of $\textrm{APTC}^{\textrm{dat}}_{\tau}$ + Rec + renaming]
The axiomatization of $\textrm{APTC}^{\textrm{dat}}_{\tau}$ + Rec + renaming is sound modulo rooted branching truly concurrent bisimulation equivalences $\approx_{rbp}$, $\approx_{rbs}$, and $\approx_{rbhp}$. That is,
\begin{enumerate}
  \item let $x$ and $y$ be $\textrm{APTC}^{\textrm{dat}}_{\tau}$ + Rec + renaming terms. If $\textrm{APTC}^{\textrm{dat}}_{\tau}\textrm{ + Rec + renaming}\vdash x=y$, then $x\approx_{rbs} y$;
  \item let $x$ and $y$ be $\textrm{APTC}^{\textrm{dat}}_{\tau}$ + Rec + renaming terms. If $\textrm{APTC}^{\textrm{dat}}_{\tau}\textrm{ + Rec + renaming}\vdash x=y$, then $x\approx_{rbp} y$;
  \item let $x$ and $y$ be $\textrm{APTC}^{\textrm{dat}}_{\tau}$ + Rec + renaming terms. If $\textrm{APTC}^{\textrm{dat}}_{\tau}\textrm{ + Rec +renaming}\vdash x=y$, then $x\approx_{rbhp} y$.
\end{enumerate}
\end{theorem}

\begin{proof}
Since $\approx_{rbp}$, $\approx_{rbs}$, and $\approx_{rbhp}$ are both equivalent and congruent relations, we only need to check if each axiom in Table \ref{AxiomsForAPTCDATTauRen} is sound modulo $\approx_{rbp}$, $\approx_{rbs}$, and $\approx_{rbhp}$ respectively.

\begin{enumerate}
  \item Each axiom in Table \ref{AxiomsForAPTCDATTauRen} can be checked that it is sound modulo rooted branching step bisimulation equivalence, by transition rules in Table \ref{TRForAPTCDATTauRen}. We omit them.
  \item From the definition of rooted branching pomset bisimulation $\approx_{rbp}$, we know that rooted branching pomset bisimulation $\approx_{rbp}$ is defined by weak pomset transitions, which are labeled by pomsets with $\underline{\tau}$. In a weak pomset transition, the events in the pomset are either within causality relations (defined by $\cdot$) or in concurrency (implicitly defined by $\cdot$ and $+$, and explicitly defined by $\between$), of course, they are pairwise consistent (without conflicts). We have already proven the case that all events are pairwise concurrent, so, we only need to prove the case of events in causality. Without loss of generality, we take a pomset of $P=\{\underline{a},\underline{b}:\underline{a}\cdot \underline{b}\}$. Then the weak pomset transition labeled by the above $P$ is just composed of one single event transition labeled by $\underline{a}$ succeeded by another single event transition labeled by $\underline{b}$, that is, $\xRightarrow{P}=\xRightarrow{a}\xRightarrow{b}$.

        Similarly to the proof of soundness modulo rooted branching step bisimulation $\approx_{rbs}$, we can prove that each axiom in Table \ref{AxiomsForAPTCDATTauRen} is sound modulo rooted branching pomset bisimulation $\approx_{rbp}$, we omit them.

  \item From the definition of rooted branching hp-bisimulation $\approx_{rbhp}$, we know that rooted branching hp-bisimulation $\approx_{rbhp}$ is defined on the weakly posetal product $(C_1,f,C_2),f:\hat{C_1}\rightarrow \hat{C_2}\textrm{ isomorphism}$. Two process terms $s$ related to $C_1$ and $t$ related to $C_2$, and $f:\hat{C_1}\rightarrow \hat{C_2}\textrm{ isomorphism}$. Initially, $(C_1,f,C_2)=(\emptyset,\emptyset,\emptyset)$, and $(\emptyset,\emptyset,\emptyset)\in\approx_{rbhp}$. When $s\xrightarrow{a}s'$ ($C_1\xrightarrow{a}C_1'$), there will be $t\xRightarrow{a}t'$ ($C_2\xRightarrow{a}C_2'$), and we define $f'=f[a\mapsto a]$. Then, if $(C_1,f,C_2)\in\approx_{rbhp}$, then $(C_1',f',C_2')\in\approx_{rbhp}$.

        Similarly to the proof of soundness modulo rooted branching pomset bisimulation equivalence, we can prove that each axiom in Table \ref{AxiomsForAPTCDATTauRen} is sound modulo rooted branching hp-bisimulation equivalence, we just need additionally to check the above conditions on rooted branching hp-bisimulation, we omit them.
\end{enumerate}

\end{proof}

\subsubsection{Completeness}

\begin{theorem}[Completeness of $\textrm{APTC}^{\textrm{dat}}_{\tau}$ + CFAR + guarded linear Rec + renaming]
The axiomatization of $\textrm{APTC}^{\textrm{dat}}_{\tau}$ + CFAR + guarded linear Rec + renaming is complete modulo rooted branching truly concurrent bisimulation equivalences $\approx_{rbs}$, $\approx_{rbp}$, and $\approx_{rbhp}$. That is,
\begin{enumerate}
  \item let $p$ and $q$ be closed $\textrm{APTC}^{\textrm{dat}}_{\tau}$ + CFAR + guarded linear Rec + renaming terms, if $p\approx_{rbs} q$ then $p=q$;
  \item let $p$ and $q$ be closed $\textrm{APTC}^{\textrm{dat}}_{\tau}$ + CFAR + guarded linear Rec + renaming terms, if $p\approx_{rbp} q$ then $p=q$;
  \item let $p$ and $q$ be closed $\textrm{APTC}^{\textrm{dat}}_{\tau}$ + CFAR + guarded linear Rec + renaming terms, if $p\approx_{rbhp} q$ then $p=q$.
\end{enumerate}

\end{theorem}

\begin{proof}
Firstly, we know that each process term in $\textrm{APTC}^{\textrm{dat}}_{\tau}$  + CFAR + guarded linear Rec + renaming is equal to a process term $\langle X_1|E\rangle$ with $E$ a linear recursive specification.

It remains to prove the following cases.

\begin{enumerate}
  \item If $\langle X_1|E_1\rangle \approx_{rbs} \langle Y_1|E_2\rangle$ for linear recursive specification $E_1$ and $E_2$, then $\langle X_1|E_1\rangle = \langle Y_1|E_2\rangle$.

        It can be proven similarly to the completeness of $\textrm{APTC}_{\tau}$ + CFAR + linear Rec + renaming, see \cite{ATC}.

  \item If $\langle X_1|E_1\rangle \approx_{rbp} \langle Y_1|E_2\rangle$ for linear recursive specification $E_1$ and $E_2$, then $\langle X_1|E_1\rangle = \langle Y_1|E_2\rangle$.

        It can be proven similarly, just by replacement of $\approx_{rbs}$ by $\approx_{rbp}$, we omit it.
  \item If $\langle X_1|E_1\rangle \approx_{rbhp} \langle Y_1|E_2\rangle$ for linear recursive specification $E_1$ and $E_2$, then $\langle X_1|E_1\rangle = \langle Y_1|E_2\rangle$.

        It can be proven similarly, just by replacement of $\approx_{rbs}$ by $\approx_{rbhp}$, we omit it.
\end{enumerate}
\end{proof}

\subsection{Continuous Relative Timing}

\begin{definition}[Signature of $\textrm{APTC}^{\textrm{srt}}_{\tau}$ + Rec + renaming]
The signature of $\textrm{APTC}^{\textrm{srt}}_{\tau}$ + Rec + renaming consists of the signature of $\textrm{APTC}^{\textrm{srt}}_{\tau}$ + Rec, and the renaming operator $\rho_f: \mathcal{P}_{\textrm{rel}}\rightarrow\mathcal{P}_{\textrm{rel}}$.
\end{definition}

The axioms of $\textrm{APTC}^{\textrm{srt}}_{\tau}$ + Rec + renaming include the laws of $\textrm{APTC}^{\textrm{srt}}_{\tau}$ + Rec, and the axioms of renaming operator in Table \ref{AxiomsForAPTCSRTTauRen}.

\begin{center}
\begin{table}
  \begin{tabular}{@{}ll@{}}
\hline No. &Axiom\\
  $SRTRN1$ & $\rho_f(\tilde{\tilde{a}})=f(\tilde{\tilde{a}})$\\
  $SRTRN2$ & $\rho_f(\dot{\delta})=\dot{\delta}$\\
  $SRTRN3$ & $\rho_f(\tilde{\tilde{\tau}})=\tilde{\tilde{\tau}}$\\
  $SRTRN$ & $\rho_f(\sigma^p_{\textrm{rel}}(x)) = \sigma^p_{\textrm{rel}}(\rho_f(x))$\\
  $RN3$ & $\rho_f(x+y)=\rho_f(x)+\rho_f(y)$\\
  $RN4$ & $\rho_f(x\cdot y)=\rho_f(x)\cdot\rho_f(y)$\\
  $RN5$ & $\rho_f(x\parallel y)=\rho_f(x)\parallel\rho_f(y)$\\
\end{tabular}
\caption{Additional axioms of renaming operator $(a\in A_{\tau\delta},p\geq 0)$}
\label{AxiomsForAPTCSRTTauRen}
\end{table}
\end{center}

The additional transition rules of renaming operator is shown in Table \ref{TRForAPTCSRTTauRen}.

\begin{center}
    \begin{table}
        $$\frac{x\xrightarrow{a}\surd}{\rho_f(x)\xrightarrow{f(a)}\surd}
        \quad\frac{x\xrightarrow{a}x'}{\rho_f(x)\xrightarrow{f(a)}\rho_f(x')}$$

        $$\frac{x\mapsto^r x'}{\rho_f(x)\mapsto^r\rho_f(x')}\quad\frac{x\uparrow}{\rho_f(x)\uparrow}$$
        \caption{Transition rule of renaming operator $(a\in A_{\tau},r>0,p\geq 0)$}
        \label{TRForAPTCSRTTauRen}
    \end{table}
\end{center}

\begin{theorem}[Elimination theorem]
Let $p$ be a closed $\textrm{APTC}^{\textrm{srt}}_{\tau}$ + Rec + renaming term. Then there is a basic $\textrm{APTC}^{\textrm{srt}}_{\tau}$ + Rec term $q$ such that $\textrm{APTC}^{\textrm{srt}}_{\tau}\textrm{ + renaming}\vdash p=q$.
\end{theorem}

\begin{proof}
It is sufficient to induct on the structure of the closed $\textrm{APTC}^{\textrm{srt}}_{\tau}$ + Rec + renaming term $p$. It can be proven that $p$ combined by the constants and operators of $\textrm{APTC}^{\textrm{srt}}_{\tau}$ + Rec + renaming exists an equal basic $\textrm{APTC}^{\textrm{srt}}_{\tau}$ + Rec term $q$, and the other operators not included in the basic terms, such as $\rho_f$ can be eliminated.
\end{proof}

 \subsubsection{Connections}

\begin{theorem}[Conservativity of $\textrm{APTC}^{\textrm{srt}}_{\tau}$ + Rec + renaming]
$\textrm{APTC}^{\textrm{srt}}_{\tau}$ + Rec + renaming is a conservative extension of $\textrm{APTC}^{\textrm{srt}}_{\tau}$ + Rec.
\end{theorem}

\begin{proof}
It follows from the following two facts.

    \begin{enumerate}
      \item The transition rules of $\textrm{APTC}^{\textrm{srt}}_{\tau}$ + Rec are all source-dependent;
      \item The sources of the transition rules of $\textrm{APTC}^{\textrm{srt}}_{\tau}$ + Rec + renaming contain an occurrence of $\rho_f$.
    \end{enumerate}

So, $\textrm{APTC}^{\textrm{srt}}_{\tau}$ + Rec + renaming is a conservative extension of $\textrm{APTC}^{\textrm{srt}}_{\tau}$ + Rec, as desired.
\end{proof}

\subsubsection{Congruence}

\begin{theorem}[Congruence of $\textrm{APTC}^{\textrm{srt}}_{\tau}$ + Rec + renaming]
Rooted branching truly concurrent bisimulation equivalences $\approx_{rbp}$, $\approx_{rbs}$ and $\approx_{rbhp}$ are all congruences with respect to $\textrm{APTC}^{\textrm{srt}}_{\tau}$ + Rec + renaming. That is,
\begin{itemize}
  \item rooted branching pomset bisimulation equivalence $\approx_{rbp}$ is a congruence with respect to $\textrm{APTC}^{\textrm{srt}}_{\tau}$ + Rec + renaming;
  \item rooted branching step bisimulation equivalence $\approx_{rbs}$ is a congruence with respect to $\textrm{APTC}^{\textrm{srt}}_{\tau}$ + Rec + renaming;
  \item rooted branching hp-bisimulation equivalence $\approx_{rbhp}$ is a congruence with respect to $\textrm{APTC}^{\textrm{srt}}_{\tau}$ + Rec + renaming.
\end{itemize}
\end{theorem}

\begin{proof}
It is easy to see that $\approx_{rbp}$, $\approx_{rbs}$, and $\approx_{rbhp}$ are all equivalent relations on $\textrm{APTC}^{\textrm{srt}}_{\tau}$ + Rec + renaming terms, it is only sufficient to prove that $\approx_{rbp}$, $\approx_{rbs}$, and $\approx_{rbhp}$ are all preserved by the operators $\rho_f$. It is trivial and we omit it.
\end{proof}

\subsubsection{Soundness}

\begin{theorem}[Soundness of $\textrm{APTC}^{\textrm{srt}}_{\tau}$ + Rec + renaming]
The axiomatization of $\textrm{APTC}^{\textrm{srt}}_{\tau}$ + Rec + renaming is sound modulo rooted branching truly concurrent bisimulation equivalences $\approx_{rbp}$, $\approx_{rbs}$, and $\approx_{rbhp}$. That is,
\begin{enumerate}
  \item let $x$ and $y$ be $\textrm{APTC}^{\textrm{srt}}_{\tau}$ + Rec + renaming terms. If $\textrm{APTC}^{\textrm{srt}}_{\tau}\textrm{ + Rec + renaming}\vdash x=y$, then $x\approx_{rbs} y$;
  \item let $x$ and $y$ be $\textrm{APTC}^{\textrm{srt}}_{\tau}$ + Rec + renaming terms. If $\textrm{APTC}^{\textrm{srt}}_{\tau}\textrm{ + Rec + renaming}\vdash x=y$, then $x\approx_{rbp} y$;
  \item let $x$ and $y$ be $\textrm{APTC}^{\textrm{srt}}_{\tau}$ + Rec + renaming terms. If $\textrm{APTC}^{\textrm{srt}}_{\tau}\textrm{ + Rec +renaming}\vdash x=y$, then $x\approx_{rbhp} y$.
\end{enumerate}
\end{theorem}

\begin{proof}
Since $\approx_{rbp}$, $\approx_{rbs}$, and $\approx_{rbhp}$ are both equivalent and congruent relations, we only need to check if each axiom in Table \ref{AxiomsForAPTCSRTTauRen} is sound modulo $\approx_{rbp}$, $\approx_{rbs}$, and $\approx_{rbhp}$ respectively.

\begin{enumerate}
  \item Each axiom in Table \ref{AxiomsForAPTCSRTTauRen} can be checked that it is sound modulo rooted branching step bisimulation equivalence, by transition rules in Table \ref{TRForAPTCSRTTauRen}. We omit them.
  \item From the definition of rooted branching pomset bisimulation $\approx_{rbp}$, we know that rooted branching pomset bisimulation $\approx_{rbp}$ is defined by weak pomset transitions, which are labeled by pomsets with $\tilde{\tilde{\tau}}$. In a weak pomset transition, the events in the pomset are either within causality relations (defined by $\cdot$) or in concurrency (implicitly defined by $\cdot$ and $+$, and explicitly defined by $\between$), of course, they are pairwise consistent (without conflicts). We have already proven the case that all events are pairwise concurrent, so, we only need to prove the case of events in causality. Without loss of generality, we take a pomset of $P=\{\tilde{\tilde{a}},\tilde{\tilde{b}}:\tilde{\tilde{a}}\cdot \tilde{\tilde{b}}\}$. Then the weak pomset transition labeled by the above $P$ is just composed of one single event transition labeled by $\tilde{\tilde{a}}$ succeeded by another single event transition labeled by $\tilde{\tilde{b}}$, that is, $\xRightarrow{P}=\xRightarrow{a}\xRightarrow{b}$.

        Similarly to the proof of soundness modulo rooted branching step bisimulation $\approx_{rbs}$, we can prove that each axiom in Table \ref{AxiomsForAPTCSRTTauRen} is sound modulo rooted branching pomset bisimulation $\approx_{rbp}$, we omit them.

  \item From the definition of rooted branching hp-bisimulation $\approx_{rbhp}$, we know that rooted branching hp-bisimulation $\approx_{rbhp}$ is defined on the weakly posetal product $(C_1,f,C_2),f:\hat{C_1}\rightarrow \hat{C_2}\textrm{ isomorphism}$. Two process terms $s$ related to $C_1$ and $t$ related to $C_2$, and $f:\hat{C_1}\rightarrow \hat{C_2}\textrm{ isomorphism}$. Initially, $(C_1,f,C_2)=(\emptyset,\emptyset,\emptyset)$, and $(\emptyset,\emptyset,\emptyset)\in\approx_{rbhp}$. When $s\xrightarrow{a}s'$ ($C_1\xrightarrow{a}C_1'$), there will be $t\xRightarrow{a}t'$ ($C_2\xRightarrow{a}C_2'$), and we define $f'=f[a\mapsto a]$. Then, if $(C_1,f,C_2)\in\approx_{rbhp}$, then $(C_1',f',C_2')\in\approx_{rbhp}$.

        Similarly to the proof of soundness modulo rooted branching pomset bisimulation equivalence, we can prove that each axiom in Table \ref{AxiomsForAPTCSRTTauRen} is sound modulo rooted branching hp-bisimulation equivalence, we just need additionally to check the above conditions on rooted branching hp-bisimulation, we omit them.
\end{enumerate}

\end{proof}

\subsubsection{Completeness}

\begin{theorem}[Completeness of $\textrm{APTC}^{\textrm{srt}}_{\tau}$ + CFAR + guarded linear Rec + renaming]
The axiomatization of $\textrm{APTC}^{\textrm{srt}}_{\tau}$ + CFAR + guarded linear Rec + renaming is complete modulo rooted branching truly concurrent bisimulation equivalences $\approx_{rbs}$, $\approx_{rbp}$, and $\approx_{rbhp}$. That is,
\begin{enumerate}
  \item let $p$ and $q$ be closed $\textrm{APTC}^{\textrm{srt}}_{\tau}$ + CFAR + guarded linear Rec + renaming terms, if $p\approx_{rbs} q$ then $p=q$;
  \item let $p$ and $q$ be closed $\textrm{APTC}^{\textrm{srt}}_{\tau}$ + CFAR + guarded linear Rec + renaming terms, if $p\approx_{rbp} q$ then $p=q$;
  \item let $p$ and $q$ be closed $\textrm{APTC}^{\textrm{srt}}_{\tau}$ + CFAR + guarded linear Rec + renaming terms, if $p\approx_{rbhp} q$ then $p=q$.
\end{enumerate}

\end{theorem}

\begin{proof}
Firstly, we know that each process term in $\textrm{APTC}^{\textrm{srt}}_{\tau}$  + CFAR + guarded linear Rec + renaming is equal to a process term $\langle X_1|E\rangle$ with $E$ a linear recursive specification.

It remains to prove the following cases.

\begin{enumerate}
  \item If $\langle X_1|E_1\rangle \approx_{rbs} \langle Y_1|E_2\rangle$ for linear recursive specification $E_1$ and $E_2$, then $\langle X_1|E_1\rangle = \langle Y_1|E_2\rangle$.

        It can be proven similarly to the completeness of $\textrm{APTC}_{\tau}$ + CFAR + linear Rec + renaming, see \cite{ATC}.

  \item If $\langle X_1|E_1\rangle \approx_{rbp} \langle Y_1|E_2\rangle$ for linear recursive specification $E_1$ and $E_2$, then $\langle X_1|E_1\rangle = \langle Y_1|E_2\rangle$.

        It can be proven similarly, just by replacement of $\approx_{rbs}$ by $\approx_{rbp}$, we omit it.
  \item If $\langle X_1|E_1\rangle \approx_{rbhp} \langle Y_1|E_2\rangle$ for linear recursive specification $E_1$ and $E_2$, then $\langle X_1|E_1\rangle = \langle Y_1|E_2\rangle$.

        It can be proven similarly, just by replacement of $\approx_{rbs}$ by $\approx_{rbhp}$, we omit it.
\end{enumerate}
\end{proof}

\subsection{Continuous Absolute Timing}

\begin{definition}[Signature of $\textrm{APTC}^{\textrm{sat}}_{\tau}$ + Rec + renaming]
The signature of $\textrm{APTC}^{\textrm{sat}}_{\tau}$ + Rec + renaming consists of the signature of $\textrm{APTC}^{\textrm{sat}}_{\tau}$ + Rec, and the renaming operator $\rho_f: \mathcal{P}_{\textrm{abs}}\rightarrow\mathcal{P}_{\textrm{abs}}$.
\end{definition}

The axioms of $\textrm{APTC}^{\textrm{sat}}_{\tau}$ + Rec + renaming include the laws of $\textrm{APTC}^{\textrm{sat}}_{\tau}$ + Rec, and the axioms of renaming operator in Table \ref{AxiomsForAPTCSATTauRen}.

\begin{center}
\begin{table}
  \begin{tabular}{@{}ll@{}}
\hline No. &Axiom\\
  $SATRN1$ & $\rho_f(\tilde{a})=f(\tilde{a})$\\
  $SATRN2$ & $\rho_f(\dot{\delta})=\dot{\delta}$\\
  $SATRN3$ & $\rho_f(\tilde{\tau})=\tilde{\tau}$\\
  $SATRN$ & $\rho_f(\sigma^p_{\textrm{abs}}(x)) = \sigma^p_{\textrm{abs}}(\rho_f(x))$\\
  $RN3$ & $\rho_f(x+y)=\rho_f(x)+\rho_f(y)$\\
  $RN4$ & $\rho_f(x\cdot y)=\rho_f(x)\cdot\rho_f(y)$\\
  $RN5$ & $\rho_f(x\parallel y)=\rho_f(x)\parallel\rho_f(y)$\\
\end{tabular}
\caption{Additional axioms of renaming operator $(a\in A_{\tau\delta},p\geq 0)$}
\label{AxiomsForAPTCSATTauRen}
\end{table}
\end{center}

The additional transition rules of renaming operator is shown in Table \ref{TRForAPTCSATTauRen}.

\begin{center}
    \begin{table}
        $$\frac{\langle x,p\rangle\xrightarrow{a}\langle\surd,p\rangle}{\langle\rho_f(x),p\rangle\xrightarrow{f(a)}\langle\surd,p\rangle}
        \quad\frac{\langle x,p\rangle\xrightarrow{a}\langle x',p\rangle}{\langle \rho_f(x),p\rangle\xrightarrow{f(a)}\langle\rho_f(x'),p\rangle}$$

        $$\frac{\langle x,p\rangle\mapsto^r \langle x,p+r\rangle}{\langle \rho_f(x),p\rangle\mapsto^r\langle\rho_f(x),p+r\rangle}\quad\frac{\langle x,p\rangle\uparrow}{\langle\rho_f(x),p\rangle\uparrow}$$
        \caption{Transition rule of renaming operator $(a\in A_{\tau},r>0,p\geq 0)$}
        \label{TRForAPTCSATTauRen}
    \end{table}
\end{center}

\begin{theorem}[Elimination theorem]
Let $p$ be a closed $\textrm{APTC}^{\textrm{sat}}_{\tau}$ + Rec + renaming term. Then there is a basic $\textrm{APTC}^{\textrm{sat}}_{\tau}$ + Rec term $q$ such that $\textrm{APTC}^{\textrm{sat}}_{\tau}\textrm{ + renaming}\vdash p=q$.
\end{theorem}

\begin{proof}
It is sufficient to induct on the structure of the closed $\textrm{APTC}^{\textrm{sat}}_{\tau}$ + Rec + renaming term $p$. It can be proven that $p$ combined by the constants and operators of $\textrm{APTC}^{\textrm{sat}}_{\tau}$ + Rec + renaming exists an equal basic $\textrm{APTC}^{\textrm{sat}}_{\tau}$ + Rec term $q$, and the other operators not included in the basic terms, such as $\rho_f$ can be eliminated.
\end{proof}

 \subsubsection{Connections}

\begin{theorem}[Conservativity of $\textrm{APTC}^{\textrm{sat}}_{\tau}$ + Rec + renaming]
$\textrm{APTC}^{\textrm{sat}}_{\tau}$ + Rec + renaming is a conservative extension of $\textrm{APTC}^{\textrm{sat}}_{\tau}$ + Rec.
\end{theorem}

\begin{proof}
It follows from the following two facts.

    \begin{enumerate}
      \item The transition rules of $\textrm{APTC}^{\textrm{sat}}_{\tau}$ + Rec are all source-dependent;
      \item The sources of the transition rules of $\textrm{APTC}^{\textrm{sat}}_{\tau}$ + Rec + renaming contain an occurrence of $\rho_f$.
    \end{enumerate}

So, $\textrm{APTC}^{\textrm{sat}}_{\tau}$ + Rec + renaming is a conservative extension of $\textrm{APTC}^{\textrm{sat}}_{\tau}$ + Rec, as desired.
\end{proof}

\subsubsection{Congruence}

\begin{theorem}[Congruence of $\textrm{APTC}^{\textrm{sat}}_{\tau}$ + Rec + renaming]
Rooted branching truly concurrent bisimulation equivalences $\approx_{rbp}$, $\approx_{rbs}$ and $\approx_{rbhp}$ are all congruences with respect to $\textrm{APTC}^{\textrm{sat}}_{\tau}$ + Rec + renaming. That is,
\begin{itemize}
  \item rooted branching pomset bisimulation equivalence $\approx_{rbp}$ is a congruence with respect to $\textrm{APTC}^{\textrm{sat}}_{\tau}$ + Rec + renaming;
  \item rooted branching step bisimulation equivalence $\approx_{rbs}$ is a congruence with respect to $\textrm{APTC}^{\textrm{sat}}_{\tau}$ + Rec + renaming;
  \item rooted branching hp-bisimulation equivalence $\approx_{rbhp}$ is a congruence with respect to $\textrm{APTC}^{\textrm{sat}}_{\tau}$ + Rec + renaming.
\end{itemize}
\end{theorem}

\begin{proof}
It is easy to see that $\approx_{rbp}$, $\approx_{rbs}$, and $\approx_{rbhp}$ are all equivalent relations on $\textrm{APTC}^{\textrm{sat}}_{\tau}$ + Rec + renaming terms, it is only sufficient to prove that $\approx_{rbp}$, $\approx_{rbs}$, and $\approx_{rbhp}$ are all preserved by the operators $\rho_f$. It is trivial and we omit it.
\end{proof}

\subsubsection{Soundness}

\begin{theorem}[Soundness of $\textrm{APTC}^{\textrm{sat}}_{\tau}$ + Rec + renaming]
The axiomatization of $\textrm{APTC}^{\textrm{sat}}_{\tau}$ + Rec + renaming is sound modulo rooted branching truly concurrent bisimulation equivalences $\approx_{rbp}$, $\approx_{rbs}$, and $\approx_{rbhp}$. That is,
\begin{enumerate}
  \item let $x$ and $y$ be $\textrm{APTC}^{\textrm{sat}}_{\tau}$ + Rec + renaming terms. If $\textrm{APTC}^{\textrm{sat}}_{\tau}\textrm{ + Rec + renaming}\vdash x=y$, then $x\approx_{rbs} y$;
  \item let $x$ and $y$ be $\textrm{APTC}^{\textrm{sat}}_{\tau}$ + Rec + renaming terms. If $\textrm{APTC}^{\textrm{sat}}_{\tau}\textrm{ + Rec + renaming}\vdash x=y$, then $x\approx_{rbp} y$;
  \item let $x$ and $y$ be $\textrm{APTC}^{\textrm{sat}}_{\tau}$ + Rec + renaming terms. If $\textrm{APTC}^{\textrm{sat}}_{\tau}\textrm{ + Rec +renaming}\vdash x=y$, then $x\approx_{rbhp} y$.
\end{enumerate}
\end{theorem}

\begin{proof}
Since $\approx_{rbp}$, $\approx_{rbs}$, and $\approx_{rbhp}$ are both equivalent and congruent relations, we only need to check if each axiom in Table \ref{AxiomsForAPTCSATTauRen} is sound modulo $\approx_{rbp}$, $\approx_{rbs}$, and $\approx_{rbhp}$ respectively.

\begin{enumerate}
  \item Each axiom in Table \ref{AxiomsForAPTCSATTauRen} can be checked that it is sound modulo rooted branching step bisimulation equivalence, by transition rules in Table \ref{TRForAPTCSATTauRen}. We omit them.
  \item From the definition of rooted branching pomset bisimulation $\approx_{rbp}$, we know that rooted branching pomset bisimulation $\approx_{rbp}$ is defined by weak pomset transitions, which are labeled by pomsets with $\tilde{\tau}$. In a weak pomset transition, the events in the pomset are either within causality relations (defined by $\cdot$) or in concurrency (implicitly defined by $\cdot$ and $+$, and explicitly defined by $\between$), of course, they are pairwise consistent (without conflicts). We have already proven the case that all events are pairwise concurrent, so, we only need to prove the case of events in causality. Without loss of generality, we take a pomset of $P=\{\tilde{a},\tilde{b}:\tilde{a}\cdot \tilde{b}\}$. Then the weak pomset transition labeled by the above $P$ is just composed of one single event transition labeled by $\tilde{a}$ succeeded by another single event transition labeled by $\tilde{b}$, that is, $\xRightarrow{P}=\xRightarrow{a}\xRightarrow{b}$.

        Similarly to the proof of soundness modulo rooted branching step bisimulation $\approx_{rbs}$, we can prove that each axiom in Table \ref{AxiomsForAPTCSATTauRen} is sound modulo rooted branching pomset bisimulation $\approx_{rbp}$, we omit them.

  \item From the definition of rooted branching hp-bisimulation $\approx_{rbhp}$, we know that rooted branching hp-bisimulation $\approx_{rbhp}$ is defined on the weakly posetal product $(C_1,f,C_2),f:\hat{C_1}\rightarrow \hat{C_2}\textrm{ isomorphism}$. Two process terms $s$ related to $C_1$ and $t$ related to $C_2$, and $f:\hat{C_1}\rightarrow \hat{C_2}\textrm{ isomorphism}$. Initially, $(C_1,f,C_2)=(\emptyset,\emptyset,\emptyset)$, and $(\emptyset,\emptyset,\emptyset)\in\approx_{rbhp}$. When $s\xrightarrow{a}s'$ ($C_1\xrightarrow{a}C_1'$), there will be $t\xRightarrow{a}t'$ ($C_2\xRightarrow{a}C_2'$), and we define $f'=f[a\mapsto a]$. Then, if $(C_1,f,C_2)\in\approx_{rbhp}$, then $(C_1',f',C_2')\in\approx_{rbhp}$.

        Similarly to the proof of soundness modulo rooted branching pomset bisimulation equivalence, we can prove that each axiom in Table \ref{AxiomsForAPTCSATTauRen} is sound modulo rooted branching hp-bisimulation equivalence, we just need additionally to check the above conditions on rooted branching hp-bisimulation, we omit them.
\end{enumerate}

\end{proof}

\subsubsection{Completeness}

\begin{theorem}[Completeness of $\textrm{APTC}^{\textrm{sat}}_{\tau}$ + CFAR + guarded linear Rec + renaming]
The axiomatization of $\textrm{APTC}^{\textrm{sat}}_{\tau}$ + CFAR + guarded linear Rec + renaming is complete modulo rooted branching truly concurrent bisimulation equivalences $\approx_{rbs}$, $\approx_{rbp}$, and $\approx_{rbhp}$. That is,
\begin{enumerate}
  \item let $p$ and $q$ be closed $\textrm{APTC}^{\textrm{sat}}_{\tau}$ + CFAR + guarded linear Rec + renaming terms, if $p\approx_{rbs} q$ then $p=q$;
  \item let $p$ and $q$ be closed $\textrm{APTC}^{\textrm{sat}}_{\tau}$ + CFAR + guarded linear Rec + renaming terms, if $p\approx_{rbp} q$ then $p=q$;
  \item let $p$ and $q$ be closed $\textrm{APTC}^{\textrm{sat}}_{\tau}$ + CFAR + guarded linear Rec + renaming terms, if $p\approx_{rbhp} q$ then $p=q$.
\end{enumerate}

\end{theorem}

\begin{proof}
Firstly, we know that each process term in $\textrm{APTC}^{\textrm{sat}}_{\tau}$  + CFAR + guarded linear Rec + renaming is equal to a process term $\langle X_1|E\rangle$ with $E$ a linear recursive specification.

It remains to prove the following cases.

\begin{enumerate}
  \item If $\langle X_1|E_1\rangle \approx_{rbs} \langle Y_1|E_2\rangle$ for linear recursive specification $E_1$ and $E_2$, then $\langle X_1|E_1\rangle = \langle Y_1|E_2\rangle$.

        It can be proven similarly to the completeness of $\textrm{APTC}_{\tau}$ + CFAR + linear Rec + renaming, see \cite{ATC}.

  \item If $\langle X_1|E_1\rangle \approx_{rbp} \langle Y_1|E_2\rangle$ for linear recursive specification $E_1$ and $E_2$, then $\langle X_1|E_1\rangle = \langle Y_1|E_2\rangle$.

        It can be proven similarly, just by replacement of $\approx_{rbs}$ by $\approx_{rbp}$, we omit it.
  \item If $\langle X_1|E_1\rangle \approx_{rbhp} \langle Y_1|E_2\rangle$ for linear recursive specification $E_1$ and $E_2$, then $\langle X_1|E_1\rangle = \langle Y_1|E_2\rangle$.

        It can be proven similarly, just by replacement of $\approx_{rbs}$ by $\approx_{rbhp}$, we omit it.
\end{enumerate}
\end{proof}

\section{Conclusions}{\label{con}}

Our previous work on truly concurrent process algebra APTC \cite{ATC} is an axiomatization for true concurrency. There are correspondence between APTC and process algebra ACP \cite{ACP}, in this paper, we extend APTC with timing related properties. Just like ACP with timing \cite{T1} \cite{T2} \cite{T3}, APTC with timing also has four parts: discrete relative timing, discrete absolute timing, continuous relative timing and continuous absolute timing.

APTC with timing is formal theory for a mixture of true concurrency and timing, which can be used to verify the correctness of systems in a true concurrency flavor with timing related properties support.

\newpage

%

\label{lastpage}

\end{document}